\documentclass[12pt,dvips]{scrbook}

\input{SecSci-macros}

\begin{document}

\frontmatter
\pagenumbering{roman}
\thispagestyle{empty}

\begin{titlepage}
\title{{\Huge Security Science (SecSci)}
\\[1ex]
{\Large --- Basic Concepts and Mathematical Foundations ---}
}

\author{\sf Dusko Pavlovic and Peter-Michael Seidel}
\date{
\vspace{2cm}
\begin{center}
\includegraphics[height=10cm
]{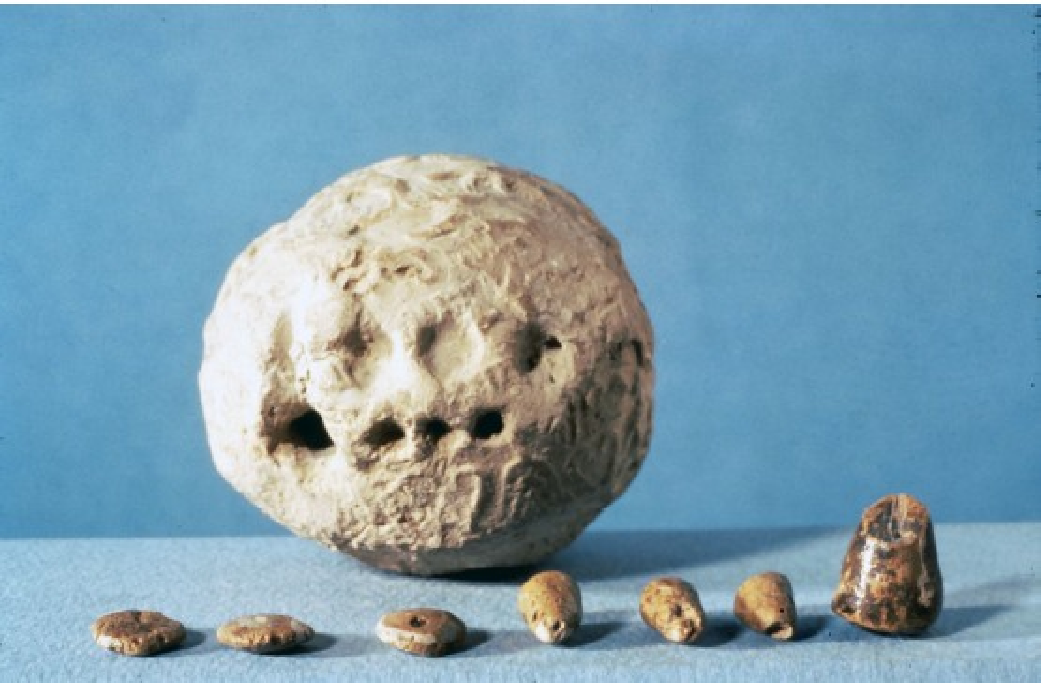}
\\[1ex]
{\large Tamper-resistant tokens and their clay envelope from Uruk, 4500 BC}
\end{center}
}

\end{titlepage}

\maketitle

\pagestyle{plain}
\tableofcontents

\chapter{Preface}\label{Chap:Preface}

\para{Origins.} This book originates from lecture notes started by the first author in Oxford in 2008 and developed by the coauthors together in Hawaii since 2015. The lecture notes evolved in style and content, but the approach to security through science persisted and did not evolve much.

The idea of security science originates from the research network initiated by Brad Martin of the NSA in 2008, allegedly in response to the questions asked by a member of the NSA Science Board: 
\begin{quote}
\em Why is it that communications security is based on a science, and for cybersecurity, we are competing who will hire more hackers? Why is there no general science of security?
\end{quote} 

\para{Background.} At the time, many software and protocol development efforts were plagued by a remarkable phenomenon: carefully designed protocols, thoroughly scrutinized and analyzed for years by top  expert committees, would finally get standardized after everyone agreed that they were secure --- only to turn out to be insecure upon deployment\footnote{The examples include the protocols from the IPSec IKE \cite{JFK,rfc6071,rfc2409,KrawczykH:SKEME,rfc2412}, and the IPSec GDoI \cite{PavlovicD:ESORICS04,MeadowsC:GDOI,rfc6407} protocol suites, as well as the relatives of  the MQV protocol \cite{KaliskiB:MQV,HMQV,MQV}, which kicked off the long series of ``another look'' articles \cite{Koblitz-Menezes:another-perspective}, questioning the utility of the product line of security proofs and experts, and in some cases even its soundness.}. The vulnerabilities varied from case to case and the only commonality was that they were abstracted away from the analyzed models. This seemed hard to avoid since you cannot model everything. So each case was dismissed as a fluke. For an outsider coming from mathematics, though, the regularity with which security experts were finding attacks on systems and protocols that other security experts had proven secure was mystifying \cite{HMQV,MenezesA:HMQV,PaulsonL:bull,Ryan-Schneider:bull}. In most cases, both the security proofs and the insecurity proofs were logically sound, just talking about different things. Time after time.

\para{Revelation.} On this background, the concept of security science came as a revelation. It suggested that \emph{security claims should not be viewed as mathematical theorems but as scientific theories}! 

A mathematical theorem is a statement about a mathematical model, derived from mathematical axioms. It remains valid as long as the axioms are considered valid. This may be forever, or until more interesting axioms are encountered. 

A scientific theory is a statement about reality, derived from the hypotheses that survive experimental testing. It remains valid until new tests or observations disprove it. For example, Newton's theory of gravitation survived for 350 years, until it was observed that the orbit of Mercury disproved it. Einstein's General Relativity improved it. While mathematicians work to \emph{prove}\/ their theorems, scientists work to \emph{disprove and improve}\/ their theories \cite{KuhnT:structure,PopperK:refutations}. The idea of Security Science is that \emph{\textbf{security is like science}}. This is spelled out in Sec.~\ref{Sec:Feynman}. 

\para{Latency.} However, the research network for the science of security did not pursue this idea. The leaders of the community argued that the security models that they studied in the 1980s were science, and the community followed suit. No one should be blamed for this. Community behaviors, of course, arise from community dynamics, not from individual qualities and shortcomings. The dynamics of security research are driven by incentives that do not always align to make the world more secure.

\para{Drift.} History is a quest for security. Every war is fought to secure something. In the year 1990, the US won the Cold War. The nuclear threat disappeared from one day to another. The Internet, developed as the communication infrastructure capable of surviving the fragmentation caused by a nuclear war, was released to the general public. Email brought together social and professional lives and the World Wide Web elevated the greed-is-good spirit of Reaganomics into a global, economy-driven engine of science, culture, and technology. The times of dark libraries with dusty books were behind us. All the information you could ever need was under the tips of your fingers. You could type a question on the keyboard; the world would process it and respond on the screen. The world could display a question on the screen; you would process it and respond on the keyboard. The world became a computer. Security became computer security. Privacy became data ownership. 

Throughout history, security has been the job of guards, generals, and diplomats. Since 1990, security became a career choice of software engineers. The echoes of this narrow view still dominate security education. Security professionals collect professional certificates issued by software giants, who in the meantime rule the world. Wars are fought in cyberspace. Political movements are implemented in social networks. Physical security is based on cybersecurity. Parents rely on cyber devices for their children's security, while predators rely on the same devices for prey. Security is a natural process. Together with everyone else, security researchers and their students are grains of sand in the rising sandstorms of their subject.

\para{Studying security.} This book is an effort to study security itself. Some of the obstacles to thinking about security on its own are that a part of it (national security) doesn't like to be observed, whereas another part (security industry) is hampered by the Principal-Agent Problem\footnote{Look it up if you are not sure what it is! The Agent defends the Principal from risks and becomes Principal's only source of information about the risks. Praetorian Guard was hired to defend Roman emperors, but ended up auctioning the position of Roman emperor for 400 years.}. But we owe it to the students to try. It is unlikely to help them get any particular security certificates, but we hope that it will help them understand the underlying processes of security. Or at least that it will help some of them to write a better book.

\subsubsection*{Thanks and apologies}
%
%
%

\label{Preface}

\mainmatter

\def\thechapter{1}
\setchaptertoc
\chapter{Introduction: On bugs and elephants}\label{Chap:intro}

\begin{flushright}
\parbox{10cm}{{\it \footnotesize A number of blind men 
came to an elephant. Somebody told them that it was an 
elephant. The blind men asked, "What is the elephant 
like?" and they began to touch its body. One of them said: 
"It is like a pillar." This blind man had only touched its leg. 
Another man said, "The elephant is like a husking basket." 
This person had only touched its ears. Similarly, he who 
touched its trunk or its belly talked of it differently.}}\\[3ex] 
{\footnotesize Ramakrishna Paramahamsa} 
\end{flushright}

\section{On security engineering}

Security means many things to many people. For a 
software engineer, it often means that there are no buffer 
overflows or dangling pointers in the code. For a 
cryptographer, it means that any successful attack on the 
cipher can be reduced to an algorithm for computing discrete logarithms, or to integer factorization. For a diplomat, security means that the enemy cannot read the confidential messages. For a credit card operator, it means that the total costs of the fraudulent transactions and of the measures to prevent them are low, 
relative to the revenue. For a bee, security means that no intruder into the beehive  will escape her sting\ldots

Is it an accident that all these different ideas 
go under the same name? What do they really have in common? They are studied in different sciences, ranging from computer science to biology, by a wide variety of different methods. Would it be useful to study them together?

\subsection{What is security engineering?}
If all avatars of security have one thing in common, it is surely the idea that \emph{there are attackers out there}. All security concerns, from computation to politics and biology, come down to averting the adversarial processes in the Environment, that are poised to subvert the goals of the System. There are, for instance, many kinds of bugs in software, but only those that the hackers use are a security concern.  

In all engineering disciplines, the System guarantees a 
functionality, provided that the Environment satisfies some 
assumptions. This is the standard \emph{assume-guarantee} format of the engineering correctness statements. Such statements are useful when the Environment is passive so that the assumptions about it remain valid for a while. \emph{The essence of security engineering is that the Environment actively seeks to invalidate the System's assumptions.} 

Security is thus an \emph{adversarial process}. In all 
engineering disciplines, failures usually arise from engineering errors. In security, failures arise \emph{in spite}\/ of compliance with the best engineering practices of the moment. \emph{\textbf{Failures are the first class citizens of security.}}\/ In most software systems, we normally expect security updates, which usually arise from attacks, and often inspire them. 

\subsection{Where did security engineering come from?}
The earliest examples of security technologies are found 
among the earliest documents of civilization. Fig.~\ref{bulla} shows security tokens with tamper protection 
technology from almost 6000 years ago. Fig.~\ref{MLsec} 
depicts the situation where this technology was 
probably used. Alice has a lamb and Bob has built a secure vault, perhaps with multiple security levels, spacious enough to store both Bob's and Alice's assets. For each of Alice's assets deposited in the vault, Bob issues a 
clay token, with an inscription identifying the asset. Alice's tokens are then encased into a \emph{bulla}, a round, hollow "envelope" of clay, which is then baked to prevent tampering. When she wants to withdraw her deposits, Alice submits her bulla to 
Bob, he breaks it, extracts the tokens, and returns the goods. Alice can also give her bulla to Carol, who can also submit it to Bob, to withdraw the goods, or to pass it on to Dave. Bull\ae\ can thus be traded, and they facilitate an exchange 
economy. The tokens used in the bull\ae\ evolved into the earliest forms of money, and the inscriptions on them led to the earliest numeral systems, as well as to the Sumerian cuneiform script, which was one of the earliest alphabets. Security thus predates literature, science, mathematics, and even money.
\begin{figure}[ht!]
    \begin{minipage}[t!]{0.45\linewidth}
    \centering
    \includegraphics[height=5cm]
    {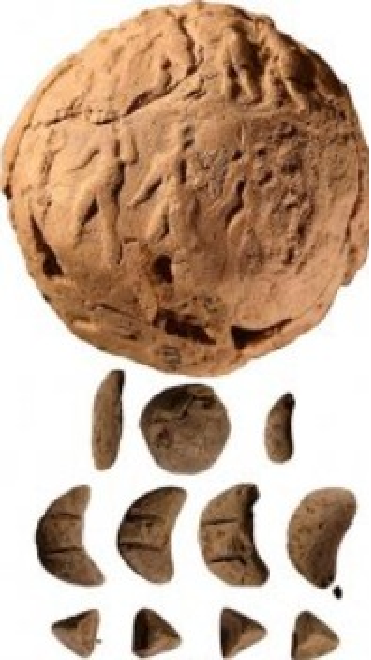}
    \medskip
    \caption{Tamper protection from 3700 BC: Bulla with tokens from Uruk (Iraq)}
    \label{bulla}
    \end{minipage}
    \hspace{0.1\linewidth}
    \begin{minipage}[t!]{0.45\linewidth}
    \centering
    \includegraphics[height=5cm]
    {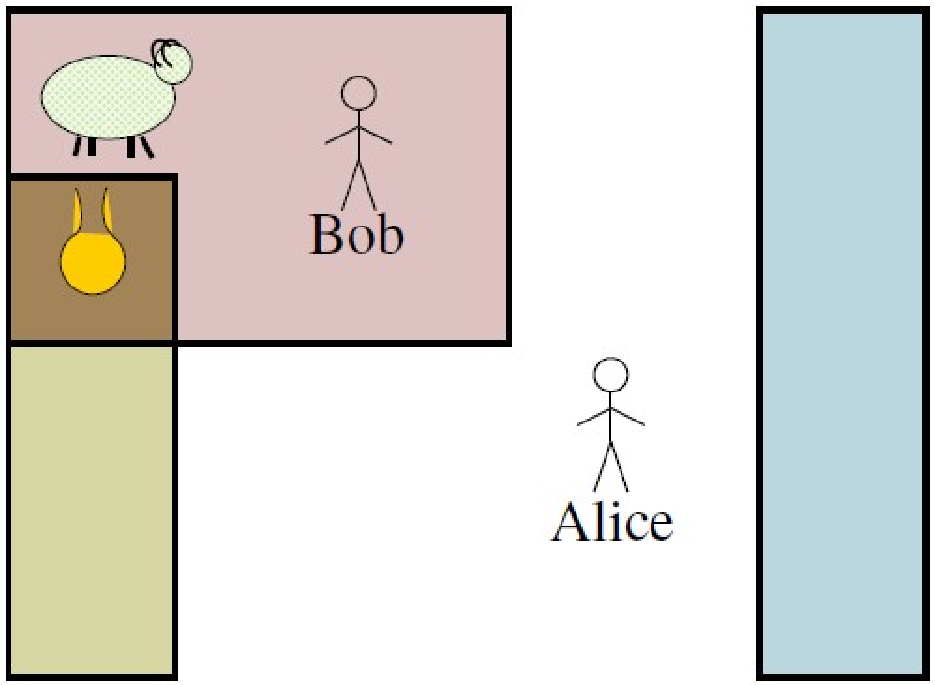}
    \medskip
    \caption{To withdraw her sheep from Bob's secure vault, 
    Alice submits a tamper-proof token, like in Fig.~\ref{bulla}.}
    \label{MLsec}
    \end{minipage}
    \end{figure}

\subsection{Where is security engineering going?}
Throughout history, security technologies evolved gradually, serving the purposes of war and peace, and protecting public resources and private property. As computers pervaded all aspects of social life, security became interlaced with computation, and security engineering came to be closely related to computer science. The developments in the realm of security are nowadays inseparable from the developments in the realm of computation. The most notable such development is the \emph{cyberspace}.

\subsubsection*{A brief history of cyberspace}
In the beginning, engineers built computers and wrote 
programs to control computations. The platform of computation was the computer, and it was used to execute algorithms and calculations, allowing people to discover, e.g., the fractals, and to invent 
compilers, that allowed them to write and execute more 
algorithms and more calculations more efficiently. Then the 
operating system became the platform of computation, and 
software was developed on top of it. The era of personal 
computing and enterprise software broke out. And then the 
Internet happened, followed by cellular networks, and 
wireless networks, and ad hoc networks, and mixed networks. \emph{Cyberspace emerged as the distance-free space of instant, costless communication.} 
Nowadays, software is developed to run in cyberspace. The Web is, strictly speaking, just a software system, albeit a formidable one. A botnet is also a software system. As social space blends with cyberspace, many social (business, collaborative) processes can be usefully construed as software systems, that run on social networks as hardware. Many social and computational processes become inextricable. Table~\ref{ages-comp} summarizes the crude picture of the paradigm shifts that led to this remarkable situation.

\begin{table}[ht]
\begin{center}
{\footnotesize
\begin{tabular}{|c||c|c|c|c|}
\hline
\textit{\textbf{age}} & \textit{ancient times} & \textit{middle ages} & \textit{modern times} \\[1ex]
\hline \hline
\textbf{platform}  & computer & operating system & network \\[1ex]
\hline
 \textbf{applications} & Quicksort, compilers & MS Word, Oracle  & 
WWW, botnets  \\[1ex]
\hline
\textbf{requirements} & correctness, termination & liveness, safety 
& trust, privacy \\[1ex]
 \hline
\textbf{tools} & programming languages & specification languages & scripting 
languages \\[1ex]
 \hline
\end{tabular}
}
\end{center}
\caption{Paradigms of computation}
\label{ages-comp}
\end{table}%

But as every person got connected to a computer, and every computer to a  network, and every network to a network of networks, computation became interlaced with communication and ceased to be programmable.  The functioning of the Web and of web applications is not determined by the code in the same sense as in a traditional software system: after all, web applications do include human users as a part of their runtime. The fusion of social and computational processes gave rise to the new sociotechnical space and led to new kinds of information processing, where the purposeful program executions at the network nodes are supplemented by the spontaneous data-driven evolution of network links. While the network emerges as the new computer, data and metadata become inseparable, and new kinds of security problems arise.

\subsubsection*{A brief history of cybersecurity}
In early computer systems, security tasks were mainly concerned with sharing of the computing resources. In computer networks, security goals expanded to include information protection. Those developments are reflected in the structure of the book, which progresses from resources to information. Both computer security and information security depend on a clear  distinction between the secure areas and the insecure areas, separated by a security perimeter. Security engineering caters to both by providing tools and methods for building security perimeters. In cyberspace, the secure areas are separated from the insecure areas by ``walls'' of cryptography; and they are connected by the ``gates'' of protocols.  

But as networks of computers and devices spread through physical and social spaces, the distinctions between the secure and the insecure areas become blurred. 
\begin{table}[ht]
\begin{center}
{\footnotesize
\begin{tabular}{|c||c|c|c|c|}
\hline
\textit{\textbf{age}} &  \textit{middle ages} & \textit{modern times} & \textit{postmodern times} \\[1ex]
\hline \hline
\textbf{space}  & computer center & cyberspace & sociotechnical space \\[1ex]
\hline
 \textbf{assets} & computing resources  & 
information & attention  \\[1ex]
\hline
\textbf{requirements} & availability, authorization & integrity, confidentiality 
& trust, privacy \\[1ex]
 \hline
\textbf{tools} &  locks, tokens, passwords & cryptography, protocols & search and intelligence \\[1ex]
 \hline
\end{tabular}
}
\end{center}
\caption{Paradigms of security}
\label{ages-sec}
\end{table}%
In such areas of the sociotechnical space, information processing does not yield to programming anymore. It cannot be secured just by cryptography and protocols. Security cannot be assured anymore by the engineering methodologies alone. Data mining, concept analysis, search, and intelligence span and cover new spaces. Like our cities, our sociotechnical spaces were originally built by engineers, but life took over, and some of our problems cannot be engineered away anymore. Enter science.
\label{Sec:Eng}

\section{On security science}

Science inputs processes that exist in nature and outputs their descriptions.  
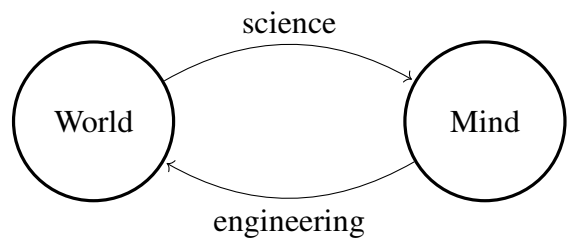
\begin{figure}[!ht]
\begin{center}
\begin{tikzpicture}[shorten >=1pt,node distance=3cm]
  \node[state,minimum size=35pt,draw=black,very thick]   
  (nocarry)                {\begin{minipage}[c]{1.75cm}\begin{center} World \end{center}\end{minipage}};
   \node[state,minimum size=0pt,draw=black,very thick] (carry)    [right=of nocarry]              
   {\begin{minipage}[c]{1.75cm}\begin{center} Mind \end{center}\end{minipage}};
  \path[->] 
  (nocarry) 
  edge [bend left] node [above] {science} 
  (carry)
  (carry) 
  edge [bend left] node [below] {engineering}  
  (nocarry);
   \end{tikzpicture}
   \caption{Civilization as conversation} 
\label{Fig:sci-eng}
\end{center}
\end{figure}
Engineering inputs descriptions of processes that do not exist in nature and outputs their realizations. Fig.~\ref{Fig:sci-eng} shows the conversation between the two which gave rise to our civilization. 

\subsection{Security space}
In the beginning, computation was engineered: we built computers and wrote programs. After we built lots of computers and wrote lots of programs, we connected them, and cyberspace emerged from the Big Bang of the Internet. Cyberspace expanded into many galaxies, and life emerged in it: viruses and chatbots, influencers and celebrities, global media, and genomic databases. Cyberspace is the space where we live. It was built by engineers, just like our cities, but it took a life of its own. \textbf{Cyberspace is the space of \emph{natural processes}.}

Every science is characterized by its space. Natural sciences place the world in physical space, connected by light, full of dark matter and  energy. Security science deals with the world in cyberspace, connected by networks, full of lies and deceit. The physical space is described using vectors and coordinate systems, cyberspace is structured by traces and protocols. 
This book provides an overview. The early chapters are suitable for a first course. Later chapters have been used in advanced courses. Research problems lurk throughout.

\subsection{Range of security science}
The challenge of security is that it defies spatial intuitions: \emph{we are used to thinking locally, whereas security requirements concern non-local interactions}. Strategic thinking about security requires expanding our views beyond our local horizons and taking into account the non-local interactions. Formal models are not just convenient means to increase precision. They are indispensable tools for non-local reasoning, just like airplanes are our indispensable tool for flying. Without such tools, we could not overcome our limitations. 

Using tools that overcome our limitations is, of course, a challenge in itself. It requires science. We need science to build airplanes.  Taming a horse requires a form of science: developing a common language with the horse, interpreting his behavior, hypotheses, experiments, and better hypotheses. Understanding other people, members of our community, and members of distant cultures, requires science. Interacting with thieves, attackers, deceivers, influencers, and all kinds of security experts, including textbook writers, requires a science. You need security science even to secure science. 

Science consists of theories and experiments. Theories and experiments need intuitive interfaces. The geometric view provides an intuitive interface for security science.

\subsection{The Earth is flat, and its perimeter is secure}
The problem with security is that things are not what they seem. A website provides free advice for dog owners, but tracks them and sells their private information. A bank offers a credit line for recent graduates but repossesses one in five of their houses. You look around, and you see that the Earth is obviously flat. The horizons where valleys and oceans meet the sky are straight. If Earth wasn't flat, oceans would flow like rivers. It's obviously flat. To understand that the Earth is round, you need science. To understand security, you also need science. Security Science. We call it SecSci.

Security science is easier than other sciences because it is more intuitive. We have a better sense for lies than for the roundness of the Earth. But security science is also harder than other sciences because its subject does not like to be observed. National security requires secrecy. Personal security requires privacy. Enterprise security requires top dollars to be paid to security experts. Otherwise, they say, other security experts, called hackers, will penetrate your defenses, and you will have to pay them more. You have to pay the good guys to defend you from the bad guys. But if the defenders are rational, they have to maximize their revenue, and they charge you just a penny less than what the attackers would steal from you. And the estimates of how much the attackers would steal are provided by the defenders, as a free service. So for all you know, they may be charging you a penny more than the attackers would steal.

\subsection{Science and deceit} 
Science begins from the idea that testing reality and eliminating false theories is a good thing. The underlying assumption is that reality strikes against false beliefs sooner or later. Employing science, our species took control of large parts of reality.  However, the assumption that reality strikes against false beliefs is just another theory. In reality, disproving false beliefs takes time and resources. If time is short and the resources are scarce, deception is a rational strategy. 

While science is the effort to find and disprove unintentionally false hypotheses, security is the effort to find and disprove intentional deceptions. Scientific hypotheses are the simplest explanations of past observations, chosen so that they are easy to test and disprove. Deceptions are the explanations beneficial for the deceiver, chosen so that they are hard to test and disprove. The rational strategy is to apply one to the other. Science and deceit are the only ones who stand a chance against each other.

\subsection{The best-kept secret}\label{Sec:Feynman}
\para{Feynman on science.}  {\it "If we have a definite theory, from which we can compute the consequences which can be compared with experiment, then in principle we can prove that theory wrong. } {\it But notice that we can never prove it right!} {\it Suppose that you invent a theory, calculate the consequences, and discover every time that the consequences agree with the experiment.  The theory is then right? No, it is simply not proven wrong. 
In the future you could compute a wider range of consequences, there could be a wider range of experiments, and you might then discover that the thing is wrong. [\ldots]} \begin{wrapfigure}
{R}
{5cm}
\centering \includegraphics[width=4cm
]
{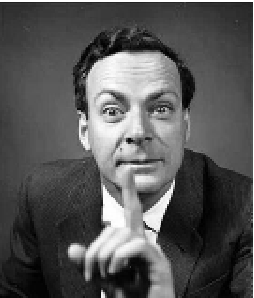}

\end{wrapfigure} 
\textit{\textbf{We never are definitely right; we can only be sure when we are wrong.}" } \cite{FeynmanR:character}

The best-kept secret of science is thus that it does not provide persistent laws, or definite assertions of truth. Science only provides methods to disprove and improve theories.

The best-kept secret of security is that all security claims have a lifetime. In cryptography, the fact that every key has a lifetime is well known. But the empiric fact that security protocols regularly fail is not well understood and is often received with astonishment. How can a verified security protocol still fail? Just like every scientific theory is stated with respect to a given set of observables, which may need to be expanded in the future, every security claim is stated with respect to a given system and attack model, which may need to be refined.

Feynman never said much about security, but if he did, he could have said something like this:

\para{"Feynman on security".} {\it If we have a precisely defined security claim about a system, from which we can derive the consequences which can be tested, then in principle we can prove that the system is insecure. 
But notice that we can never prove that it is secure! 
Suppose that you design a system, calculate some security claims, and discover every time that the system remains secure under all tests. The system is then secure? No, it is simply not proven insecure. In the future you could refine the security model, there could be a wider range of tests and attacks, and you might then discover that the thing is insecure. \textit{\textbf{We never are definitely secure; 
we can only be sure when we are insecure.}} }

\begin{figure}[htp]
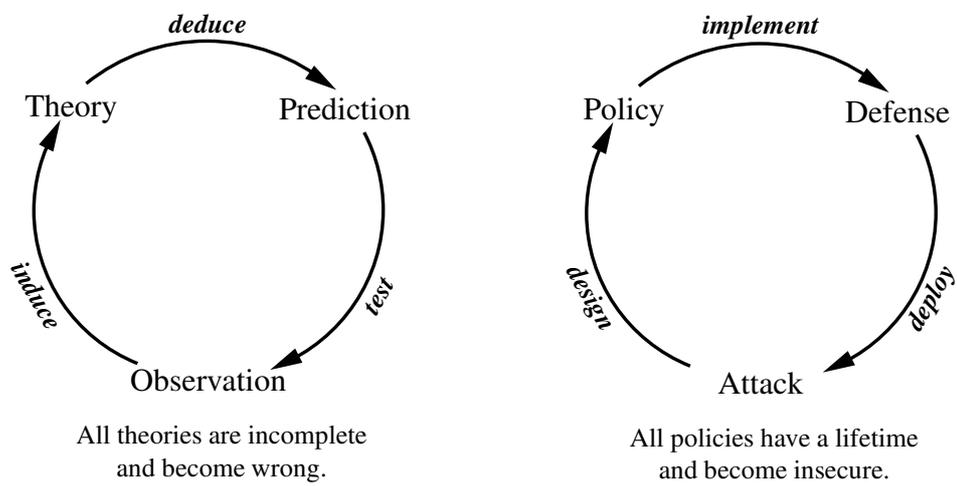


\caption{The process of science and the process of security}

\vspace{\baselineskip}

\begin{minipage}{.45\linewidth}
\newcommand{\Hypothesis}{Theory}
\newcommand{\Prediction}{Prediction}
\newcommand{\Example}{Observation}
\newcommand{\deduction}{\footnotesize \textit{\textbf{deduce}}}
\newcommand{\testing}{\footnotesize \textit{\textbf{test}}}
\newcommand{\induction}{\footnotesize \textit{\textbf{induce}}}
\newcommand{\Security}{}
\newcommand{\squig}{}
\begin{center}
\def\JPicScale{.45}
\input{evolution-tri}

\bigskip
\footnotesize All theories are incomplete\\ and become wrong.
\end{center}
\end{minipage}
\begin{minipage}{.45\linewidth}
\newcommand{\Hypothesis}{Policy}
\newcommand{\Prediction}{Defense}
\newcommand{\Example}{Attack}
\newcommand{\deduction}{\footnotesize \textit{\textbf{implement}}}
\newcommand{\testing}{\footnotesize \textit{\textbf{deploy}}}
\newcommand{\induction}{\footnotesize \textit{\textbf{design}}}
\newcommand{\Security}{}
\newcommand{\squig}{}
\begin{center}
\def\JPicScale{.45}
\input{evolution-tri}

\bigskip
\footnotesize All policies have a lifetime\\
{and become insecure}.
\end{center}
\end{minipage}
\label{Fig:science}
\end{figure}

\def\thechapter{2}
\setchaptertoc
\chapter{Security concepts}\label{Chap:concepts} 

We separate concepts to be able to reason. When concepts depend on each other, separating them seems wrong. How can I make sense of an egg without a chicken? But I also cannot make sense if I mix them up. Many security concepts depend on each other, don't make sense without each other, but get mixed up. Let us try to make sense.

\section{The dark side}
\label{Sec:why}

We all need security. We are anxious when we are insecure. We seek a secure home, we try to find a secure job. We crave security so much that we even study it. But what is it? 

\begin{figure}[htbp]
\begin{center}
\includegraphics[height=6cm]{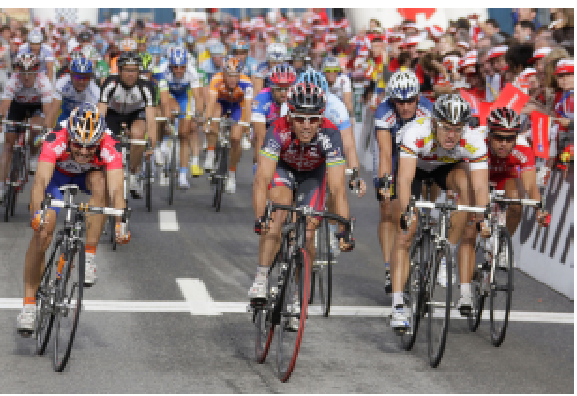}
\caption{A race can be won by being the best and the fastest\ldots}
\label{Fig:race}
\end{center}
\end{figure}

Yes, it means many things to many people: one thing to the software engineer, another to the soldier, and yet another to the diplomat and the little bee. We kicked off with that in Sec.~\ref{Sec:Eng}. But what is the common denominator of all the different notions of security?

In \emph{\textbf{Star Wars}}\/ terminology, security is the battle against the \textbf{dark side of the Force}. In words of Vilfredo Pareto\footnote{Pareto was one of the fathers of political economy. His most famous quote is: \emph{"The rich get richer"}.} \cite{Pareto},
\begin{quote}
\em "The efforts of humans are organized in two directions:
\begin{itemize}
\item to the production of goods, or else
\item to the appropriation of goods produced by others."
\end{itemize}
\end{quote}
Security is concerned with the second direction: those who produce goods want to \emph{secure}\/ them from being appropriated by others. When that fails, then those who have appropriated the goods want to \emph{secure}\/ them from being taken back; and so on. Security is the realm of such ownership conflicts. But note that Pareto, when he speaks of "the efforts of humans", restricts a broader process, as the ownership conflicts also rage among animals, and many of the security solutions that we will study emerged as evolutionary strategies. Even the simplest forms of life already compete to \emph{secure}\/ their resources from others in one way or another. Wherever there is competition, the is security. In terms of racing, Pareto's observation is illustrated in \cref{Fig:race} and \cref{Fig:trip}.
 \begin{figure}[htbp]
\begin{center}
\includegraphics[height=6cm]{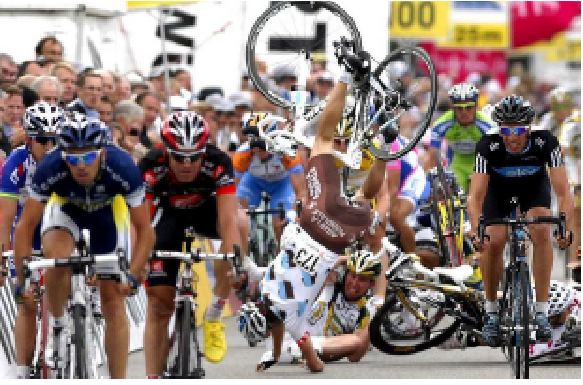}
\caption{\ldots or by tripping up the opponent.}
\label{Fig:trip}
\end{center}
\end{figure}

The war against the Dark Side of the Force, of course, persists in the world of computers and networks. The cyberspace only makes it harder to tell who is who, and what is what. In cyberspace, every force has a dark side. Is Facebook my friend or my enemy? How about YouTube? The memes all together seem to be taking me somewhere, one leading to the other; but where are they taking me? Am I safe there? Am I secure? What is the difference? We need \emph{\textbf{Security Science (SecSci)}}\/ because the eyes may deceive, and the minds are actively deceived.
 
 \begin{figure}[htbp]
\begin{center}
\includegraphics[height=6cm]{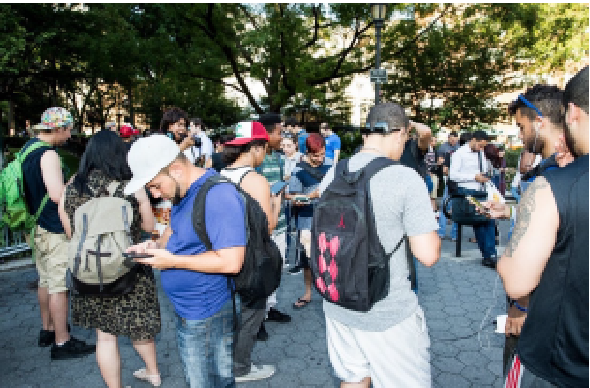}
\caption{If we are all plugged into a Matrix, then it is hard to tell what is secure.}
\label{Fig:matrix}
\end{center}
\end{figure}

Spam is obviously spam, an attack is obviously an attack, but not always. When should I start defending myself?

\section{The timing}
\label{Sec:when}

Security can be implemented 
\begin{itemize}
\item {\bf before attacks:} to \emph{\textbf{prevent}}\/ them by cryptography, protocols, firewalls; \sindex{prevention}

\item  {\bf during attacks:} \sindex{detection}\sindex{forensics}\sindex{intrusion detection} to \emph{\textbf{detect}}\/ them by intrusion detection systems, or to detect earlier phases by forensic methods; 

\item {\bf after attacks:}\sindex{deterrence} to \emph{\textbf{deter}}\/ them by punishing (by legal measures) or attacking the attackers (by illegal measures), and by balancing the incentives: the likely cost of an attack must be higher than the likely profit.  
\end{itemize}

\section{The good and the bad}
\label{Sec:good-bad}

In life and regarding software, our needs and requirements are always expressed in the same way: \sindex{good stuff}\sindex{bad stuff}
\begin{itemize}
\item we want that good stuff happens, and
\item we want that bad stuff doesn't happen.
\end{itemize}
For example, the good stuff that we want from a computational process is that it eventually terminates and gives the correct output. The bad stuff that should not happen is that the process should not crash or get hijacked. 
\sindex{liveness}\sindex{safety}

\sindex{properties!dependability}

\begin{figure}[!ht]
\begin{center}
\caption{Types of requirement specifications}
\renewcommand{\top}{Requirements}
\newcommand{\oneleft}{\begin{minipage}{2cm}
\center\small 
Good stuff\\
should\\ happen
\end{minipage}}
\newcommand{\oneright}{\begin{minipage}[c]{2cm} \center
\small Bad stuff\\
 should not happen
\end{minipage}}
\newcommand{\twoleft}{liveness}
\newcommand{\tworight}{\bf security}
\newcommand{\twomiddle}{safety}
\newcommand{\captionone}{\begin{minipage}[t]{2cm} \footnotesize functionality
\end{minipage}}
\newcommand{\captiontwo}{\begin{minipage}[t]{2cm} \footnotesize \flushright no accidents\\
{\color{red}(natural)}
\end{minipage}}
\newcommand{\captionthree}{\begin{minipage}[t]{3cm} \footnotesize no attacks\\
{\color{red}(adversarial)}
\end{minipage}}

\vspace{1.5\baselineskip}

\def\JPicScale{.5}
\input{requirements}

\vspace{\baselineskip}

\begin{minipage}{.25\linewidth}
\begin{center}
%
\vspace{.5\baselineskip}
\includegraphics[height=4cm
]{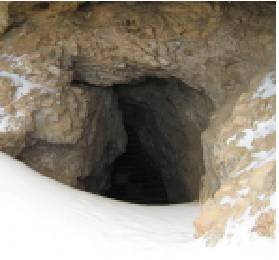}

A \emph{functional} dwelling,
\end{center}
\end{minipage}
\ \ \ \ \ \ \ \ \ \ \ \ \  
\begin{minipage}{.25\linewidth}
\begin{center}
%
\vspace{.5\baselineskip}
\includegraphics[height=4cm
]{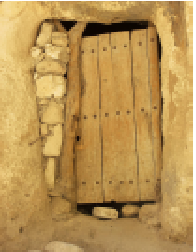}

a shelter from \emph{hazards},
\end{center}
\end{minipage}
\ \ \ \ \ \ \ \ \ \ \ \ \ 
\begin{minipage}{.25\linewidth}
\begin{center}
%
\vspace{.5\baselineskip}
\includegraphics[height=4cm
]{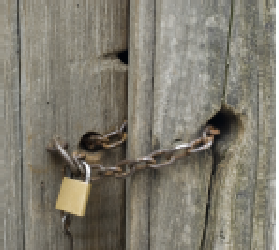}

a lock against~\emph{intruders}.
\end{center}
\end{minipage}

\vspace{.5\baselineskip}
\label{Fig:live-safe-sec}
\end{center}
\end{figure}
So there are two kinds of bad stuff that we want to avoid: accidental, and intentional. A requirement that good stuff happens is called a \emph{liveness}\/ requirement. \sindex{liveness} A requirement that the accidental bad stuff should not happen is called a \emph{safety}\/ requirement. \sindex{safety} A requirement that the intentional bad stuff perpetrated by an attacker should not happen is called a \emph{security}\/ requirement. \sindex{security} 
Note that
\begin{itemize}
\item the \textbf{liveness}\/ requirements specify some desired \emph{functionality}\/ that should be achieved; 
\item the \textbf{safety}\/ requirements specify some undesirable and unintentional \emph{hazards}\/ that should be prevented;
\item the \textbf{security}\/ requirements specify some undesirable intentional \emph{attacks}\/ that should be defeated.
\end{itemize}
This subdivision of the requirement specifications is displayed in \cref{Fig:live-safe-sec}. 
\begin{figure}[!ht]
\begin{center}
\includegraphics[width=9cm
]{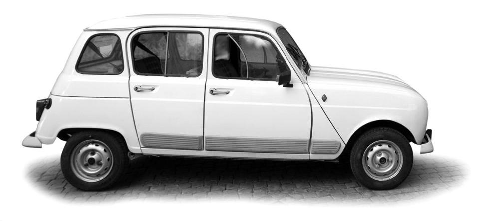}
\end{center}
\caption{A modest car provides modest functions, average safety, and affordable security}
\label{Fig:renault}
\end{figure}
Further down the road, cars are also built to satisfy the same three kinds of requirements: 
\begin{itemize}
\item the \emph{engine}\/ provides the driving functionality, \hfill $\leftsquigarrow$ \textbf{liveness}
\item the \emph{brakes}\/ prevent accidents and the loss of control, \hfill $\leftsquigarrow$ \textbf{safety}
\item the \emph{locks and alarm}\/ prevent theft and intrusions. \hfill $\leftsquigarrow$ \textbf{security}
\end{itemize}

\para{Summary.} The conceptual distinctions made so far are that
\begin{itemize}\sindex{safety}\sindex{bad stuff}\sindex{good stuff}
\item safety is a \emph{negative}\/ requirement that bad stuff does not happen --- in contrast with liveness, which is a \emph{positive}\/ requirement that good stuff happens; 
\item security protects from \emph{intentional}\/ bad stuff --- in contrast with safety, which protects from the \emph{unintentional}\/ bad stuff.
\end{itemize}

\begin{table}[!ht]
{\small
\begin{center}
\begin{tabular}{|c||c|c|c|}
\hline
& \textit{positive requirement} & \multicolumn{2}{c|}{\textit{negative requirements}} \\
\cline{2-4}
\textit{domain}& \textbf{liveness} & \begin{minipage}[c][.75cm][c]{3cm}\textbf{\hspace{2.5em} safety}\end{minipage} & \textbf{security}\\
\hline \hline
mountain & reach the peak   & do not slip on ice & \begin{minipage}[c][1cm][c]{3cm}do not get pushed
\end{minipage} \\
\hline
kitchen & \begin{minipage}[c][1cm][c]{3.5cm}\ \ \ \ \ \  prepare food \end{minipage} & \begin{minipage}[c][1cm][c]{3.5cm}\ do not bite your tongue \end{minipage}  & \begin{minipage}[c][1cm][c]{3.5cm}\ \ do not get poisoned
\end{minipage} \\
\hline
airport & board passengers & mark slippery floor & \begin{minipage}[c][1cm][c]{3cm}prevent terrorism
\end{minipage} \\
\hline
cryptography & $D(k, E(k,m)) = m$ & no bugs & 
\begin{minipage}[c][1cm][c]{3cm}$A(E(k,m)) = m \Rightarrow$\\[-1ex]
 $A(y) = D(k,y)$
\end{minipage}\\
\hline
\end{tabular}
\caption{Positive and negative requirements in various application domains}
\label{Tab:posneg}
\end{center}}
\end{table}

For more instances of these high-level distinctions, see Table~\ref{Tab:posneg}. The same distinctions arise in all areas of engineering, since every system specification includes a liveness requirement, most of them include some safety requirements, and all those that involve network interactions also include some security requirements.

\section{Know, Have, Be}\label{Sec:know-have-be}


All security requirement specifications begin by distinguishing three types of entities:
\begin{itemize}
\item \textbf{data}: what you know, \sindex{data}
\item \textbf{things}:  what you have, \sindex{thing}
\item \textbf{traits}:  what you are. \sindex{trait}
\end{itemize}
This \emph{know-have-be tripod}, \sindex{know-have-be tripod} displayed in \cref{Fig:know-have-be}, provides  a 'metaphysical foundation' for security analyses and designs.  The three types are distinguished by two actions:
\begin{itemize}
\item \textbf{what you know} can be copied and given away; 
\item \textbf{what you have} cannot be copied, but it can be given away; and
\item \textbf{what you are} cannot be either copied, or given away.
\end{itemize}
For example, in a data network, Alice can give copies of her digital keys to Bob, and then both Alice and Bob will \emph{know}\/ Alice's keys.  Digital keys are \textbf{\emph{data}}\/ because they can be copied and given away. In a physical network, on the other hand, it is not as easy for Alice to copy her physical key, or her tamper-resistant smart card; yet she can still give them to Bob, and then Bob will have them, but Alice will not. The physical keys are therefore \textbf{\emph{things}}, which means that they cannot be copied but can be given away. Finally, since Alice cannot easily copy or give away her \textbf{\emph{traits}}, such as genes, iris patterns, fingerprints, or handwriting, her identity is often identified with such indivisible individual \emph{biometric}{\/ features;\sindex{biometric feature} they are who she is. This security typing is summarized in Table~\ref{Table:tripod} and \cref{Fig:tripod}.

\begin{table}[!ht]
\begin{center}
\begin{tabular}{|r|c||c|c|}
\hline
type & what &can be copied & can be given away \\
\hline \hline 
\textbf{data} & what you know & $\checkmark$ & $\checkmark$\\
\hline
\textbf{things} &what you have & $\mathbf{\times}$ & $\checkmark$ \\
\hline 
\textbf{traits} & what you are &$\mathbf{\times}$ & $\mathbf{\times}$  \\
\hline
\end{tabular}
\caption{Distinguishing the security types}
\end{center}
\label{Table:tripod}
\end{table}%

\para{Comment.} There are, of course, methods to \emph{clone}\/ a smart card so that Alice still has it after she gives it to Bob; and there are methods to \emph{forge}\/ handwriting and fingerprints, which may make Alice and Bob indistinguishable even biometrically. But these methods provide attack avenues on particular \emph{implementations}\/ of security based on what you are, or what you have. Here we are not talking about the implementations, but about the basic ideas of security. The particular attack avenues can always be eliminated by improved implementations. The basic ideas remain the same. 

\sindex{properties!security}
\sindex{authenticity}
\sindex{integrity}
\sindex{availability}\sindex{authority}\sindex{authorization}\sindex{freedom}\sindex{health}

\begin{figure}[!ht]
\begin{center}

\renewcommand{\top}{Security}
\newcommand{\oneleft}{\begin{minipage}{2cm}
\center\small 
Resource\\
security
\end{minipage}}
\newcommand{\oneup}{\begin{minipage}{2cm}
\center\small 
Physical\\
security
\end{minipage}}
\newcommand{\oneright}{\begin{minipage}[c]{2cm} \center
\small Data\\
 security
\end{minipage}}
\newcommand{\captionone}{\begin{minipage}[t]{3cm} \footnotesize \flushleft what you are
\end{minipage}}
\newcommand{\captiontwo}{\begin{minipage}[t]{3cm} \footnotesize \flushright what you have
\end{minipage}}
\newcommand{\captionthree}{\begin{minipage}[t]{3cm} \footnotesize what you know
\end{minipage}}
\newcommand{\goode}{\begin{minipage}[t]{2cm} \small \centering good\\stuff
\end{minipage}}
\newcommand{\baade}{\begin{minipage}[t]{2cm} \small \centering bad\\ stuff
\end{minipage}}
\newcommand{\availability}{\begin{minipage}[t]{2cm} \footnotesize \centering \bf availability
\end{minipage}}
\newcommand{\authority}{\begin{minipage}[t]{2cm} \footnotesize \centering \bf authority
\end{minipage}}
\newcommand{\authenticity}{\begin{minipage}[t]{2cm} \scriptsize\centering \bf authenticity\\
integrity
\end{minipage}}
\newcommand{\secrecy}{\begin{minipage}[t]{2cm} \scriptsize \centering \bf secrecy\\
confidentiality
\end{minipage}}

\begin{center}
\def\JPicScale{.3}
\input{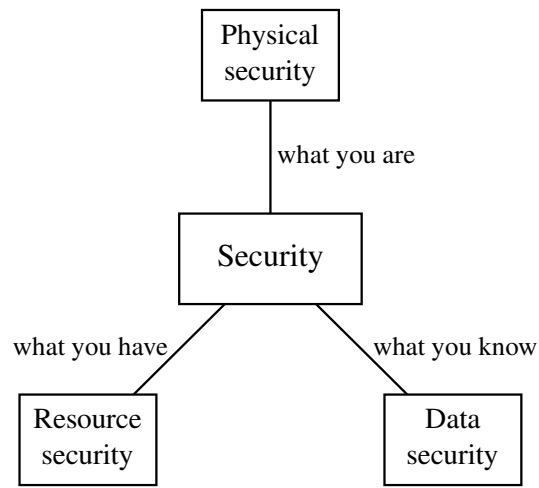}
\end{center}
\caption{The security tripod}
\label{Fig:tripod}
\end{center}
\end{figure}

\para{Using the tripod.} Of the three legs of the security tripod  displayed in \cref{Fig:know-have-be}, we leave aside the top one, the Physical Security, and focus on the remaining two: Resource Security and Data Security. So we ignore \emph{what you are}, and explore the methods to secure and to use \emph{what you have} and \emph{what you know}. In each case, the security requirements can again be subdivided into a \emph{``good-stuff-should-happen''}\/ part and a \emph{``bad-stuff-should-not-happen''}\/ part. This is displayed in \cref{Figure:taxo}. Security requirements specify the good stuff that should eventually happen (``yes good''). Security constraints specify the bad stuff that should never happen (``no bad'').  
\sindex{physical security} \sindex{data security} \sindex{resource security}
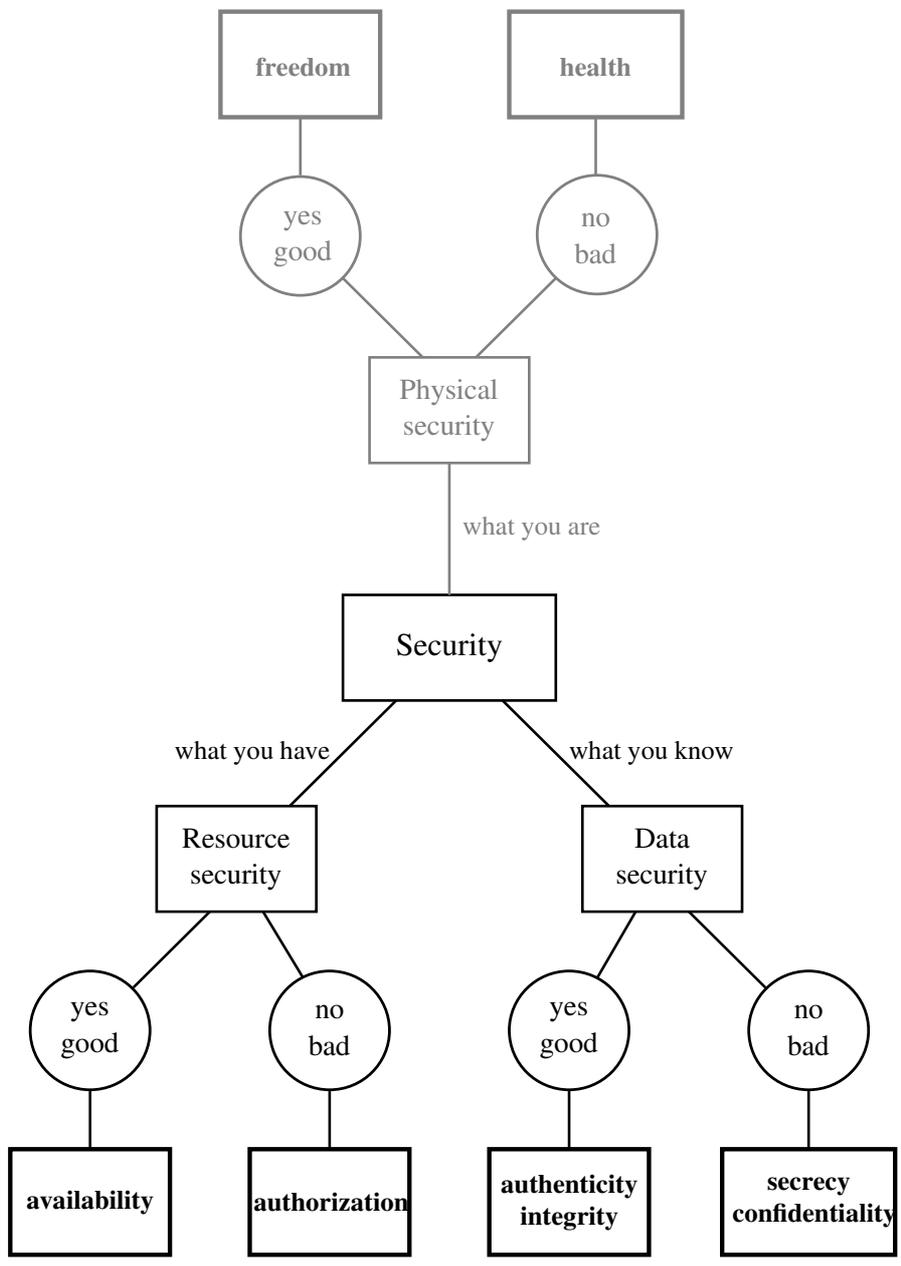
\begin{figure}[htbp]
\begin{center}
\renewcommand{\top}{Security}
\newcommand{\oneleft}{\begin{minipage}{2cm}
\center\small 
Resource\\
security
\end{minipage}}
\newcommand{\oneright}{\begin{minipage}[c]{2cm} \center
\small Data\\
 security
\end{minipage}}
\newcommand{\captiontwo}{\begin{minipage}[t]{3cm} \footnotesize \flushright what you have
\end{minipage}}
\newcommand{\captionthree}{\begin{minipage}[t]{3cm} \footnotesize what you know
\end{minipage}}
\newcommand{\goode}{\begin{minipage}[t]{2cm} \small \centering yes\\good
\end{minipage}}
\newcommand{\baade}{\begin{minipage}[t]{2cm} \small \centering no\\ bad
\end{minipage}}
\newcommand{\ggoode}{\begin{minipage}[t]{2cm} \small \centering \color{gray} yes \\ good
\end{minipage}}
\newcommand{\bbaade}{\begin{minipage}[t]{2cm} \small \centering \color{gray} no\\bad
\end{minipage}}
\newcommand{\availability}{\begin{minipage}[t]{2cm} \footnotesize \centering \bf availability
\end{minipage}}
\newcommand{\authority}{\begin{minipage}[t]{2cm} \footnotesize \centering \bf authorization
\end{minipage}}
\newcommand{\authenticity}{\begin{minipage}[t]{2cm} \footnotesize\centering \bf authenticity\\
integrity
\end{minipage}}
\newcommand{\secrecy}{\begin{minipage}[t]{2cm} \footnotesize 
\centering \bf secrecy\\
confidentiality
\end{minipage}}
\newcommand{\oneup}{\begin{minipage}{2cm}
\center\small 
\color{gray} Physical \\
security
\end{minipage}}
\newcommand{\captionone}{\begin{minipage}[t]{3cm} \footnotesize \flushleft \color{gray} what you are
\end{minipage}}
\newcommand{\health}{\begin{minipage}[t]{2cm} \footnotesize \centering \bf \color{gray} freedom
\end{minipage}}
\newcommand{\medicine}{\begin{minipage}[t]{2cm} \footnotesize \centering \bf \color{gray} health
\end{minipage}}

\bigskip
\def\JPicScale{.35}
\ifx\JPicScale\undefined\def\JPicScale{1}\fi
\psset{unit=\JPicScale mm}
\psset{linewidth=0.3,dotsep=1,hatchwidth=0.3,hatchsep=1.5,shadowsize=1,dimen=middle}
\psset{dotsize=0.7 2.5,dotscale=1 1,fillcolor=black}
\psset{arrowsize=1 2,arrowlength=1,arrowinset=0.25,tbarsize=0.7 5,bracketlength=0.15,rbracketlength=0.15}
\begin{pspicture}(0,0)(327.5,471)
\pspolygon[linewidth=1](50,170)(110,170)(110,130)(50,130)
\pspolygon[linewidth=1](210,170)(270,170)(270,130)(210,130)
\pspolygon[linewidth=1](120,250)(200,250)(200,210)(120,210)
\psline[linewidth=1](180,210)(220,170)
\psline[linewidth=1](140,210)(100,170)
\rput(160,230){\top}
\rput(80,150){\oneleft}
\rput(240,150){\oneright}
\rput[r](115,190){\captiontwo}
\rput[l](205,190){\captionthree}
\psline[linewidth=1](230,130)(215.62,105)
\psline[linewidth=1](70,130)(41.25,101.25)
\psline[linewidth=1](90,130)(105,105)
\psline[linewidth=1](250,130)(278.75,101.25)
\rput{0}(25,85){\psellipse[linewidth=1.05](0,0)(22.5,-22.5)}
\rput{0}(115,85){\psellipse[linewidth=1.05](0,0)(22.5,-22.5)}
\rput{0}(205,85){\psellipse[linewidth=1.05](0,0)(22.5,-22.5)}
\rput{0}(295,85){\psellipse[linewidth=1.05](0,0)(22.5,-22.5)}
\pspolygon[linewidth=1.65](-5,40)(55,40)(55,0)(-5,0)
\psline[linewidth=1](25,62)(25,40)
\pspolygon[linewidth=1.65](85,40)(145,40)(145,0)(85,0)
\psline[linewidth=1](115,62)(115,40)
\pspolygon[linewidth=1.65](262.5,40)(327.5,40)(327.5,0)(262.5,0)
\psline[linewidth=1](295,62)(295,40)
\pspolygon[linewidth=1.65](175,40)(235,40)(235,0)(175,0)
\psline[linewidth=1](205,62)(205,40)
\rput(25,85){\goode}
\rput(205,85){\goode}
\rput(115,85){\baade}
\rput(295,85){\baade}
\rput(295,20){\secrecy}
\rput(205,20){\authenticity}
\rput(115,20){\authority}
\rput(25,20){\availability}
\pspolygon[linewidth=1,linecolor=gray](130,340)(190,340)(190,300)(130,300)
\psline[linewidth=1,linecolor=gray](160,300)(160,250)
\psline[linewidth=1,linecolor=gray](200,370)(170,340)
\psline[linewidth=1,linecolor=gray](120,370)(150,340)
\rput{0}(104,386){\psellipse[linewidth=1.05,linecolor=gray](0,0)(22.5,-22.5)}
\rput{0}(215.5,386.5){\psellipse[linewidth=1.05,linecolor=gray](0,0)(22.5,-22.5)}
\pspolygon[linewidth=1.65,linecolor=gray](182.5,471)(247.5,471)(247.5,431)(182.5,431)
\psline[linewidth=1,linecolor=gray](215,431)(215,409)
\pspolygon[linewidth=1.65,linecolor=gray](74,471)(134,471)(134,431)(74,431)
\psline[linewidth=1,linecolor=gray](104,431)(104,409)
\rput(160,320){\oneup}
\rput[l](165,275){\captionone}
\rput(105,450){\health}
\rput(215,450){\medicine}
\rput(105,385){\ggoode}
\rput(215,385){\bbaade}
\end{pspicture}
\caption{Security requirements are ``yes good''. Security constraints are ``no bad''}
\label{Figure:taxo}
\end{center}
\end{figure}

\para{The coming chapters \emph{roughly}\/ correspond to the  properties at the bottom of \cref{Figure:taxo}.} Resource security specifications will be studied  in Ch.~\ref{Chap:Resource}, resource security requirements expressed in terms of \emph{availability}\/ and \emph{authorization} (a.k.a. \emph{authority}), in  Ch.\ref{Chap:Process}. Channel security is discussed in Ch.~\ref{Chap:Channel}, and it brings us to network security in \ref{Chap:Auth}. The crucial security properties required from channel and network flows are \emph{integrity}\/ and \emph{secrecy}. The crucial security requirements concerning channel sources are  \emph{authenticity}\/ and \emph{confidentiality}. \sindex{availability}
\sindex{authorization}
\sindex{authority}
\sindex{confidentiality}
\sindex{integrity}
\sindex{secrecy}
\sindex{authenticity}
We say that the chapters \emph{roughly}\/ correspond to the ``know-have-be'' types because data can also be resources, not just things; and things can also flow through the channels, not just data. But at least the rough correspondence follows the typical examples, and it seems useful early on. 

\sindex{good stuff}\sindex{bad stuff}

\para{Remark about CIA.}\sindex{CIA-triad}
Most security courses and textbooks begin with the \emph{CIA}-triad of security properties, which refers to \emph{C}\/onfidentiality, \emph{I}\/ntegrity}, and \emph{A}\/vailability. The question: \emph{``Why these properties, and not some other?''} often arises in conversations. The usual answer is that their importance has been established through years of experience and expertise. In special situations, other properties, such as \emph{non-repudiation}, are also considered, with similar justifications. Fig.~\ref{Figure:taxo} suggests that the CIA-canon could perhaps be naturally reconstructed along the \emph{good stuff / bad stuff}\/ axis, by saying that

\medskip
\makebox[\textwidth][l]{$\left.\begin{minipage}{.8\textwidth}\begin{itemize}
\item \textbf{Confidentiality} means that \emph{bad data flows} do not happen;
\item \textbf{Integrity} means \emph{good data flows} do happen;
\end{itemize}
\end{minipage}\right\}\text{\emph{\textbf{what you know}}}$}

\medskip
\makebox[\textwidth][l]{$\left.\begin{minipage}{.8\textwidth}\begin{itemize}
\item \textbf{Availability} means that \emph{good resource calls} are accepted.
\end{itemize}
\end{minipage}\right\}\text{\emph{\textbf{what you have}}}$}

However, the requirement that \emph{bad resource calls} are rejected is missing. It is called  \emph{authorization}\/ in \cref{Figure:taxo}. Maybe the ``A''  in CIA should be double counted, although we never encountered such a suggestion. \sindex{non-repudiation} On the other hand, the problem that non-repudiation is not a part of CIA is often discussed, and there was a proposal that the triad should be extended to CIAN. It is easy to see that non-repudiation is in fact a  \emph{``good-information-flows-do-happen''} requirement. It was interpreted as \emph{authentication-after-the-fact}\/ by some researchers. We wrote authenticity and integrity in the same box in \cref{Figure:taxo}, and they are often used interchangeably. The whole CIA thing should probably be taken with a grain of salt. 
\sindex{non-repudiation}
\sindex{authenticity}

\section{What did we learn?}\label{Sec:Intro:Summary}

\subsection{Process properties}
\subsubsection{Dependability}
\sindex{dependability}
In \textbf{software engineering}, the requirements are expressed in terms of \emph{dependability properties} that include
\begin{itemize}
\item {\bf safety:} \emph{bad stuff (actions) does not happen}, and \sindex{safety}
\item {\bf liveness:} \emph{good stuff (actions) does happen}, \sindex{liveness}
\end{itemize}
and their logical combinations. For sequential computations, every first order property can be expressed as a conjunction of safety and liveness properties. Safety and liveness are semantically independent properties: whether one property is satisfied or not does not depend on the other property.

\subsubsection{Security}
In \textbf{security engineering}, the requirements are expressed as \emph{security properties}. 

The \textbf{resource security requirements} are expressed in terms of the \emph{access control properties}, which are the combinations of
\begin{itemize}
\item {\bf authority:}\footnote{The process of verifying authority is called \textbf{authorization}.} \emph{bad resource requests are not granted}, and
\item {\bf availability:} \emph{good resource requests are granted}. \sindex{authority} \sindex{authorization} \sindex{availability}
\end{itemize}
In distributed computation (e.g., within a computer system, or controlled by an operating system), all security requirements are conjunctions of authorization and availability requirements. These are just the \emph{local}\/ versions of safety and liveness properties: to specify authority means to specify what is considered safe for a particular user; to specify availability means to specify what kind of liveness guarantees are provided for each particular user's resource calls.

The \textbf{channel security requirements} are expressed as combinations of 
\begin{itemize}
\item {\bf secrecy:} \emph{bad data flows do not happen}; and \sindex{secrecy}
\item {\bf integrity:} \emph{good data flows do happen}. \sindex{integrity}
\end{itemize}
Since the data flows that identify the originator often require special treatment, the same channel security requirements instantiated to identifications are often called by special names:
\begin{itemize}
\item {\bf confidentiality:} \emph{bad identifications do not happen}; and \sindex{confidentiality}
\item {\bf authenticity:} \emph{good identifications do happen}. \sindex{authenticity}
\end{itemize}
But the usage varies, and confidentiality and secrecy are often bundled together, confused, or switched, as are authenticity and integrity. We have to be flexible with words but precise with concepts.

\subsubsection{Static vs dynamic} 
While safety and liveness are independent of each other and freely combined to express arbitrary dependability properties, and while the resource security properties just extend the dependability properties by subjects' identities, and thus leave authority and availability independent, and even orthogonal in a certain formal sense, the channel security properties dynamically depend on each other, and require a substantially different treatment.

In \textbf{network computation} (where computers communicate by messages), all \emph{data flow constraints} are logical combinations of secrecy and authenticity. Note, however, that the secrecy and the authenticity properties are not independent. In fact, if only data are processed (while the objects and the subjects are fixed), then
\begin{itemize}
\item every secret must be authenticated, and
\item every authentication is based on some secrets.
\end{itemize}
In more general systems, authentications can also be based on secure tokens, and on biometric properties, but establishing and maintaining secrecy still requires authentications. Table~\ref{Table:depend} compares and contrasts dependability and security properties.

\begin{table}[!ht]
\begin{center}
\begin{tabular}{|l||c|c|}
\hline
processing & dependability & {\bf security} \\
\hline \hline 
System & centralized & {\bf distributed}\\
\hline
observations & global & {\bf local} \\
\hline \hline
Environment & neutral & {\bf adversarial}\\
\hline
threats &  accidents & {\bf attacks}\\
\hline
\end{tabular}
\caption{The differences between dependability and  security}
\label{Table:depend}
\end{center}
\end{table}%

\para{Terminology.} The concept of secrecy is closely related to the concept of \emph{confidentiality}. They are not entirely synonymous in the colloquial usage (e.g., a "confidential meeting" is not the same thing as a "secret meeting"), but we will not make a distinction yet, since both require that some undesired data and information flows do not happen. 

Similarly, the concept of authenticity is closely related to the concept of \emph{integrity}. In this case, the difference in the usage is perhaps easier to pin down: authenticity of a message means that if it is signed by Bob, then it is really a message from Bob; whereas the integrity of the same message means that no part of it was altered on the way. While such distinctions are, of course, of great interest in particular analyses, in order to keep the general conceptual framework as simple as possible, we shall use both authenticity and integrity to refer to the same requirement that some desired data and information flows (e.g., that Bob is online) do happen.

\subsection{So what?} 
Don't skip the questions in the next section. There are several correct answers in some cases, and it is useful to discuss them. It may be equally easy to argue for different answers. When you need to decide which security policy to impose, you need clear and unique answers. You will need methods to unify different answers. Remembering the dimensions of security, its frame of reference discussed in this chapter, will help.

We learned that security concepts are
\begin{itemize}
\item simple and easy to understand and define, \emph{but} that they are
\item complicated and hard to reason about, protect, and implement. 
\end{itemize}
The complications arise from the ease with which we think and talk about security, and develop the narratives we repeat until we and everyone around us believe them. Until they fail, and then we develop other narratives. Our common sense is primed for compelling security arguments, and the common-sense arguments are primed to eventually fail and be replaced by other arguments. The common-sense \emph{security}\/ arguments are primed for \emph{security}\/ failures, which are usually based on deceit, which is usually based on self-deceit. \emph{That is why we need to go beyond \emph{common sense}\/ reasoning and study security by the methods of \emph{science}.}

The language that evolved to address the shortcomings of our common-sense reasoning and provide foundations for scientific testing is the language of mathematics. In the coming chapters, we take up the task of spelling out the mathematical models of security.

\def\thechapter{3}
\setchaptertoc
\chapter{Static resource security: access control and multi-level security}\label{Chap:Resource}

\section[What]{What is resource security?}

\sindex{security!static}
\sindex{access control}
\sindex{security!multi-level}
\sindex{MLS}

Recall that the two basic types of resource security requirements are
\begin{itemize}
\item {\bf authority:} \emph{bad resource requests are rejected}, and
\item {\bf availability:} \emph{good resource requests are fulfilled}.
\end{itemize}
Before studying such requirements, we describe the methods of \emph{access control} that are used to distinguish the good resource requests from the bad. But we first need to define what is a resource request.

\subsection{What is a resource?}\label{Sec:what-resource} \sindex{resource}
The typical examples of resources are the fossil fuels: coal, petroleum, and natural gas. They gave us power to race in cars, 10 times faster than we can run, and to fly in planes, higher than any bird. They store the energy from the Sun, accumulated by plants and bacteria for 550 million years, and then fossilized. We have been burning that energy for about 200 years, and we will probably burn it up in another 50 years, or less.

\emph{Resources are one-way functions: they are easy to use, but hard to get by.}\/ Burning the hydrocarbons from fossil fuels releases energy; making the hydrocarbons requires energy. Capturing that energy took 550 million years; releasing it took 250 years. Going down a hill is easier than going back up. That is the essence of a resource: it is a one-way function. The idea is illustrated in Fig.~\ref{Fig:resources} on the left and in the middle.


\begin{figure}[t!]
\begin{minipage}{.3\linewidth}
\begin{center}
\newcommand{\Resource}{\includegraphics[height=.7cm]{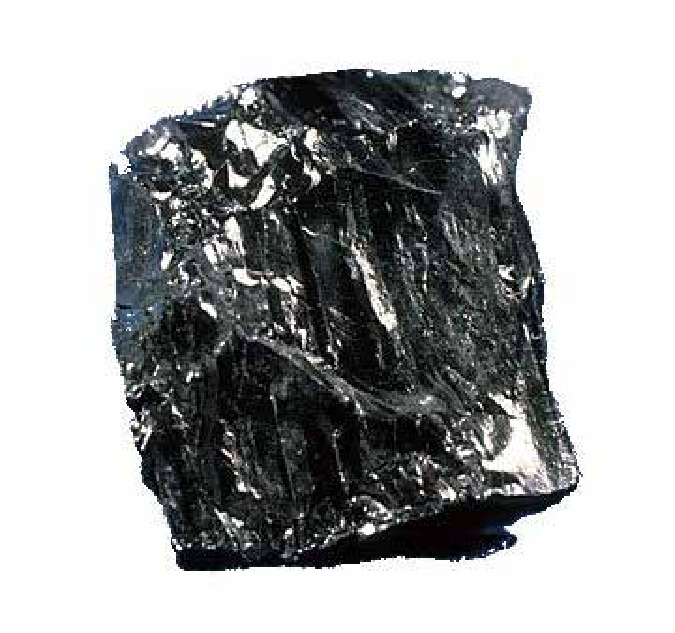}}
\newcommand{\Residue}{\includegraphics[height=.8cm]{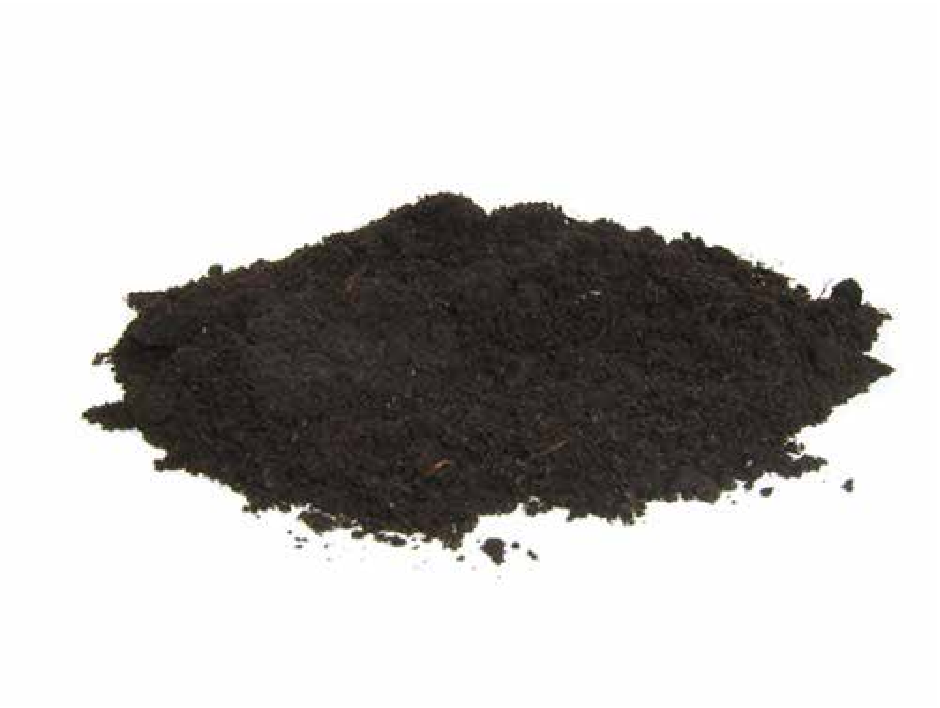}}
\newcommand{\utility}{\includegraphics[height=.85cm]{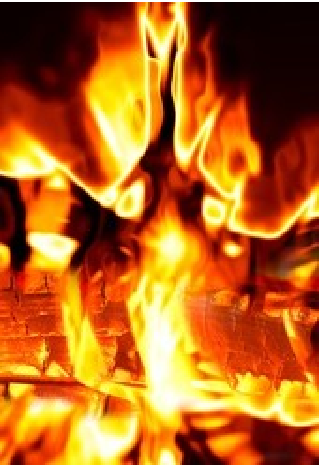}}
\newcommand{\investment}{\includegraphics[height=.85cm]{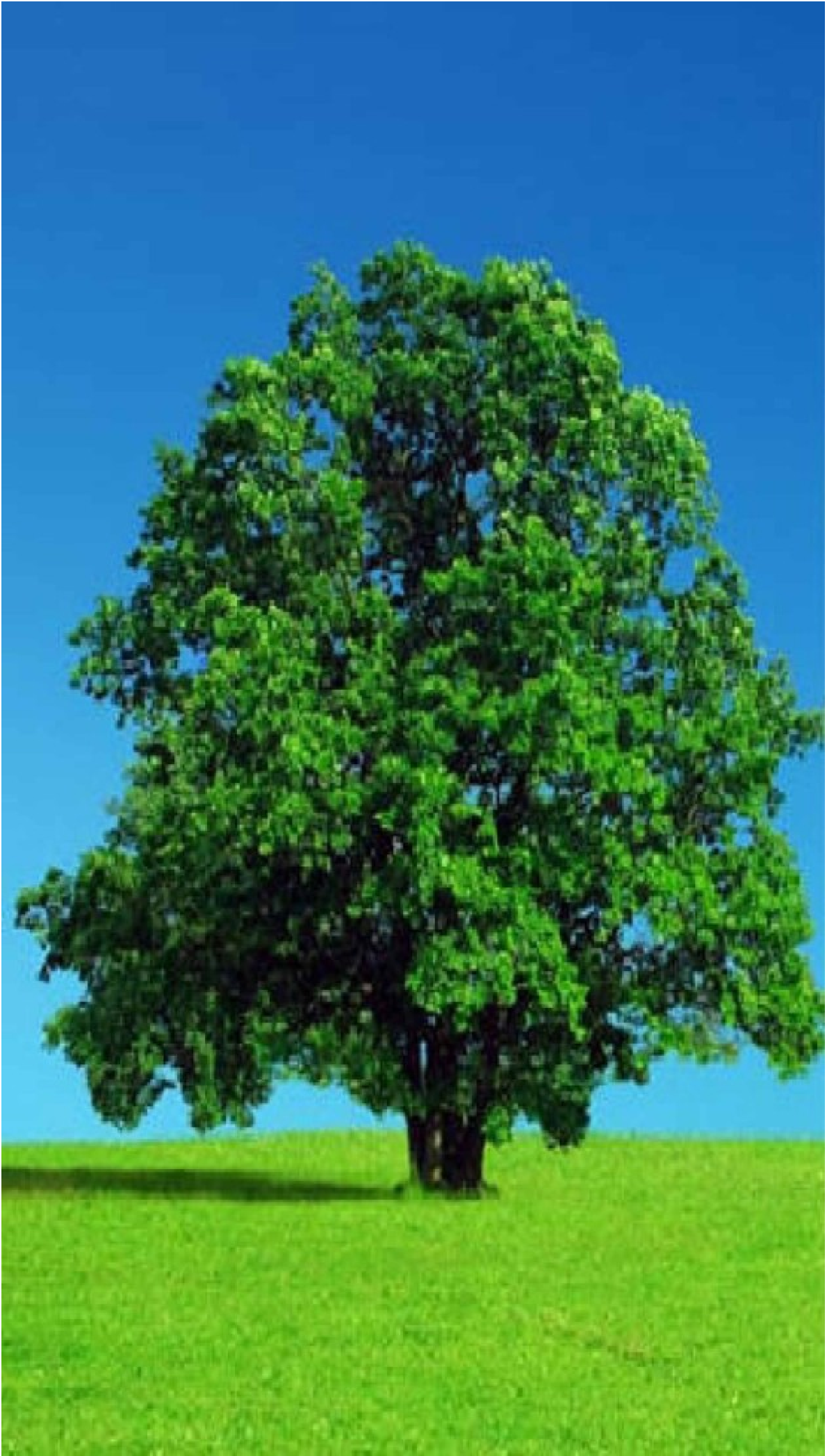}}
\newcommand{\coal}{coal}
\newcommand{\ashes}{ashes}
\newcommand{\burn}{\small burn}
\newcommand{\store}{\small store}
\def\JPicScale{.5}
\input{resource-photos.tex}

A resource is easy to use\\ but hard to come by. 
\end{center}
\end{minipage}
\begin{minipage}{.33\linewidth}
\begin{center}
\newcommand{\Resource}{Resource}
\newcommand{\Residue}{Residue}
\newcommand{\utility}{$\rightsquigarrow$ \small utility}
\newcommand{\investment}{\small investment $\rightsquigarrow$}
\def\JPicScale{.5}
\input{resource.tex}

\vspace{.75\baselineskip}
A resource provides utility\\ but requires investment.
\end{center}
\end{minipage}
\begin{minipage}{.33\linewidth}
\begin{center}
\newcommand{\Resource}{$11,213 \cdot 756,839$}
\newcommand{\Residue}{$8,486,435,707$}
\newcommand{\utility}{$\rightsquigarrow$ \small security}
\newcommand{\investment}{\small attack $\rightsquigarrow$}
\def\JPicScale{.5}
\input{resource.tex}

\vspace{.7\baselineskip}
A resource is\\ a one-way function.
\end{center}
\end{minipage}
\caption{The essence of resources.} \label{Fig:resources}
\end{figure}

Modern cryptography is based on one-way functions as computational resources. That is illustrated in the figure on the right. A one-way function is easy to compute, but hard to \emph{"uncompute"}. To compute a function $f:A\to B$ means to take an input $a$ of type $A$, and compute the output $f(a)$ of type $B$. To "uncompute" the function $f:A\to B$ means to take a value $b\in B$ and find a value $x$ of type $A$ such that $f(x) = b$. For instance, if we take the inputs to be pairs of natural numbers (nonnegative integers), i.e., $A = \NNn\times \NNn$, and multiply them, i.e., take $B = \NNn$ and set the function $f$ to be the multiplication $(\cdot ) : \NNn\times \NNn \to \NNn$, then the "uncomputing" can be construed as factoring a number $y\in B$ into primes, and partitioning them into two numbers, $x_0, x_1\in \NNn$, such that $x_0\cdot x_1 = y$. If $y$ is a product of two primes, then there is a unique way to "uncompute" the computation of $y = x_0\cdot x_1$. However, while multiplying $x_0$ and $x_1$ takes the number of computational steps proportional to the \emph{lengths}\/ of $x_0$ and $x_1$, factoring $y = x_0\cdot x_1$ requires that we somehow try to divide  $y$ with the various primes below it. This is why "uncomputing" a product of large primes is thought to be much harder than computing it, and the product is used as a one-way function. It is a computational resource of modern cryptography. (Sorry for repeating. It is sometimes useful.)

\subsection{What does it mean to secure a resource?}
A resource can be useful for many people. For example, a water well can be used by people and animals from a wide area. If it is so large that none of them can prevent others from accessing it, then they have to share it. It remains a public resource. However, if some of them can control access to the water well, then they can assert private ownership over it and claim it as an \emph{asset}.
\sindex{asset}

Resource security is \emph{access control}. It makes resources into assets. In a sense, access control is the stepping stone both into security, and into economy, which at this level boil down to the same. An asset is a resource that can be secured, owned, sold, and hence it acquires an economic value. Without security, there is no economy. However, if the cost of securing a resource is greater than its value, then claiming any ownership makes no sense. You don't spend \$200 on a lock to secure a \$20 bike. Without economy, there is no security. In a sense, economy and security are two sides of the same coin.

\section[Access Control (AC)]{Access control}\label{Sec:AC}

\sindex{access control}
\para{Example 1.} Access control had to be implemented long before there were computers and operating systems, as soon as there were some people, and they owned some goods. But things get more general in the science of security, so we will call people \emph{subjects}, and their goods \emph{objects}. This will be formalized in Def.~\ref{Def:AC}. To begin, consider our hero subjects and their valuable objects in \cref{Fig:aliceboob}: Alice had a sheep, and Bob had a jar of oil. \sindex{subjects} \sindex{objects}
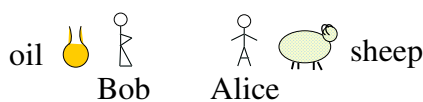
\begin{figure}[htbp]
\begin{center}
\newcommand{\sheep}{\mbox{sheep}}
\newcommand{\oil}{\mbox{oil}}
\newcommand{\alice}{\mbox{Alice}}
\newcommand{\bob}{\mbox{Bob}}
\def\JPicScale{.85}
\ifx\JPicScale\undefined\def\JPicScale{1}\fi
\psset{unit=\JPicScale mm}
\psset{linewidth=0.3,dotsep=1,hatchwidth=0.3,hatchsep=1.5,shadowsize=1,dimen=middle}
\psset{dotsize=0.7 2.5,dotscale=1 1,fillcolor=black}
\psset{arrowsize=1 2,arrowlength=1,arrowinset=0.25,tbarsize=0.7 5,bracketlength=0.15,rbracketlength=0.15}
\begin{pspicture}(0,0)(52.5,12.04)
\rput{0}(36.12,10.12){\psellipse[linewidth=0.15](0,0)(1,-1)}
\newrgbcolor{userFillColour}{0.77 0.89 0.88}
\psline[linewidth=0.15,fillcolor=userFillColour,fillstyle=solid](36.12,9.12)(36.12,6.12)
\newrgbcolor{userFillColour}{0.77 0.89 0.88}
\psline[linewidth=0.15,fillcolor=userFillColour,fillstyle=solid](36.12,6.12)(35.12,3.12)
\newrgbcolor{userFillColour}{0.77 0.89 0.88}
\psline[linewidth=0.15,fillcolor=userFillColour,fillstyle=solid](36.12,6.12)(37.12,3.12)
\newrgbcolor{userFillColour}{0.77 0.89 0.88}
\psline[linewidth=0.15,fillcolor=userFillColour,fillstyle=solid](36.12,8.12)(38.12,7.12)
\newrgbcolor{userFillColour}{0.77 0.89 0.88}
\psline[linewidth=0.15,fillcolor=userFillColour,fillstyle=solid](36.12,8.12)(34.12,7.12)
\rput{0}(17,10.38){\psellipse[linewidth=0.15](0,0)(1,-1)}
\newrgbcolor{userFillColour}{0.77 0.89 0.88}
\psline[linewidth=0.15,fillcolor=userFillColour,fillstyle=solid](17,9.38)(17,6.38)
\newrgbcolor{userFillColour}{0.77 0.89 0.88}
\psline[linewidth=0.15,fillcolor=userFillColour,fillstyle=solid](17,6.38)(16.88,3.12)
\newrgbcolor{userFillColour}{0.77 0.89 0.88}
\psline[linewidth=0.15,fillcolor=userFillColour,fillstyle=solid](17,6.38)(18.75,3.12)
\newrgbcolor{userFillColour}{0.77 0.89 0.88}
\psline[linewidth=0.15,fillcolor=userFillColour,fillstyle=solid](16.88,8.12)(18.75,7.5)
\newrgbcolor{userFillColour}{0.77 0.89 0.88}
\psline[linewidth=0.15,fillcolor=userFillColour,fillstyle=solid](18.75,7.5)(16.88,6.25)
\newrgbcolor{userFillColour}{0.77 0.89 0.88}
\newrgbcolor{userHatchColour}{1 1 0.8}
\rput{0}(45.5,6.12){\psellipse[linewidth=0.05,fillcolor=userFillColour,fillstyle=crosshatch*,hatchwidth=0.05,hatchsep=0.3,hatchcolor=userHatchColour](0,0)(4,-2.5)}
\newrgbcolor{userFillColour}{0.77 0.89 0.88}
\psline[fillcolor=userFillColour,fillstyle=solid](43.75,3.75)(43.5,2.62)
\newrgbcolor{userFillColour}{0.77 0.89 0.88}
\psline[fillcolor=userFillColour,fillstyle=solid](44.5,3.62)(44.5,2.62)
\newrgbcolor{userFillColour}{0.77 0.89 0.88}
\psline[fillcolor=userFillColour,fillstyle=solid](47.5,3.75)(48.12,2.5)
\newrgbcolor{userFillColour}{0.77 0.89 0.88}
\psline[fillcolor=userFillColour,fillstyle=solid](46.88,3.75)(46.88,2.5)
\newrgbcolor{userFillColour}{0.77 0.89 0.88}
\newrgbcolor{userHatchColour}{1 1 0.8}
\rput{90}(49.38,8.12){\psellipse[linewidth=0.05,fillcolor=userFillColour,fillstyle=crosshatch*,hatchwidth=0.05,hatchsep=0.3,hatchcolor=userHatchColour](0,0)(1.19,-1.12)}
\rput{112.29}(48.75,8.75){\psellipticarc[linewidth=0.2](0,0)(0.85,-0.58){-25.6}{109.12}}
\rput{111.93}(49.38,8.75){\psellipticarc[linewidth=0.2](0,0)(0.85,-0.59){-25.77}{108.95}}
\newrgbcolor{userFillColour}{0.77 0.89 0.88}
\rput[t](36.25,1.25){$\alice$}
\newrgbcolor{userFillColour}{0.77 0.89 0.88}
\rput[t](17.5,1.25){$\bob$}
\newrgbcolor{userFillColour}{0.77 0.89 0.88}
\rput[l](52.5,5){$\sheep$}
\rput[r](5,5){$\oil$}
\newrgbcolor{userFillColour}{1 0.8 0}
\rput{90}(10.06,4.91){\psellipticarc[linewidth=0.1,fillcolor=userFillColour,fillstyle=solid](0,0)(1.97,1.94){29.04}{324.46}}
\newrgbcolor{userFillColour}{1 0.8 0}
\rput{108.03}(8.17,8.81){\psellipticarc[linewidth=0.1,fillcolor=userFillColour,fillstyle=solid](0,0)(2.42,1){179.97}{281.38}}
\newrgbcolor{userFillColour}{1 0.8 0}
\rput{77.96}(11.88,8.75){\psellipticarc[linewidth=0.1,fillcolor=userFillColour,fillstyle=solid](0,0)(2.33,1.12){71.07}{176.01}}
\newrgbcolor{userFillColour}{0.77 0.89 0.88}
\psline[linewidth=0.15,fillcolor=userFillColour,fillstyle=solid](35.62,4.38)(36.88,4.38)
\newrgbcolor{userFillColour}{0.77 0.89 0.88}
\psline[linewidth=0.15,fillcolor=userFillColour,fillstyle=solid](16.88,3.12)(18.75,3.12)
\end{pspicture}
\caption{Subjects Alice and Bob, and their objects sheep and oil}
\label{Fig:aliceboob}
\end{center}
\end{figure}
The situation (or \emph{state}\/ of the world) $s$, where only Alice has access to sheep's wool, milk, and eventually meat, whereas only Bob has access to cooking with his oil, can be represented using a matrix like $M^s$, on the left in \cref{Fig:AliceBobStart}. \sindex{state} 
\begin{figure}[htbp]
\begin{center}
{\small \begin{tabular}{|l||l|l|}
\hline
$s$ & sheep & oil \\
 \hline\hline
Alice & \{\text{milk, wool, meat}\} 
& $\emptyset$ \\
\hline
Bob  & $\emptyset$ & \{\text{cook}\} 
\\
\hline
\end{tabular}
\hspace{1em}
{\Large $\to$}
\hspace{1em}
\begin{tabular}{|l||l|l|}
\hline
$q$ & sheep & oil \\
 \hline\hline
Alice & \{\text{milk, wool, meat}\} 
& \{\text{bottle of oil}\} \\
\hline
Bob  & \{\text{bottle of milk}\} & \{\text{cook}\} 
\\
\hline
\end{tabular}}
\caption{Alice and Bob trade their private resources}
\label{Fig:AliceBobStart}
\end{center}
\end{figure}
If Alice gives Bob a bottle of her sheep's milk in exchange for a bottle of Bob's oil, then the state $s$ changes to the state $q$,  which is represented by the matrix $M^{q}$, in \cref{Fig:AliceBobStart} on the right. 
Such representations are formalized as follows. 

\medskip
\begin{definition}\label{Def:AC} \sindex{access control!model}  \sindex{access control!types}
An \emph{access control (AC)}\/ model consists of
\begin{itemize}
\item \emph{access control types}:    
\begin{itemize}
\item \emph{objects} (or \emph{items}) $\Obj = \{i,j,\ldots\}$,
\item \emph{subjects} (or \emph{users}) $\Subj = \{s,u,\ldots\}$, and
\item \emph{actions} (or \emph{labels}) $\Act = \{a,b,\ldots\}$,
\end{itemize}
\item \emph{access control matrices} in the form $
M, B\  :\  \Subj \times \Obj  \to  \WP \Act$, 
where $\WP \Act$ is the set of subsets of $\Act$, and where
\begin{itemize} \sindex{access control!matrix}
\item $M \ =\  \left(M_{ui}\right)_{\Subj \times \Obj}$  is the \emph{permission matrix}: each entry $M_{ui}$ specifies  which actions the {\subj} $u$ is permitted to apply on the {\obj} $i$;
\item $B \ =\  \left(B_{ui}\right)_{\Subj \times \Obj}$ is the \emph{access matrix}: each entry $B_{ui}$ specifies  which actions has the {\subj} $u$ requested for the {\obj} $i$.
\end{itemize} \sindex{permission matrix} \sindex{access matrix}
\end{itemize} \sindex{actions} \sindex{access control!requirements}
The \emph{access control (AC) requirement}\/ is that for all $u\in \Subj$ and $i\in \Obj$ holds
\bea \label{eq:AC} \sindex{access control!requirements}
B_{ui} &\subseteq & M_{ui}.  
\eea
\end{definition}

\para{Remarks.} The terminological alternatives object/subject/action vs item/user/label are used in different communities that study the same phenomena from slightly different angles. There are still other alternatives. While using several terminological alternatives can be confusing, different alternatives sometimes support different intuitions, which can be valuable. We initially stick with the object/subject/action variant.  \sindex{item} \sindex{user} \sindex{label}

\para{How is access controlled?} Although the matrices $M$ and $B$ have the same types, their purposes and implementations are substantially different. The permission matrix $M$ is specified at design time by some \emph{access policy authority}, \sindex{access policy authority} e.g., the system designer or administrator. The sets of actions $M_{ui}\in \WP\Act$ are usually implemented as bitstrings, assigning each action $a\in \Act$ the value $1$ if it is permitted and the value $0$ otherwise; we write the permitted actions as subsets for convenience. The access matrix $B$ is maintained by the system itself, usually by a \emph{system monitor}. \sindex{system monitor}  Whenever a {\subj} $u\in \Subj$ issues a request to access an {\obj} $i\in \Obj$ by an action $a\in \Act$, the monitor tests whether $a\in M_{ui}$, and if the test succeeds, it adds $a$ to $B_{ui}$. The monitor thus maintains the AC requirement~\eqref{eq:AC} as a system invariant. The set $B_{ui}$ thus consists of the actions $a$ that the subject $u$ has requested for the object $i$, and the permissions were granted. An action that was requested but not permitted by $M_{ui}$, or an action that is permitted by $M_{ui}$, but not requested, will not be recorded in $B_{ui}$.\footnote{The requested actions must, of course, be recorded before the permission is issued or denied. An implementer might thus be tempted to record all requests in $B_{ui}$, and to purge them to maintain \eqref{eq:AC}. However, a separate list of all resource calls is maintained as a basic component of the operating system, not just for the purposes of access control. The access matrix $B$ is the intersection of that list of the actual calls and of the matrix $M$ of permissible calls.}

\para{Example 2.} Access control applies not only to goods but also to space, which is also a resource. When Alice and Bob build houses, like in \cref{Fig:housing} on the left, 
\begin{figure}[ht!]
\newcommand{\Carol}{}
\newcommand{\one}{\mbox{\footnotesize Uruk Lane}}
\newcommand{\two}{\mbox{\footnotesize \#31}}
\newcommand{\four}{\mbox{\footnotesize \#30}}
    \centering
    \caption{Access control of public and private spaces}
    \bigskip
    \begin{minipage}[b]{0.4\textwidth}
    \centering
\def\JPicScale{.8}
\ifx\JPicScale\undefined\def\JPicScale{1}\fi
\psset{unit=\JPicScale mm}
\psset{linewidth=0.3,dotsep=1,hatchwidth=0.3,hatchsep=1.5,shadowsize=1,dimen=middle}
\psset{dotsize=0.7 2.5,dotscale=1 1,fillcolor=black}
\psset{arrowsize=1 2,arrowlength=1,arrowinset=0.25,tbarsize=0.7 5,bracketlength=0.15,rbracketlength=0.15}
\begin{pspicture}(0,0)(55,40)
\newrgbcolor{userFillColour}{0.87 0.76 0.76}
\pspolygon[fillcolor=userFillColour,fillstyle=solid](0,40)(30,40)(30,20)(0,20)
\newrgbcolor{userFillColour}{0.73 0.84 0.87}
\pspolygon[fillcolor=userFillColour,fillstyle=solid](45,40)(55,40)(55,0)(45,0)
\newrgbcolor{userFillColour}{0.77 0.89 0.88}
\rput[br](29.38,20.62){$\two$}
\newrgbcolor{userFillColour}{0.77 0.89 0.88}
\rput[br](54,1){$\four$}
\rput{0}(20.75,34.12){\psellipse[linewidth=0.15](0,0)(1,-1)}
\newrgbcolor{userFillColour}{0.77 0.89 0.88}
\psline[linewidth=0.15,fillcolor=userFillColour,fillstyle=solid](20.75,33.12)(20.75,30.12)
\newrgbcolor{userFillColour}{0.77 0.89 0.88}
\psline[linewidth=0.15,fillcolor=userFillColour,fillstyle=solid](20.75,30.12)(20.62,26.88)
\newrgbcolor{userFillColour}{0.77 0.89 0.88}
\psline[linewidth=0.15,fillcolor=userFillColour,fillstyle=solid](20.75,30.12)(22.5,26.88)
\newrgbcolor{userFillColour}{0.77 0.89 0.88}
\psline[linewidth=0.15,fillcolor=userFillColour,fillstyle=solid](20.62,31.88)(22.5,31.25)
\newrgbcolor{userFillColour}{0.77 0.89 0.88}
\psline[linewidth=0.15,fillcolor=userFillColour,fillstyle=solid](22.5,31.25)(20.62,30)
\newrgbcolor{userFillColour}{0.77 0.89 0.88}
\psline[linewidth=0.15,fillcolor=userFillColour,fillstyle=solid](20.62,26.88)(22.5,26.88)
\rput{0}(38.62,34.5){\psellipse[linewidth=0.15](0,0)(1,-1)}
\newrgbcolor{userFillColour}{0.77 0.89 0.88}
\psline[linewidth=0.15,fillcolor=userFillColour,fillstyle=solid](38.62,33.5)(38.62,30.5)
\newrgbcolor{userFillColour}{0.77 0.89 0.88}
\psline[linewidth=0.15,fillcolor=userFillColour,fillstyle=solid](38.62,30.5)(37.62,27.49)
\newrgbcolor{userFillColour}{0.77 0.89 0.88}
\psline[linewidth=0.15,fillcolor=userFillColour,fillstyle=solid](38.62,30.5)(39.62,27.49)
\newrgbcolor{userFillColour}{0.77 0.89 0.88}
\psline[linewidth=0.15,fillcolor=userFillColour,fillstyle=solid](38.62,32.5)(40.62,31.5)
\newrgbcolor{userFillColour}{0.77 0.89 0.88}
\psline[linewidth=0.15,fillcolor=userFillColour,fillstyle=solid](38.62,32.5)(36.62,31.5)
\newrgbcolor{userFillColour}{0.77 0.89 0.88}
\psline[linewidth=0.15,fillcolor=userFillColour,fillstyle=solid](38.12,28.75)(39.38,28.75)
\newrgbcolor{userFillColour}{0.77 0.89 0.88}
\rput[bl](0,1.25){$\one$}
\rput{0}(31.12,15.75){\psellipse[linewidth=0.15](0,0)(1,-1)}
\newrgbcolor{userFillColour}{0.77 0.89 0.88}
\psline[linewidth=0.15,fillcolor=userFillColour,fillstyle=solid](31.12,14.75)(31.12,11.75)
\newrgbcolor{userFillColour}{0.77 0.89 0.88}
\psline[linewidth=0.15,fillcolor=userFillColour,fillstyle=solid](31.12,11.75)(30.12,8.74)
\newrgbcolor{userFillColour}{0.77 0.89 0.88}
\psline[linewidth=0.15,fillcolor=userFillColour,fillstyle=solid](31.12,11.75)(32.12,8.74)
\newrgbcolor{userFillColour}{0.77 0.89 0.88}
\psline[linewidth=0.15,fillcolor=userFillColour,fillstyle=solid](31.25,12.5)(32.5,11.88)
\newrgbcolor{userFillColour}{0.77 0.89 0.88}
\psline[linewidth=0.15,fillcolor=userFillColour,fillstyle=solid](32.5,14.38)(31.25,13.75)
\newrgbcolor{userFillColour}{0.77 0.89 0.88}
\psline[linewidth=0.15,fillcolor=userFillColour,fillstyle=solid](32.5,14.38)(33.75,13.75)
\newrgbcolor{userFillColour}{0.77 0.89 0.88}
\psline[linewidth=0.15,fillcolor=userFillColour,fillstyle=solid](33.75,12.5)(32.5,11.88)
\rput(31.25,6.25){\Carol}
\end{pspicture} 
\\[2ex]
Architecture
    \end{minipage}%
    ~ 
    \begin{minipage}[b]{0.5\textwidth}
\begin{center}
{\small \begin{tabular}{|c||c|c|c|}
\hline
$s$ & House \#30 & House \#31 & Uruk Lane \\
 \hline\hline
Alice & \{occupy\} & $\emptyset$ & \{use\} \\
\hline
Bob  & $\emptyset$ &  \{occupy\} & \{use\}
\\
\hline
Carol  & $\emptyset$ &  $\emptyset$ & \{use\}
\\
\hline
\end{tabular}}
\\[8ex]
Permissions
\end{center}
    \end{minipage}
        \label{Fig:housing}
\end{figure}
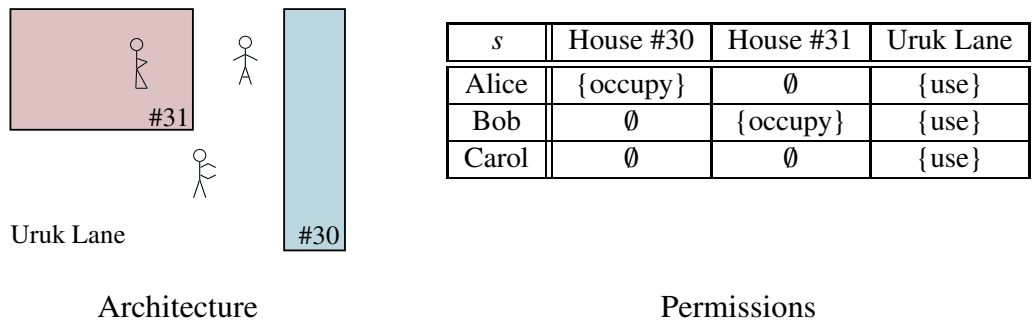
then the separation of the private space inside their houses and the public space that they share is an instance of access control. The permission matrix is displayed in \cref{Fig:housing} on the right. The public space is accessible to all public, represented by Carol, who does not own a house in the neighborhood.

\para{Matchstick names.} To minimize clutter, we avoid writing out subjects' names in the diagrams, but include their initials in their bodies: Alice's legs form an A, Carol's arms form a C, and Bob uses his arms and legs to form a B.

\para{Exercise.} Specify permission matrices that describe the states where Alice invites Bob and Carol to her house.

\para{Remark.} Note that the current locality of a subject cannot be directly expressed using permission matrices alone. For example, the situation in \cref{Fig:housing}, where Alice is not in her private home but in the public space of Uruk Lane, cannot be expressed. Location specifications will be introduced in Sec.~\ref{Sec:MLS}.

\subsection{Access control in computer security}
It is easy to see that each of the AC matrices is, in fact, a ternary relation, as there is a function
\[
\prooftree
\Subj \times \Obj \tto R \WP \Act
\justifies
\widehat R \in \WP(\Subj \times \Obj \times \Act)
\endprooftree
\] 
This means that the size of AC matrices grows linearly with the product of the number of \subjss, \objss, and actions. This presents an implementation problem, since the access matrix $B$ needs to be updated, and compared with the permission matrix $M$, at each access, potentially in each clock cycle. So a direct implementation of abstract access control is impractical.

The space of practical implementation strategies is spanned between the two extremal approaches:
\begin{itemize} \sindex{access control!lists} \sindex{access control!capability-based}
\item \textbf{access control lists (ACLs):}\ \ \ $\Obj \to \WP(\Act \times \Subj)$
\item \textbf{capability-based AC:}\hspace{3em} $\Subj \to \WP (\Act \times \Obj)$
\end{itemize}
In both cases, the AC matrices are thus decentralized: in the former case, the rows of the permission matrix $M$ are distributed among the \objs of the system as the ACLs, one for each $i\in \Obj$; in the latter case, the columns of $M$ are distributed among the \subjs as capabilities, one for each $u\in \Subj$. The access matrix $B$ is not stored at all: each access request of $u$ for $i$ is checked as it comes, be it against $u$'s capability, or against $i$'s ACL.

The ACL-based approach was first implemented in UNIX. \sindex{access control!UNIX} This is the common-sense approach, since the number of \objs in $\Obj$ is usually much greater than the number of \subjs in $\Subj$ or actions in $\Act$. The fact that the datatype $\WP(\Act \times \Subj)$ still appears exponential is resolved by reducing $\Act$ to some standard access tools. In UNIX, they are just {\tt r}ead, {\tt w}rite and e{\tt  x}ecute, and the Windows operating systems advanced the state of the art by extending this set to up to $7$ elements. The ACL entries for \subjs are also reduced to {\tt u}ser, {\tt o}thers, and {\tt g}roup. The \emph{user}\/ is $u\in \Subj$ extracted from the login, the entry \emph{others}\/ accommodates all of $\Subj$, and the rest of the $\WP \Subj$ are the possible groups. The impractical part of the ACLs is thus deferred to the group management, which is the task of the system manager.  Problem solved.

The capability-based approaches go back to Symbian, which was Ericson's early operating system for mobile phones. The development was driven by the fact that the early mobile devices were not powerful enough to manage ACLs. For different reasons, a version of capability-based AC was adopted in SELinux, and that AC model was adopted in Android and then modified many times. It is not easy to tell whether there is still a capability-based model. 

\subsection{Implicit clearance and classification}\label{Sec:implicit}
Intuitively, Alice's authority is at least as great as Bob's authority if she can do anything he can do. Such comparisons can be derived from any AC model as a \emph{clearance}\/ preordering of {\subjs}, defined by \sindex{clearance}
\bear
u\leq v & \iff & \forall i\in \Obj.\  M_{ui}\subseteq M_{vi}
\eear
for all $u,v\in \Subj$. It is easy to see that $\leq$ is a reflexive and transitive relation on $\Subj$. It is not antisymmetric, because different {\subjs} can have the same authority, i.e., there may be $u\neq v$ with $M_{ui}= M_{vi}$, and thus $u\leq v$ and $u\geq v$.

An analogous \emph{classification}\/ preordering of {\objs} can be derived similarly, except that it is more meaningful to use access requests than permissions. For example, oil is at least as classified as the sheep if anyone who requests the sheep must also request oil, i.e., \sindex{classification}
\bear
i\leq j & \iff & \forall u\in \Subj.\  B_{ui}\subseteq B_{uj}
\eear

In this way, an AC model implicitly assigns some clearance levels to the \subjs and some classification levels to the \objss. It is often more convenient to specify access policies by making such security levels explicit.

\section[Multi-Level Security (MLS)]{Multi-level security}\label{Sec:MLS}

While the AC models have been formalized to facilitate access control in operating systems, the multi-level security (MLS)models are familiar from the various military, diplomatic, industrial, and banking security systems. There, the access to resources is controlled by assigning different \emph{security classification}\/  \sindex{security!classification} levels to the \sindex{locality}\sindex{clearance}\sindex{classification} objects and different \emph{security clearance}\/ \sindex{security!clearance} levels to the subjects. While these two security frameworks have quite different origins, and different formalisms, they turn out to be logically equivalent. In Sec.~\ref{Sec:implicit}, \sindex{access policy} we saw that every AC model contains an implicit MLS model. In this section, we shall show how to reduce any given MLS model to an AC model. 

\subsection{Explicit clearance and classification}\label{Sec:explicit}
\newcommand{\one}{\ell_1}
\newcommand{\two}{\ell_2}
\newcommand{\three}{\ell_3}
\newcommand{\four}{\ell_4}
\newcommand{\five}{\ell_5}
In the meantime, as the Neolithic broke out, Bob built some secure vaults to store his assets, with a high-security chamber where he stored his crate of oil. This is displayed in \cref{Fig:vault}. Alice and her sheep remained outside at the lowest security level, denoted $\one$.
\begin{figure}[ht!]
\newcommand{\alice}{\mbox{\footnotesize Alice}}
\newcommand{\bob}{\mbox{\footnotesize Bob}}
\newcommand{\sheep}{\mbox{\footnotesize sheep}}
\newcommand{\oil}{\mbox{\footnotesize oil}}
    \centering
    \caption{Aspects of multi-level security (MLS)}
    \bigskip
    \begin{minipage}[b]{0.3\textwidth}
    \centering
\def\JPicScale{.8}
\ifx\JPicScale\undefined\def\JPicScale{1}\fi
\psset{unit=\JPicScale mm}
\psset{linewidth=0.3,dotsep=1,hatchwidth=0.3,hatchsep=1.5,shadowsize=1,dimen=middle}
\psset{dotsize=0.7 2.5,dotscale=1 1,fillcolor=black}
\psset{arrowsize=1 2,arrowlength=1,arrowinset=0.25,tbarsize=0.7 5,bracketlength=0.15,rbracketlength=0.15}
\begin{pspicture}(0,0)(55,40)
\newrgbcolor{userFillColour}{0.87 0.76 0.76}
\pspolygon[fillcolor=userFillColour,fillstyle=solid](0,40)(30,40)(30,20)(0,20)
\newrgbcolor{userFillColour}{0.84 0.84 0.64}
\pspolygon[fillcolor=userFillColour,fillstyle=solid](0,0)(10,0)(10,30)(0,30)
\newrgbcolor{userFillColour}{0.64 0.52 0.35}
\pspolygon[fillcolor=userFillColour,fillstyle=solid](0,20)(10,20)(10,30)(0,30)
\newrgbcolor{userFillColour}{0.73 0.84 0.87}
\pspolygon[fillcolor=userFillColour,fillstyle=solid](45,40)(55,40)(55,0)(45,0)
\newrgbcolor{userFillColour}{0.77 0.89 0.88}
\rput[br](29.38,0.62){$\one$}
\newrgbcolor{userFillColour}{0.77 0.89 0.88}
\rput[br](29.38,20.62){$\two$}
\newrgbcolor{userFillColour}{0.77 0.89 0.88}
\rput[bl](0.62,0.62){$\three$}
\newrgbcolor{userFillColour}{0.77 0.89 0.88}
\rput[br](54,1){$\four$}
\newrgbcolor{userFillColour}{0.77 0.89 0.88}
\newrgbcolor{userHatchColour}{1 1 0.8}
\rput{0}(31,9.75){\psellipse[linewidth=0.05,fillcolor=userFillColour,fillstyle=crosshatch*,hatchwidth=0.05,hatchsep=0.3,hatchcolor=userHatchColour](0,0)(4,-2.49)}
\newrgbcolor{userFillColour}{0.77 0.89 0.88}
\psline[fillcolor=userFillColour,fillstyle=solid](29.25,7.38)(29,6.24)
\newrgbcolor{userFillColour}{0.77 0.89 0.88}
\psline[fillcolor=userFillColour,fillstyle=solid](30,7.25)(30,6.24)
\newrgbcolor{userFillColour}{0.77 0.89 0.88}
\psline[fillcolor=userFillColour,fillstyle=solid](33,7.38)(33.62,6.12)
\newrgbcolor{userFillColour}{0.77 0.89 0.88}
\psline[fillcolor=userFillColour,fillstyle=solid](32.38,7.38)(32.38,6.12)
\newrgbcolor{userFillColour}{0.77 0.89 0.88}
\newrgbcolor{userHatchColour}{1 1 0.8}
\rput{90}(34.88,11.75){\psellipse[linewidth=0.05,fillcolor=userFillColour,fillstyle=crosshatch*,hatchwidth=0.05,hatchsep=0.3,hatchcolor=userHatchColour](0,0)(1.19,-1.12)}
\rput{112.29}(34.25,12.37){\psellipticarc[linewidth=0.2](0,0)(0.85,-0.58){-25.6}{109.12}}
\rput{111.93}(34.88,12.37){\psellipticarc[linewidth=0.2](0,0)(0.85,-0.59){-25.77}{108.95}}
\newrgbcolor{userFillColour}{0.77 0.89 0.88}
\rput[bl](0.62,20.62){$\five$}
\newrgbcolor{userFillColour}{1 0.8 0}
\rput{90}(5.06,24.91){\psellipticarc[linewidth=0.1,fillcolor=userFillColour,fillstyle=solid](0,0)(1.97,1.94){29.04}{324.46}}
\newrgbcolor{userFillColour}{1 0.8 0}
\rput{108.03}(3.17,28.81){\psellipticarc[linewidth=0.1,fillcolor=userFillColour,fillstyle=solid](0,0)(2.42,1){179.97}{281.38}}
\newrgbcolor{userFillColour}{1 0.8 0}
\rput{77.96}(6.88,28.75){\psellipticarc[linewidth=0.1,fillcolor=userFillColour,fillstyle=solid](0,0)(2.33,1.12){71.07}{176.01}}
\rput{0}(20.75,34.12){\psellipse[linewidth=0.15](0,0)(1,-1)}
\newrgbcolor{userFillColour}{0.77 0.89 0.88}
\psline[linewidth=0.15,fillcolor=userFillColour,fillstyle=solid](20.75,33.12)(20.75,30.12)
\newrgbcolor{userFillColour}{0.77 0.89 0.88}
\psline[linewidth=0.15,fillcolor=userFillColour,fillstyle=solid](20.75,30.12)(20.62,26.88)
\newrgbcolor{userFillColour}{0.77 0.89 0.88}
\psline[linewidth=0.15,fillcolor=userFillColour,fillstyle=solid](20.75,30.12)(22.5,26.88)
\newrgbcolor{userFillColour}{0.77 0.89 0.88}
\psline[linewidth=0.15,fillcolor=userFillColour,fillstyle=solid](20.62,31.88)(22.5,31.25)
\newrgbcolor{userFillColour}{0.77 0.89 0.88}
\psline[linewidth=0.15,fillcolor=userFillColour,fillstyle=solid](22.5,31.25)(20.62,30)
\newrgbcolor{userFillColour}{0.77 0.89 0.88}
\psline[linewidth=0.15,fillcolor=userFillColour,fillstyle=solid](20.62,26.88)(22.5,26.88)
\rput{0}(38.62,34.5){\psellipse[linewidth=0.15](0,0)(1,-1)}
\newrgbcolor{userFillColour}{0.77 0.89 0.88}
\psline[linewidth=0.15,fillcolor=userFillColour,fillstyle=solid](38.62,33.5)(38.62,30.5)
\newrgbcolor{userFillColour}{0.77 0.89 0.88}
\psline[linewidth=0.15,fillcolor=userFillColour,fillstyle=solid](38.62,30.5)(37.62,27.49)
\newrgbcolor{userFillColour}{0.77 0.89 0.88}
\psline[linewidth=0.15,fillcolor=userFillColour,fillstyle=solid](38.62,30.5)(39.62,27.49)
\newrgbcolor{userFillColour}{0.77 0.89 0.88}
\psline[linewidth=0.15,fillcolor=userFillColour,fillstyle=solid](38.62,32.5)(40.62,31.5)
\newrgbcolor{userFillColour}{0.77 0.89 0.88}
\psline[linewidth=0.15,fillcolor=userFillColour,fillstyle=solid](38.62,32.5)(36.62,31.5)
\newrgbcolor{userFillColour}{0.77 0.89 0.88}
\psline[linewidth=0.15,fillcolor=userFillColour,fillstyle=solid](38.12,28.75)(39.38,28.75)
\end{pspicture} \\[2ex]
Architecture
    \end{minipage}%
    ~ 
    \begin{minipage}[b]{0.3\textwidth}
\begin{center}
\def\JPicScale{1}
\ifx\JPicScale\undefined\def\JPicScale{1}\fi
\psset{unit=\JPicScale mm}
\psset{linewidth=0.3,dotsep=1,hatchwidth=0.3,hatchsep=1.5,shadowsize=1,dimen=middle}
\psset{dotsize=0.7 2.5,dotscale=1 1,fillcolor=black}
\psset{arrowsize=1 2,arrowlength=1,arrowinset=0.25,tbarsize=0.7 5,bracketlength=0.15,rbracketlength=0.15}
\begin{pspicture}(0,0)(36,26)
\newrgbcolor{userFillColour}{0.77 0.89 0.88}
\rput(20,1){$\one$}
\newrgbcolor{userFillColour}{0.77 0.89 0.88}
\rput(5,15){$\two$}
\newrgbcolor{userFillColour}{0.77 0.89 0.88}
\rput(20,15){$\three$}
\newrgbcolor{userFillColour}{0.77 0.89 0.88}
\rput(36,15){$\four$}
\newrgbcolor{userFillColour}{0.77 0.89 0.88}
\rput(12,26){$\five$}
\newrgbcolor{userFillColour}{0.77 0.89 0.88}
\psline[fillcolor=userFillColour,fillstyle=solid](20,12)(20,4)
\newrgbcolor{userFillColour}{0.77 0.89 0.88}
\psline[fillcolor=userFillColour,fillstyle=solid](35,13)(23,4)
\newrgbcolor{userFillColour}{0.77 0.89 0.88}
\psline[fillcolor=userFillColour,fillstyle=solid](6,13)(17,4)
\newrgbcolor{userFillColour}{0.77 0.89 0.88}
\psline[fillcolor=userFillColour,fillstyle=solid](11,24)(6,17)
\newrgbcolor{userFillColour}{0.77 0.89 0.88}
\psline[fillcolor=userFillColour,fillstyle=solid](14,24)(19,17)
\end{pspicture}\\[2ex]
Security levels
\end{center}
    \end{minipage}
    ~
        \begin{minipage}[b]{0.3\textwidth}
\begin{center}
\begin{tabular}{|c||c|c|}
\hline
$q_0$ & $\loc^{q_0}$ & $\clr^{q_0}$  \\
 \hline\hline
$\alice$ & $\one$ & $\four$ \\
\hline
$\bob$  & $\two$ & $\five$ \\
\hline
\cline{1-2}
$\sheep$ & $\one$  \\
\cline{1-2}
$\oil$  & $\five$ \\
\cline{1-2}
\end{tabular} \\[2ex]
Security assignments
\end{center}
    \end{minipage}
        \label{Fig:vault}
\end{figure}
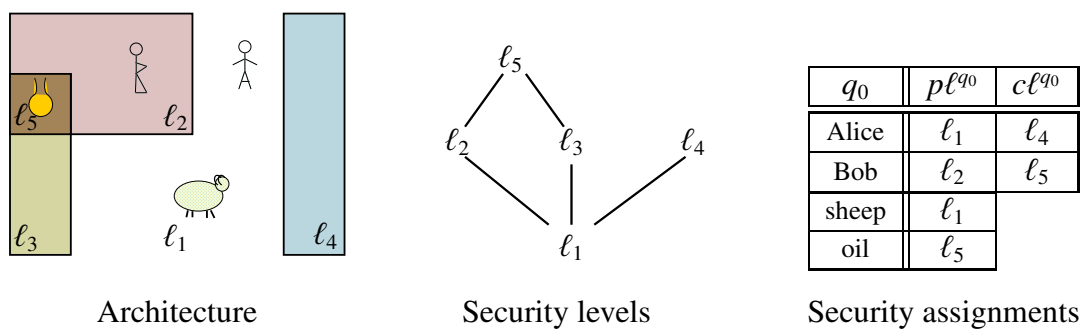

The location of each \subjj and each \objj is expressed by the function $\locc: \Subj \cup \Obj \to \Levels$, which we tabulated in the same diagram. The current location of each \objj expresses its security \emph{classification}. Authorized \subjss can change their classification level by moving it to another security level, provided that they are cleared to access that level themselves. The \emph{clearance}\/ level of each \subjj is specified by the function $\clr: \Subj \to \Levels$, also tabulated in the diagram. A \subjj\ $u$ is cleared to access the security levels up to its clearance level $\clr_u$. This is formalized in the following definition.

\bigskip
\begin{definition}\label{Def:MLS} \sindex{security!multi-level} \sindex{MLS} \sindex{security!levels} \sindex{security!clearance} \sindex{security!location}
A \emph{multi-level security (MLS)}\/ model consists of:
\begin{itemize}
\item security types $\Subj$ and $\Obj$ like in Def.~\ref{Def:AC}, and moreover, the type of 
\begin{itemize}
\item \emph{security levels} $\Levels = \{\ell_1, \ell_2,\ldots \}$, partially ordered by $\leq$,
\end{itemize}
\item two \emph{security assignment functions}:
\begin{itemize}
\item \emph{clearance}\/ $\clr: \Subj \to \Levels$, and 
\item \emph{location}\/ $\loc: \Subj \cup \Obj \to \Levels$.
\end{itemize}
\end{itemize}
The \emph{multi-level security (MLS) requirement}\/ is that for all $u\in \Subj$ holds
\bea \label{eq:MLS}
\loc_u &\leq & \clr_u.
\eea
\end{definition}

For the Neolithic example above, the MLS requirement  means that Alice cannot enter Bob's vault because she is not cleared for anything above the level $\ell_1$.

While AC models and MLS models appear quite different and are used for different purposes, in Sec.~\ref{Sec:implicit}, we saw that every AC model contains an implicit MLS model. 
Thm.~\ref{Thm:AC-MLS} below tells that any AC model can, in fact, be expressed as 
an MLS model. The price to be paid is that capturing the distinct access control permissions for distinct \objs in terms of security levels may require lots of security levels.  Further down, Thm.~\ref{Thm:MLS-AC} will show that any MLS model, even extended by additional authority requirements, can be reduced to an AC model. 

\bigskip
\begin{theorem}\label{Thm:AC-MLS}
Given an AC model, let the type of security levels $\Levels$ be the set of functions $\ell_0, \ell_1,\ldots, : \Obj \to \WP\Act$, ordered by pointwise inclusion, i.e.
\[
\Levels \ = \ \left(\WP\Act\right)^\Obj \qquad \mbox{ with the partial order } \qquad \ell_0 \leq \ell_1  \iff  \forall i\in \Obj.\ \ell_0(i) \subseteq \ell_1(i).
\] 
Define the corresponding MLS model by:
\begin{align*}
\clr\ : \Subj & \to  \left(\WP\Act\right)^\Obj & \loc\ : \Subj & \to  \left(\WP\Act\right)^\Obj\\
u & \longmapsto  \clr_u(i) = M_{ui} & u & \longmapsto  \loc_u(i) = B_{ui}. 
\end{align*}
The MLS requirement for this MLS model is satisfied if and only if the AC requirement was satisfied for the original AC model. This means that for every {\subj} $u\in \Subj$ holds
\bear
\loc_u \leq \clr_u & \iff & \forall i\in \Obj.\ B_{ui}\subseteq M_{ui}.
\eear
\end{theorem}

The proof is easy, and it is left as an exercise.

\subsection{Reading and writing}\label{Sec:declassification}
The \obj classifications and the \subj clearances are established and managed by authorized \subjss. This is, in a sense, where security \emph{as a process}\/ begins.  
Classifying an \obj and requiring a clearance from a \subj are the most basic security operations. Declassifying an \obj or increasing a \subjj's clearance level are the simplest ways to relax security. Many different views of the problems of clearance and declassification, arising in many different contexts where these operations are needed, led to diverse but closely related security models and architectures, from which the early security research emerged \cite{Bell-Lapadula,Clark-Wilson,Biba,HRU,ChineseWall}. The process of classifying and declassifying \objs is modeled in terms of two basic actions which we write $\wt$ and $\rd$. Since most security models have been concerned with data, these are usually called \emph{write}\/ and \emph{read}; but the formalism applies to other action/reaction couples, such as give/take, and send/receive. It is often convenient and sometimes fun to call them all \emph{write/read}. The formal models in this area of security research largely subsume under the same structure that we will present, although the application domains vary vastly, and the interpretations are often genuinely different. We stick with a Neolithic one.

Suppose that Alice is leaving for a vacation and wants to "deposit" her sheep in Bob's secure vault (a Neolithic bank of sorts). Since Alice is not cleared to enter the secure vault, she \emph{gives}\/ the sheep to Bob, and he \emph{takes}\/ it to  a higher security level. When Alice returns from the vacation, Bob has to come out of the vault to return the sheep to the lower security level. This security process evolves through the states $q_0 \to q_1\to q_2\to q_3\to q_0$, displayed in \cref{Figure:sheep-machine}.

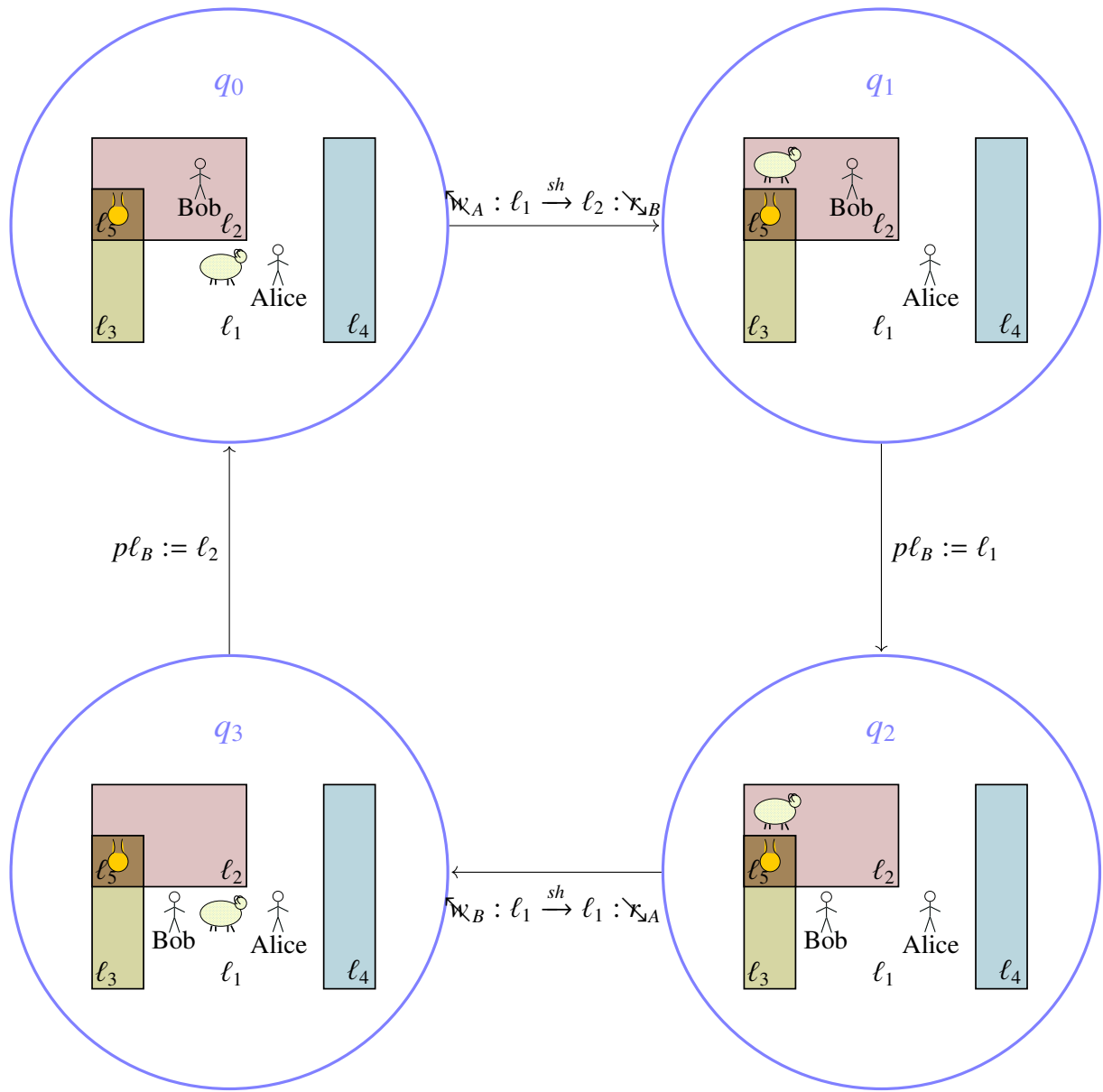
\begin{figure}[!ht]
\newcommand{\alice}{\mbox{\small Alice}}
\newcommand{\bob}{\mbox{\small Bob}}
\begin{center}
\begin{tikzpicture}[shorten >=1pt,node distance=9.5cm,on grid]
\tikzstyle{every state}=[draw=blue!50,very thick
]
\tikzstyle{initial}= [initial by arrow,initial text=]
\tikzstyle{every node}=[font=\normalsize]
  \node[state]   (ul)                {\begin{minipage}{4cm}
  \centering
  \mbox{\color{blue!50}\large $q_0$} \\[3ex]
  \def\JPicScale{.75}
\ifx\JPicScale\undefined\def\JPicScale{1}\fi
\psset{unit=\JPicScale mm}
\psset{linewidth=0.3,dotsep=1,hatchwidth=0.3,hatchsep=1.5,shadowsize=1,dimen=middle}
\psset{dotsize=0.7 2.5,dotscale=1 1,fillcolor=black}
\psset{arrowsize=1 2,arrowlength=1,arrowinset=0.25,tbarsize=0.7 5,bracketlength=0.15,rbracketlength=0.15}
\begin{pspicture}(0,0)(55,40)
\newrgbcolor{userFillColour}{0.87 0.76 0.76}
\pspolygon[fillcolor=userFillColour,fillstyle=solid](0,40)(30,40)(30,20)(0,20)
\newrgbcolor{userFillColour}{0.84 0.84 0.64}
\pspolygon[fillcolor=userFillColour,fillstyle=solid](0,0)(10,0)(10,30)(0,30)
\newrgbcolor{userFillColour}{0.64 0.52 0.35}
\pspolygon[fillcolor=userFillColour,fillstyle=solid](0,20)(10,20)(10,30)(0,30)
\newrgbcolor{userFillColour}{0.73 0.84 0.87}
\pspolygon[fillcolor=userFillColour,fillstyle=solid](45,40)(55,40)(55,0)(45,0)
\newrgbcolor{userFillColour}{0.77 0.89 0.88}
\rput[br](29.38,0.62){$\one$}
\newrgbcolor{userFillColour}{0.77 0.89 0.88}
\rput[br](29.38,20.62){$\two$}
\newrgbcolor{userFillColour}{0.77 0.89 0.88}
\rput[bl](0.62,0.62){$\three$}
\newrgbcolor{userFillColour}{0.77 0.89 0.88}
\rput[br](54,1){$\four$}
\rput{0}(36.12,18.12){\psellipse[linewidth=0.15](0,0)(1,-1)}
\newrgbcolor{userFillColour}{0.77 0.89 0.88}
\psline[linewidth=0.15,fillcolor=userFillColour,fillstyle=solid](36.12,17.12)(36.12,14.12)
\newrgbcolor{userFillColour}{0.77 0.89 0.88}
\psline[linewidth=0.15,fillcolor=userFillColour,fillstyle=solid](36.12,14.12)(35.12,11.12)
\newrgbcolor{userFillColour}{0.77 0.89 0.88}
\psline[linewidth=0.15,fillcolor=userFillColour,fillstyle=solid](36.12,14.12)(37.12,11.12)
\newrgbcolor{userFillColour}{0.77 0.89 0.88}
\psline[linewidth=0.15,fillcolor=userFillColour,fillstyle=solid](36.12,16.12)(38.12,15.12)
\newrgbcolor{userFillColour}{0.77 0.89 0.88}
\psline[linewidth=0.15,fillcolor=userFillColour,fillstyle=solid](36.12,16.12)(34.12,15.12)
\rput{0}(21,35){\psellipse[linewidth=0.15](0,0)(1,-1)}
\newrgbcolor{userFillColour}{0.77 0.89 0.88}
\psline[linewidth=0.15,fillcolor=userFillColour,fillstyle=solid](21,34)(21,31)
\newrgbcolor{userFillColour}{0.77 0.89 0.88}
\psline[linewidth=0.15,fillcolor=userFillColour,fillstyle=solid](21,31)(20,28)
\newrgbcolor{userFillColour}{0.77 0.89 0.88}
\psline[linewidth=0.15,fillcolor=userFillColour,fillstyle=solid](21,31)(22,28)
\newrgbcolor{userFillColour}{0.77 0.89 0.88}
\psline[linewidth=0.15,fillcolor=userFillColour,fillstyle=solid](21,33)(23,32)
\newrgbcolor{userFillColour}{0.77 0.89 0.88}
\psline[linewidth=0.15,fillcolor=userFillColour,fillstyle=solid](21,33)(19,32)
\newrgbcolor{userFillColour}{0.77 0.89 0.88}
\newrgbcolor{userHatchColour}{1 1 0.8}
\rput{0}(25,14.74){\psellipse[linewidth=0.05,fillcolor=userFillColour,fillstyle=crosshatch*,hatchwidth=0.05,hatchsep=0.3,hatchcolor=userHatchColour](0,0)(4,-2.49)}
\newrgbcolor{userFillColour}{0.77 0.89 0.88}
\psline[fillcolor=userFillColour,fillstyle=solid](23.25,12.38)(23,11.24)
\newrgbcolor{userFillColour}{0.77 0.89 0.88}
\psline[fillcolor=userFillColour,fillstyle=solid](24,12.25)(24,11.24)
\newrgbcolor{userFillColour}{0.77 0.89 0.88}
\psline[fillcolor=userFillColour,fillstyle=solid](27,12.38)(27.62,11.12)
\newrgbcolor{userFillColour}{0.77 0.89 0.88}
\psline[fillcolor=userFillColour,fillstyle=solid](26.38,12.38)(26.38,11.12)
\newrgbcolor{userFillColour}{0.77 0.89 0.88}
\newrgbcolor{userHatchColour}{1 1 0.8}
\rput{90}(28.88,16.75){\psellipse[linewidth=0.05,fillcolor=userFillColour,fillstyle=crosshatch*,hatchwidth=0.05,hatchsep=0.3,hatchcolor=userHatchColour](0,0)(1.19,-1.13)}
\rput{112.29}(28.25,17.37){\psellipticarc[linewidth=0.2](0,0)(0.85,-0.58){-25.6}{109.12}}
\rput{111.93}(28.88,17.37){\psellipticarc[linewidth=0.2](0,0)(0.85,-0.59){-25.77}{108.95}}
\newrgbcolor{userFillColour}{0.77 0.89 0.88}
\rput[t](36.25,10.5){$\alice$}
\newrgbcolor{userFillColour}{0.77 0.89 0.88}
\rput(20.88,26.5){$\bob$}
\newrgbcolor{userFillColour}{0.77 0.89 0.88}
\rput[bl](0.62,20.62){$\five$}
\newrgbcolor{userFillColour}{1 0.8 0}
\rput{90}(5.06,24.91){\psellipticarc[linewidth=0.1,fillcolor=userFillColour,fillstyle=solid](0,0)(1.97,1.94){29.04}{324.46}}
\newrgbcolor{userFillColour}{1 0.8 0}
\rput{108.02}(3.17,28.81){\psellipticarc[linewidth=0.1,fillcolor=userFillColour,fillstyle=solid](0,0)(2.42,1){179.98}{281.38}}
\newrgbcolor{userFillColour}{1 0.8 0}
\rput{77.86}(6.88,28.75){\psellipticarc[linewidth=0.1,fillcolor=userFillColour,fillstyle=solid](0,0)(2.33,1.12){71.11}{176.06}}
\end{pspicture}
\end{minipage}};
    \node[state]   (dl)  [below=of ul]              {\begin{minipage}{4cm}
  \centering
  \mbox{\color{blue!50}\large $q_3$} \\[3ex]
    \def\JPicScale{.75}
\ifx\JPicScale\undefined\def\JPicScale{1}\fi
\psset{unit=\JPicScale mm}
\psset{linewidth=0.3,dotsep=1,hatchwidth=0.3,hatchsep=1.5,shadowsize=1,dimen=middle}
\psset{dotsize=0.7 2.5,dotscale=1 1,fillcolor=black}
\psset{arrowsize=1 2,arrowlength=1,arrowinset=0.25,tbarsize=0.7 5,bracketlength=0.15,rbracketlength=0.15}
\begin{pspicture}(0,0)(55,40)
\newrgbcolor{userFillColour}{0.87 0.76 0.76}
\pspolygon[fillcolor=userFillColour,fillstyle=solid](0,40)(30,40)(30,20)(0,20)
\newrgbcolor{userFillColour}{0.84 0.84 0.64}
\pspolygon[fillcolor=userFillColour,fillstyle=solid](0,0)(10,0)(10,30)(0,30)
\newrgbcolor{userFillColour}{0.64 0.52 0.35}
\pspolygon[fillcolor=userFillColour,fillstyle=solid](0,20)(10,20)(10,30)(0,30)
\newrgbcolor{userFillColour}{0.73 0.84 0.87}
\pspolygon[fillcolor=userFillColour,fillstyle=solid](45,40)(55,40)(55,0)(45,0)
\newrgbcolor{userFillColour}{0.77 0.89 0.88}
\rput[br](29.38,0.62){$\one$}
\newrgbcolor{userFillColour}{0.77 0.89 0.88}
\rput[br](29.38,20.62){$\two$}
\newrgbcolor{userFillColour}{0.77 0.89 0.88}
\rput[bl](0.62,0.62){$\three$}
\newrgbcolor{userFillColour}{0.77 0.89 0.88}
\rput[br](54,1){$\four$}
\rput{0}(36.12,18.12){\psellipse[linewidth=0.15](0,0)(1,-1)}
\newrgbcolor{userFillColour}{0.77 0.89 0.88}
\psline[linewidth=0.15,fillcolor=userFillColour,fillstyle=solid](36.12,17.12)(36.12,14.12)
\newrgbcolor{userFillColour}{0.77 0.89 0.88}
\psline[linewidth=0.15,fillcolor=userFillColour,fillstyle=solid](36.12,14.12)(35.12,11.12)
\newrgbcolor{userFillColour}{0.77 0.89 0.88}
\psline[linewidth=0.15,fillcolor=userFillColour,fillstyle=solid](36.12,14.12)(37.12,11.12)
\newrgbcolor{userFillColour}{0.77 0.89 0.88}
\psline[linewidth=0.15,fillcolor=userFillColour,fillstyle=solid](36.12,16.12)(38.12,15.12)
\newrgbcolor{userFillColour}{0.77 0.89 0.88}
\psline[linewidth=0.15,fillcolor=userFillColour,fillstyle=solid](36.12,16.12)(34.12,15.12)
\rput{0}(16,18){\psellipse[linewidth=0.15](0,0)(1,-1)}
\newrgbcolor{userFillColour}{0.77 0.89 0.88}
\psline[linewidth=0.15,fillcolor=userFillColour,fillstyle=solid](16,17)(16,14)
\newrgbcolor{userFillColour}{0.77 0.89 0.88}
\psline[linewidth=0.15,fillcolor=userFillColour,fillstyle=solid](16,14)(15,11)
\newrgbcolor{userFillColour}{0.77 0.89 0.88}
\psline[linewidth=0.15,fillcolor=userFillColour,fillstyle=solid](16,14)(17,11)
\newrgbcolor{userFillColour}{0.77 0.89 0.88}
\psline[linewidth=0.15,fillcolor=userFillColour,fillstyle=solid](16,16)(18,15)
\newrgbcolor{userFillColour}{0.77 0.89 0.88}
\psline[linewidth=0.15,fillcolor=userFillColour,fillstyle=solid](16,16)(14,15)
\newrgbcolor{userFillColour}{0.77 0.89 0.88}
\newrgbcolor{userHatchColour}{1 1 0.8}
\rput{0}(25,14.74){\psellipse[linewidth=0.05,fillcolor=userFillColour,fillstyle=crosshatch*,hatchwidth=0.05,hatchsep=0.3,hatchcolor=userHatchColour](0,0)(4,-2.49)}
\newrgbcolor{userFillColour}{0.77 0.89 0.88}
\psline[fillcolor=userFillColour,fillstyle=solid](23.25,12.38)(23,11.24)
\newrgbcolor{userFillColour}{0.77 0.89 0.88}
\psline[fillcolor=userFillColour,fillstyle=solid](24,12.25)(24,11.24)
\newrgbcolor{userFillColour}{0.77 0.89 0.88}
\psline[fillcolor=userFillColour,fillstyle=solid](27,12.38)(27.62,11.12)
\newrgbcolor{userFillColour}{0.77 0.89 0.88}
\psline[fillcolor=userFillColour,fillstyle=solid](26.38,12.38)(26.38,11.12)
\newrgbcolor{userFillColour}{0.77 0.89 0.88}
\newrgbcolor{userHatchColour}{1 1 0.8}
\rput{90}(28.88,16.75){\psellipse[linewidth=0.05,fillcolor=userFillColour,fillstyle=crosshatch*,hatchwidth=0.05,hatchsep=0.3,hatchcolor=userHatchColour](0,0)(1.19,-1.13)}
\rput{112.29}(28.25,17.37){\psellipticarc[linewidth=0.2](0,0)(0.85,-0.58){-25.6}{109.12}}
\rput{111.93}(28.88,17.37){\psellipticarc[linewidth=0.2](0,0)(0.85,-0.59){-25.77}{108.95}}
\newrgbcolor{userFillColour}{0.77 0.89 0.88}
\rput[t](36.25,10.5){$\alice$}
\newrgbcolor{userFillColour}{0.77 0.89 0.88}
\rput(15.88,9.5){$\bob$}
\newrgbcolor{userFillColour}{0.77 0.89 0.88}
\rput[bl](0.62,20.62){$\five$}
\newrgbcolor{userFillColour}{1 0.8 0}
\rput{90}(5.06,24.91){\psellipticarc[linewidth=0.1,fillcolor=userFillColour,fillstyle=solid](0,0)(1.97,1.94){29.04}{324.46}}
\newrgbcolor{userFillColour}{1 0.8 0}
\rput{108.03}(3.17,28.81){\psellipticarc[linewidth=0.1,fillcolor=userFillColour,fillstyle=solid](0,0)(2.42,1){179.97}{281.38}}
\newrgbcolor{userFillColour}{1 0.8 0}
\rput{77.96}(6.88,28.75){\psellipticarc[linewidth=0.1,fillcolor=userFillColour,fillstyle=solid](0,0)(2.33,1.12){71.07}{176.01}}
\end{pspicture}
\end{minipage}};
  \node[state]           (ur) [right=of ul] {\begin{minipage}{4cm}
  \centering
  \mbox{\color{blue!50}\large $q_1$} \\[3ex]
  \def\JPicScale{.75}
\ifx\JPicScale\undefined\def\JPicScale{1}\fi
\psset{unit=\JPicScale mm}
\psset{linewidth=0.3,dotsep=1,hatchwidth=0.3,hatchsep=1.5,shadowsize=1,dimen=middle}
\psset{dotsize=0.7 2.5,dotscale=1 1,fillcolor=black}
\psset{arrowsize=1 2,arrowlength=1,arrowinset=0.25,tbarsize=0.7 5,bracketlength=0.15,rbracketlength=0.15}
\begin{pspicture}(0,0)(55,40)
\newrgbcolor{userFillColour}{0.87 0.76 0.76}
\pspolygon[fillcolor=userFillColour,fillstyle=solid](0,40)(30,40)(30,20)(0,20)
\newrgbcolor{userFillColour}{0.84 0.84 0.64}
\pspolygon[fillcolor=userFillColour,fillstyle=solid](0,0)(10,0)(10,30)(0,30)
\newrgbcolor{userFillColour}{0.64 0.52 0.35}
\pspolygon[fillcolor=userFillColour,fillstyle=solid](0,20)(10,20)(10,30)(0,30)
\newrgbcolor{userFillColour}{0.73 0.84 0.87}
\pspolygon[fillcolor=userFillColour,fillstyle=solid](45,40)(55,40)(55,0)(45,0)
\newrgbcolor{userFillColour}{0.77 0.89 0.88}
\rput[br](29.38,0.62){$\one$}
\newrgbcolor{userFillColour}{0.77 0.89 0.88}
\rput[br](29.38,20.62){$\two$}
\newrgbcolor{userFillColour}{0.77 0.89 0.88}
\rput[bl](0.62,0.62){$\three$}
\newrgbcolor{userFillColour}{0.77 0.89 0.88}
\rput[br](54,1){$\four$}
\rput{0}(21,35){\psellipse[linewidth=0.15](0,0)(1,-1)}
\newrgbcolor{userFillColour}{0.77 0.89 0.88}
\psline[linewidth=0.15,fillcolor=userFillColour,fillstyle=solid](21,34)(21,31)
\newrgbcolor{userFillColour}{0.77 0.89 0.88}
\psline[linewidth=0.15,fillcolor=userFillColour,fillstyle=solid](21,31)(20,28)
\newrgbcolor{userFillColour}{0.77 0.89 0.88}
\psline[linewidth=0.15,fillcolor=userFillColour,fillstyle=solid](21,31)(22,28)
\newrgbcolor{userFillColour}{0.77 0.89 0.88}
\psline[linewidth=0.15,fillcolor=userFillColour,fillstyle=solid](21,33)(23,32)
\newrgbcolor{userFillColour}{0.77 0.89 0.88}
\psline[linewidth=0.15,fillcolor=userFillColour,fillstyle=solid](21,33)(19,32)
\newrgbcolor{userFillColour}{0.77 0.89 0.88}
\newrgbcolor{userHatchColour}{1 1 0.8}
\rput{0}(6,34.75){\psellipse[linewidth=0.05,fillcolor=userFillColour,fillstyle=crosshatch*,hatchwidth=0.05,hatchsep=0.3,hatchcolor=userHatchColour](0,0)(4,-2.5)}
\newrgbcolor{userFillColour}{0.77 0.89 0.88}
\psline[fillcolor=userFillColour,fillstyle=solid](4.25,32.38)(4,31.24)
\newrgbcolor{userFillColour}{0.77 0.89 0.88}
\psline[fillcolor=userFillColour,fillstyle=solid](5,32.25)(5,31.24)
\newrgbcolor{userFillColour}{0.77 0.89 0.88}
\psline[fillcolor=userFillColour,fillstyle=solid](8,32.38)(8.62,31.12)
\newrgbcolor{userFillColour}{0.77 0.89 0.88}
\psline[fillcolor=userFillColour,fillstyle=solid](7.38,32.38)(7.38,31.12)
\newrgbcolor{userFillColour}{0.77 0.89 0.88}
\newrgbcolor{userHatchColour}{1 1 0.8}
\rput{90}(9.88,36.75){\psellipse[linewidth=0.05,fillcolor=userFillColour,fillstyle=crosshatch*,hatchwidth=0.05,hatchsep=0.3,hatchcolor=userHatchColour](0,0)(1.19,-1.12)}
\rput{112.29}(9.25,37.37){\psellipticarc[linewidth=0.2](0,0)(0.85,-0.58){-25.6}{109.12}}
\rput{111.93}(9.88,37.37){\psellipticarc[linewidth=0.2](0,0)(0.85,-0.59){-25.77}{108.95}}
\newrgbcolor{userFillColour}{0.77 0.89 0.88}
\rput(20.88,26.5){$\bob$}
\newrgbcolor{userFillColour}{0.77 0.89 0.88}
\rput[bl](0.62,20.62){$\five$}
\rput{0}(36.12,18.12){\psellipse[linewidth=0.15](0,0)(1,-1)}
\newrgbcolor{userFillColour}{0.77 0.89 0.88}
\psline[linewidth=0.15,fillcolor=userFillColour,fillstyle=solid](36.12,17.12)(36.12,14.12)
\newrgbcolor{userFillColour}{0.77 0.89 0.88}
\psline[linewidth=0.15,fillcolor=userFillColour,fillstyle=solid](36.12,14.12)(35.12,11.12)
\newrgbcolor{userFillColour}{0.77 0.89 0.88}
\psline[linewidth=0.15,fillcolor=userFillColour,fillstyle=solid](36.12,14.12)(37.12,11.12)
\newrgbcolor{userFillColour}{0.77 0.89 0.88}
\psline[linewidth=0.15,fillcolor=userFillColour,fillstyle=solid](36.12,16.12)(38.12,15.12)
\newrgbcolor{userFillColour}{0.77 0.89 0.88}
\psline[linewidth=0.15,fillcolor=userFillColour,fillstyle=solid](36.12,16.12)(34.12,15.12)
\newrgbcolor{userFillColour}{0.77 0.89 0.88}
\rput[t](36.25,10.5){$\alice$}
\newrgbcolor{userFillColour}{1 0.8 0}
\rput{90}(5.06,24.91){\psellipticarc[linewidth=0.1,fillcolor=userFillColour,fillstyle=solid](0,0)(1.97,1.94){29.04}{324.46}}
\newrgbcolor{userFillColour}{1 0.8 0}
\rput{108.02}(3.17,28.81){\psellipticarc[linewidth=0.1,fillcolor=userFillColour,fillstyle=solid](0,0)(2.42,1){179.98}{281.38}}
\newrgbcolor{userFillColour}{1 0.8 0}
\rput{77.86}(6.88,28.75){\psellipticarc[linewidth=0.1,fillcolor=userFillColour,fillstyle=solid](0,0)(2.33,1.12){71.11}{176.06}}
\end{pspicture}
\end{minipage}};
    \node[state]           (dr) [right=of dl] {\begin{minipage}{4cm}
  \centering
  \mbox{\color{blue!50}\large $q_2$} \\[3ex]
    \def\JPicScale{.75}
\ifx\JPicScale\undefined\def\JPicScale{1}\fi
\psset{unit=\JPicScale mm}
\psset{linewidth=0.3,dotsep=1,hatchwidth=0.3,hatchsep=1.5,shadowsize=1,dimen=middle}
\psset{dotsize=0.7 2.5,dotscale=1 1,fillcolor=black}
\psset{arrowsize=1 2,arrowlength=1,arrowinset=0.25,tbarsize=0.7 5,bracketlength=0.15,rbracketlength=0.15}
\begin{pspicture}(0,0)(55,40)
\newrgbcolor{userFillColour}{0.87 0.76 0.76}
\pspolygon[fillcolor=userFillColour,fillstyle=solid](0,40)(30,40)(30,20)(0,20)
\newrgbcolor{userFillColour}{0.84 0.84 0.64}
\pspolygon[fillcolor=userFillColour,fillstyle=solid](0,0)(10,0)(10,30)(0,30)
\newrgbcolor{userFillColour}{0.64 0.52 0.35}
\pspolygon[fillcolor=userFillColour,fillstyle=solid](0,20)(10,20)(10,30)(0,30)
\newrgbcolor{userFillColour}{0.73 0.84 0.87}
\pspolygon[fillcolor=userFillColour,fillstyle=solid](45,40)(55,40)(55,0)(45,0)
\newrgbcolor{userFillColour}{0.77 0.89 0.88}
\rput[br](29.38,0.62){$\one$}
\newrgbcolor{userFillColour}{0.77 0.89 0.88}
\rput[br](29.38,20.62){$\two$}
\newrgbcolor{userFillColour}{0.77 0.89 0.88}
\rput[bl](0.62,0.62){$\three$}
\newrgbcolor{userFillColour}{0.77 0.89 0.88}
\rput[br](54,1){$\four$}
\newrgbcolor{userFillColour}{0.77 0.89 0.88}
\newrgbcolor{userHatchColour}{1 1 0.8}
\rput{0}(6,34.75){\psellipse[linewidth=0.05,fillcolor=userFillColour,fillstyle=crosshatch*,hatchwidth=0.05,hatchsep=0.3,hatchcolor=userHatchColour](0,0)(4,-2.5)}
\newrgbcolor{userFillColour}{0.77 0.89 0.88}
\psline[fillcolor=userFillColour,fillstyle=solid](4.25,32.38)(4,31.24)
\newrgbcolor{userFillColour}{0.77 0.89 0.88}
\psline[fillcolor=userFillColour,fillstyle=solid](5,32.25)(5,31.24)
\newrgbcolor{userFillColour}{0.77 0.89 0.88}
\psline[fillcolor=userFillColour,fillstyle=solid](8,32.38)(8.62,31.12)
\newrgbcolor{userFillColour}{0.77 0.89 0.88}
\psline[fillcolor=userFillColour,fillstyle=solid](7.38,32.38)(7.38,31.12)
\newrgbcolor{userFillColour}{0.77 0.89 0.88}
\newrgbcolor{userHatchColour}{1 1 0.8}
\rput{90}(9.88,36.75){\psellipse[linewidth=0.05,fillcolor=userFillColour,fillstyle=crosshatch*,hatchwidth=0.05,hatchsep=0.3,hatchcolor=userHatchColour](0,0)(1.19,-1.12)}
\rput{112.29}(9.25,37.37){\psellipticarc[linewidth=0.2](0,0)(0.85,-0.58){-25.6}{109.12}}
\rput{111.93}(9.88,37.37){\psellipticarc[linewidth=0.2](0,0)(0.85,-0.59){-25.77}{108.95}}
\newrgbcolor{userFillColour}{0.77 0.89 0.88}
\rput[bl](0.62,20.62){$\five$}
\rput{0}(36.12,18.12){\psellipse[linewidth=0.15](0,0)(1,-1)}
\newrgbcolor{userFillColour}{0.77 0.89 0.88}
\psline[linewidth=0.15,fillcolor=userFillColour,fillstyle=solid](36.12,17.12)(36.12,14.12)
\newrgbcolor{userFillColour}{0.77 0.89 0.88}
\psline[linewidth=0.15,fillcolor=userFillColour,fillstyle=solid](36.12,14.12)(35.12,11.12)
\newrgbcolor{userFillColour}{0.77 0.89 0.88}
\psline[linewidth=0.15,fillcolor=userFillColour,fillstyle=solid](36.12,14.12)(37.12,11.12)
\newrgbcolor{userFillColour}{0.77 0.89 0.88}
\psline[linewidth=0.15,fillcolor=userFillColour,fillstyle=solid](36.12,16.12)(38.12,15.12)
\newrgbcolor{userFillColour}{0.77 0.89 0.88}
\psline[linewidth=0.15,fillcolor=userFillColour,fillstyle=solid](36.12,16.12)(34.12,15.12)
\newrgbcolor{userFillColour}{0.77 0.89 0.88}
\rput[t](36.25,10.5){$\alice$}
\newrgbcolor{userFillColour}{1 0.8 0}
\rput{90}(5.06,24.91){\psellipticarc[linewidth=0.1,fillcolor=userFillColour,fillstyle=solid](0,0)(1.97,1.94){29.04}{324.46}}
\newrgbcolor{userFillColour}{1 0.8 0}
\rput{108.03}(3.17,28.81){\psellipticarc[linewidth=0.1,fillcolor=userFillColour,fillstyle=solid](0,0)(2.42,1){179.97}{281.38}}
\newrgbcolor{userFillColour}{1 0.8 0}
\rput{77.96}(6.88,28.75){\psellipticarc[linewidth=0.1,fillcolor=userFillColour,fillstyle=solid](0,0)(2.33,1.12){71.07}{176.01}}
\rput{0}(16,18){\psellipse[linewidth=0.15](0,0)(1,-1)}
\newrgbcolor{userFillColour}{0.77 0.89 0.88}
\psline[linewidth=0.15,fillcolor=userFillColour,fillstyle=solid](16,17)(16,14)
\newrgbcolor{userFillColour}{0.77 0.89 0.88}
\psline[linewidth=0.15,fillcolor=userFillColour,fillstyle=solid](16,14)(15,11)
\newrgbcolor{userFillColour}{0.77 0.89 0.88}
\psline[linewidth=0.15,fillcolor=userFillColour,fillstyle=solid](16,14)(17,11)
\newrgbcolor{userFillColour}{0.77 0.89 0.88}
\psline[linewidth=0.15,fillcolor=userFillColour,fillstyle=solid](16,16)(18,15)
\newrgbcolor{userFillColour}{0.77 0.89 0.88}
\psline[linewidth=0.15,fillcolor=userFillColour,fillstyle=solid](16,16)(14,15)
\newrgbcolor{userFillColour}{0.77 0.89 0.88}
\rput(15.88,9.5){$\bob$}
\end{pspicture}
\end{minipage}};
  \path[->] 
  (ul) 
  edge node [above] {$\wt_A : \one\tto{sh} \two : \rd_B$}  (ur)
    (ur)
  edge node [right] {$\loc_B := \one$} (dr)
 (dl) 
  edge node [left] {$\loc_B := \two$}  (ul)
  (dr)
  edge node [below] {$\wt_B : \one\tto{sh} \one : \rd_A$} (dl);
\end{tikzpicture}
\caption{The process of storing and withdrawing Alice's sheep from Bob's secure vault}
\label{Figure:sheep-machine}
\end{center}
\end{figure}
The pairs of complementary actions analogous to the \emph{give/take}, applied to sheep, are usually called \emph{write/read}\/ when applied to data. Alice, in a sense, ``writes'' her sheep up the security ladder, but she is not permitted to ``read'' the sheep back down from the higher security level of the vault to her lower clearance level. Bob does have a clearance to be in his vault, but ``writing'' down is also not generally permitted since that would allow a malicious insider with a high clearance to declassify everything. The general principle is
\begin{itemize}
\item \emph{\textbf{no-read-up}}: 
\begin{itemize}
\item only receive data at or below your clearance level
\end{itemize}
and
\item \emph{\textbf{no-write-down}}: 
\begin{itemize}
\item only send data at or above your classification level.
\end{itemize}\end{itemize}
To remember this, we write $\rd$ and  $\wt$ for the actions of reading and writing data (or taking and giving objects), respectively. 

To return the sheep to Alice, Bob has to descend from his vault $\ell_{2}$ to the level $\ell_{1}$, where he can ``write'' the sheep to Alice. The state transition $\wt_A: \one \tto{sh} \two: \rd_B$ is an \emph{interaction}\/ between Alice whose action $\wt_A: \one \tto{sh} \two$ elicits Bob's reaction $\one  
\tto{?} \two : \rd_B$. The question mark means that Bob is ready to take on this channel, whatever Alice gives. Alice writes the sheep from $\one$ to $\two$, and Bob accepts to read on the same levels.  Similar couplings arise, e.g., when Alice \emph{sends}\/ a message, and Bob \emph{receives}\/ it. Alice's action and Bob's reaction are coupled into an interaction. The action and the interaction are bound together with an object: in the above case the sheep, and in the case of the send/receive cases the message. For the moment, though, we are only interested in the fact that security is this \emph{process}, evolving from state to state and from transition to transition:
\begin{itemize}
\item \textbf{$q_0$:} The crate of oil is highly classified. \emph{The sheep is unclassified}.
\begin{itemize}
\item \textbf{$\wt_A : \one\tto{sh} \two : \rd_B$: } Alice writes up, Bob reads down: \emph{they classify the sheep}.
\end{itemize}

\item\textbf{$q_1$:} Alice cannot read up, Bob cannot write down: \emph{the sheep is classified}.
\begin{itemize}
\item \textbf{$\clr_B := \one$ : } To be able to return the sheep, Bob descends to Alice's security level.
\end{itemize}

\item\textbf{$q_2$: } Bob is cleared to read the sheep in his vault; Alice is not.
\begin{itemize}
\item \textbf{$\wt_B : \one\tto{sh} \one : \rd_A$: } Bob writes and Alice reads at the same level: \emph{they declassify the sheep}.
\end{itemize}

\item\textbf{$q_3$: } All (except oil) are insecure.
\begin{itemize}
\item \textbf{$\clr_B := \two$ : } Bob is back to security. \emph{The sheep is unclassified}.
\end{itemize}

\end{itemize}

\bigskip
\begin{definition}\label{Def:SM}
An abstract \emph{(resource) security}\/ model consists of:
\begin{itemize}
\item security types $\Subj$, $\Obj$, $\Act$, and $\Levels$, as in Definitions~\ref{Def:AC} and \ref{Def:MLS}, together with
\begin{itemize}
\item distinguished actions $\rd, \wt \in \Act$, called respectively \emph{read}\/ and \emph{write}, 
\end{itemize}

\item the AC and the MLS structures:
\begin{itemize}
\item permission and access matrices $M,B: \Subj \times \Obj \to \WP\Act$, 
\item clearance and location assignments $\clr, \loc: \Subj \to \Levels$,
\item 
      location assignment $\loc: \Obj \to \Levels$.
\end{itemize}
\end{itemize}
The \emph{no-read-up}\/ and \emph{no-write-down}\/ requirements are respectively
\bea
\rd \in B_{ui} & \Longrightarrow & \clr_u \geq \loc_i, \label{eq:nru}\\
\wt \in B_{ui}  & \Longrightarrow & \loc_u \leq \loc_i. \label{eq:nrd}
\eea 
The states where resource security models satisfy the no-read-up and the no-write-down requirements are often called \emph{secure states}. 
\end{definition}

\para{Duality of reading and writing.} Requirements \eqref{eq:nru} and \eqref{eq:nrd} can be written together:
\[ \loc_u \ \stackrel{\wtt} \leq \ \loc_i  \ \stackrel{\rdd} \leq \ \clr_u.
\]
Bob is thus permitted to both read and write the sheep if the sheep's classification is in the interval between Bob's location and his clearance, i.e., if $\loc_{sh} \in  [\loc_B, \clr_B]$. In the early days of security research, the duality of reading and writing was viewed as echoing the duality of the confidentiality and integrity pair~\cite{Biba}. It is probably more justified to view it in terms of the duality of authority and availability.

\para{Remarks.} The term \emph{"secure state"}\/ does not imply that the state satisfies any particular security requirement. A security process may satisfy the no-read-up and no-write-down requirements and still contain a state where all subjects have the top clearance, and all objects are declassified to the lowest security level so that everyone can access all resources. Some subjects would consider such a state undesirable and completely insecure, while others would consider it desirable and completely secure. Secure states are secure only with respect to the two particular, very basic security requirements: no-read-up and no-write-down. More refined security requirements capture more refined concepts of security.

\subsection{All resource security boils down to access control}\label{Sec:MLSisAC}

Since every authorization model subsumes an AC model, extends it by an MLS model, and moreover imposes the no-read-up and no-write-down requirements on the combination, it seems that the authorization model formalism is more expressive than the AC model formalism. The next theorem says that this is not the case: every resource security policy specified by an authorization model can be equivalently expressed as an access control policy.

\bigskip
\begin{theorem}\label{Thm:MLS-AC}
For every authorization model
\begin{itemize}
\item $\PM, \AM :\Subj\times \Obj \to \WP\Act$
\item $\clr: \Subj \to \Levels$ and $\loc: \Subj \cup \Obj \to \Levels$
\end{itemize}
there is an access control model
\begin{itemize}
\item $\widehat \PM, \widehat \AM :\Subj\times \Obj \to \WP\widehat \Act$
\end{itemize}
such that
\bea\label{eq:propone}
\left(\begin{split} 
\AM_{ui}\subseteq \PM_{ui}\ \wedge\\  
\loc_u\leq \clr_u\ \wedge \\
\rd\in \AM_{ui} \Rightarrow  \clr_u \geq \loc_i\ \wedge\\ 
\wt\in \AM_{ui} \Rightarrow \loc_u\leq \loc_i\ \ \ 
\end{split}\right)  & \iff & \Big(\widehat{\AM}_{ui}\subseteq \widehat{\PM}_{ui}\Big)\eea
\end{theorem}

\bpr
Define the new set of actions as the disjoint union $\widehat \Act\   = \  \Act + \Levels$, so that $\WP \widehat \Act \cong \WP \Act \times \WP \Levels$. Define furthermore the new permission and access matrices $\widehat \PM, \widehat \AM : \Subj\times \Obj \to \WP\widehat \Act$ to have the following entries:  
\begin{align}\label{eq:MB}
\widehat \PM_{ui} & = \ \coprod_{a\in \PM_{ui}} \PML_{ui}(a) &
\widehat \AM_{ui} & = \ \coprod_{a\in \AM_{ui}} \AML_{ui}(a)
\end{align}
where $\coprod$s denote the disjoint unions of the families
\begin{align}
\PML_{ui}(a) & =  \left.\begin{cases} \nabla\clr_u \cap \nabla\loc_i & \mbox{ if } a = \wt\\
\nabla\clr_u& \mbox{ otherwise} \end{cases}\right\} &
\AML_{ui}(a) & =  \left.\begin{cases} \nabla\loc_u \cup \nabla\loc_i & \mbox{ if } a = \rd\\
\nabla\loc_u& \mbox{ otherwise} \end{cases}\right\}
\end{align}
and $\nabla \ell = \{x\in \Levels\ |\ x\leq \ell\}$ is the set of security levels at or below $\ell$. Using the fact that $\ell \leq \ell'$ holds if and only if $\nabla \ell \subseteq \nabla\ell'$ holds, the direction $\Longrightarrow$ of \eqref{eq:propone} is straightforward. 

Towards the direction $\Longleftarrow$, first note that  $\widehat \AM_{ui}\subseteq \widehat \PM_{ui}$ always implies $\AM_{ui}\subseteq \PM_{ui}$. To see why, first note that, by definition in \eqref{eq:MB}, an element of $\widehat \AM_{ui}$ is a pair $<a,\elll>$ where $a\in \AM_{ui}$ and $\ell \in \beta_{ui}(a)$. The inclusion $\widehat \AM_{ui}\subseteq \widehat \PM_{ui}$ implies that $<a,\ell>\in \PM_{ui}$, i.e., $a\in \PM_{ui}$ and $\ell \in \mu_{ui}(a)$. So if we take an arbitrary $a\in \AM_{ui}$, then $\ell = \loc_{u}\in \beta_{ui}(a)$ gives $<a,\ell> \in \widehat \AM_{ui}\subseteq \widehat \PM_{ui}$, which implies  $a\in \PM_{ui}$. Hence, $\AM_{ui}\subseteq \PM_{ui}$, since $a$ was arbitrary.

Furthermore, whenever $\PM_{ui}$ does not contain either  $\rd$ or $\wt$, then $\AM_{ui}\subseteq \PM_{ui}$ does not contain them either, and therefore
\begin{align*}
\widehat \PM_{ui} & = \ \coprod_{a\in \PM_{ui}} \nabla\clr_u &
\widehat \AM_{ui} & = \ \coprod_{a\in \AM_{ui}} \nabla\loc_u
\end{align*}
Using the fact that $\loc_u \leq \clr_u$ holds if and only if $\nabla \loc_u \subseteq \nabla \clr_u$, we get
\bear
\AM_{ui}\subseteq \PM_{ui}\  \wedge\ \loc_u \leq \clr_u & \iff & \widehat \AM_{ui}\subseteq \widehat \PM_{ui} 
\eear
This means that for the case without $\rd$ or $\wt$,  the claim is proven, because the two bottom 
conjuncts on the left-hand side of \eqref{eq:propone} are true trivially since the premises $\rd \in \AM_{ui}$ and $\wt\in \AM_{ui}$ are false.

When $\rd \in \AM_{ui}$, then $\widehat \AM_{ui}\subseteq \widehat \PM_{ui}$ implies $\AML_{ui}(\rd) \subseteq \PML_{ui}(\rd)$, i.e., $\nabla\loc_u \cup \nabla\loc_i \subseteq \nabla \clr_u$. Hence, $\nabla\loc_i \subseteq \nabla \clr_u$ and thus $\loc_i\leq \clr_u$, so the no-read-up requirement, third conjunct on the left-hand side of \eqref{eq:propone} is proven.

Finally, when $\wt \in \AM_{ui}$, then $\widehat \AM_{ui}\subseteq \widehat \PM_{ui}$ implies $\AML_{ui}(\wt) \subseteq \PML_{ui}(\wt)$, which means that $\nabla\loc_u \subseteq \nabla \clr_u\cap \nabla \loc_i$. Hence, $\nabla\loc_u \subseteq \nabla \loc_i$ and thus $\loc_u\leq \loc_i$. This means that the no-write-down requirement, the fourth conjunct on the left-hand side of \eqref{eq:propone}, is also proven.
\epr

\para{Remark.} Theorem~\ref{Thm:AC-MLS} showed how any AC policy can be expressed as an MLS policy. The price to be paid was a large and usually unintuitive poset of security levels. Theorem~\ref{Thm:MLS-AC} showed how any authorization policy (and thus any MLS policy as a special case) can be expressed as an AC policy. The price to be paid in this direction is that the security levels are captured as actions. This is not so unintuitive since each security level can be construed as the action of accessing that level. In practice, though, separating actions and security levels is usually more convenient than bundling them together. 

The upshot of Theorems~\ref{Thm:AC-MLS} and \ref{Thm:MLS-AC} is that the apparently different languages of AC, MLS and authorization are logically equivalent. However, it is useful to have different languages that express the same things.

\def\thechapter{4}
\setchaptertoc
\chapter{Dynamic resource security: authorization and availability}\label{Chap:Process}

\section{Histories, properties, safety, liveness}\label{Sec:Traces}

\subsection{What happens: Sequences of events and properties}\label{Sec:histories}
In~\cref{Fig:AliceBobStart}, different states of the world are presented as different permission matrices. In~\cref{Figure:sheep-machine}, transitions between the different states arise from subjects' actions. Now, we need to study what happens as time goes by. 

\para{Notation.} Our notational conventions for lists and strings can be found in Prerequisites \ref{Prereq:list}.

\begin{definition}\label{Def:property}
For a type $\Event$ \sindex{event} of \emph{events}, the finite sequences
\bear
\vec x = \seq{x_0 \, x_1 \,  \ldots \,  x_n} & \in & \Event^\ast
\eear
are called the $\Event$-\emph{traces}, or \sindex{history} \sindex{trace|see{history}}\emph{histories}, or \emph{contexts}\sindex{context|see{history}}. A \emph{$\Event$-property}, \sindex{property} is expressed as a set of histories 
\bear
P & \subseteq &\Event^\ast.
\eear
\end{definition}

\para{Remark.} The words \emph{``trace''}, \emph{``history''}, and \emph{``context''}\/ are used in different communities to denote the same thing: a sequence of events. ``Traces'' are used in automata theory, ``histories'' in process theory, and ``contexts'' in computational linguistics. Sequences of events are also studied in other research areas and perhaps called different names. The notion of events in time is ubiquitous. 

\para{Trivial properties.} \sindex{property!trivial} The simplest properties are the set of all traces $\Event^\ast$, and the empty set $\emptyset$. 
If properties are thought of as requirement constraints, then $\Event^\ast$ is the least constraining requirement since every history satisfies it, whereas $\emptyset$ is the impossible requirement, satisfied. A singleton $\left\{\vec t\right\} \subseteq  \Event^\ast$ is a maximally constraining non-trivial requirement. The larger properties are less constraining. 

\para{Standard and refined $\Event$ types.} In most models, an event is an action of a subject on an object. Formally, an event $x\in \Event$ is thus a triple $x = <a,i,u>$, and the event type is the product 
\bea\label{eq:event}
\Event & = & \Act \times \Obj \times \Subj
\eea 
However, not all actions are applicable to all objects by all subjects. Alice can drive a car and eat an apple, but she cannot drive an apple or eat a car. And Bob might be prohibited from driving Alice's car or eating her apple. Formally, a subject $u\in \Subj$ is often given a specified subtype of actions $\Act_{ui} \subseteq \Act$ that they can apply to an object $i\in \Obj$. Sometimes $\Act_{ui}$ is empty, which means that the subject $u$ is not permitted to use the object $i$. That can also be captured by specifying a subtype $\Obj_u \subseteq \Obj$ of objects accessible by the subject $u\in \Subj$., The set of events that may actually occur can then be expressed more explicitly as the disjoint union\footnote{In dependent type theory, this type would be written as $\sum{u: \Subj}\  \sum{i: \Obj(u)}.\ \  \Act({u,i})$. The notational clash of the symbol $\Sigma$ used for labels, alphabets, and events in semantics of computation, and for sum types in type theory, is accidental.
}
\bea\label{eq:event-deptype}
\widetilde{\Event} & = & \coprod_{u\in \Subj} \coprod_{i \in \Obj_u} \Act_{ui}
\eea 
But since $\widetilde \Event \subseteq
 \Event$, the crude form \eqref{eq:event} generally suffices, and the more precise form  \eqref{eq:event-deptype} is used when additional clarity is needed, e.g., in Sec.~\ref{Sec:au-av}.

\para{Access matrices as event spaces.} If an event is an action of  a subject on an object, then a permission matrix corresponds to a set of permitted events along the one-to-one correspondence 
\[
\prooftree
\Subj \times \Obj \tto M \WP\Act
\justifies
\widehat{M} \subseteq \Event
\endprooftree
\] 
defined by
\bear
<a,i,u> \in \widehat M & \iff & a \in M_{ui}
\eear
Since any access matrix $\Subj \times \Obj \tto{B} \WP\Act$ similarly corresponds to a set of events $\widehat B^q\subseteq \Event$, the access control security requirement  in \eqref{eq:AC} now becomes $\widehat B^q \subseteq \widehat M$. The access matrices from Ch.~\ref{Chap:Resource} correspond to subsets of events. At each moment in time, the matrix view, and the event space view are equivalent. The latter is, however, more convenient when we need to take a "historic perspective", and study resource security of processes in time.

To get going, let us take a look at a couple of examples of histories and properties. In general, events are actions of subjects on objects. But when subjects don't matter, or when there is just one, then the events are actions on objects. When objects don't matter, or there is just one, then the events are just actions. We begin from this simplest case, and then broaden to richer frameworks.

\para{Example 1: sheep.} We begin with Alice as the only subject, and her sheep as the only object, so that an event is just one of Alice's actions on the sheep, i.e., 
\bear
\Event\  \ =\ \  \Act  & = &   \{\mbox{milk}, \mbox{wool}, \mbox{meat}\}.
\eear 
The trace 
\bear
\vec m & = & \seq{ \milk \ \milk \ \milk \ \milk \ \mbox{wool} \ \milk } 
\eear
means that Alice milked her sheep 4 times, then sheared the wool from the sheep, and then milked her again. Consider the following trace properties: 
\bea
\mbox{MilkWool} & = & \{\milk,   \mbox{wool}\}^\ast \label{eq:MilkWool}\\ 
\mbox{MilkWoolMeat} & = & \{\milk,  \mbox{wool}\}^\ast ::  \mbox{meat} \label{eq:MilkWoolMeat}\\
\mbox{MilkWoolWool} & = & \{\milk,  \mbox{wool}\}^\ast ::  \mbox{wool} 
\label{eq:MilkWoolWool}\\
\mbox{MilkWoolMeatMeat} & = & \{\milk,   \mbox{wool},\mbox{meat}\}^\ast ::  \mbox{meat} \label{eq:MilkWoolMeatMeat}\\
\mbox{MeatWoolMilkMilk} & = & \{\mbox{meat},   \mbox{wool}, \milk \}^\ast ::  \milk 
\label{eq:MeatWoolMilkMilk}\\ 
\mbox{MilkWoolAnnual} & = & \milk^\ast \cup \big[\underbrace{\milk\, \milk\,  \ldots\,  \milk}_{365\ {\rm times}}\  \mbox{wool}\big]::\mbox{MilkWoolAnnual}\label{eq:MilkWoolAnnual}
\eea
They constrain how Alice uses her sheep resources as follows:
\begin{itemize}
\item \eqref{eq:MilkWool} says that Alice can milk and shear her sheep at will, but cannot use its meat;
\item \eqref{eq:MilkWoolMeat} says that Alice can milk and shear her sheep at will, but only until she eats its meat;

\item \eqref{eq:MilkWoolWool} says that Alice can milk and shear her sheep at will, and in the end has to shear it;

\item \eqref{eq:MilkWoolMeatMeat} says that Alice can milk, shear, and eat her sheep at will, and in the end has to eat it;

\item \eqref{eq:MeatWoolMilkMilk} says that Alice can milk, shear, and eat her sheep at will, and in the end has to milk it;

\item \eqref{eq:MilkWoolAnnual} says that Alice can use milk and wool, provided that she shears the wool at most once per year
\end{itemize}

\para{Example 2: elevator.}  Time passed, and Alice and Bob moved from their houses on Uruk Lane into apartments in a tall building with an elevator, like in \cref{Fig:elevator}. 
\begin{figure}[htp]
\begin{center}
\includegraphics[height=6cm]{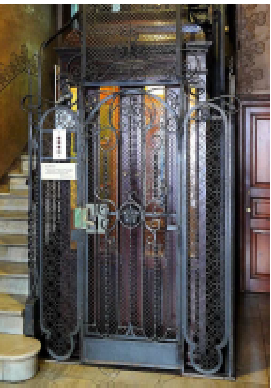}
\caption{The elevator has a call button at each floor, and the call buttons for all floors in the cabin.}
\label{Fig:elevator}
\end{center}
\end{figure}
They are now sharing not only the public spaces in their neighborhood but also the elevator in their building. Making sure that the elevator is safe and secure requires understanding how it works. For simplicity, we first assume that it works the same for everyone and omit subjects from the model. The relevant types are thus:
\begin{itemizz}{.25}
\item objects $\Obj\  =\  \big\{0, 1, 2, \ldots, n\big\}$, where
\begin{itemize}
\item $i$ denotes the $i$-th floor of the building;
\end{itemize}
\item actions $\Act \ =\  \big\{\send{}, \recv{} \big\}$, where 
\begin{itemize}
\item $\send{}$ means ``call'' and
\item $\recv{}$ means ``go'';
\end{itemize}
\item  events $\Event\  =\  \Obj \times \Act\ = \ \big\{\send{i}, \recv{i} \  \big|\ i = 0, 1, 2,\ldots, n \big\}$, where\footnote{An action $a$ performed on an object $i$ is usually written as $a_i$. Here we write $\send i$ instead of $<>_i$ and $\recv i$ instead of $()_i$.}
\begin{itemize}
\item $\send{i}$  means ``call elevator to floor $i$'' and
\item $\recv{i}$  means ``go in elevator to floor $i$''.
\end{itemize}
\end{itemizz}
An elevator history is thus a sequence of calls $\send i$  and services $\recv i$. Such sequences of events are the elements of the set 
\bear \Event^\ast & = & \left\{\, \big[\send{x_{0}}\, \send{x_{1}}\, \recv{x_0}\, \send{x_2}\, \send{x_{3}}\ldots\ \recv{x_k} \big]\  \big| \ x_0, x_1, \ldots, x_k   \in \Obj \
\right\}.
\eear
The variables $x_k\colon \Obj$ denote the floor numbers where the elevator is called or sent. The elevator receives the calls as the inputs and provides its services as the outputs.  A process diagram is displayed in \cref{Fig:elevator-process}. 
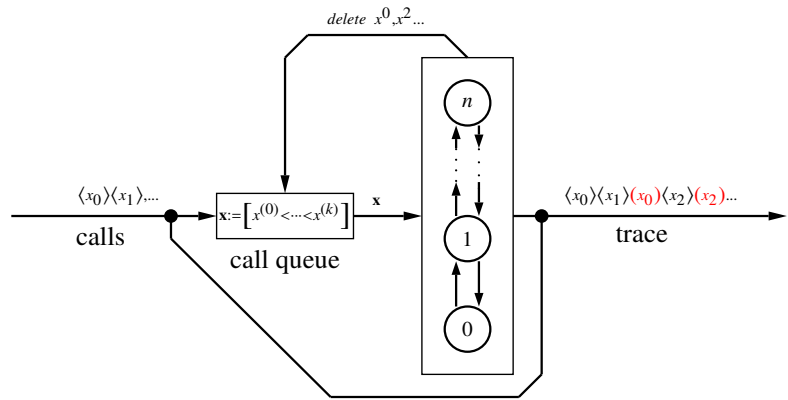
\begin{figure}[htp]
\newcommand{\remove}{\scriptscriptstyle delete\ x^0, x^2\ldots}
\newcommand{\add}{\scriptscriptstyle\vec x}
\newcommand{\xss}{\scriptscriptstyle \left<x_{0}\right> \left<x_{1}\right>,\ldots}
\newcommand{\calls}{\scriptscriptstyle\vec x := \left[x^{(0)} \lt \cdots\lt x^{(k)}\right]}
\newcommand{\callls}{\footnotesize calls}
\newcommand{\trace}{\footnotesize trace}
\newcommand{\stack}{\footnotesize call queue}
\newcommand{\ffloor}{\scriptstyle 0}
\newcommand{\sfloor}{\scriptstyle 1}
\newcommand{\tfloor}{\scriptstyle n}
\newcommand{\yss}{\scriptscriptstyle\left<x_{0}\right> \left<x_{1}\right> {\color{red}\left(x_0\right)} \left<x_{2}\right>{\color{red}\left(x_2\right)}\ldots}
\begin{center}
\def\JPicScale{.6}
\ifx\JPicScale\undefined\def\JPicScale{1}\fi
\psset{unit=\JPicScale mm}
\psset{linewidth=0.3,dotsep=1,hatchwidth=0.3,hatchsep=1.5,shadowsize=1,dimen=middle}
\psset{dotsize=0.7 2.5,dotscale=1 1,fillcolor=black}
\psset{arrowsize=1 2,arrowlength=1,arrowinset=0.25,tbarsize=0.7 5,bracketlength=0.15,rbracketlength=0.15}
\begin{pspicture}(0,0)(165,82.5)
\psline[linewidth=0.5,arrowlength=2,arrowinset=0]{<-}(85,40)(70,40)
\psline[linewidth=0.5,arrowlength=1.5,arrowinset=0,bracketlength=0.65]{->}(92.5,55)(92.5,60)
\psline[linewidth=0.5,arrowlength=1.5,arrowinset=0,bracketlength=0.65]{->}(97.5,60)(97.5,55)
\psline[linewidth=0.5,arrowlength=1.5,arrowinset=0,bracketlength=0.65]{->}(92.5,20)(92.5,30)
\psline[linewidth=0.5,arrowlength=1.5,arrowinset=0,bracketlength=0.65]{->}(97.5,30)(97.5,20)
\psline[linewidth=0.5,linestyle=dotted,dotsep=2](92.5,47.5)(92.5,55)
\psline[linewidth=0.5,linestyle=dotted,dotsep=2](97.5,55)(97.5,45.88)
\rput{0}(95,15){\psellipse[linewidth=0.5](0,0)(5,-5)}
\pspolygon[](85,75)(105,75)(105,5)(85,5)
\rput{0}(95,35){\psellipse[linewidth=0.5](0,0)(5,-5)}
\rput{0}(95,65){\psellipse[linewidth=0.5](0,0)(5,-5)}
\psline[linewidth=0.5,arrowlength=1.5,arrowinset=0,bracketlength=0.65]{->}(92.5,40)(92.5,47.5)
\psline[linewidth=0.5,arrowlength=1.5,arrowinset=0,bracketlength=0.65]{->}(97.5,47.5)(97.5,40)
\pspolygon[](40,45)(70,45)(70,35)(40,35)
\rput(55,40){$\calls$}
\psline[linewidth=0.5,arrowlength=2,arrowinset=0]{<-}(40,40)(-5,40)
\rput{0}(30,40){\psellipse[linewidth=0.5,fillstyle=solid](0,0)(1.25,-1.25)}
\psline[linewidth=0.5](65,0)(30,35)
\psline[linewidth=0.5](106.25,0)(65,0)
\psline[linewidth=0.5,arrowlength=2,arrowinset=0]{<-}(165,40)(105,40)
\psline[linewidth=0.5](106.25,0)(111.25,5)
\rput{0}(111.25,40){\psellipse[linewidth=0.5,fillstyle=solid](0,0)(1.25,-1.25)}
\psline[linewidth=0.5](111.25,5)(111.25,40)
\psline[linewidth=0.5](95,75)(90,80)
\psline[linewidth=0.5](90,80)(60,80)
\psline[linewidth=0.5,arrowlength=2,arrowinset=0]{<-}(55,45)(55,75)
\rput[B](75,82.5){$\remove$}
\rput[b](75,42.5){$\add$}
\rput[br](27.5,42.5){$\xss$}
\rput[bl](116.25,42.5){$\yss$}
\rput(95,15){$\ffloor$}
\rput(95,35){$\sfloor$}
\rput(95,65){$\tfloor$}
\psline[linewidth=0.5](30,35)(30,40)
\psline[linewidth=0.5](60,80)(55,75)
\rput[tr](20,37.5){\callls}
\rput[tl](127.5,37.5){\trace}
\rput[t](55,32.5){\stack}
\end{pspicture}
\caption{Elevator requests flow in on the left and come shuffled with services on the right}
\label{Fig:elevator-process}
\end{center}
\end{figure}
The stream of calls is entered on the left, and the process trace streams out on the right, where the elevator arrivals to where it was called are inserted between the calls. In-between is the service scheduler, which maintains the queue of open elevator calls. The calls are pushed on a queue and popped from it after they are responded to. This is implemented as a feedback loop from the elevator controller, which informs the scheduler about the elevator arrivals. When the scheduler matches an arrival with a previous call, it pops the call. The process can be viewed as a channel with feedback. The scheduler usually sorts the calls in order of the floors and changes the priorities depending on the state: it pops the next floor up on the way up and the next floor down on the way down. This is obviously a crude, high-level picture that needs to be refined for actual implementations. 

But even this crude picture illustrates the main point: that the required properties of processes, such as an elevator, or Bob's sheep bank from Sec.~\ref{Sec:MLS}, can be expressed as trace properties, which we discuss next. 

%

\subsection{Dependability properties: safety and liveness} \label{Sec:safe-live}\sindex{dependability}\sindex{safety}\sindex{liveness}
The idea of safety and liveness is that a process is
\begin{itemize}
\item safe if \emph{bad stuff never happens}, and
\item alive if \emph{good stuff eventually happens}.
\end{itemize}
As they emerged in software engineering \cite{Lamport:safliv}, the properties expressible in terms of safety and liveness came to be called \emph{dependability}\/ properties. When a geometric view of such properties got worked out \cite{Alpern-Schneider:safliv}, it turned out that \emph{all}\/ properties expressible as sets, as in Def.~\ref{Def:property}, can be expressed as an intersection of a safety and a liveness property. This will be spelled out in Sec.~\ref{Sec:Geo-saf}.

%

\para{Unsafety} implements the idea that \sindex{unsafety} there is no good stuff (actions) happening in any future: 
\begin{quote} $U$ is an unsafety property if every unsafe history remains unsafe in all futures.
\end{quote}

\para{Safety}\sindex{safety} implements the idea that there is no bad stuff (actions) happening in the past: 
\begin{quote} $S$ is a safety property if every safe history has always been safe in the past.
\end{quote}

A convenient way to formalize safety in terms of traces is to declare that \emph{unsafety is persistent}: once bad stuff (actions) happens, it cannot become good in any future.  \emph{An unsafe history remains unsafe forever.} More precisely, if a history $\vec x$ is unsafe, then no future $\vec z$ can become safe again, i.e.
\bea\label{eq:saf-conv}
\forall \vec x\vec z \in \Event^\ast. \ \ \vec x \not \in S\ \wedge \vec x \sqsubseteq \vec z & \implies & \vec z  \not \in S.
\eea
The other way around (and equivalently), if a history is safe, then its past must also be safe, i.e.
\bea\label{eq:saf}
\forall \vec x\vec z \in \Event^\ast.\ \ \vec x \sqsubseteq \vec z \in S\ \ \Longrightarrow \ \ \vec x \in S.
\eea

\para{Liveness} implements the idea that \sindex{liveness} there is always a good stuff (actions) in the future: 
\begin{quote} $L$ is a liveness property if every history can reach it in some future.
\end{quote}
Intuitively, a set of histories is a liveness property if every given history may become \emph{alive}\footnote{This is not everyone's idea of the meaning of the word "alive", but it does make sense for computations. You look at a process, do not see that it works, and ask yourself "Is this process alive?" The answer "yes" usually means that you will see that it works sometime in the future.} in the future. Formally,  $L\subseteq \Event^\ast$ is a liveness property if
\bea\label{eq:liv} 
\forall \vec x\in \Event^\ast\  \exists \vec z \in L. \ \ \   \vec x \sqsubseteq \vec z.
\eea
In terms of the prefix ordering $\sqsubseteq$ from Prerequisites~\ref{Prereq:list}\eqref{eq:prefix-order}, this means that  safety properties are the $\sqsubseteq$-lower closed sets, whereas liveness properties are reached from below.
\begin{definition}\label{Def:safliv}
The families of safety properties and liveness properties over the event set $\Event$ are respectively
\bea\label{eq:safdef}
\SAF_\Event &= & \left\{S \in \WP\left(\Event^\ast\right)\ |\ 
\forall \vec x \in \Event^\ast\ \forall \vec y \in \Event^\ast.\ \vec x::\vec y \in S \implies \vec x\in S
\right\}, \\
\LIV_\Event & = & \left\{L \in \WP\left(\Event^\ast\right)\ |\ 
\forall \vec x \in \Event^\ast\  \exists \vec y\in \Event^\ast.\ \vec x::\vec y \in L
\right\}.
\label{eq:livdef}
\eea
When confusion seems unlikely, we omit the subscript $\Event$ from $\SAF_\Event$ and $\LIV_\Event$. 
\end{definition}

\para{Remark.} Note that the definition of safety in \eqref{eq:safdef} is equivalent to \eqref{eq:saf}, whereas the definition of liveness in \eqref{eq:livdef} is equivalent to \eqref{eq:liv}.

\para{Different situations require different safety and liveness requirements.} 
There are examples where both $S\subseteq \Event^\ast$ and $\neg S  = \Event^\ast \setminus S$ are safety properties. There are examples where both $L\subseteq \Event^\ast$ and $\neg L  = \Event^\ast \setminus L$ are liveness properties.  This does not mean that a system can be both safe and unsafe at the same time, or both alive and dead. It means that desirable concepts of safety and liveness may vary from situation to situation. The intent of Def.~\ref{Def:safliv} is not to provide objective or universal properties that every safe process should satisfy.  Def.~\ref{Def:safliv} describes the properties of properties that can be meaningfully declared to be the desired safety or liveness notions for a particular family of processes.

While there are many properties that are both safe and unsafe and properties that are both alive and dead, the only property that is both safe and alive is the trivial one, i.e., $\SAF\ \cap\  \LIV = \left\{\Event^\ast \right\}$, whereby every history is safe. More examples follow.

\para{Example 1: Safety and liveness of sheep.}\label{Exmpl:safliv} The properties of Alice's interactions with her sheep using $\Event=\{\milk, \mbox{wool, meat}\}$ are as follows:
\begin{itemize}
\item \eqref{eq:MilkWool} $\mbox{MilkWool}\ \in\ \SAF \ \setminus\  \LIV $;
\item \eqref{eq:MilkWoolMeat} $\mbox{MilkWoolMeat}\ \not \in\ \SAF \ \cup\  \LIV $;
\item \eqref{eq:MilkWoolWool} $\mbox{MilkWoolWool}\ \not \in\ \SAF \ \cup\  \LIV $;

\item \eqref{eq:MilkWoolMeatMeat} $\mbox{MilkWoolMeatMeat}\ \in\ \LIV \ \setminus\  \SAF $;

\item \eqref{eq:MeatWoolMilkMilk} $\mbox{MeatWoolMilkMilk}\ \in\ \LIV \ \setminus\  \SAF $;

\item \eqref{eq:MilkWoolAnnual} $\mbox{MilkWoolAnnual}\ \in\ \SAF \ \setminus\  \LIV $.
\end{itemize}

\para{Example 2: Safety and liveness of an elevator.}\label{Exmpl:saflivelev} An example of reasonable dependability properties required from an elevator are:
\begin{itemize}
\item\emph{\textbf{safety:}} the elevator should \emph{only}\/ come to a floor to which it was called, and
\item\emph{\textbf{liveness:}} the elevator should eventually go to \emph{every}\/ floor where it was called. 
\end{itemize}
Using the notation from Prerequisites \ref{Prereq:list}, these properties can be formally stated as
\bear
 \overline{\recv i} & \subseteq & \overline{\exists \send i \prc \recv i}, \hspace{4em}  \leftsquigarrow \text{\textit{safety}}
 \\
\overline{\send i} & \subseteq &\overline{\send i \prc \exists \recv i}. \hspace{4em} \leftsquigarrow \text{\textit{liveness}}
\eear
Safety says that every arrival $\recv i$  must be preceded by a call $\send i$. Liveness says that every call $\send i$ must eventually be followed by an arrival $\recv i$.
The properties that we specified are thus
\bea\label{eq:SafE}
\SafE & = &
 \left\{\ \vec t\in \Event^\ast \ \Big|\ \forall i\in \Obj.\ \vec t\in \overline{\recv i} \ \Longrightarrow\ \vec t\in \overline{\exists \send i \prc \recv i}\right\}\ \ =\ \ \bigcap_{i=0}^n\ \neg \overline{\recv i}\ \cup \ \overline{\exists \send i \prc \recv i} \  \notag  \\
 \label{eq:Safe-elev} \\
 LivE & = &  \left\{\ \vec t\in \Event^\ast \ |\ \forall i\in \Obj.\ \vec t\in \overline{\send i} \ \Longrightarrow\ \vec t\in \overline{\send i \prc \exists \recv i}\right\}\ \ =\ \ \bigcap_{i=0}^n\ \neg \overline{\send i}
\ \cup \   \overline{\send i \prc \exists \recv i}, \ \notag \\ 
\label{eq:Live-elev}
 \eea
where $\neg X = \Event^\ast \setminus X$ denotes the complement set.

\section{Authority and availability}\label{Sec:Authorization}

\subsection{Strictly local events and properties}
\label{Sec:Local}
As indicated in Sec.~\ref{Sec:histories}, we model security properties as sets of histories $P\subseteq \Event^\ast$, where $\Event \ = \ \Act \times \Obj \times \Subj$, or more generally $\Event \ = \ \Act \times \Obj \times \Subj\times \Levels$, when security levels need to be taken into account.  An event $x \in \Event$ is thus an action of a subject on an object, i.e., a tuple $x = <a, i, u>$, or $x = <a,i,u,\ell>$, where $\ell$ is a security level. A more general notion of events will be needed for channel modeling in Sec.~\ref{Sec:shared}, but for the moment we stick with the tuples. The reason is that they conveniently display the \emph{strict localities}\/ \sindex{locality!strict} of each event as the component of the tuple that represents it. 

\para{Strictly local views of events.} The type $\Event$ of events thus comes with the projections
\beq \label{eq:projections} \suct\colon\Event \to \Subj,\qquad\qquad  \obct\colon\Event \to \Obj,\qquad\qquad \actn\colon\Event \to \Act, \qquad\qquad \levl \colon\Event \to \Levels\eeq
assigning to each event $x\in \Event$ the unique subject $\suct(x)$ that observes or enacts it, the unique action $\actn(x)$ that takes place, the unique object $\obct(x)$ that is acted on, and the security level $\levl(x)$ where the event takes place. All such projections $\ppP \colon \Event \to \VVv$ induce the \emph{strictly local} event types $\Event_v \ = \  \left\{x\in \Event \ |\ \ppP(x) = v\right\}$, which partition the event type as the disjoint unions 
\bea\label{eq:partitions} \Event & = & \coprod_{v\in \VVv} \Event_v\qquad \mbox{ where }\hfill \VVv\in \{\Act, \Obj, \Subj, \Levels\}
\eea
The strictly local event types that we will be working with here are thus:
\begin{itemize}
\item $\Event_u = \big\{<a, i, u>\ |\ a\in \Act, i\in \Obj\big\}$ --- the events observable by the subject $u$;
\item  $\Event_i = \big\{<a, i, u>\ |\ a\in \Act, u\in \Subj\big\}$ --- the events involving the object $i$;
\item  $\Event_a = \big\{<a, i, u>\ |\ i\in \Obj, u\in \Subj\big\}$ --- the events where someone performs the action $a$ on something.
\end{itemize}
These types are \textbf{strictly}\/ local in the sense that they are disjoint, i.e., $\Event_{v}\cap \Event_{w} = \emptyset$ as soon as $v\neq w$, for $v,w$ both of any of the above types $\VVv\in \{\Act, \Obj, \Subj, \Levels\}$.

\para{Locality vs \emph{strict}\/ locality.} Localizing the events at the various security levels, various subjects, various actions, and various objects is necessary for studying security dynamics. In Chapters~\ref{Chap:Process}--\ref{Chap:Geometry}, we always assume the \emph{strict}\/ localities, partitioning the events into the disjoint families in~\eqref{eq:partitions}. In the later chapters (starting from Sec.~\ref{Sec:shared}), we will use a more general and more realistic notion of locality. We start with the special case of strict localities, to keep things simple until we get used to them. \textbf{\emph{In Chapters~\ref{Chap:Process}--\ref{Chap:Geometry}, \emph{``local''}\/ always means \emph{``strictly local''}, unless specified otherwise.}} We will keep repeating ``strict'' to introduce new concepts, but will elide it eventually.

\para{Strictly local histories.}  \sindex{history!strictly local}Process models are usually set up so that Alice only sees her own actions. An event $x = <a,i,u>\in \Event$ is thus observable for Alice just when $u = A$. This assumption is imposed on the model using the \emph{strict purge}\/ operation\sindex{purge!strict} $\restriction_A : \Event^\ast \to \Event^\ast_A$  defined as
\bea\label{eq:purge}
\seq{}\restr_A & = & \seq{}, \label{eq:restrict}\\
\left(\vec x\cons y\right)\restr_{A} & = & \begin{cases} \left(\vec x\restr_{A}\right)\cons y & \mbox{ if } y = <a,i,A> \mbox{ for some } a\in \Act\mbox{ and } i\in \Obj \\
\vec x\restr_{A} & \mbox{ otherwise.}
\end{cases}  \notag
\eea 
Similar strict purge operations can be defined for any of the projections in \eqref{eq:projections}. The general (nonstrict) purge operations will be defined and used in  Ch.~\ref{Chap:Channel}.

\para{Notation.} We write $x_A\in \Event_A$ for local events, and $\vec x_A \in \Event_A^\ast$ for strictly local histories. Note the difference between
\begin{itemize}
\item  the strict purge $\vec x\restr_A\in \Event_A^\ast$ of a global history $\vec x \in \Event^\ast$, and
\item a strictly local history $\vec x_A\in \Event_A^\ast$, specified strictly locally, with no reference to a global context.
\end{itemize}
Alice's \emph{strictly local view}\/ of the \sindex{view!strictly local} property $\Property\subseteq \Event^\ast$ is written
\bea
\Property_A &= & \left\{\vec x\restr_A\ |\  \vec x \in \Property\right\}.
\eea

\para{Strictly localized properties.}  A history $\vec t_A\in \Event_A^\ast$ observed by Alice is a strict localization of some global history $\vec z \in \Event^\ast$ such that $\vec z \restr_A = \vec t_A$. While $\vec z$ may be just $\vec t_{A}$ if no one except Alice did anything in the given global history; or there may be lots of events unobservable for Alice. She cannot know. But she does know that there are infinitely many possible global histories $\vec z$ consistent with her observation $t_{A}$. 

If Alice observes $\vec z\restr_A$, Bob observes $\vec z\restr_B$, and they tell each other what they have seen, they will still not be able to derive $\vec z$, even if they are alone in the world. The reason is that neither of them can tell how exactly their actions were mixed: which of Alice's actions preceded Bob's actions, and the other way around.

\para{Example.} Let $\Event = \Event_A \coprod \Event_B$ where $\Event_A = \{a\}$ and $\Event_B = \{b\}$. Suppose that a history $\vec t$ satisfies the property $\Property = \big\{\sseq{aaaa},\sseq{aabb}, \sseq{baab},\sseq{bbbb}\big\}$. If Alice observes $\vec t_A = \sseq{aa}$ and Bob observes $\vec t_B = \sseq{bb}$, can they be sure that the property $\Property$ has been satisfied?
The possible global histories consistent with Alice's and Bob's local observations are
\bear
\widehat t & = & \big\{\vec z\in \Event^\ast\ |\  \vec z\restr_A = \sseq{aa} \wedge \vec z\restr_B = \sseq{bb}\big\}\\
& = & \big\{\sseq{aabb}, \sseq{abab}, \sseq{abba}, \sseq{baab}, \sseq{baba}, \sseq{bbaa}\big\}.
\eear
Any of these actions could have taken place. Alice and Bob will only be able to verify locally a property $\Property$ is satisfied if it is a \emph{strictly localized}\/ property, according to the next definition.

\bigskip
\begin{definition}\label{Def:localized}\sindex{localization!strict} \sindex{strict localization|see{localization, strict}}
The \emph{strict localization} of a  property $\Property\subseteq \Event^\ast$ is the set $\widehat \Property$ of all histories with the projections satisfying $\Property$, i.e.
\bea\label{eq:localization}
\widehat \Property & = & \left\{\vec z\in \Event^\ast\ |\  \forall u\in \Subj.\ \vec z \restr_u \in \Property_u\right\}\\ && \mbox{ where } \Property_u = \left\{ \vec x\restr_u\ |\ \vec x\in \Property\right\}. \notag
\eea

A property $\Property\subseteq \Event^\ast$ is \emph{strictly localized}  when $\Property = \widehat \Property$, i.e.
\bea\label{eq:proj}
\vec t \in \Property & \iff & \forall u\in \Subj.\ \vec t \restr_u \in \Property_u.
\eea
The family of strictly localized properties is
\sindex{property!strictly localized}\sindex{strictly localized property|see{property, strictly localized}}
\bea
\Local & = & \{\Property \in \WP(\Event^\ast)\ |\ \Property = \widehat \Property\}.
\eea
\end{definition}

\subsection{Resource security as localized safety and liveness} \label{Sec:au-av}
The simplest security properties arise as localized dependability properties. In particular, authority and availability can be construed as safety and liveness from Alice's, Bob's, and other subjects' points of view. A process is thus
\begin{itemize}
\item authorized if bad stuff (actions) does not happen {\color{red} to anyone}: \emph{all bad resource requests are rejected};
\item available if good stuff (actions) happens {\color{red} to everyone}: \emph{all good resource requests are eventually accepted}.
\end{itemize}
On a closer look, it turns out that there are several reasonable views of what are "good resource requests". Do all subjects need to coordinate to make the request; or is it enough that some subjects make the request, and no one objects; or should a majority of some sort be required? To study such questions, we need a formalization.

Refining the idea of safety from \eqref{eq:saf-conv}, we say that $\Closed \subseteq \Event^\ast$ is an authorization property if it satisfies the following condition:
\bea\label{eq:au-conv}
\forall \vec x\vec z \in \Event^
\ast. \ \ \vec x \not \in \Closed\ \wedge \vec x \sqsubseteq \vec z & \implies & \exists u\in \Subj.\ \vec z \restr_u  \not \in \Closed_u.
\eea
In words, if there is an authority breach in some history, then in every future of that history, some subject will observe that their authority has been breached. Every authority breach is a breach of someone's authority. The logical converse of \eqref{eq:au-conv}, refining \eqref{eq:saf}, characterizes authority (or synonymously, an authorization property) $\Closed$ by
\bea\label{eq:authorization}
\forall \vec x\vec z \in \Event^
\ast.\ \ \Big(\vec x \sqsubseteq \vec z\ \wedge\ \forall u\in \Subj.\ \vec z \restr_u \in \Closed_u\Big)\ \ \Longrightarrow \ \ \vec x \in \Closed.
\eea
In other words, if the local views $\vec z \restr_u$ of a history $\vec z$ appear authorized to all subjects $u\in \Subj$, then all past histories $\vec x\sqsubseteq \vec z$ must have been authorized globally. Note that the statement packs two intuitively different requirements. One is that if a history $\vec z$ is authorized, then every past history $\vec x\sqsubseteq \vec z$ is authorized; i.e., any authorization property is a safety property. The other one is that if all local projections of $\vec z \restr_u$ are authorized, then $\vec z$ is authorized globally; i.e., any authorization property is a local property. The fact that these two requirements are equivalent to authority is stated in Prop.~\ref{Prop:saf-liv-au-av}(b) as claim \eqref{prop:auu}.

Liveness from  \eqref{eq:liv} can be localized in more than one way. We first consider what seems to be the weakest reasonable localization: 
\bea\label{eq:availability} 
\forall \vec x\in \Event^\ast\  \forall u\in \Subj\ \ \exists \vec z \in \Dense. \ \ \  \vec x\restr_u \sqsubseteq \vec z\restr_u
\eea
In other words, an availability property is a liveness property where any subject on their own can assure the liveness.

\smallskip
\begin{definition}\label{Def:AuthAvail}
For any family of subjects $\Subj$ for events partitioned into $\Event  =  \coprod_{u\in \Subj} \Event_u$, the authorization and the availability properties are respectively defined
\bea\label{eq:AU}
\AU_{\Event} & = & \left\{\Closed \in\WP\left(\Event^*\right)\ |\ \forall \vec x\in \Event^\ast\ \forall \vec y \in \Event^\ast.\ \Big(\forall u\in \Subj.\ \left(\vec x::\vec y\right) \restr_u \in \Closed_u\Big) \implies \vec x\in \Closed
\right\}, \quad\quad \\
\AV_{\Event} & = & \Big\{\Dense \in \WP\left( \Event^*\right)\ |\ \forall \vec x \in \Event^\ast\  \forall u\in \Subj\ \ \exists \vec y_u\in \Event^\ast_u.\ \ \vec x\restr_u ::\vec y_u \in \Dense_u
\Big\}. \quad\quad \label{eq:AV}
\eea
When confusion seems unlikely, we omit the subscript $\Event$.
\end{definition}

\begin{proposition}\label{Prop:saf-liv-au-av}
Let $\Property \subseteq \Event^\ast$ where $\Event = \coprod_{u\in \Subj} \Event_u$. Then the following statements are true.
%
\begin{enumerate}[a)]
%
\item Authorization is strictly local safety. Availability implies strictly local liveness. Formally, this means: 
\bea 
\Property \in \AU_\Event & \iff & \forall u\in \Subj.\ \ \Property_u \in \SAF_{\Event_u}\ \wedge \ P=\widehat P,
\label{prop:au}\\
\Property \in \AV_\Event & \iff & \forall u\in \Subj.\ \ \Property_u \in \LIV_{\Event_u}.
\label{prop:av}
\eea
\item Authorization is just the global safety that is strictly local. The strict localization of an availability property is a liveness property. Formally:
\bea
\AU_\Event & = & \{\Property \in \WP(\Event^\ast)\ |\ P \in \SAF_\Event\ \wedge \ P=\widehat P \}, \label{prop:auu}\\
\AV_\Event & \subseteq & \{\Property \in \WP(\Event^\ast)\ |\ \widehat P \in \LIV_\Event\}.
\label{prop:avv}
\eea
\end{enumerate}
\end{proposition}

\bpr \textbf{\emph{a)}} Since $\AU_\Event \subseteq \Local_\Event$ follows from \eqref{eq:AU} for $\vec x \cons \vec y \ = \ \vec x$, proving \eqref{prop:au} boils down to showing that the following two implications are equivalent
\bear
\Big(\forall u\in \Subj.\ \left(\vec x \cons \vec y\right)\restr_u \in P_u\Big) & \Longrightarrow & \Big(\forall u\in \Subj.\ \vec x \restr_u \in P_u\Big) \hspace{2.5em} \mbox{--- which means }\hspace{3.6em} P \in \AU_\Event,
\\
 \Big(\forall u\in \Subj.\ \left(\vec x_u \cons \vec y_u\right) \in P_u \ & \Longrightarrow & \hspace{4em}\vec x_u \in P_u\Big) \hspace{2.5em} \mbox{--- which means }\  \forall u\in \Subj.\  P_u \in \SAF_{\Event_u}.
\eear
for all $\vec x,\vec y \in \Event^\ast$ and all $\vec x_u,\vec y_u \in \Event_u^\ast$. The bottom-up direction is valid for all predicates in first-order logic: the second implication is stronger than the first one. Towards the top-down direction, fix a subject $A$, take arbitrary histories $\vec x_A, \vec y_A\in \Event^\ast_A$ such that $\left(\vec x_A::\vec y_A\right) \in P_A$, and set $\vec x = \vec x_A$ and $\vec y = \vec y_A$. For an arbitrary subject $u\in \Subj$ we have 
\bear
\left(\vec x \cons \vec y\right)\restr_u\  \ =\ \ \left(\vec x_A::\vec y_A\right)\restr_u &  = &  \begin{cases} 
\ \left(\vec x_A::\vec y_A\right) \in P_A & \mbox{ if } u=A,\\
\ \seq{} \in P_u & \mbox{ if } u \neq A.
\end{cases}
\eear
Since $\seq{}\in P_u$ follows from part \emph{(b)} below, and $\left(\vec x_A::\vec y_A\right) \in P_A$ was assumed, we have $\forall~u\in~\Subj.\ \vec x \restr_u \in P_u$, as claimed. 

Towards \eqref{prop:av}, we need to show that the following statements are equivalent:
\bear 
\forall \vec x\in \Event^\ast\ \ \exists \vec y_u \in \Event^\ast_u. &\hspace{-1.2em}& \left(\vec x\restr_u :: \vec y_u\right) \in P_u \hspace{2em} \mbox{--- which means }\hspace{3.6em} P \in \AV_\Event, \\
\forall \vec x_u\in \Event_u^\ast\   \exists \vec y_u \in \Event_u^\ast.&\hspace{-1.2em}&   \left(\vec x_u \ :: \vec y_u\right) \in P_u\hspace{2em}\mbox{--- which means }\  \forall u\in \Subj.\  \, P_u \in \LIV_{\Event_u}.
\eear
To show that the first implies the second, consider that the second is just a special case of the first one, 
obtained by taking $\vec x = \vec x_u$, and noting that $\vec x_u\restr_u \ = \ \vec x_u$.
For the converse direction, we consider arbitrary $\vec x_u \in \Event_u^\ast$ and any $\vec x$ with $\vec x\restr_u \  = \ \vec x_u$ for all $u \in \Subj$. Then, the choices of $\vec y_u$ provided by the local liveness
requirements, to meet $\left(\vec x_u \ :: \vec y_u\right) \in P_u$, will also meet the requirements for availability: $\left(\vec x\restr_u \ :: \ \vec y_u\right) \in P_u$, so that availability follows.

\smallskip

\textbf{\emph{b)}} Towards \eqref{prop:auu}, consider
\bear \Big(\forall u\in \Subj.\ \left(\vec x \cons \vec y\right)\restr_u \in P_u\Big) & \Longrightarrow & \vec x \in P \hspace{3em} \mbox{--- which means } P \in \AU_\Event, \\
\left(\vec x \cons \vec y\right) \in P & \implies & \vec x \in P \hspace{3em} \mbox{--- which means } P \in \SAF_\Event.
\eear
The fact that $\AU_\Event\subseteq \SAF_\Event$ means that the first implication follows from the second one. This is true because  $\left(\vec x \cons \vec y\right) \in P$ implies $\left(\vec x \cons \vec y\right)\restr_u \ \in P_u$  for all $u$. The converse clearly also holds as soon as $P$ is local. And \eqref{prop:avv} is obvious, since $\vec y_u \in \Event_u^\ast \subseteq \Event^\ast$.
\epr

\medskip
\para{Example 1: Authority and availability of sheep and oil.} To model sheep security, we zoom out again and go back to the situation where Alice and Bob need to share some of their resources. The subject type is thus $\Subj = \{A,B\}$, the objects are from $\Obj = \{\mbox{sheep}, \mbox{oil}\}$, and the actions are $\Act = \{\mbox{shear}, \mbox{cook}\}$. The possible events in the simple form of \ref{eq:event} would thus be all triples from $\Act \times \Obj \times \Subj$. But since Alice and Bob will not shear the oil or cook the sheep, we take\footnote{See the remark about \emph{Standard and refined $\Event$ types}\/ in Sec.~\ref{Sec:histories} The $\Event$ type used here is in a refined form of \eqref{eq:event-deptype}.}
\bear
\Event & = &   \left\{\mbox{shear sheep}_A, \mbox{cook oil}_B\right\}.
\eear 
Consider the following properties of Alice's and Bob's interactions:
\bea
\mbox{Either} & = & \mbox{shear sheep}_A^\ast \cup \mbox{cook oil}_B^\ast, \label{eq:Either}\\
\mbox{Alternate} & = & \left(\mbox{shear sheep}_A :: \mbox{cook oil}_B\right)^\ast, \label{eq:Alternate}\\ 
\mbox{Finish} & = & \left\{\mbox{shear sheep}_A,  \mbox{cook oil}_B\right\}^\ast :: \mbox{shear sheep}_A :: \mbox{cook oil}_B. \label{eq:Finish}
\eea
These properties provide counterexamples for the converses of the claims of Prop.~\ref{Prop:saf-liv-au-av}:
\begin{itemize}
\item \eqref{eq:Either} $\mbox{Either}\ \ \in\ \ \SAF_\Event  \setminus  \left(\AU_\Event\cup \Local\right)$;
\item \eqref{eq:Alternate} 
$\mbox{Alternate}_A\ \in\ \LIV_{\Event_A}$, and $\mbox{Alternate}_B\ \in\ \LIV_{\Event_B}$, but \ \  $\mbox{Alternate} \not\in \LIV_{\Event} $;
\item \eqref{eq:Finish} $\mbox{Finish}\  \in  \ \LIV_\Event \cap \AV_\Event$.
\end{itemize}

\medskip
\para{Example 2: Authorization and availability of elevator.}  To model the security of the elevator, we consider the events from the point of view of the subjects, i.e., in the form
\bear
\Event & = & \Obj \times \Act \times \Subj\ \ = \ \ \big\{\send{i}_u, \recv{i}_u \  \big|\ i \in \Obj, u \in \Subj \big\}
\eear
where 
\begin{itemize}
\item $\Obj\  =\  \big\{0, 1, 2, \ldots, n\big\}$ are the objects again: the floors of the building (denoted by variables $x_0, x_1, \ldots$);
\item $\Act \ =\  \big\{\send{-}, \recv{-} \big\}$ are the actions:  ``call/go'' and ``arrive'', respectively; 
\item $\Subj \ =\ \big\{A, B\big\}$ are the subjects: Alice and Bob (denoted by variables $Y_0, Y_1, \ldots$).
\end{itemize}
A history is now  in the form
\[\seq{\ssend {x_0}{ Y_0} 
\ssend {x_1}{ Y_1} 
\rrecv{x_2}{ Y_2} 
\ssend {x_0}{ Y_0} 
\rrecv{x_0}{ Y_2}
\ldots}
\]
meaning that  ``$Y_0$ calls to $x_0$, $Y_1$ calls to $x_1$, ${ Y_2}$ arrives to $x_2$, 
$Y_0$ calls to $x_0$ again, ${ Y_2}$ arrives to $x_0$
\ldots''. The dependability properties (which required that the elevator should \emph{only}\/ go where called, and that it should \emph{eventually}\/ go where called) are now refined to
\begin{itemize}
\item\emph{\textbf{authorization:}} the elevator should \emph{only}\/ take \emph{some}\/ subject to a floor if they have a clearance and if they requested it, and
\item\emph{\textbf{availability:}} the elevator should \emph{eventually}\/ take \emph{every}\/ subject with a clearance for the floor that they requested,
\end{itemize}
which generalizes Example 2 from Sec.~\ref{Sec:safe-live} to
\bear
 \overline{\recv i_u} & \subseteq & \overline{\exists \send i_u \prc \recv i_u}, \hspace{4em}  \leftsquigarrow \text{\textit{authority}}
 \\
\overline{\send i_u} & \subseteq &\overline{\send i_u \prc \exists \recv i_u}. \hspace{4em} \leftsquigarrow \text{\textit{availability}}
\eear
\para{Note}, however, that here we need to assume that the subject $u$ is somehow authorized to go to the floor $i$. This is where the static resource security formalism from Chap.~\ref{Chap:Resource} needs to be imported into the dynamic resource security formalisms of histories and properties. In terms of multi-level security, the clearance assumption would be $\clr(u)\geq \loc(i)$. In terms of permission matrices, it would be $\send -, \recv - \in M_{ui}$. Writing for simplicity either of these assumptions as the predicate $\Clr(u,i)$, the above crude idea of authorization and availability properties of the elevator can be formalized to
\bea
\AuthoE & = &
 \Big\{\vec t\in \Event^\ast \ |\ \forall u\in \Subj\ \ \forall i\in \Obj.
  \label{eq:Autho-elev}\\  
 && \hspace{3.6em} \vec t\in \overline{\recv i _u} \ \Longrightarrow\ \left({\color{red}\Clr(u,i)}\ \wedge\  \vec t\in \overline{\exists \send i_u \prc \recv i_u} \right)\Big\}\notag\\
\AvailE & = &  \Big\{\vec t\in \Event^\ast \ |\ \forall u\in \Subj\ \ \forall i\in \Obj.\label{eq:Avail-elev}
\\ 
&& \hspace{3.6em} \left(\vec t\in \overline{\send i _u} \ \wedge\ {\color{red}\Clr(u,i)}\right) \ \Longrightarrow\ \vec t\in \overline{\send i_u \prc \exists \recv i_u}\Big\}\notag
\eea
The reason why $\Clr(u,i)$ occurs on the left of the implication for authorization and on the right for availability is that the lift should be available to $u$ for $i$ \emph{only if}\/ $\Clr(u,i)$ holds, and $u$ should be authorized to go to $i$ \emph{whenever} $\Clr(u,i)$ holds. Recall from Ch.~\ref{Chap:Resource} that permissions and clearances may be declared differently for different states of the system, which means that the clearance predicate would be in the more general form $\Clr\left(\vec t, u,i\right)$, which makes it dependent on the history $\vec t$. Alice could thus have different authorizations in different historic contexts.

\para{Example 3: Questions and answers.}  Suppose that Alice and Bob are having a conversation: one asks a question, and the other one answers or asks another question. We denote the action of asking a question by "$?$", and the action of answering a question by "$!$". To distinguish between different questions, and to identify the question that can be repeated, and to bind questions and the corresponding answers, we assume that there is a fixed set of questions $\Obj$, which we take to be a set of numbers. Thus, we have the types
\begin{itemize}
\item $\Obj\  =\  \big\{0, 1, 2, \ldots, n\big\}$ are the questions that may be asked or answered (viewed as objects, and denoted by variables $x_0, x_1, \ldots$ again),
\item $\Act \ =\  \big\{ ?, ! \big\}$ are the actions:  "ask question" and "answer question", respectively; and moreover
\item $\Subj \ =\ \big\{A, B\big\}$ are the subjects: Alice and Bob (denoted by variables $Y_0, Y_1, \ldots$),
\end{itemize}
which induce the event space in the form
\bear\Event & = & \Subj\times \Obj \times \Act \ \ = \ \ \big\{
\qu {i} {u},\  \an {i} {u} \ \big| \ \ u \in \Subj,\  i \in \Obj \big\}
\eear
with histories as sequences of events such as
\[\seq{\qu{x_0}{ Y_0}\ \  \ 
\qu {x_1}{ Y_1} \ \ \ 
\an{x_2}{ Y_2} \ \ \ 
\qu {x_0}{ Y_0} \ \ \ 
\an{x_2}{Y_2}
\ldots}\]
which says that "$Y_0$ asks the question $x_0$, then $Y_1$ asks the question $x_1$, then ${ Y_2}$ answers the question $x_2$,  then $Y_0$ asks $x_0$ again, then ${ Y_2}$ answers $x_0$", etc. Formally, such conversations between Alice and Bob are of course similar to the operations of the elevator that they share in their building. 

The difference is, of course, that when Alice calls the elevator, she usually wants to use the service herself; whereas when she asks a question, then she usually expects Bob to provide an answer. This leads to the following properties that may be required in a formal conversation or perhaps interrogation:
\begin{itemize}
\item\emph{\textbf{all answers questioned:}} an answer should \emph{only}\/ be given for a question previously asked by someone, and
\item\emph{\textbf{all questions answered:}} every question that was asked will  \emph{eventually}\/ be answered by someone.
\end{itemize}
Questions and answers in a conversation are matched similarly like calls and services of an elevator, but the matching is slightly more general:
\bear
 \overline{\an i u} & \subseteq & \overline{\exists \qu i w \prc \an i u}, \hspace{4em}  \leftsquigarrow \text{\textit{answers questioned}}
 \\
\overline{\qu  i u} & \subseteq &\overline{\qu i u  \prc \exists \an i w}. \hspace{4em} \leftsquigarrow \text{\textit{questions answered}}
\eear
A clearance predicate could moreover be used to constrain who is permitted to ask which questions, and who is permitted to answer them. A predicate in the form $\Clr(u,i,a)$, meaning that the subject $u$ has permission to perform on the object $i$ the action $a$, would thus correspond to the declaration $a\in M_{ui}$ in a permission matrix $M$. The requirements that all answers should be preceded by corresponding questions, and that all questions should eventually be answered, can be refined to
\bea
\SafC & = &
 \Big\{\vec t\in \Event^\ast \ |\ \forall u\in \Subj\ \ \forall i\in \Obj.  \label{eq:Saf-conv}\\  
 &&\hspace{3em}\Big(\vec t\in \overline{\an i u} \ \wedge\ {\Clr(u,i,!)}\Big) \ \Longrightarrow\ \exists w\in \Subj.\ \Big(\Clr(w,i,?)\ \wedge\  \vec t\in \overline{\exists \qu i w \prc \an i u} \Big)\Big\},\notag\\
\LivC & = &  \Big\{\vec t\in \Event^\ast \ |\ \forall u\in \Subj\ \ \forall i\in \Obj.\label{eq:Liv-conv} \\
&& \hspace{3em}\Big(\vec t\in \overline{\qu i u} \ \wedge\ {\Clr(u,i,?)}\Big) \ \Longrightarrow\ \exists w\in \Subj.\ \Big(\Clr(w,i,!)\ \wedge\ \vec t\in \overline{\qu i u \prc \exists \an i w}\Big)\Big\}.\notag
\eea
The fact that $\SafC$ does not guarantee authority and that $\LivC$ does not guarantee availability points to the limitations of the presented formalizations of authority and availability as tools of resource security. A particular way to mitigate some of these limitations is worked out in the exercises. The general way to resolve the issue is authentication, explored in Ch.~\ref{Chap:Auth}.

\para{Remark.} The presented examples are meant to be simple. In some cases, they are oversimplified. For example, the precedence relation in $\exists \send i\prec \recv i$ does not discharge the calls that the elevator has responded to. Even the simplest requirements used in real elevators are significantly more complex than a textbook permits. And elevators are among the simplest systems that need to be secured.

\subsection{General relativity of histories}
If there are just two possible events, i.e.
$\Event =  \{0,1\}$, then both $0\Event^{\ast}$ and $1\Event^{\ast}$ are safety properties. 
But they are each other's complements, since $0\Event^{\ast} \cup 1\Event^{\ast} = \Event^{\ast}$ and $0\Event^{\ast} \cap 1\Event^{\ast} = \emptyset$. Each of them is, therefore, both a safety property and an unsafety property. Similarly, there are authorization properties whose complements are also authorizations. If Alice is Bob's commanding officer, then she needs to be able to authorize him to use a weapon and advance in one situation; and she also needs to be able to authorize him to not use the weapon and retreat in another situation. Authorizations are not objective or absolute, but subjective and relative.  On one hand, they are relative to the subjective standpoints. In a zero-sum game, one subject's safe position is the other subject's unsafe position, and vice versa. On the other hand, they are relative to security contexts, e.g., in an evolving conflict where Alice may need to modify Bob's authorizations. Within each subject's frame of preferences, different situations induce different authorities. 

Formal models, scientific methods, and general relativity are needed not only in astronomy but also in security --- to protect us from intuitions that the Earth is flat and that authority is absolute.

Conflicting availability requirements are even more common, and occasionally more subtle, than the conflicting authority claims. If both Alice and Bob need water, but the fountain cannot serve both at the same time, then the fountain may be alive but not available to either of them. The concept of availability evolved to characterize such situations.
It signals opportunities for attacks on shared resources, whereby Bob may use the fountain just to deny it to Alice. Other such attacks are described in the next section.

\section{Denial-of-Service}\label{Sec:DoS}

Although  liveness is not a security property and is generally not a matter of conflict, Alice and Bob may pursue different \emph{good stuff (actions)}\/ that they want to happen. If each of them pursues different \emph{good stuff (actions)}\/, and Alice's \emph{good stuff (actions)}\/ preclude Bob's \emph{good stuff (actions)}, then a system that is alive for Alice will be dead for Bob. 

If, moreover,  Alice's liveness property happens to be availability, so that every subject can secure it on their own, then Alice can make her \emph{good stuff (actions) happen}, which then means that Bob's opposite good stuff (actions) will not happen. In such situations, we say that Alice has launched a \emph{Denial-of-Service (DoS)}\/ attack against Bob.

In general, a DoS phenomenon occurs when both a property $\Property$ and its complement $\neg \Property$ are liveness properties, and a history that is alive in one sense must be dead in another sense. When the property $\neg \Property$ is moreover \emph{available} for some subjects, then those subjects can cause a \emph{Denial-of-Service} $\Property$. 

\para{Example 3 again: More questions and answers.} One thing we didn't mention so far is Alice's and Bob's ages. They are, in fact, 2 years old. They learned to speak very recently, and they are trying to learn the rules of conversation. While Alice believes that all questions should be answered, Bob's standpoint is that all answers should be questioned. 

For Alice, a conversation is alive only when every question can be answered by someone.  That is when Alice's requirement is a liveness property. It is, moreover, an availability property when everyone knows answers to all questions. Only then can anyone satisfy Alice's requirement.

Bob's view is, on the other hand, that a conversation is only alive when there are some unanswered questions. Bob's requirement is \emph{also}\/ a liveness property whenever, for any question, there is someone who is permitted to ask it. If, moreover, everyone is permitted to ask any question, then Bob's requirement is even an availability property.

Formally, Alice's requirement is thus $\Service_A  =  \LivC$, whereas Bob's requirement is opposite: $\Service_B  =  \neg \LivC$. 

On the other hand, using  \eqref{eq:Liv-conv} and the definitions in the sections preceding it, it can be proven that 
\bear
\forall i \in \Obj \ \exists w \in \Subj.\ \Clr(w,i,!) & \implies & \LivC \in \LIV,\\
\forall i \in \Obj \ \forall w \in \Subj.\ \Clr(w,i,!) & \implies & \LivC \in \AV,\\
\forall i \in \Obj \ \exists w \in \Subj.\ \Clr(w,i,?) & \implies & \neg\LivC \in \LIV,\\
\forall i \in \Obj \ \forall w \in \Subj.\ \Clr(w,i,?) & \implies & \neg\LivC \in \AV.
\eear
If there is a particular question $q\in \Obj$ that Bob is permitted to ask, in the sense that $\Clr(B,q,?)$ holds true, then the property $\Service_B = \neg \LivC$ is available to him, as he can extend any history of questions $\vec t\in \Event^\ast$  to $\vec t:: \left(\qu q B\right) \in \neg \LivC = \Service_B$. Since $\Service_A = \LivC$, Alice has herewith been denied the service of a final answer, that she normally requires from conversations. 

If Alice is, on the other hand, permitted to answer the question $q$, in the sense that $\Clr(A,q,!)$ is true, then she can extend the current history to $\vec t:: \left(\qu q B\right):: \left(\an q A\right) \in \LivC = \Service_A$. Now Bob has been denied service. 

Alice and Bob can continue to deny service to each other like this until one of them finds a way to break the symmetry. For example, Alice may be able to convince their parents, Carol and Dave, that the questions that have been answered in the past should not be asked again. Bob's Denial-of-Service attack will then depend on how many other questions he is permitted to ask, i.e., for how many different $i\in \Obj$ does he have a clearance $\Clr(B,i,?)$. Whether he will win or not also depends on Alice's capability $\Clr(A,i,!)$ to answer questions.

\para{Example 4: TCP-flooding.}  The internet transport layer connections are established using the Transmission Control Protocol (TCP). Its rudimentary structure, displayed on the left in Fig.~\ref{Fig:SYN-flood}, can be construed as an extension of the basic question-answer protocol, where asking and answering a question is followed by a final action, where the answer is accepted.  
\begin{figure}[!h]
\begin{center}
\begin{minipage}[t]{.8\linewidth}
\includegraphics[height=6.3cm
]{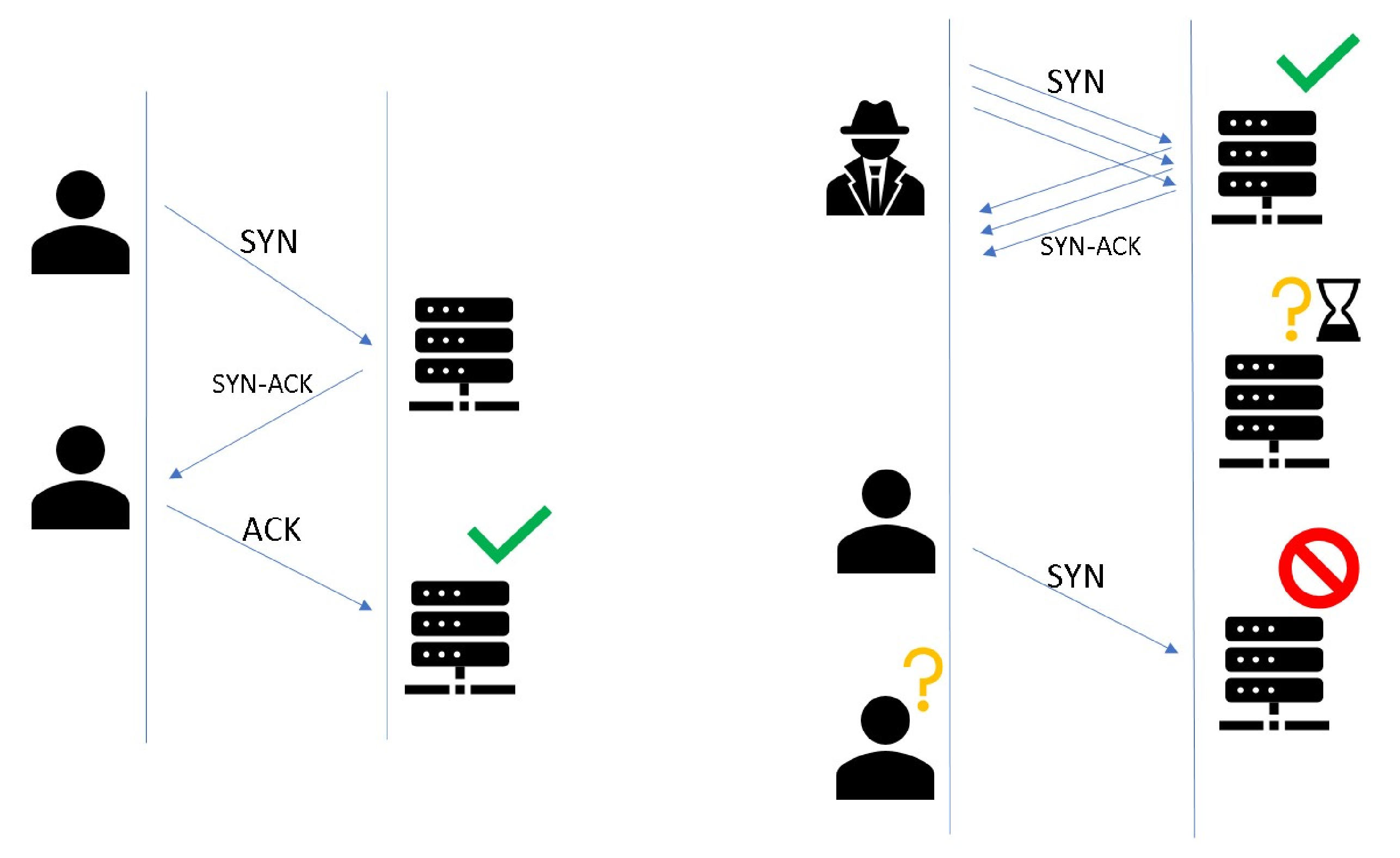}
\end{minipage}
\caption{{TCP: 3-way handshake and the SYN-flooding}}
\end{center}
\label{Fig:SYN-flood}
\end{figure}
In the TCP terminology, the action of asking a question is called SYN, the action of answering it is called SYN-ACK, and the acceptance of the answer is called ACK. Alice really likes this part, as it makes the TCP connections available, by closing the question-answer sessions with the answer acceptance. After an answer has been accepted, a TCP server establishes the TCP network socket and releases the protocol state, i.e., forgets the question. If Bob, however, never accepts any answer, then the TCP server has in principle to remember lots of questions, i.e., keeps lots of TCP protocol sessions open and, at one point, runs out of memory. That is the SYN-flooding attack, displayed in Fig.~\ref{Fig:SYN-flood} on the right.

This basic idea of a DoS attack on the Internet is very old, almost as old as the Internet, and there are in the meantime many methods to prevent it; but there are even more methods to circumvent these preventions. DoS attacks are a big business, both on the Internet and in everyday conversations between the 2-year-olds.

The theory presented so far just provides a formal model of security properties and suggests a basic classification. The next chapter spells out the security space where this classification turns out to be universal. 

\def\thechapter{5}
\setchaptertoc
\chapter{Geometry of security$^{\star}$}
\label{Chap:Geometry}

\section{Geometry?}\label{Sec:Geo-why}

Geometry is the science of space. Space is the way we observe things. Our physical space has 3 dimensions, and we usually subdivide it into cubes because we see and hear things up or down, left or right, forward or backward. Ants observe the world by way of smells and tastes, and their physical space is structured differently.  
\begin{figure}[h!t]
\begin{center}
\includegraphics[height=4.5cm]{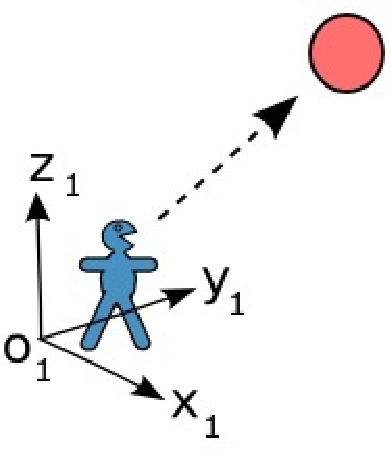}
\hspace{8em}\includegraphics[height=4.5cm]{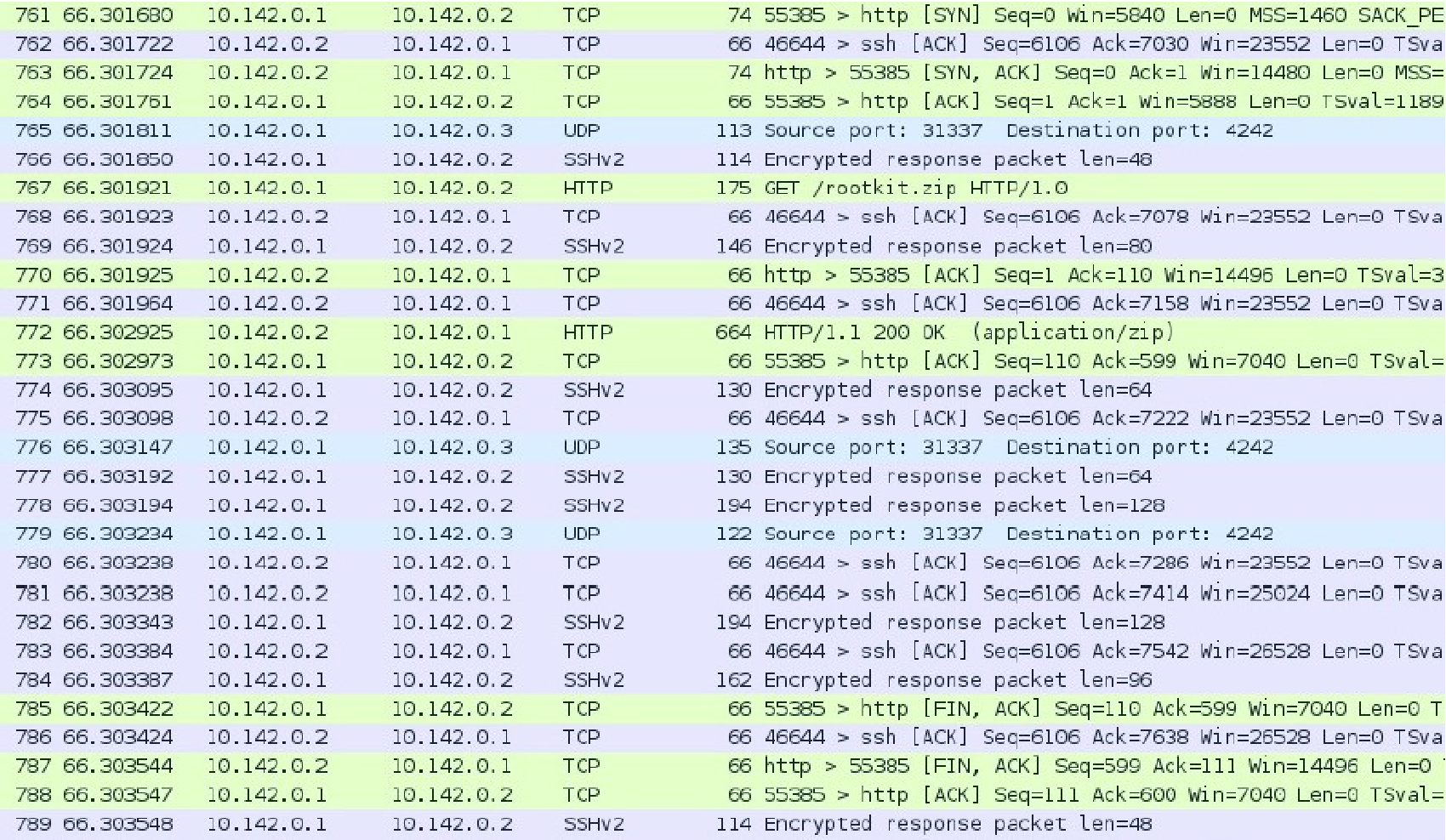}
\caption{An observation in physical space vs an observation in cyberspace.}
\label{Fig:space}
\end{center}
\end{figure}
Computers observe the world as traces of computations, as streams of data, as network logs, as listings of one sort or another. We have been calling such listings \emph{histories} because they have a known past and an unknown future, separated by the present.

\section{Geometry of histories and properties}\label{Sec:Geo-hist}

\para{Observations.} If we record a string of past events $ \vec a = \sseq{ a_0\ a_1\ a_2\ \cdots\  a_n   }\in \Event^\ast$, then we have observed that our history has the property 
\bear
\vec a & \hspace{-.5em}::\hspace{-.5em}& \Event^\ast\  \in\  \WP\left(\Event^\ast\right).\\
\uparrow &\hspace{-1em}& \uparrow\\
past &\hspace{-1em}& futures
\eear
This is an \emph{observation}. Observations are collected in the family of \emph{basic}\/ sets
\bea\label{eq:BBB}
\BBB & = & \left\{\vec a :: \Event^\ast\ |\ \vec a \in \Event^\ast\right\}.
\eea

\para{Open observables} represent possible observations of bad stuff (actions) that may happen. If such an observation occurs, it will persist, i.e., remain open under all futures. A family $\OOO$ of open observables is thus required to satisfy the following properties:
\begin{enumerate}[a)]
\item $\BBB \subseteq \OOO$ --- any observation is open;
\item $\XXX \subseteq \OOO\ \implies \ \bigcup \XXX \in \OOO$ --- any set of open observables corresponds to a possible observation, i.e., to an open observable;
\item\label{item:intersection} $U, V\in \OOO\  \implies\  U\cap V\in \OOO$ --- any finite set of open observables can be observed together as a single open observable.
\end{enumerate}
The set of observables $\OOO$ thus has to contain $\BBB$, and all unions of the elements of $\BBB$. It follows that
\bea\label{eq:OOO}
\OOO & =  \left\{ U \in\WP\Event^\ast \ |\ \vec x \in U \wedge \vec x \sqsubseteq \vec y \implies \vec y \in U \right\},
\eea
because $U \in\WP\Event^\ast$ satisfies $\vec x \in U \wedge \vec x \sqsubseteq \vec y \implies \vec y \in U$ if and only if it also satisfies $U = \bigcup \{\vec a::\Event^\ast\subseteq U\ |\ \vec a\in \Event^\ast\}$. It is easy to see that $\OOO$ from \eqref{eq:OOO} also satisfies \eqref{item:intersection} since already $\BBB$ is closed under $\cap$. According to \eqref{eq:OOO}, the observables $U\in \OOO$ are thus the \emph{upper-closed}\/ sets under the prefix order $\sqsubseteq$. This captures that they are \emph{open}\/ into all futures.

\para{Closed observables} represent observations that bad stuff (actions) has not happened so far; i.e., they are the complements $\Closed = \neg U$ of open observables $U\in \OOO$. Since for $\vec x\sqsubseteq \vec y$ the implication $\vec x \in U\implies \vec y\in U$ is equivalent with $\vec y \not \in U \implies \vec x\not \in U$, it follows that the family of closed observables must be in the form
\bea\label{eq:FFF}
\FFF & = & \left\{ F \in\WP\Event^\ast \ |\ \vec x \sqsubseteq \vec y \wedge \vec y \in F  \implies \vec x \in F \right\}.
\eea
The closed observables are thus the \emph{lower-closed}\/ sets under the prefix ordering. This captures that they are \emph{closed}\/ under the past: if bad stuff (actions) did not happen now, then that statement was also true at any moment in the past.

\para{Dense observables} are those that always remain possible: an observable is dense if it must be observed eventually, after any future-open observation. Each of the three previously defined families of properties can be used to characterize dense properties, which leads to three equivalent characterizations:
\bea\label{eq:DDD}
\DDD & = & \left\{ D \in\WP\left(\Event^\ast\right) \ |\ \forall \vec a\in\Event^\ast.\ D\cap \left(\vec a::\Event^\ast\right) \neq \emptyset \right\}
\\
& = & \left\{ D \in\WP\left(\Event^\ast\right) \ |\ \forall U\in \OOO.\ D\cap U= \emptyset\implies U=\emptyset \right\}
\notag\\
& = & \left\{ D \in\WP\left(\Event^\ast\right) \ |\ \forall F\in \FFF.\ D\subseteq  F\implies F=\Event^\ast \right\}\notag
\eea
Dense observations are used to characterize the {\nice} stuff (actions) that is required to eventually happen. The third characterization says that if {\nice} stuff (actions) is not bad, then nothing is bad; i.e., if a dense property is past-closed, then it contains all histories. The definition is then relativized to any $Y\subseteq \Event^\ast$ as
\bea\label{eq:DDDY}
\DDD_Y & = & \left\{ D \in\WP(Y) \ |\ \forall F\in \FFF.\ D\subseteq  F\implies Y\subseteq F \right\}.
\eea

\section{Geometry of safety and liveness}\label{Sec:Geo-saf}

\begin{proposition}
A history is
\begin{itemize}
\item[a)] safe if and only if it is closed: $\SAF\  = \ \FFF$ 

\item[b)] alive if and only if it is dense: $\LIV \  = 
\  \DDD$
\end{itemize}
\end{proposition}

\bpr  By inspection of \eqref{eq:safdef}$\stackrel{{\rm(a)}}\iff$\eqref{eq:FFF} and \eqref{eq:livdef}$\stackrel{{\rm(b)}}\iff$\eqref{eq:DDD}.
\epr

\bigskip
\begin{corollary}\label{Corollary:safe}
Any property $X\subseteq \Event^\ast$ factors through the induced lower set $\closure X  =   \left\{\vec y \sqsubseteq \vec x \in X\right\}$ 
\[
\begin{tikzar}[column sep = 1.5cm,row sep = 1.5cm]
X\ar[hook]{rr} \ar[two heads]{rd}{\DDD_{\closure X}}[swap]{\LIV_{\closure X}} \&\& \Event^\ast \\
\& \closure X \ar[tail]{ur}{\FFF}[swap]{\SAF}
\end{tikzar}
\]
The set $\closure X \subseteq \Event^\ast$ is closed in the sense of \eqref{eq:FFF} and  $X\subseteq \closure X$ is dense in the sense \eqref{eq:DDD}.
\end{corollary}

\bigskip
\para{Requirements are future-open properties.} Positive properties, e.g., in the form $\overline{b \prc c}$, can be expressed as unions of basic properties, e.g.
\bear
\overline{b \prc c} & = & \bigcup_{\vec x, \vec y \in \Event^\ast} \vec x:: b:: \vec y :: c:: \Event^\ast \ \in \  \OOO
\eear
and thus, remain open. Further properties can be expressed as finite intersections of those, e.g.
\bear
\overline{b_0 \prc b_1\prc \cdots b_n} & = & 
\overline{b_0 \prc b_1}\cap \overline{b_1 \prc b_2} \cap \cdots \overline{b_{n-1} \prc b_n}\ \in \ \OOO.
\eear

\para{Constraints are past-closed} because they correspond to negative requirements, e.g., 
\bear
\bigcap_{i\in\Obj} \neg \overline{i! \prc i?} \ \ \in \ \ \FFF & \mbox{ because } & \bigcup_{i\in\Obj} \overline{i! \prc i?}  \ \ \in \ \ \OOO
\eear
 This explains why safety properties from Sec.~\ref{Sec:safe-live} often occur as negative requirements.

\section{Geometry of security}\label{Sec:Geo-sec}

Security problems arise from conflicting goals: what is good for me is bad for you, and vice versa. We discussed in Sec.~\ref{Sec:Local} that they also arise from different standpoints, local observability, and hiding: I can deceive you if there is something that I see and you do not. 
\begin{figure}[ht!]
\begin{minipage}[t]{0.4\linewidth}
\centering
\includegraphics[height=6cm]{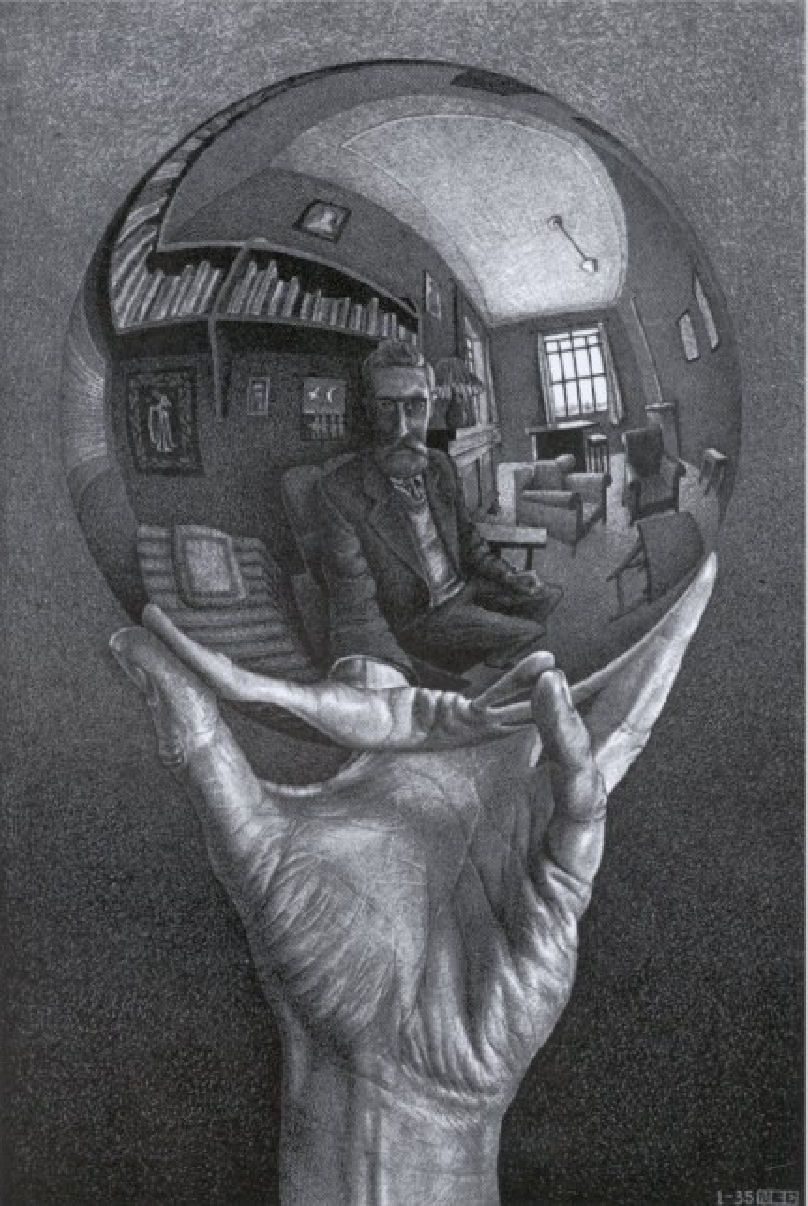}
\caption{Global observer:\\ ``I think, therefore I exist.''}
\label{Fig:cogito}
\end{minipage}
\hspace{0.1\linewidth}
\begin{minipage}[t]{0.45\linewidth}
\centering
\includegraphics[height=6cm]{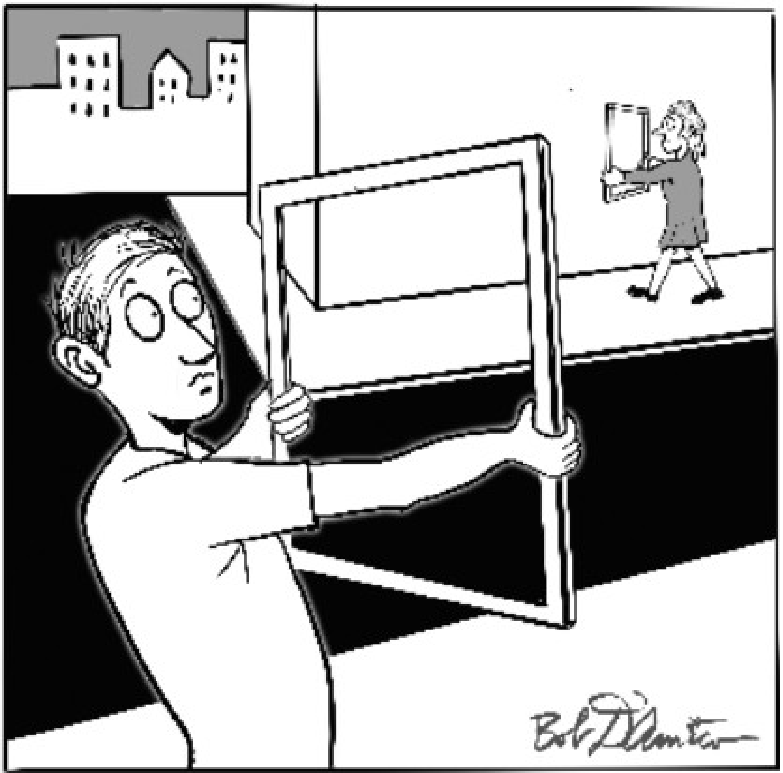}
\caption{Local observers:\\ relativistic frames of reference}
\label{Fig:frames}
\end{minipage}
\end{figure}
Burglary, as an attack on physical security, is more likely to succeed if there is no one home to see it. Cybersecurity often requires reconciling Alice's and Bob's views of transaction histories. National security also requires reconciling different views of history. Behind every security problem, there are different views of some space of histories.

The difference between cyber security and physical security is the difference between the underlying geometries. Physical security is based on physical distances and velocities: animals prevent attacks from predators by maintaining a distance that allows them to outrun the attacker. Cybersecurity  is not based on outrunning the attacker because the concept of distance is unreliable in cyberspace, as many copies of a message travel in parallel.   
\begin{figure}[h!t]
\begin{center}
\includegraphics[height=4.5cm]{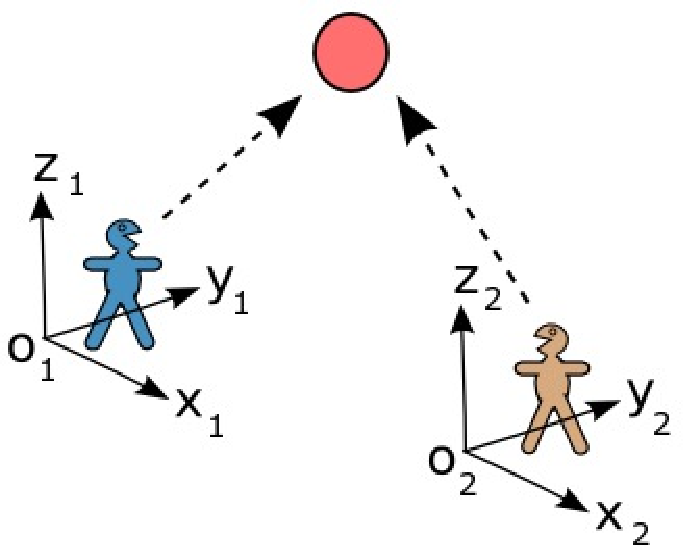}
\hspace{3em}\includegraphics[height=4.5cm]{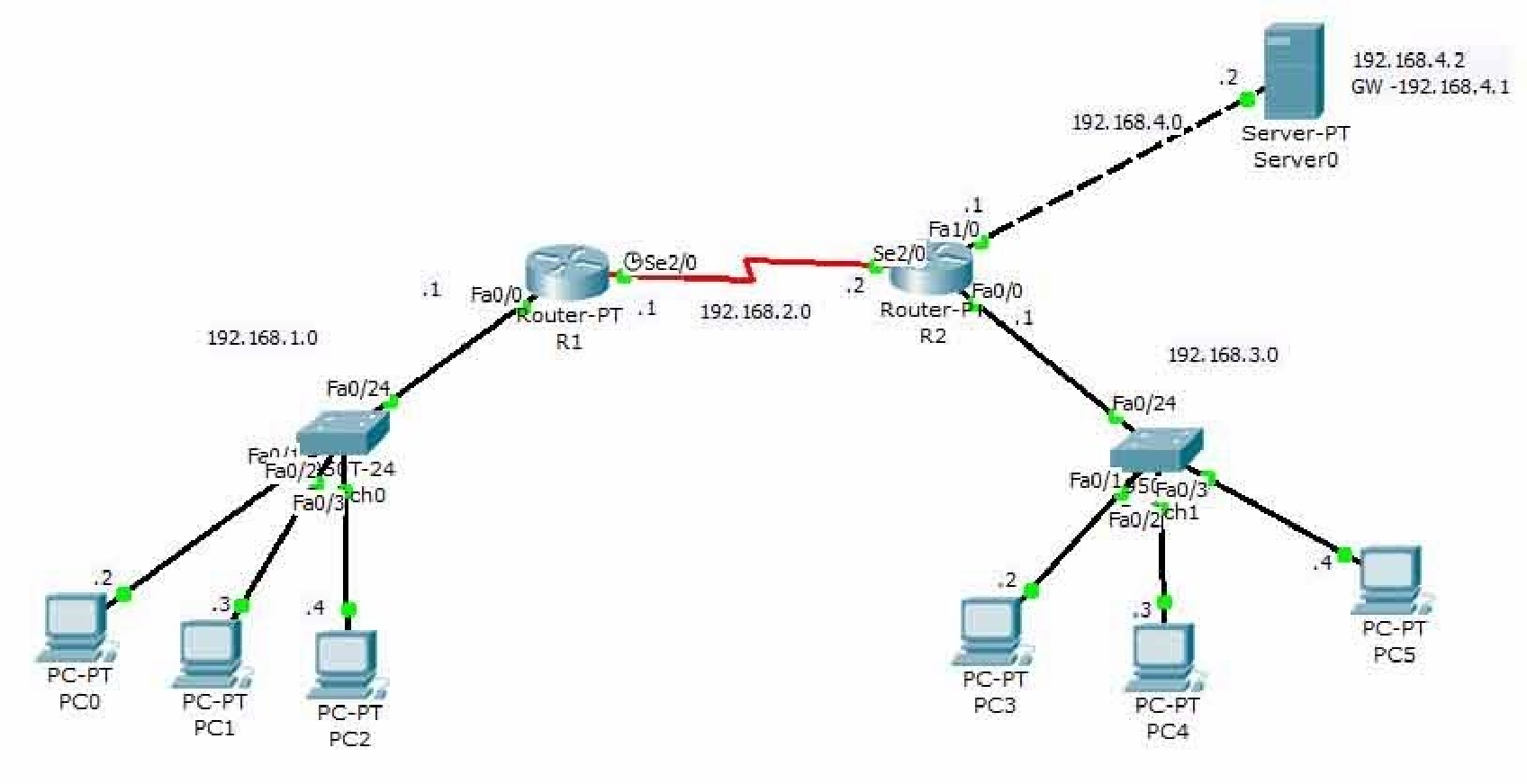}
\caption{Linking and separating frames of reference in physical space and in cyberspace}
\label{Fig:frame-linking}
\end{center}
\end{figure}
If the implementation details of network services are abstracted away, then every two nodes  in cyberspace look like neighbors because the underlying networks, in principle, do their best to route traffic as fast as possible. This geometric peculiarity of cyberspace, as a space where every two points are at the same negligible distance\footnote{In terms of a disc, as a set of points bounded by circle, as a set of points at equal distances from a center, cyberspace can be viewed as a disc whose center is everywhere, and whose bounding circle is nowhere. That very same geometric property was often proposed as the defining characteristic of God. The proposal is said to have originated from Hermes Trismegistos, but it was also repeated e.g., by Voltaire.}, is the source of many cybersecurity problems.

\section{Locality and cylinders}\label{Sec:Geo-loc}

In Sec.~\ref{Sec:au-av}, we analyzed how authority and availability can be construed as \emph{localized}\/ versions of safety and liveness. Here we spell out the geometric meaning of that observation. It is based on interpreting the purge \sindex{purge} operations $\restriction_u:\Event^\ast \to \Event_u^\ast$ \eqref{eq:purge} as spatial projections, that reduce global histories to local views. Instantiating the \emph{cylinder localizations}\/ from Sec.~\ref{Sec:localization} in the Appendix to the family of views $V = \left\{\restriction_u: \Event^\ast \to \Event_u^\ast\ |\ u\in \Subj\right\}$, the task of localizing security properties boils down to determining how well they are approximated by the corresponding cylinders, which are the local,  subjective views.

\bigskip
\begin{lemma}\label{Lemma:cylinder}
A property is localized in the sense of \sindex{property!localized} Def.~\ref{Def:localized} if and only if it is external cylindric in the sense of Def.~\ref{Def:cylindric} in the Prerequisites: 
\bear
\Local & = & \Cylout
\eear
\end{lemma}

\bpr We prove that the localization $\widehat Y$ from Def.~\ref{Def:localized} is the special case of cylindrification $\ana Y$ from Def.~\ref{Def:cylindric}:
\[\widehat Y \  \  =\  \  \left\{\vec y\ |\ \forall u\in \Subj.\ \vec y\restr_u \in Y_u
\right\}\ \
 = \ \ \bigcap_{u\in \Subj} \left\{\vec y\ |\ \vec y\restr_u \in Y_u\right\}\ \ =\ \ 
\ \  \bigcap_{u\in \Subj} u^*u_!(Y) \ \ = \ \ \ana{Y}
\] 
Hence $\Local \ = \ \{Y\in \WP(\Event^\ast)\ |\ Y = \widehat Y\} \ = \ \{Y\in \WP(\Event^\ast)\ |\ Y=\ana Y\} \ = \ \Cylout$.
\epr

\section{Geometry of authority and availability}
\label{Sec:Geo-avai}

\begin{proposition}\label{Prop:geo-auth-avail}
A history is
\begin{itemize}
\item[a)] authorized if and only if it is closed and cylindric: $$\AU \  = \   \left\{\Property\in \WP(\Event^\ast)\ |\ \Property \in \FFF \wedge \Property = \ana \Property\right\}\ =\  \ana\FFF$$

\item[b)] available only if its cylindrification is dense: 
$$\AV \ \subseteq \ \left\{\Property\in \WP(\Event^\ast)\ |\ \ana\Property \in \DDD\right\}\ =\ \ana{\DDD}^{-1}$$
\end{itemize}
\end{proposition}

\bpr (a) follows from \eqref{prop:auu} and Lemma~\ref{Lemma:cylinder}. (b) follows from \eqref{eq:availability}.
\epr

\bigskip
\begin{corollary}\label{Corollary:au-av}
The closure operator $\closure{\ana -} :\WP(\Event^\ast) \to \WP(\Event^\ast)$ factors any property $X\subseteq \Event^\ast$ through its cylindric closure $\closure{\ana X}   =   
\left\{\vec y \sqsubseteq \vec x\ |\ \forall u\in \Subj.\  \vec x\restr_u \in X_u\right\}$
\[
\begin{tikzar}[column sep = 1.5cm,row sep = 1.5cm]
X\ar[hook]{rr} \ar[dashed,thin,two heads]{rd} \ar[thick,two heads,crossing over]{rrd}[description]{\overline{\AV}} \&\& \Event^\ast \\
\& \closure X \ar[tail,dashed,thin]{ur} \rar[tail,thin] \& \closure{\ana X} \ar[tail,thick]{u}[description]{\AU}
\end{tikzar}
\]
where $\AU$ is the set of close cylindric sets $\ana{\closure X} \subseteq \Event^\ast$,  whereas $\overline \AV$ is the set of dense cylindric sets $X\subseteq \ana{\closure X}$.
\end{corollary}

\bigskip
\para{Remark.} The closure $\closure{\ana X}$ is obtained by sending $X$ through 
$\WP(\Event^\ast) \tto{\ana -} \WP(\Event^\ast) \tto{\ \closure{}\ } \WP(\Event^\ast)
$, which is the composite of the lower closure operator $\closure{}$ from Prerequisites~\ref{Prereq:Topology}\eqref{eq:closure}, instantiated to the space of histories, and the cylinder closure operator $\ana -$ from Prerequisites~\ref{Sec:localization}\eqref{eq:ana}. It is easy to check that $\ana{\closure X} = \closure{\ana X}$.

\section[Decompositions]{Normal decompositions of properties}
\label{Sec:Geo-decomp}

Corollaries \ref{Corollary:safe} and \ref{Corollary:au-av} establish how any specified constraint $X\subseteq \Event^\ast$ can be approximated by a safety constraint $\closure X$ or by an authority constraint $\ana{\closure X}$. This is useful because safety constraints and authority constraints are generally easier to implement and validate than arbitrary constraints. 

\para{Safety-driven decompositions.} The closure $\closure X$, in general, declares more histories to be safe than intended by $X$. This deviation of $\closure X$ from $X$ can, however, always be corrected by intersecting $\closure X$ with a liveness property in such a way that the intersection contains precisely the original property $X$. The histories that are in $\closure X$ but not in $X$ will thus be declared safe but not alive. This is a special case of the geometric decomposition from prerequisites section ~\ref{Prereq:Topology} of arbitrary sets into closed and dense sets.

\bigskip
\begin{proposition}\label{Prop:closed-dense}
Any property can be expressed as an intersection of a safety property and a liveness property:
\bear
X & = & \closure X\ \cap \ \left(X\cup \neg \closure X\right).
\eear
\end{proposition}

\bpr Recalling that in the space $\Event^\ast$ of traces the closure is $\closure X = \{\vec y \sqsubseteq \vec x \in X\}$, we have
\[
\begin{tikzar}[column sep = 1.5cm,row sep = 1.5cm]
\& X\cup \neg \closure X \ar[two heads]{dr}[swap]{\DDD}{\LIV}\\
X\ar[tail]{ur} \ar[equal]{r} \ar[two heads]{rd} \&\closure X \cap \left(X\cup \neg \closure X\right)\ar[hook]{d}\ar[hook']{u}\& \Event^\ast \\
\& \closure X \ar[tail]{ur}{\FFF}[swap]{\SAF}
\end{tikzar}
\]
\epr

\para{Example 1: sheep life.} If Alice wants to shear her sheep at most once every year, and only to use meat in the end, the requirement might be
\bear
\mbox{SheepLife} & = & \mbox{MilkWoolMeat}\ \cap\ \mbox{MilkMeatWoolAnnual} \hspace{4em} \mbox{ --- where}\\
\mbox{MilkWoolMeat} & = & \{\milk,  \mbox{wool}\}^\ast ::  \mbox{meat} \\
\mbox{MilkMeatWoolAnnual} & = & \{\milk, \mbox{meat}\}^\ast \cup \big[\underbrace{\milk\, \milk\,  \ldots\,  \milk}_{365\ {\rm times}}\  \mbox{wool}\big]::\mbox{MilkWoolAnnual}
\eear
The normal decomposition is then
\bear
\mbox{SheepLife} & = & \closure{\mbox{SheepLife}}\  \cap\ \left(\mbox{SheepLife} \ \cup \ \interior{\mbox{NoSheepLife}}\right)
\eear 
where 
\begin{itemize}
\item $\closure{\mbox{SheepLife}}$ consists of all histories where sheep's wool is only ever used after 365 consecutive milkings, and where sheep's meat is only ever used at the end, after which there are no further events; and
\item $\interior{\mbox{NoSheepLife}} = \neg \closure{\mbox{SheepLife}}$ consists of all histories where the rules of $\closure{\mbox{SheepLife}}$ are broken. 
\end{itemize}

\para{Remark.} Note that the liveness part of the SheepLife decomposition is thus unsafe, whereas the safety part of the decomposition is not alive. It is, of course, reasonable to expect this. Yet, when the liveness part and the safety part are tested independently, this sometimes causes problems.

\bigskip

\para{Authority-driven property decompositions.} The cylinder closure $\ana{\closure X}$ in Corollary~\ref{Corollary:au-av} is the smallest authorization property containing all histories from a given $X\subseteq \Event^\ast$. If a security designer wants to make sure that all histories contained in $X$ are authorized, he will thus use $\ana{\closure X}$. Some of the histories authorized in $\ana{\closure X}$ may not be in $X$; but they can be eliminated as unavailable, by intersecting  $\ana{\closure X}$ with an availability property, as spelled out in Prop.~\ref{Prop:au-av-decomp}. The desired property $X$ now remains as the set of precisely those histories that are both authorized and available, just like it remained as the set of histories that are both safe and alive in the safety-driven decomposition above.

\begin{proposition}\label{Prop:au-av-decomp}
Any property can be expressed as an intersection of authorization and availability:
\bear
X & = & \ana{\closure X} \cap \left(X\cup \neg \ana{\closure X}\right).
\eear
\end{proposition}

\bpr Recalling that the cylinder closure is $\ana{\closure X} = \{\vec y\in \Event^\ast\ |\  \exists \vec x\in X \ \forall u\in \Subj.\ \vec y\restr_u \sqsubseteq \vec x\restr_u\}$, and expanding the diagram from Corollary~\ref{Corollary:au-av}, we have
\[
\begin{tikzar}[column sep = 1.5cm,row sep = 1.5cm]
\&\& X\cup \neg \ana{\closure X} \ar[two heads]{d}[description]{\AV}\\
X\ar{urr} \ar[equal]{r} \ar{rrd} \&\ana{\closure X} \cap \left(X\cup \neg \ana{\closure X}\right)\ar[hook]{dr}\ar[hook']{ur}\& \Event^\ast \\
\&\& \ana{\closure X} \ar[tail]{u}[description]{\AU}.
\end{tikzar}
\]
\epr

\para{Availability-driven decomposition.} The decomposition in Prop.~\ref{Prop:au-av-decomp} is authority-driven because it is obtained by embedding the desired property $X$ into the smallest authorization property $\ana{\closure X}$ that contains it. One of the equivalent logical forms of authority derived in Sec.~\ref{Sec:au-av} was
\bea\label{eq:aau-conv}
\forall \vec x\vec z \in \Event^
\ast. \ \ \vec x \not \in \Closed\ \wedge \vec x \sqsubseteq \vec z & \implies & \exists u\in \Subj.\ \vec z \restr_u  \not \in \Closed_u
\eea
It says that any unauthorized event will be observed by someone, independently on others. \emph{\textbf{No cooperation is needed.}} On the other hand, the corresponding availability property 
\bea\label{eq:aavailability}
\forall \vec x\in \Event^\ast\  \forall u\in \Subj\ \ \exists \vec z \in \Event^\ast. \ \Big(  \vec x\restr_u \sqsubseteq \vec z\restr_u \ \wedge\  \ \vec z \in \Dense\Big)
\eea
says that in any situation there is a future for all \emph{together}. More precisely, it says that Alice on her own can find $\vec y_A$ such that $\vec x\restr_A :: \vec y_A$ is in $\Dense_A$; and Bob can find $\vec y_B$ on his own, such that $\vec x\restr_B :: \vec y_B$ is in $\Dense_B$; but they must come together to find  $\vec z \in \Dense$ such that $\vec z\restr_A = \vec x\restr_A :: \vec y_A$ and $\vec z\restr_B = \vec x\restr_B :: \vec y_B$. Finding such $\vec z$ requires scheduling local histories. 
\emph{\textbf{This requires cooperation.}}

A \emph{strong availability}\/ requirement, strengthening \eqref{eq:aavailability} so that the required actions can be realized by individual subjects with no need for cooperation, is
 \beq\label{eq:stongav}
 {\forall \vec x \in\Event^\ast\ \forall   \vec z\in\Event^\ast \ \exists u\in \Subj.\Big( \vec x\restr_u \sqsubseteq \vec z\restr_u\ \ \implies \ \ \vec z\in\Dense\Big)}
 \eeq
saying there is someone for whom a future is available \emph{independently of others}.

\bigskip
\begin{proposition}
Any property can be expressed as a union of a strong availability and an authority breach:
\bea\label{eq:union-decomp}
X & = & \cata{\interior{X}} \ \cup \ \Big(X \cap \neg \cata{\interior{X}}  \Big)
\eea
\end{proposition}
\bigskip

\bpr
The largest strong availability property contained in $X$ is
\begin{multline*}
\cata{\interior X} \ \  = \ \  
\left\{\vec y\ |\ \exists u\in \Subj.\ \vec y\restr_u \in X_u \Rightarrow \vec y \in \interior X\right\}
\ \   = \ \  \bigcup_{u\in \Subj} \left\{\vec y\ |\ \vec y\restr_u \in X_u \Rightarrow \vec y \in \interior X
\right\}\ \  = \\  \bigcup_{u\in \Subj} u^*u_*\left(\interior X\right) \ \  = \ \  \neg \ana{ \closure{\neg X}}
\end{multline*}
Since Prop.~\ref{Prop:au-av-decomp} gives
$\neg X \  = \ \ana{\closure{\neg X}} \ \cap \ \left(\neg X \cup \neg \ana{\closure{\neg X}}  \right)$, the claim follows by taking complements on both sides. \epr

%
%
%
%
%

\section{What did we learn about resource security?}
\label{Sec:Geo-lesson}

\para{Security properties of histories and interactions are like syntactic properties of sentences.} Security properties as families of trace properties are the ``syntactic types'' of the ``language'' of network interactions. A correctly typed security policy is like a syntactically well-formed sentence in this language. A syntactically incorrect sentence surely does not convey the intended meaning. But syntactic correctness does not guarantee that it does. An authorization policy that is not expressed in terms of an authorization property surely does not implement the desired authority; but an authorization policy expressed in a sound authority property does not guarantee the desired security effects. The meaning of a sentence, a program, or a protocol is up to the designer. The science of security only provides the tools to achieve the desired goals, just like the syntax of a language facilitates communication --- but does not create the content to be communicated.

The geometric view allows us to express any set as an intersection of a closed and a dense set. Since safety and liveness properties of computations turn out to be past-closed and dense sets, any property of computations $X$ can thus be approximated by a safety property, and precisely recovered by intersecting that safety property with a liveness property. 

In cyberspace, the authority properties turn out to be past-closed \emph{cylindric}\/ sets, whereas the availability properties are the sets whose cylindrifications are dense. Therefore, any resource security property can be approximated by an authority property, and precisely recovered by intersecting it with an availability property.

In this chapter, we formalized arbitrary, often complex, resource security requirements and decomposed them into normal forms. Verifying whether security requirements are satisfied can thus be simplified to standardized tests. In the next chapter, we shall see how non-local properties, that cannot be tested locally at all, can be established and assured through non-local interactions.

\def\thechapter{6}
\setchaptertoc
\chapter{Relational channels and noninterference}
\label{Chap:Channel}

\section{Networks and channels}\label{Sec:chanwhat}

\subsection{Networks}
We live in networks. \sindex{network} Our computers and devices are connected to electronic networks, our friends and families to social networks, our species to biological networks, and so on. The resources studied so far are the nodes of the various networks that we live in. The channels studied in this chapter are the network links.

A network is an abstraction of space. It abstracts away the   irrelevant features of space and displays some features of interest as network nodes. The network links present the connections or relations between the nodes. Fig.~\ref{Fig:london} 
\begin{figure}
\begin{center}
\includegraphics[width = 10.5cm
]{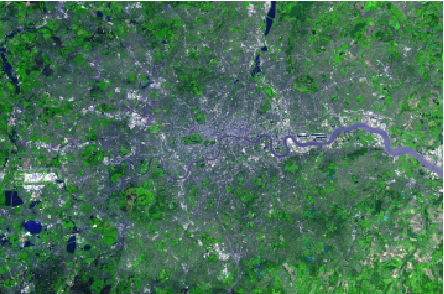}

\medskip
\includegraphics[width = 10.5cm
]{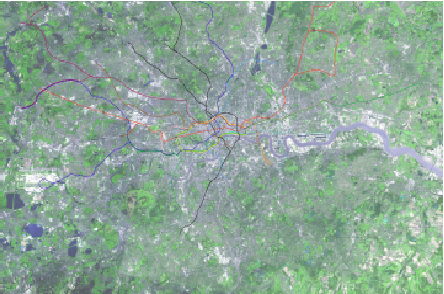}

\medskip
\includegraphics[width = 10.5cm
]{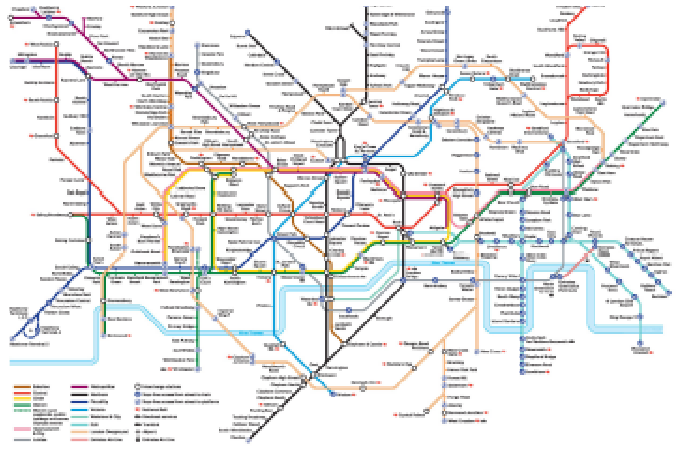}

\caption{The London Underground network is an abstraction of London}
\end{center}
\label{Fig:london}
\end{figure}
shows a physical space (the City of London) on the top, a network of interest (the London Underground) at the bottom, and the two projected over each other in the middle. The network nodes are the tube stations, whereas the network links are the railway connections between them. The traffic channels are realized by the trains traveling over the network links. They input the passengers who enter and output the passengers who exit. The passengers may enter through many entrances, but at each entrance, they enter one after the other. For simplicity, we assume that they all enter in a sequence (since that will be the case in most examples of interest here). But different passengers travel to different destinations, which are usually unknown. The traffic models, therefore, take into account possibilities or probabilities. Since each passenger's trip depends on the traffic load, corresponding to where the other passengers are going, the traffic channels take the sequences of entering passengers at the input but single exits at the output, taking into account all passengers that may possibly exit there, and even assigning the probabilities to each of them, if known. The fact that all passengers eventually exit the network is recovered by accumulating the sequences of single-exit outputs. 

We first focus on modeling a single channel and then expand the view to networks with many channels and the protocols regulating the network traffic. The individual model is formalized in Sec.~\ref{Sec:chandef}, and the cumulative view in Sec.~\ref{Sec:chanlist}.

\subsection{Channels}\label{Sec:chandef}
\para{Idea.} Channels transmit what we know, what we have, or what we are. Most organisms are equipped with senses to recognize each other and to display their traits to each other. The senses are the basic channels. Traffic channels transmit what we have, and communication channels transmit what we know. Our first communication channels were the body gestures and the speech. To be communicable, data must be encoded in a language. Communication channels transmit codes and languages. They are the content that the channels transmit. 

Channels are the media that transmit content. An electronic channel can be a copper wire that conducts electronic signals between two nodes in a circuit. The copper wire is the medium that carries the communication channel. The radio waves are the medium that carries the content of radio channels. Radio channels carry the radio communications. Physicists used to think that the radio waves were carried by a specific medium, the \emph{ether}, but the theory of relativity established that the only carrier of radio waves and light was the space itself. The theory of relativity says that space is a channel and that information cannot travel through it faster than light. All channels transmit information and can be used for communication. A communication channel can be implemented on top of a radio channel, or a copper wire, or even the English Channel as the medium. A social channel can be implemented on top of a communication channel. There are layers and layers of channels, implemented on top of each other, but the basic picture of each of them is always the same. 

\para{Channel as a history-sensitive function.} A channel \sindex{channel} is a function that inputs histories.  While an ordinary function is in the form $\Inp\to \Oup$, a channel is in the form $\Inp^{+}\to \Oup$, where $\Inp^{+}$ is the type of strings of events from $\Inp$, representing histories. (See Prerequisite~\ref{Prereq:list}.) A channel is thus a \emph{history-sensitive}\/ function. The other way around, a function is a \emph{memoryless}\/ channel, meaning that its output only depends on the most recent event, and not on the memory of earlier events.

\subsection{Possibilistic channels}
\para{Subsets and relations.} A relation \sindex{relation} is a function in the form $\Inp \to \WP\Oup$, where $\WP\Oup = \{Y\subseteq \Oup\}$ is the set of subsets of $\Oup$, its powerset. A relation is thus a function that may output multiple values, or nothing. A function is a single-valued, total relation.  The bijection $\WP\Oup \cong \{0,1\}^{\Oup}$, i.e., the correspondence
\beq\label{eq:hommset}
\prooftree
Y\in  \WP{\Oup}
 \Justifies
 \homm - _{Y}\colon \Oup\to \{0,1\}
 \endprooftree
 \qquad \mbox{ where } \homm y _{Y} = 1\  \iff \ y\in Y
\eeq
induces the bijective correspondence 
 \beq\label{eq:homset}
 \prooftree
 r\colon \Inp \to \WP{\Oup}
 \Justifies
 \hhom- r - \colon \Inp\times \Oup \to \{0,1\}
 \endprooftree
 \qquad \mbox{ where } \hhom x r y = 1 \ \iff\  y\in r_{x}
 \eeq
 A subset $Y$ of $\Oup$ can thus be viewed as a 01-vector $\homm - _{Y}\in \{0,1\}^{\Oup}$, whereas a relation $r$ from $\Inp$ to $\Oup$ can be viewed as a 01-matrix $\hhom- r - \in \{0,1\}^{\Inp\times \Oup}$.

\para{Function view.}\label{Sec:relchan} 
A possibilistic (or relational) channel \sindex{channel!relational} \sindex{channel!possibilistic@ see {relational}} is a binary relation between strings of inputs from $\Inp$ and outputs in $\Oup$. They can be viewed as functions $f\colon \Inp^{+}\to \WP\Oup$, mapping the input histories $\vec x = \seq{x_{0}x_{1}\ldots x_{n}}$ to the sets of \emph{possible}\/ outputs  $f_{\vec x}\subseteq \Oup$. 

\para{Matrix view.} \sindex{sequent} \sindex{channel!sequent notation} According to \eqref{eq:hommset}, the output subsets can  be equivalently viewed as functions to the set $\{0,1\}$.  According to \eqref{eq:homset}, the functions $f\colon\Inp^+\to \WP\Oup$ can be equivalently viewed as the matrices $\hhom - f - \colon \Inp^{+}\times \Oup\to \{0,1\}$ whose entries are the sequents $\hhom{x_{0}x_{1}\ldots x_{n}} f y$.  The bijective correspondence in \eqref{eq:homset} is thus instantiated to
\beq\label{eq:homchan}
 \prooftree
 f\colon \Inp^{+} \to \WP{\Oup}
 \Justifies
 \hhom- f - \colon \Inp^{+}\times \Oup \to \{0,1\}
 \endprooftree
 \qquad \mbox{ where } \hhom{x_{0}x_{1}\ldots x_{n}} f y = f_{x_{0}x_{1}\ldots x_{n}}(y)
 \eeq
When the channel name $f$ is clear from the context, we elide it and write $\hom{x_{0}x_{1}\ldots x_{n}} y$, or more succinctly as $\hom{\vec x}{y}=f_{\vec x}(y)$ for $\vec x = \seq{x_{0}x_{1}\ldots x_{n}}$. 

\para{Special case: deterministic channels.} \sindex{channel!deterministic} A deterministic channel is a relational channel $f$ where the inputs determine unique outputs, which means that $f_{\vec x}\subseteq \Oup$ is either a singleton $\{y\}$ or the empty set $\emptyset$. A deterministic channel is, therefore, in the form $f\colon \Inp^{+}\to \Oup\uplus 1$, where $1=\{\emptyset\}$ and $\uplus$ is the disjoint union. It can be viewed as the \emph{partial}\/ function $f\colon \Inp^{+}\pto \Oup$, considered undefined when the output is empty. The empty output must be allowed if computable channels are to be taken into account, since computations may search infinitely long and not return any output.

\para{Special case: memoryless channel.} \sindex{channel!memoryless} While a channel is a \emph{history-sensitive}\/ function, a function is a \emph{memoryless}\/ channel, whose output only depends on the most recent input, i.e.
\bea\label{eq:memoryless}
f_{x_{0}x_{1}\ldots x_{n}} & = & f_{x_{n}}. 
\eea

\para{Examples.} The English Channel links the North Sea and the Atlantic Ocean. It is a deterministic channel, because the same ships that enter eventually exit, except those that sink. A channel for land traffic can be a freeway between two cities or a railway line between two stations. The sequence of cars  $\seq{x_{0} x_{1}\ldots x_{n}}$ that enters the freeway in the city $\Inp$ is generally different from the set of cars $Y= \{y_{0}, y_{1},\ldots,y_{m}\}$ that may exit in the city $\Oup$, since some cars that entered go to other cities, and some that exit come from other cities. The set $Y$ is not ordered because its elements are the cars that may exit at a given moment. Below we introduce an equivalent relational channel formalism, which also displays the sequence of outputs on a given sequence of inputs. The relation $\hom {x_{0} x_{1}\ldots x_{n}}{y}$ means that the car $y$ may exit the freeway after the sequence of inputs $\seq {x_{0} x_{1}\ldots x_{n}}$.

\subsection{Possibilistic channel types and structure}
\subsubsection{Cumulative channels}\sindex{channel!cumulative}\label{Sec:chanlist}
\para{Chatbots as channels.} 
A chatbot \sindex{chatbot} inputs prompts as sequences of words\footnote{They actually partition words into \emph{tokens}, and also take punctuation into account, also as tokens. Here, it only matters that the contexts are sequences. It does not matter whether they are sequences of words or sequences of tokens. For simplicity, we ignore the difference.}\sindex{words vs tokens}\sindex{tokens vs words} and outputs the corresponding responses, which are also sequences of words. But the crucial point is that it does not generate the response all at once. It first guesses the most likely first word of the response, which it then adds to the context and guesses the most likely next word, and so on. So the chatbot, in principle, inputs sequences of words and outputs words, one at a time. Although some words are more likely than others, taking into account just which words $y\in \Oup$ are \emph{possible}\/ as continuations of the contexts $\vec x\in \Inp^{+}$ gives a channel $\Inp^{+}\to \WP\Oup$. Later, in Ch.~\ref{Chap:Info}, we will spell out tools that will allow us to take into account the different \emph{probabilities}\/ of different words in the same context. In any case, \emph{a chatbot is a channel}. But taking just a single word of the chatbot output into account provides a poor picture of its performance. The \emph{miracle of language}\/ is that the single words that we say form sentences; and that sentences form stories, and conversations. The chatbot miracle boils down to accumulating the single word outputs of the channel $\Inp^{+}\to \WP\Oup$ and presenting it in the form $\Inp^{\ast}\to \WP\Oup^{\ast}$, mapping the prompts as contexts to the sequences of words that are its responses. This is the \emph{cumulative}\/ form of the channel. 

\para{Accumulating and projecting strings.} The string constructors \sindex{string!constructor} given in Prerequisites~\ref{Prereq:list} induce the bijection
\beq
\begin{tikzar}[column sep = large]
X^{+}\times X \ar[bend left]{r}[pos=0.45]{(\cons)} \ar[phantom]{r}[description,pos=0.4]{\cong} \& X^{+} \ar[bend left]{l}[pos=0.53]{<\pxx,\sxx>}
\end{tikzar}
\eeq
where 
\bear
\pxx\seq{x_{1}\, x_{2}\ldots x_{n-1}\, x_{n}} & = & \seq{x_{1}\, x_{2}\ldots x_{n-1}},\\
\sxx\seq{x_{1}\, x_{2}\ldots x_{n-1}\, x_{n}} & = & x_{n}. 
\eear
The condition in \eqref{eq:memoryless} can \sindex{channel!memoryless} now be written  in the form
\bear
f_{\vec x}  & = &  f_{\sxx(\vec x)}.
\eear
A channel $f$ is thus memoryless if and only if the following diagram commutes 
\beq\label{eq:memoryless2}
\begin{tikzar}[column sep = 5pc]
\Inp^+ \arrow[two heads]{d}[swap]{\sxx} \arrow{drr}{f}\\
\Inp \arrow[tail]{d}[swap]{(-)}\&\& \WP\Oup\\
\Inp^+ \arrow{urr}[swap]{f}
\end{tikzar}
\eeq
with the embedding $(-)\colon X\mono X^+$ from  Prerequisites~\ref{Prereq:list}\eqref{eq:string}.

\para{Accumulating and projecting channels.} The string constructors given in Prerequisites~\ref{Prereq:list} induce the maps
\beq\label{eq:cumul}
\begin{tikzar}[column sep = large]
\Big\{\Inp^{+}\to \WP\Oup\Big\} \ar[bend left,tail]{r}{\ladj{(-)}} \& \Big\{\Inp^{\ast}\to \WP\Oup^{\ast}\Big\} \ar[bend left, two heads]{l}{\radj{(-)}}
\end{tikzar}
\eeq
where $f\colon \Inp^{+}\to \WP\Oup$ and $g\colon \Inp^{\ast}\to \WP\Oup^{\ast}$ are mapped to
\begin{center}
\begin{tabular}{cc}
$\ladj f()  = \big\{()\big\}$\hspace{8em}		& 
	\multirow{2}{*}{$\radj g\left(\vec x\right) = \begin{cases} \emptyset & \mbox{ if } g(\vec x) \subseteq\big\{()\big\} \\
\sxx\left(g(\vec x)\right) & \mbox{ otherwise}
\end{cases}$}					\\
$\ladj f\left(\vec x \cons a\right)  =  \ladj f(\vec x) \cons f\left(\vec x \cons a\right)$\hspace{5em}		&		\\
\end{tabular}
\end{center}
where the concatenation $(\cons)$ and the projection $\sxx$ are extended from elements to subsets
\[\prooftree
\Oup^{\ast}\times \Oup \tto{(\cons)} \Oup^{\ast}
\justifies 
\WP\Oup^{\ast}\times \WP\Oup \tto{(\cons)} \WP\Oup^{\ast}
\endprooftree
\qquad\qquad 
\prooftree
\Oup^{+} \tto{\sxx} \Oup
\justifies 
\WP\Oup^{+} \tto{\sxx} \WP\Oup
\endprooftree
\]
along the direct images so that 
\[ \ladj f (\vec x)\cons f(\vec x\cons a) = \left\{\vec y \cons b\ |\ \vec y\in \ladj f(\vec x), b\in f(\vec x\cons a)\right\}\qquad \quad \sxx\left(g(\vec x)\right) = \{\sxx\left( \vec y\right)\ |\ \vec y \in g(\vec x)\}\]
It is easy to verify that $\radj{\left(\ladj f\right)} = f$ always holds, whereas  $\ladj{\left(\radj g\right)} = g$ holds if and only if $g$ preserves the list lengths. 

\para{Sequent notation for cumulative channels.} \sindex{channel!cumulative!sequent notation} To capture correspondence \eqref{eq:cumul}, the notation from Sec.~\ref{Sec:chandef} extends to
\bea\label{eq:cumulseq}
 \hom{x_0 \ldots  x_n}{y_0 \ldots  y_n}\  &= &  \hom{x_0}{y_{0}}\cdot \hom{x_0 x_{1}}{y_{1}}\cdot \hom{x_0 x_{1}x_{2}}{y_{2}}\cdots  \hom{x_0 \ldots  x_n}{y_n}
\eea

\subsubsection{Continuous channels}\label{Sec:contchan}\sindex{channel!continuous possibilistic}
\begin{definition}\label{Def:contchan}\sindex{channel!continuous}
A\/ \emph{continuous (possibilistic) channel} is a function $\gamma\colon \WP\Inp^{\ast}\tto\cup \WP\Oup^{\ast}$ which preserves all unions and maps finite sets to finite sets:
\bea
\gamma\left(\bigcup \UUU\right) \ \  = \ \  \bigcup_{U\in \UUU}\gamma(U) &\quad\mbox{ and }\quad&  \gamma\Big(\{\vec x\}\Big) = \Big\{\vec y_{1}, \vec y_{2},\ldots, \vec y_{m}\Big\}
\eea 
\end{definition}

\para{Channel continuation and restriction.} 
The bijection between the cumulative channels $g\colon \Inp^{\ast}\to \WP\Oup^{\ast}$  and the continuous channels $\gamma\colon \WP\Inp^{\ast}\tto\cup \WP\Oup^{\ast}$ is in the form
\beq\label{eq:contchan}
\begin{tikzar}[column sep = large]
\Big\{\Inp^{\ast}\to \WP\Oup^{\ast}\Big\} \ar[bend left]{r}{(-)^{\#}}\ar[phantom]{r}[description]{\cong} \& \Big\{\WP\Inp^{\ast}\tto{\cup} \WP\Oup^{\ast}\Big\}  \ar[bend left]{l}{(-)_{\#}}
\end{tikzar}
\eeq
where the transformations
\bea
\prooftree 
\Inp^{\ast}\tto g \WP\Oup^{\ast}
\justifies
\WP\Inp^{\ast}\tto{g^{\#}} \WP \Oup^{\ast}
\endprooftree
&\qquad\qquad& 
\prooftree
\WP\Inp^{\ast}\tto \gamma \WP\Oup^{\ast}
\justifies
\Inp^{\ast}\tto{\gamma_{\#}} \WP \Oup^{\ast}
\endprooftree
\eea
define
\begin{align*}
g^{\#}(U) & =  \bigcup_{\vec x\in U} g(\vec x) & \gamma_{\#}\left(\vec x\right) & =  \gamma\left(\{\vec x\}\right)
\end{align*}

\para{Sequent notation for continuous channels.} \sindex{channel!continuous!sequent notation} To capture correspondence \eqref{eq:contchan}, the sequent notation is further extended from \eqref{eq:cumulseq} to $U\in \WP\Inp^{\ast}$ and $V\in \WP\Oup^{\ast}$
\bea\label{eq:contseq}
 \hom{U}{V}\  &= &  \bigvee_{\substack{\vec x\in U\\
 \vec y \in V}}\hom{\vec x}{\vec y}
 \eea

\subsubsection{Overt and covert channels}
\sindex{channel!covert}\sindex{channel!overt}
An \textbf{overt} channel is provided to serve an overtly specified intent. For example, a phone is an overt channel for conversations. The security checkpoint at the airport is an overt channel for passengers. 

A \textbf{covert} channel is, on the other hand, used covertly and against the specified intent. For example, the phone can be used as a covert channel if Alice and Bob establish a secret code to transmit a confidential message. Say, if Alice calls at midnight and hangs up as soon as Bob picks up, the message means that Bob should leave the front door open. While the overt constraint for security checkpoints is that no more than 3.4 oz of liquid should be permitted into the secure area, Alice and Bob can establish a covert channel by pooling their 3.4 oz containers together, to transmit 6.8 oz liquid into the secure area.

While overt channels can be unsafe, the purpose of channel security is to prevent covert channels.

\section{Channel safety and inference}

\subsection{Observation and inference in overt channels}
Channels are often used to model \emph{causation}.\sindex{causation}  A channel of type $\Inp^{+}\to \WP\Oup$ is then thought of as a process of observing the effects of type $\Oup$ caused by sequences of events of type $\Inp$. Typically, the input causes in $\Inp$ are unobservable, and the task is to infer information about them from the output effects in $\Oup$. For example, a microscope and a  telescope are such channels, making the invisible visible. Galileo's telescope is in \cref{Fig:newton}, together with Newton's cradle, which is a channel displaying the invisible transmission of force along a sequence of adjacent stationary metal balls. 
\begin{figure}[h!t]
\begin{center}
\includegraphics[width=9cm
]
{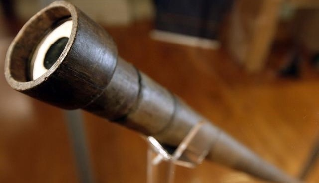}

\vspace{2ex}
\includegraphics[width=9cm]
{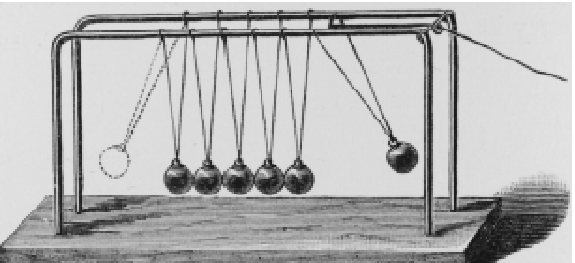}
\caption{A channel transmits $\left(\mbox{Unobservable causes}\right)^{+}\to \left(\mbox{Observable effects}\right)$}
\label{Fig:newton}
\end{center}
\end{figure} 
Most scientific instruments are channels that transform unobservable inputs into observable outputs. They are black boxes enclosing causations. The question: \emph{``What can be inferred about the unobservable causes of this observed effect?''} --- becomes the problem of channel decoding:
\begin{quote}{What can be inferred about the channel inputs from the channel outputs?}\end{quote} 
This makes a channel into the principal tool of \emph{\textbf{inductive inference}}\sindex{inductive inference}. \sindex{channel!inference} and the cumulative channels from Sec.~\ref{Sec:chanlist} into a rudimentary model of empiric induction \cite{MacKayD:book}. 

\subsection{Inverse channels}\label{Sec:chaninverse}
\sindex{channel!inverse}
If a {\relational} channel $g\colon\Inp^{\ast}\to \WP\Oup^{\ast}$ maps causal contexts $\vec x\in \Inp^{\ast}$ to sets of their possible effects $f_{\vec x} \subseteq \Oup^{\ast}$, then the task of deriving causes from effects, i.e., channel inputs from its outputs requires that we construct the inverse channel $\tilde g\colon\Oup^{\ast}\to \WP\Inp^{\ast}$ defined as
\bear
\tilde g_{\vec y} & = & \{\vec x \in \Inp^{\ast}\ |\ \vec y\in g_{\vec x}\}.
\eear
Given in the matrix form, the inverse of the relational channel $\hhom - g - \colon \Inp^{\ast}\times \Oup^{\ast}\to\{0,1\}$ is just
\bear
\hhom{\vec y}{\tilde g}{\vec x} & = & \hhom{\vec x}{g}{\vec y}.
\eear
  
\subsection{Unsafe channels}
Based on the black boxes of causation, inductive inference is always hypothetical. The claim that a sequence of events $x_{0}x_{1}\ldots x_{n}$ causes the event $y$ means that the correlation $\hom {x_{0}x_{1}\ldots x_{n}} y$ has been frequently observed on the input and the output of a channel black box. Such a claim can never be definitely proven, since further observations may always disprove it. It is accepted as inductively validated only until it is disproven. 

History of science is the history of disproven hypotheses, stated as channel correlations $$\hom {x_{0}x_{1}\ldots x_{n}} y.$$ Fig.~\ref{Fig:phrenology} shows the diagram illustrating the hypothesis that the properties $y$ of a person's mind are correlated with the shapes $x_{0}, x_{1}, \ldots,  x_{n}$ of their skull. The claim was that the scull provides an overt channel\sindex{channel!overt} into the human mind.
\begin{figure}[h!t]
\begin{center}
\includegraphics[height=7cm
]
{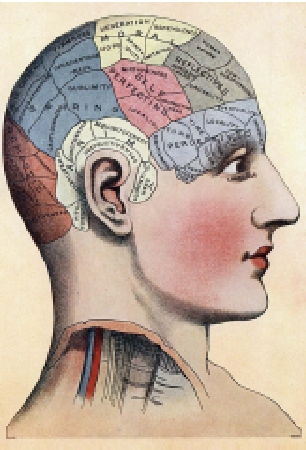}
\caption{Phrenological map of human mind}
\label{Fig:phrenology}
\end{center}
\end{figure} 
This hypothesis, which went under the name \sindex{phrenology} \emph{phrenology}, was conveniently viewed as true by many people through the XIX century. 
Although disproved early on, it continued to be attractive for the same reason for which many people still follow predictions of their horoscopes. The human mind is an ongoing quest for unobservable causes of everything we observe. When we cannot find some likely causal explanations supported by experience, we contrive them. We are primed to think in terms of causal channels $\hom{\vec x}{\vec y}$, and we construct them, no matter what.

Chatbots are primed to answer questions, and when they cannot find some likely answers supported by their training datasets, they contrive them. Their artificial mind is an ongoing process of generation for which the underlying language models are trained. They are required to extrapolate the given context and generate a response to every prompt no matter what. Hence the phenomenon of chatbots'  \emph{hallucinations}\sindex{chatbot!hallucination}. A hallucinating chatbot is unsafe because a contrived correlation $\hom{\vec x}{\vec y}$ of your prompt $\vec x$ and the chatbot's response $\vec y$ may cause bad stuff to happen. In general, \sindex{unsafety!of a channel} \sindex{channel!unsafe} a channel is unsafe when it outputs observables $\vec y$ that are correlated with the inputs $\vec x$ by flawed inferences or superstitions \cite[Ch.~4]{PavlovicD:LangEng}.

\section{Channel security and interference}

Security problems arise from sharing goals,  views, or resources; or rather from not sharing. Alice wants to do this, and Bob wants to do that, and they cannot do both. Or they both need the same thing, but only one can have it. A channel can be viewed as a resource, a tool for acquiring or transmitting information. Alice may need to protect her private information transmitted by the shared channel. To circumvent the protections, Bob may construct a covert channel within the overt channel shared with Alice. Alice's private information then covertly leaks to Bob through a channel within a channel.

\para{Idea of interference.}\sindex{interference} If Alice is using a channel on her own and no one else needs it, there are no security problems with it. Channel security problems arise when Alice and Bob share a channel. Each of them enters their inputs and observes their outputs. The channel $f\colon\Inp^{+}\to \Oup$ receives the inputs as they are entered, and its input histories are the shuffles of Alice's inputs with Bob's inputs, say 
\bear
\vec x & = & \seq{x_{0}^{A}x_{1}^{A}x_{2}^{B}x_{3}^{B}x_{4}^{B}x_{5}^{A}x_{6}^{B}\ldots}
\eear
After each of Alice's inputs, the channel produces an output for Alice; after each of Bob's inputs, it produces an output for Bob. Although each of them only observes their own outputs, Bob's output $y^{B}=f(x_{0}^{A}x_{1}^{A}x_{2}^{B})$ depends not only on his input $x_{2}^{B}$ but also on Alice's inputs $x_{0}^{A}$ and $x_{1}^{A}$. Therefore, Bob's input $x_{2}^{B}$ may cause the output $y^{B}$ in Alice's context $x_{0}^{A}x_{1}^{A}$, but it may cause a different output $\overline y^{B}$ if Alice enters a different context $\overline x_{0}^{A}\overline x_{1}^{A}$. When that happens, we say that the private inputs \emph{interfere}\/ in the shared channel. \sindex{channel!interference} Interference is the main problem of channel security. Bob cannot directly derive Alice's inputs $x_{0}^{A}x_{1}^{A}$ or $\overline x_{0}^{A}\overline x_{1}^{A}$ from his outputs $y^{B}$ and $\overline y^{B}$, but he can derive that Alice has entered different inputs. That is one bit of information about Alice's inputs. If the inputs are 0 or 1, then there are just 4 possible inputs $x_{0}^{A}x_{1}^{A}$. If Bob can get 2 bits of information about Alice's inputs, he can guess them. 

In this section, we formalize the \emph{security problem of \textbf{interference}}. In the next section, we characterize the \emph{security requirement of \textbf{noninterference}}.

\subsection{Example: Elevator as a shared channel}\label{Sec:elev-inter}
The elevator in Fig.~\ref{Fig:elev} inputs the calls to the floors $0, 1,$ up to $n$ represented as the events of type 
\bear
\Inp & = & \Big\{\send i\ |\ 0\leq i \leq n \Big\}.
\eear
It outputs the services 
\bear
\Oup & = & \{ \recv{\comme i}, \recv{\sstay i}\ |\  0\leq i \leq n\}
\eear
where 
\begin{itemize}
\item $\recv{\comme i}$ means: \emph{"The elevator comes to floor $i$"}, whereas
\item $\recv{\sstay i}$ means: \emph{"The elevator is already on floor $i$".}
\end{itemize}
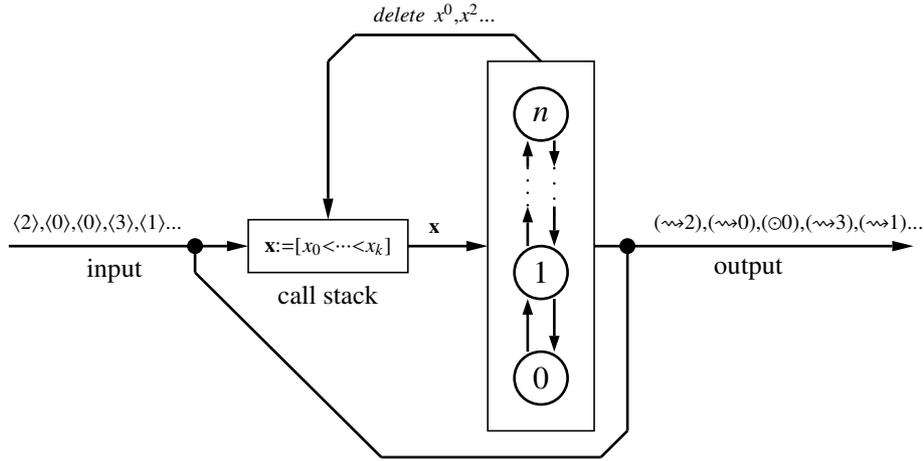
\begin{figure}[!h]
\newcommand{\remove}{\scriptstyle delete\ x^0, x^2\ldots}
\newcommand{\add}{\scriptstyle\vec x}
\newcommand{\xss}{\scriptstyle \left<{2}\right>, \left<0\right>,\left<0\right>,\left<3\right>,\left<1\right>\ldots}
\newcommand{\calls}{\scriptstyle\vec x := \left[x_0 \lt \cdots\lt x_k\right]}
\newcommand{\callls}{\footnotesize input}
\newcommand{\trace}{\footnotesize output}
\newcommand{\stack}{\footnotesize call stack}
\newcommand{\ffloor}{0}
\newcommand{\sfloor}{1}
\newcommand{\tfloor}{n}
\newcommand{\yss}{\scriptstyle\left(\comme 2\right), \left(\comme 0\right),\left(\sstay 0\right),\left(\comme 3\right),\left(\comme 1\right)\ldots}
\begin{center}
\def\JPicScale{.7}
\ifx\JPicScale\undefined\def\JPicScale{1}\fi
\psset{unit=\JPicScale mm}
\psset{linewidth=0.3,dotsep=1,hatchwidth=0.3,hatchsep=1.5,shadowsize=1,dimen=middle}
\psset{dotsize=0.7 2.5,dotscale=1 1,fillcolor=black}
\psset{arrowsize=1 2,arrowlength=1,arrowinset=0.25,tbarsize=0.7 5,bracketlength=0.15,rbracketlength=0.15}
\begin{pspicture}(0,0)(165,82.5)
\psline[linewidth=0.5,arrowlength=2,arrowinset=0]{<-}(85,40)(70,40)
\psline[linewidth=0.5,arrowlength=1.5,arrowinset=0,bracketlength=0.65]{->}(92.5,55)(92.5,60)
\psline[linewidth=0.5,arrowlength=1.5,arrowinset=0,bracketlength=0.65]{->}(97.5,60)(97.5,55)
\psline[linewidth=0.5,arrowlength=1.5,arrowinset=0,bracketlength=0.65]{->}(92.5,20)(92.5,30)
\psline[linewidth=0.5,arrowlength=1.5,arrowinset=0,bracketlength=0.65]{->}(97.5,30)(97.5,20)
\psline[linewidth=0.5,linestyle=dotted,dotsep=2](92.5,47.5)(92.5,55)
\psline[linewidth=0.5,linestyle=dotted,dotsep=2](97.5,55)(97.5,45.88)
\rput{0}(95,15){\psellipse[linewidth=0.5](0,0)(5,-5)}
\pspolygon[](85,75)(105,75)(105,5)(85,5)
\rput{0}(95,35){\psellipse[linewidth=0.5](0,0)(5,-5)}
\rput{0}(95,65){\psellipse[linewidth=0.5](0,0)(5,-5)}
\psline[linewidth=0.5,arrowlength=1.5,arrowinset=0,bracketlength=0.65]{->}(92.5,40)(92.5,47.5)
\psline[linewidth=0.5,arrowlength=1.5,arrowinset=0,bracketlength=0.65]{->}(97.5,47.5)(97.5,40)
\pspolygon[](40,45)(70,45)(70,35)(40,35)
\rput(55,40){$\calls$}
\psline[linewidth=0.5,arrowlength=2,arrowinset=0]{<-}(40,40)(-5,40)
\rput{0}(30,40){\psellipse[linewidth=0.5,fillstyle=solid](0,0)(1.25,-1.25)}
\psline[linewidth=0.5](65,0)(30,35)
\psline[linewidth=0.5](106.25,0)(65,0)
\psline[linewidth=0.5,arrowlength=2,arrowinset=0]{<-}(165,40)(105,40)
\psline[linewidth=0.5](106.25,0)(111.25,5)
\rput{0}(111.25,40){\psellipse[linewidth=0.5,fillstyle=solid](0,0)(1.25,-1.25)}
\psline[linewidth=0.5](111.25,5)(111.25,40)
\psline[linewidth=0.5](95,75)(90,80)
\psline[linewidth=0.5](90,80)(60,80)
\psline[linewidth=0.5,arrowlength=2,arrowinset=0]{<-}(55,45)(55,75)
\rput[B](75,82.5){$\remove$}
\rput[b](75,42.5){$\add$}
\rput[br](27.5,42.5){$\xss$}
\rput[bl](116.25,42.5){$\yss$}
\rput(95,15){$\ffloor$}
\rput(95,35){$\sfloor$}
\rput(95,65){$\tfloor$}
\psline[linewidth=0.5](30,35)(30,40)
\psline[linewidth=0.5](60,80)(55,75)
\rput[tr](20,37.5){\callls}
\rput[tl](127.5,37.5){\trace}
\rput[t](55,32.5){\stack}
\end{pspicture}
\caption{Elevator channel}
\label{Fig:elev}
\end{center}
\end{figure}
The elevator can be viewed as a deterministic channel $\Inp^+ \tto{ele} \Oup$  realized by the state machine in Fig.~\ref{Fig:elev-STM}. Note that it is assumed that the outputs are produced both at each transition and at each state.\footnote{This mixture of Moore and Mealy machine outputs can be reduced to either paradigm using the usual translation between the two.}

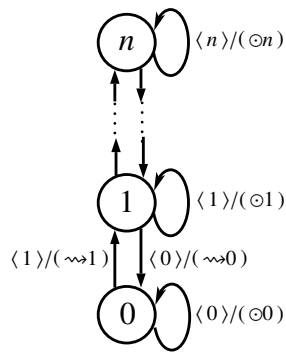
\begin{figure}[!hb]
\begin{center}
\newcommand{\oneleft}{$0$}
\newcommand{\onemid}{$1$}
\newcommand{\oneright}{$n$}
\newcommand{\transone}{$\scriptstyle \send 0/\recv{\sstay 0}$}
\newcommand{\transtwo}{$\scriptstyle \send 1/ \recv{\comme 1}$}
\newcommand{\transthree}{$\scriptstyle \send 1/ \recv{\sstay 1}$}
\newcommand{\transfive}{$\scriptstyle \send n/ \recv{\sstay n}$}
\newcommand{\transfour}{$\scriptstyle \send 0/\recv{ \comme 0}$}
\def\JPicScale{.25}
\hspace{2em} \ifx\JPicScale\undefined\def\JPicScale{1}\fi
\psset{unit=\JPicScale mm}
\psset{linewidth=0.3,dotsep=1,hatchwidth=0.3,hatchsep=1.5,shadowsize=1,dimen=middle}
\psset{dotsize=0.7 2.5,dotscale=1 1,fillcolor=black}
\psset{arrowsize=1 2,arrowlength=1,arrowinset=0.25,tbarsize=0.7 5,bracketlength=0.15,rbracketlength=0.15}
\begin{pspicture}(0,0)(51.88,176.88)
\rput{0}(15,74.37){\psellipse[linewidth=1.5](0,0)(15,-15)}
\rput{0}(15,15){\psellipse[linewidth=1.5](0,0)(15,-15)}
\psline[linewidth=1.5,arrowsize=2 2,arrowlength=1.5,arrowinset=0,bracketlength=0.65]{->}(8.75,29.37)(8.75,60)
\psline[linewidth=1.5,arrowsize=2 2,arrowlength=1.5,arrowinset=0,bracketlength=0.65]{->}(22.5,60.62)(22.5,29.37)
\rput[l](51.25,15){\transone}
\rput(15.62,15.62){\oneleft}
\rput[r](5.62,44.38){\transtwo}
\rput[l](26.88,44.38){\transfour}
\rput[l](51.88,75.62){\transthree}
\rput{89.99}(38.75,13.75){\psellipticarc[linewidth=1.5,arrowsize=2 2,arrowlength=2]{->}(0,0)(17.5,9.38){122}{414.45}}
\rput{89.99}(38.12,74.38){\psellipticarc[linewidth=1.5,arrowsize=2 2,arrowlength=2]{->}(0,0)(17.5,9.38){122}{414.45}}
\rput{0}(15,159.37){\psellipse[linewidth=1.5](0,0)(15,-15)}
\psline[linewidth=1.5,arrowsize=2 2,arrowlength=1.5,arrowinset=0,bracketlength=0.65]{->}(8.75,128.12)(8.75,145)
\psline[linewidth=1.5,arrowsize=2 2,arrowlength=1.5,arrowinset=0,bracketlength=0.65]{->}(22.5,145.62)(22.5,127.5)
\rput(15,159.37){\oneright}
\rput{89.99}(38.12,159.38){\psellipticarc[linewidth=1.5,arrowsize=2 2,arrowlength=2]{->}(0,0)(17.5,9.38){122}{414.45}}
\psline[linewidth=1.5,arrowsize=2 2,arrowlength=1.5,arrowinset=0,bracketlength=0.65]{->}(9.38,88.75)(9.38,108.75)
\psline[linewidth=1.5,arrowsize=2 2,arrowlength=1.5,arrowinset=0,bracketlength=0.65]{->}(23.75,109.38)(23.75,87.5)
\psline[linewidth=1.5,linestyle=dotted,dotsep=4](8.75,110.62)(8.75,127.5)
\psline[linewidth=1.5,linestyle=dotted,dotsep=4](23.12,128.12)(23.12,106.25)
\rput(15,75){\onemid}
\rput[l](51.25,160.62){\transfive}
\end{pspicture}
\caption{The state machine view of the elevator}
\label{Fig:elev-STM}
\end{center}
\end{figure}

\para{Sharing the elevator.} If Alice and Bob share the elevator, then each of them calls it separately. For simplicity, we assume that they also receive separate service.  For example,  Alice calls it to the ground floor $0$ by pushing the input $\send 0_A$. When the elevator arrives, Alice observes $\recv{\comme 0}_A$. To go to floor 1, Alice then enters  $\send 1_A$. Then Bob may call the elevator down again by $\send 0_B$. We do not distinguish the calls from inside the cabin from the calls to the cabin. To capture sharing, we now need to distinguish Alice's actions from Bob's actions, which doubles the input and output types:
\begin{align*}
\Inp_{\Subj} \ & =\   \hspace{5.3em} \Big\{\send i_{A}\ |\ 0\leq i \leq n \Big\}\  +\  \Big\{\send i_{B}\ |\ 0\leq i \leq n \Big\} \ &=\  \ \Inp\times \Subj  \\
\Oup_{\Subj} & =  \hspace{2em} \Big\{ \recv{\comme i}_{A}, \recv{\sstay i}_{A}\ |\  0\leq i \leq n\Big\}\  +\  \Big\{ \recv{\comme i}_{B}, \recv{\sstay i}_{B}\ |\  0\leq i \leq n\Big\} & \ \ = \Oup\times \Subj 
\end{align*}
For a subject $S \in \Subj$ 
\begin{itemize}
\item $\recv{\comme i}_S$ means: \emph{"Subject observes that the elevator comes to floor $0$"}, whereas
\item $\recv{\sstay i}_S$ means: \emph{"Subject observes that elevator is already on floor $i$".}
\end{itemize}

\para{Channel interference in the elevator.} When a subject calls the elevator, they observe whether the elevator is already there or not. If Alice and Bob are the only users then each of them will know where the other has left the elevator. If Alice arrives home and leaves the elevator on floor 1, Bob will know that she is home or not by calling the elevator. E.g. if he enters a call to floor 1, he may observe two different outputs, depending on the history of Alice's earlier action:
\bear
\Inp^{+} & \to & \Oup\\
\send 0 _A \send 1_B & \mapsto &\recv{\comme 1}_B
\\
\send 1 _A \send 1_B & \mapsto & \recv{\sstay 1}_B
\eear
Alice's use of the elevator thus interferes with Bob's use\sindex{interference}.

\subsection{Sharing}
\subsubsection{Shared world}\label{Sec:shared}
In the preceding chapters, the world has been presented in terms  of histories $\vec x\in \Event^{\ast}$ as sequences of events from a set partitioned into disjoint unions 
\[ \Event\ =\ \coprod_{u\in \Subj}\Event_{u}\ =\  \coprod_{i\in \Obj} \Event_{i}\  =\  \coprod_{a\in \Act}\Event_{a}\]
where the events $\Event_{u}$ were performed by a subject $u$, the events $\Event_{i}$ were performed on the object $i$, and the events $\Event_{a}$ were all occurrences of the action $a$, performed by all subjects on all objects. In each case, the parts $\Event_{v}$ \emph{strictly}\/ localize the events at the subject, object, action, or level $v$. The strictness meant that Alice's local events $\Event_{A}$ are disjoint from Bob's local events $\Event_{B}$. But if the clearances need to be taken into account, then the subjects at the same clearance level should be able to observe and enact the same events. And even if they don't have the same clearance level, there are often events for which multiple subjects are cleared, and their general localities are not disjoint.

\para{General localities.} The world \sindex{world} \sindex{locality} is presented as histories $\vec x\in \Inp^{\ast}$, the sequences of events $x\in \Inp$. But some events happen here, other events there. Alice and Bob may be able to observe what is happening here, but not there. The localities  ``here'' and ``there'' may be presented as network nodes. The events from $\Inp$ may be localized on a network. To formalize this, we consider a lattice $\Levels$ of clearance levels, viewed as general localities, like in Sec.~\ref{Sec:explicit}. In general, they are  partially ordered, in the sense that a house is contained in a city, the city is contained in a country, so the lattice $\Levels$ then contains the chain house$\leq$city$\leq$country. But the order is partial because, say, a \emph{blue house}\/ and a \emph{red house}, are not contained in each other. Alice may, however, have access to both houses, and the lattice $\Levels$ must contain the clearance level \emph{blue house $\vee$ red house}\/ to assign it to Alice. If Carol has access to the whole city and Bob only to the blue house, then their clearance levels are $\clr(B) \lt \clr(A) \lt \clr(C)$. To capture channel sharing, the subjects from $\Subj$ and the events from $\Inp$ are assigned clearances \sindex{clearance} by the maps 
\[\loc: \Inp \to \Levels\qquad\qquad\qquad\qquad \clr :\Subj \to \Levels 
\]

\para{Clearance types.} \sindex{clearance!scope} Alice is cleared to enter an input $x$ if its locality $\loc(x)$ is below Alice's clearance level $\clr(A)$. The \emph{clearance relation}\/ \sindex{clearance!relation} $(\propto)\subseteq \Inp\times \Subj$ is defined
\bea\label{eq:obscone}
\Big(x\propto A\Big) &\iff & \Big(\loc(x)\leq \clr(A)\Big).
\eea
The events and actions for which Alice is cleared are collected in 
\bea
\Inp_{A} & = & \{x\in \Inp\ |\ x\propto A\}.
\eea
This is Alice's \emph{clearance type}.\sindex{clearance!type} The family of clearance types $\{\Inp_{u}\ |\ u\in \Subj\}$ is generally not a partition of $\Inp$, since its elements may have nonempty intersections, and may not cover $\Inp$. This is in contrast with the partitions $\{\Event_{u}\ |\ u\in \Subj\}$ of  general events $\Event$  in Sec.~\ref{Sec:Local}.

\para{Worldviews.} \sindex{worldview} Alice's worldview is the set $\Inp_{A}^{\ast}$ of all  histories from her clearance type $\Inp_{A}$. The general histories $\vec x\in \Inp^{\ast}$ are purged of events outside Alice's clearance type along the \emph{general purge}\/ channel \sindex{purge!general} $\restr_{A}\colon \Inp^{\ast}\to \Inp_{A}^{\ast}$ defined\sindex{channel!purge}
 \beq\label{eq:privview} \seq{}\restr_{A} = \seq{}\qquad \qquad \qquad\qquad (\vec x\cons u)\restr_{A} = \begin{cases}
\left(\vec x\restr_{A}\right) \cons u & \mbox{ if } u\propto A\\
 \left(\vec x\restr_{A}\right) & \mbox{ otherwise}
 \end{cases}
 \eeq
The local history $\vec x\restr_{A}$ is Alice's view of a global history $\vec x$. \sindex{history!global}\sindex{global history|see{history}}\sindex{history!local}\sindex{local history|see{history}}
 
\para{Local states of the world.} \sindex{world!state of the world} \sindex{state of the world|see{world}} If Alice observes a local history $\vec x_{A}\in \Inp_{A}^{\ast}$, then she knows that any of the global histories $\vec x\in \Inp^{\ast}$ such that $\vec x\restr_{A}=\vec x_{A}$ could have taken place. Alice's state of the world is the set of global histories consistent with her local view. Such derivations from her local views induce Alice's state of the world channel $\state = \tilde\restr \colon \Inp^{\ast}_{A}\to \WP\Inp^{\ast}$ defined
 \bea\label{eq:stateA}
 \state_{\vec x_{A}}(\vec x) & = & \begin{cases} 1 & \mbox{ if } \vec x \restr_{A} = \vec x_{A}\\
 0 & \mbox{ otherwise}
 \end{cases}
 \eea
 It is easy to see that this is the inverse (in the sense of Sec.~\ref{Sec:chaninverse}) of Alice's worldview channel, i.e., $\state = \tilde \restr$. Equivalently, using correspondence \eqref{eq:homchan}, Alice's local state of the world can also be written in the matrix form $\hom - - \colon \Inp^{\ast}_{A}\times \Inp^{\ast}\to \{0,1\}$ with the entries
\bear
\hom{\vec x_{A}}{\vec x} & = & \state_{\vec x_{A}}(\vec x).
\eear 

\subsubsection{Shared channels}\label{Sec:gensharing}
\sindex{channel!shared}
\para{Setting for sharing.} A shared channel must be given with a set of subjects $\Subj$, a locality lattice $\Levels$, a locality assignment $\loc: \Inp \to \Levels$, and a clearance assignment $\clr :\Subj \to \Levels$, determining the observability relation $(\propto)$ defined in  \eqref{eq:obscone}.

\para{Local channel views.} \sindex{channel!local view} Alice's local view of the channel $g\colon \Inp^{\ast}\to \WP\Oup^{\ast}$ is the derived channel $g^{A}\colon \Inp^{\ast}\to \WP\Oup^{\ast}$ which displays the outputs induced by Alice's inputs, and purges all outputs induced outside Alice's clearance set $\Inp_{A}$:
 \beq\label{eq:chanview} g^{A}()= \left\{()\right\}\qquad \qquad \qquad\qquad g^{A}(\vec x\cons u) = \begin{cases}
g^{A}\left(\vec x\right) \cons g_{\ast}(\vec x\cons u) & \mbox{ if } u\propto A\\
g^{A}\left(\vec x\right) & \mbox{ otherwise}
 \end{cases}
 \eeq
where $g_{\ast}\colon \Inp^{+}\to \WP\Oup$ corresponds to $g$ by \eqref{eq:cumul}.

\begin{lemma}\label{Lemma:privchan} Alice's local view $g^{A}$ of any {\relational} channel $g\colon \Inp^{\ast}\to \WP\Oup^{\ast}$ includes all channel outputs of her worldviews $\vec x\restr_{A}$ for every global history $\vec x$:
\bea\label{eq:privchanlemma}
g\left(\vec x\restr_{A}\right) & \subseteq & g^{A}(\vec x)
\eea
\end{lemma}

The converse inclusion is not satisfied by all {\relational} channels. In some cases, Alice may learn from her local view $g^{A}(\vec x)$ of a channel $g$ more than that channel outputs in response to her local inputs as $g\left(\vec x\restr_{A}\right)$.  That is a consequence of the \emph{interferences} within some channels. The channels where the converse inclusion of \eqref{eq:privchanlemma} is also valid, so that $g^{A}(\vec x)=g\left(\vec x\restr_{A}\right)$, are characterized in Prop.~\ref{Prop:nonint-char}.

\subsubsection{Interference channels}
\label{Sec:interchan}\sindex{interference!channel}\sindex{channel!interference}
Alice can collect information about the interferences within a channel $g\colon \Inp^{\ast}\to \WP\Oup^{\ast}$ by repeatedly entering the same input $\vec x_{A}\colon \Inp_{A}^{\ast}$ and recording the different outputs that arise in the different global contexts within her local state. In this way, she builds a {\relational} channel $\intt^{A} g\colon \Inp^{\ast}_{A}\to \WP\Oup^{\ast}$ defined as
\bea\label{eq:inttposs}
\intt^{A}  g(\vec x_{A}) & = & \bigcup_{\vec u\restr_{A}= \vec x_{A}} g^{A}(\vec u).
\eea
If the channel $g^{A}$ is viewed, using (\ref{eq:hommset}--\ref{eq:homchan}), as the matrix $\hhom- {g^{A}} -\colon \Inp^{\ast}\times \Oup^{\ast}\to \{0,1\}$ where $\hhom {\vec x} {g^{A}} {\vec y} = \homm{\vec y}_{g^{A}(\vec x)}=1$ if and only if $\vec y \in g^{A}(\vec x)$, else 0, then the interference channel $\int^{A} g$ becomes the matrix $\hhom- {\int^{A} g} -\colon \Inp_{A}^{\ast}\times \Oup^{\ast}\to \{0,1\}$ defined as 
\bear
\hhom{\vec x_{A}}{\int^{A} g}{\vec y} & = & 
\bigvee_{\vec x\in \Inp^{\ast}} \hom{\vec x_{A}}{\vec x}\cdot \hhom{\vec x}{g^{A}}{\vec y}.
\eear


\section{Relational noninterference}\label{Sec:nonint}

\subsection{Ideas of covert channels and interference} 
The \emph{no-read-up}\/ and \emph{no-write-down}\/ requirements, that we studied in Chap.~\ref{Chap:Resource}, assure that data and objects can not flow through a given channel downwards, i.e., that they cannot be either pulled down by reading, or pushed down by writing. 

But those requirements are imposed and can be tested  only on the \emph{explicitly given}, i.e.,  \emph{overt}\/ channels. The ancient idea of the \emph{Trojan horse}\/ shows that an attacker can hide a \emph{covert}\/ \sindex{channel!covert} \sindex{Trojan horse} channel inside an overt channel, and transfer through it some prohibited data or objects. The profession of smugglers mainly consists of constructing covert channels inside overt channels. 
\begin{figure}[!ht]
\centering
\includegraphics[height = 7cm
]{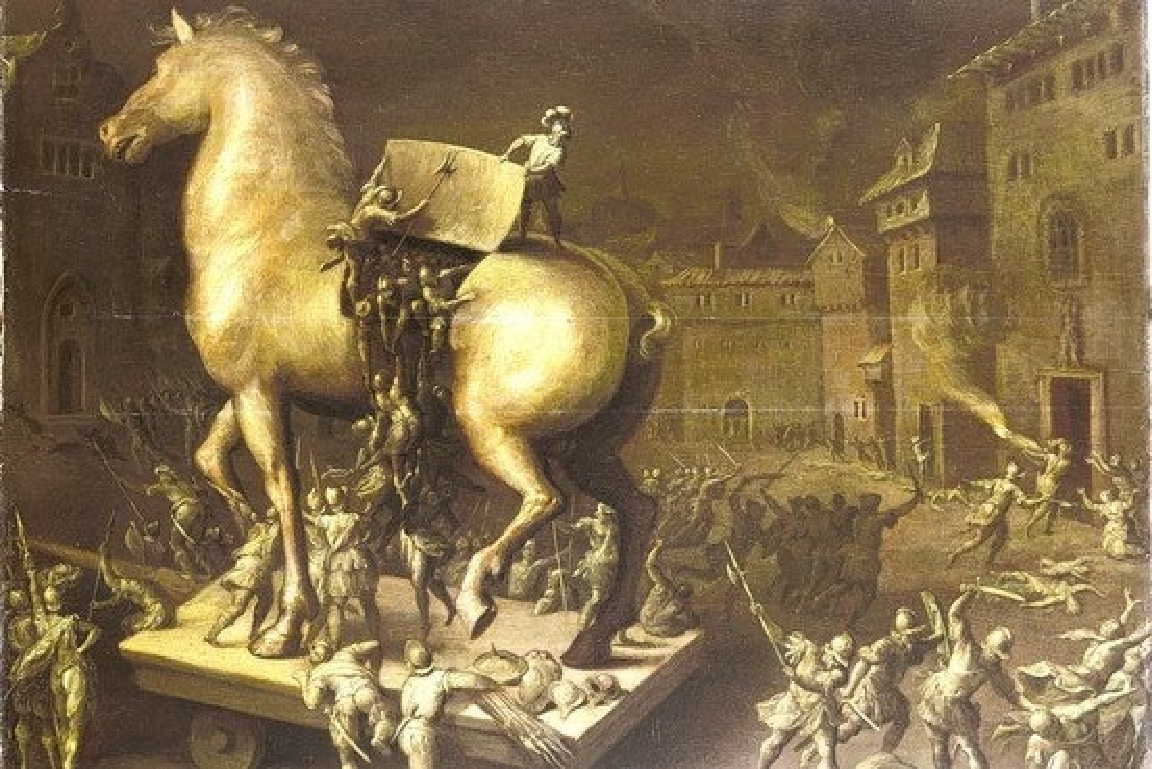}
\caption{Trojan horse covertly channels soldiers through overt gates}
\label{Fig:trojan}
\end{figure}
The story of the Trojan Horse is a legend of Odysseus' brilliant smuggling idea. The city of Troy was a walled fortification, and the only overt channel into the city was the city gate, manned by the guards who made sure that no arms enter the city. Odysseus built a wooden horse, large enough to conceal a group of armed soldiers and left it in front of the gate. To the citizens, the wooden horse looked like a work of art. An overt channel from the gods. It was actually a covert channel for armed soldiers, but it appeared to comply with the security requirement of the city gate because they were covert.

Noninterference assures that overt channels do not contain covert channels. When a channel $g:\Inp^\ast \to \left(\WP\Oup\right)^{\ast}$ is shared between Alice, Bob, Carol, and others, the overt channels are the local views $g^{A}, g_{B}, g_{C},\ldots:\Inp^\ast\to \left(\WP\Oup\right)^{\ast}$ defined in \eqref{eq:chanview}. The security requirement is that Alice's local overt channel $g^{A}$ must not provide any information about Bob's channel $g_{B}$ unless she has a clearance  $\clr(A) \geq \clr(B)$ to see Bob's private data.  But how could $g^{A}$ ever provide information about $g_{B}$ if \eqref{eq:chanview} purges Bob's outputs from Alice's view? We saw an instance of the answer (with the roles reverted) in Example~\ref{Sec:elev-inter}.
While Bob's outputs are purged from Alice's local view, Alice's outputs $g^{A}(\vec x)$ do not depend only on her inputs, but on the entire global histories $\vec x$, which contain everyone's inputs. So Alice can derive information about the global history $\vec x$, including Bob's inputs $\vec x\restr_B$, from her local observations $g^{A}(\vec x)$. Such derivations are possible when there is channel interference between Alice's outputs and Bob's inputs.

How do we gain information about the input of a function from its outputs? For example, if I know that the sum of two positive integers is $x+y = 4$, then I know that $x$ must be 1, 2, or 3. If an unknown function $f$ over positive integers outputs the values $f(x)$ and $f(y)$, I may not be able to find what $x$ and $y$ are, but if $f(x) \neq f(y)$, then I do know that $x\neq y$. This is $1$ bit of information about $x$ and $y$. If I calculate the value of $f$ at 0, and it turns out to be $f(0) = f(x)$, then I have one more bit of information about $x$ and $y$. But if the function $f$ is constant, and $f(x) = f(y)$ for all $x,y\in A$, then the outputs of $f$ provide no information about its inputs. 

Suppose Alice observes a worldview $\vec x_{A}$. She wonders what is the actual state of the world $\vec x$, of which she only sees the restriction $\vec x_{A}=\vec x\restr_{A}$. From the channel $g\colon \Inp^{\ast}\to \left(\WP\Oup\right)^{\ast}$, she also sees some outputs $y, y',\ldots \in g^{A}(\vec x)\subseteq \Oup^{\ast}$. What does $y\in g^{A}(\vec x)$ tell about $\vec x$? Alice can extract one bit of information about $\vec x$ if she can find another global extension  $\vec u$ of her worldview $\vec x \restr_A = \vec u\restr_A$ such that $y\not \in g^{A}(\vec u)$, i.e., $g^{A}(\vec x) \neq g^{A}(\vec u)$. Another such extension provides another bit of information and restricts $\vec x$ even more. Starting from a local view $\vec x_{A}$ and testing the observable outputs $g^{A}(\vec u)$ for unobservable global inputs $\vec u$ such that $\vec u\restr_{A}=\vec x_{A}$, Alice gathers information about $\vec x$ that caused her observation $g^{A}(\vec x)$ by distinguishing it from $\vec u$ that cause different effects $g^{A}(\vec u)$, although they are for Alice locally indistinguishable, since $\vec u\restr_{A}=\vec x\restr_{A} = \vec x_{A}$. 

This information gathering process, where Alice keeps re-entering her local input $\vec x_{A}$ to trigger different global inputs $\vec u$ with $\vec u\restr_{A}= \vec x_{A}$ and record the possibly different channel outputs $g^{A}(\vec u)$, is modeled by Alice's interference channel $\intt^{A}g$. There is no interference or covert channels in $g$ if the interference channels $\intt^{A}g$ do not gather any more information than the local channels $g^{A}$ display directly.

\begin{definition}\label{Def:nonint}\sindex{noninterference}
A channel $g \colon \Inp^\ast \to \left(\WP\Oup\right)^{\ast}$ satisfies the noninterference requirement if for all subjects $A$ the interference channels output just the local views that they get overtly, i.e.
\bea\label{eq:nonint}
\intt^{A} g(\vec x_{A}) & = & g(\vec x_{A})
\eea
A channel $f\colon \Inp^{+}\to \WP\Oup$ is said to satisfy noninterference if $f^{\ast}\colon \Inp^{\ast}\to \left(\WP\Oup\right)^{\ast}$ satisfies \eqref{eq:nonint}.
\end{definition}

\subsection{Characterizing noninterference}
The following proposition provides several equivalent characterizations of relational noninterference. A useful equivalent condition is stated in terms of the \sindex{purge!complementary} \sindex{complementary purge} \emph{complementary purge}\/ channel 
$\restriction_{\oth   A} : \Inp^\ast \to \Oup_A^\ast$ defined
 \beq\label{eq:complview} \seq{}\restr_{\oth A} = \seq{}\qquad \qquad \qquad\qquad (\vec x\cons u)\restr_{\oth A} = \begin{cases}
\left(\vec x\restr_{\oth A}\right) & \mbox{ if } u\propto A\\
 \left(\vec x\restr_{\oth A}\right) \cons u & \mbox{ otherwise}
 \end{cases}
 \eeq
Comparing this definition with \eqref{eq:privview} shows that this \emph{complementary}\/ purge channel keeps in $\restr_{\oth A}\vec x$ precisely those events from $\vec x$ which Alice purge operation eliminates from $\restr_{A}\vec x$; and vice versa, $\restr_{\oth A}$ purges precisely those events which $\restr_{A}$ keeps. 

\begin{proposition}\label{Prop:nonint-char}
For every shared channel $g:\Inp^\ast\to\WP\Oup^{\ast}$ and every subject $A$, the following conditions are equivalent:
\begin{enumerate}[(a)]
\item $g$ satisfies the noninterference requirement:
\bear
\intt^{A} g(\vec x_{A}) & = & g(\vec x_{A})
\eear

\item  for all $\vec x, \vec y\in \Inp^\ast$ 
\bear
\vec x\restr_{A} = \vec y\restr_{A} &  \Longrightarrow &  g^{A}(\vec x) = g^{A}(\vec y)\eear
 
 \item for all $\vec x\in \Inp^\ast$ 
 \bear g^{A}\left(\vec x\right)  &  = &    g\left(\vec x\restr_{A}\right)\eear
 
\item for all $\vec x, \vec z\in \Inp^\ast$ there is $\vec y\in \Inp^\ast$ such that  
\bear 
\vec x\restr_{A} = \vec y\restr_{A} \ \ \wedge\ \ \  \vec y \restr_{\oth   A} = \vec z\restr_{\oth   A} \ \ \wedge \ \ g^{A}(\vec x) = g^{A}(\vec y) \eear
\end{enumerate}
\end{proposition}

\bpr
\emph{(d)$\Longrightarrow$(c):}\/ For $\vec z = \seq{}$, the second conjunct of \emph{(d)} is $\vec y \restr_{\oth   A} = ()$, which by \eqref{eq:complview} and \eqref{eq:privview} implies $\vec y \restr_{A} = \vec y$. The first conjunct of \emph{(d)} then gives $\vec x\restr_{A} = \vec y\restr_{A} = \vec y$. Using this, the third conjunct of \emph{(d)} gives $g^{A}(\vec x) = g^{A}(\vec y) = g^{A}(\vec x\restr_{A})$. Hence \emph{(c)}. 

\noindent\emph{(c)$\Longrightarrow$(b): } 
Towards the implication in \emph{(b)}, suppose
\[
\vec x\restr_{A} \ \  = \ \  \vec y\restr_{A}\tag{$^{\bullet}b$}
\]
Using \emph{(c)}, it follows that
\[g^{A}(\vec x)\  \stackrel{(c)}=\  g^{A}(\vec x\restr_{A}) \stackrel{(^{\bullet}b)}= g^{A}(\vec y\restr_{A})\ \stackrel{(c)}=\ g^{A}(\vec y)
\tag{$b^{\bullet}$}
\]
So we have proven the implication $({^{\bullet}b} {\stackrel{(c)}{\implies}}\ b^{\bullet})$, which is \emph{(b)}.

\noindent \emph{(b)$\Longrightarrow$(d): } For arbitrary $\vec x, \vec z\in \Inp^\ast$, set 
\bea\label{eq:z}
\vec y  & = &  \left(\vec x\restr_{A}\right) \cons \left( \vec z\restr_{\oth A}\right)
\eea
Hence $\vec x\restr_{A} = \vec y\restr_{A}$ and $\vec y\restr_{\oth A} = \vec z_{\oth A}$. But $\vec x\restr_{A} = \vec y\restr_{A}$ implies $g^{A}(\vec x) = g^{A}(\vec y)$ by \emph{(b)}. The three conjuncts of \emph{(d)} are thus satisfied by  $\vec y$ defined as in \eqref{eq:z}.

\noindent \emph{(c)$\implies$(a)} holds because $g^{A}\left(\vec u\right) = g\left(\vec u\restr_{A}\right)$ implies
\bear\intt^{A}  g(\vec x_{A}) &  \stackrel{\eqref{eq:inttposs}}=&   \bigcup_{\vec u\restr_{A}= \vec x_{A}} g^{A}(\vec u)  \ \ \stackrel{(c)} =\ \ \bigcup_{\vec u\restr_{A}= \vec x_{A}} g (\vec u\restr_{A})\ =\  g (\vec x_{A})
\eear    

\noindent \emph{(a)$\implies$(c):} The inclusion one way follows from \emph{(a)} by
\bear
g\left(\vec x\restr_{A}\right)\ \ \stackrel{(a)} = \intt^{A}g \left(\vec x\restr_{A}\right) &\stackrel{\eqref{eq:inttposs}}= &   \bigcup_{\vec u\restr_{A}= \vec x\restr_{A}} g^{A}(\vec u)  \ \ \supseteq \ \ g^{A}(\vec x)
\eear    
The inclusion $g\left(\vec x\restr_{A}\right) \ \subseteq \ g^{A}(\vec x)$ is valid by Lemma~\ref{Lemma:privchan}, independently on the interference. 
\epr

\subsection{Noninterference as invariance}
\label{Sec:invariance}
\begin{definition}\label{Def:invariance}
A continuous channel $\gamma\colon \WP\Inp^{\ast}\to \WP\Oup^{\ast}$ is said to be \/ \emph{invariant} under the operation \sindex{invariance} $p\colon\WP\Inp^{\ast}\to \WP\Inp^{\ast}$ if $\gamma\circ p= \gamma$, i.e., if the following diagram commutes
\beq\label{eq:invariant}
\begin{tikzar}[row sep = .25pc,column sep = 5pc]
\WP\Inp^\ast \arrow{dd}[swap]{p} \arrow{dr}[description]{\gamma}\\
\& \WP\Oup^\ast\\
\WP\Inp^\ast \arrow{ur}[description]{\gamma}
\end{tikzar}
\eeq
\sindex{channel!continuous}\sindex{continuous channel|see{channel, continuous}}

A \emph{projector} is a continuous channel $p\colon \WP\Inp^{\ast}\to \WP\Inp^{\ast}$ that is invariant under itself, i.e., satisfies $p\circ p = p$.\sindex{projector}
\end{definition}

\para{Cumulative and single-output versions.} Mapped along the bijections in \eqref{eq:contchan} and \eqref{eq:cumul}, diagram \eqref{eq:invariant} becomes
\beq\label{eq:invariant-single}
\begin{tikzar}[row sep = .25pc,column sep = 5pc]
\Inp^\ast \arrow{dd}[swap]{q} \arrow{dr}[description]{g}\\
\& \WP\Oup^\ast\\
\Inp^\ast \arrow{ur}[description]{g}
\end{tikzar}
\qquad\qquad\qquad
\begin{tikzar}[row sep = .25pc,column sep = 5pc]
\Inp^+ \arrow{dd}[swap]{r} \arrow{dr}[description]{f}\\
\& \WP\Oup\\
\Inp^+ \arrow{ur}[description]{f}
\end{tikzar}
\eeq

\para{Examples.} Diagram \eqref{eq:memoryless2} thus says that a channel $f\colon \Inp^{+}\to\WP\Oup$ is memoryless if and only if it is invariant under the operation $r=\left(\Inp^{+}\eepi{\sxx} \Inp \mmono{(-)} \Inp^{+}\right)$. It is easy to check that $r=r\circ r$ is a projector.

Another example is noninterference. Condition \emph{(c)}\/ from Prop.~\ref{Prop:nonint-char} can be summarized by piping the outputs of the purges \eqref{eq:privview} from $\Inp^\ast_{A}$ to all of $\Inp^\ast$ and then bundling them into a projector:  
\[\prooftree
\prooftree
\Big\{\Inp^\ast \tto{\restr_A} \Inp_A^\ast \inclusion \Inp^\ast \Big\}_{A\in \Subj}
\justifies
\Inp^\ast \times \Subj  \tto{\restr_\bullet} \Inp^\ast
\endprooftree
\justifies
\restr\ =\ \left( \Inp^\ast\tto{<\id,\sxx>} \Inp^\ast\times \Inp\tto{\ \id\times\suct\ } \Inp^\ast\times  \Subj \tto{\ \ \ \ \restr_\bullet \ \ \ } \Inp^\ast\right)
\endprooftree
\]
where $\restr_\bullet(\vec x, A) = \vec x\restr_{A}$. Hence,                                                                                                                       \sindex{channel!bundled purge} \sindex{purge} \sindex{channel!purge} 
\bea\label{eq:restri}
\restriction \ :\ \Inp^\ast  & \to &  \Inp^\ast\\
\seq{ x_0 x_1 \ldots x_n} & \mapsto & \seq{ x_0 x_1 \ldots x_n}\restr_{\suct(x_n)}.\notag
\eea
In words, the global history $\vec x$ is projected to the local view $\vec x\restr_A$ of the subject $A=\suct(\sxx(\vec x))$ who entered the last channel input $\sxx(\vec x)$. It is easy to see that this is a projector,\sindex{projector} since $\left(\vec x\restr_{A}\right)\restr_{A} = \vec x\restr_{A}$ and hence,
\beq\begin{tikzar}[row sep = .25pc,column sep = 5pc]
\Inp^\ast \arrow{dd}[swap]{\restr} \arrow{dr}[description]{\restr}\\
\& \Inp^\ast\\
\Inp^\ast \arrow{ur}[description]{\restr}
\end{tikzar}
\eeq

\begin{proposition}\label{Prop:nonint-secure}
A channel $g:\Inp^\ast \to\WP\Oup^\ast$ satisfies the noninterference requirement if it is invariant under the purge: 
\beq\label{eq:invnonint}
\begin{tikzar}[row sep = .25pc,column sep = 5pc]
\Inp^\ast \arrow{dd}[swap]{\restr} \arrow{dr}[description]{g}\\
\& \WP\Oup^\ast\\
\Inp^\ast \arrow{ur}[description]{g}
\end{tikzar}
\eeq
\end{proposition}

\bpr
On one hand, definition \eqref{eq:chanview}  of the local view $g^{A}$ gives 
\bea\label{eq:fstchar}
g(\vec x) & = & g^{\Gamma\vec x}(\vec x)
\eea
where $\Gamma = \left(\Inp^+ \tto{\sxx} \Inp\tto{\suct}\Subj\right)$ tells who entered the last input into the channel. On the other hand, Prop.\ref{Prop:nonint-char}\emph{(c)} gives  
\bea\label{eq:sndchar}
g^{\Gamma\vec x}\left(\vec x\right) & = & g\left(\vec x\restr_{}\right)
\eea
It follows that $g$ satisfies the noninterference if and only if for all $\vec x$ 
\[
g\left(\vec x\right)\ \  \stackrel{\eqref{eq:fstchar}}{=}  \ \ g^{\Gamma\vec x}\left(\vec x\right)\ \  \stackrel{\eqref{eq:sndchar}}{=}  \ \ g\left(\vec x\restr_{}\right)
\] 
i.e., if and only if $g=g\circ\! \restr_{}$, which is \eqref{eq:invnonint}.
\epr

It will turn out that all resource and channel  security properties, that we studied in the previous chapters, can be characterized in terms of invariants --- as well as the data security properties that we will study in the next chapters.

\def\thechapter{7}
\setchaptertoc
\chapter{Communication channels, protocols, and authentication}
\label{Chap:Auth}

\section{Communication channels}

\subsection{What is communication?}\sindex{communication} \sindex{traffic and communication}
Alice and Bob communicate by sending messages, exchanging objects, and displaying features. Hence the three types of communication channels, transmitting the three types of security entities: encoded data, physical things, and individual traits. Different channels enable different network computations. They \sindex{protocol} are programmed as protocols. They give rise to network security problems and solutions. We need to understand communication channels in order to be able to understand network security problems and design protocols that solve them. 

Any form of traffic can be viewed as communication. If we think of London and Paris as subjects, then the air, land, and sea between them are communication channels. The air traffic, the shipping across the English Channel, and the trains under it, are forms of communication. 

The main feature of communication channels is that they connect a subject that enters the inputs on one side with a subject that extracts the outputs on the other side. The two sides, the input interface and the output interface are also the characteristics of functions. Channels are a special kind of functions. That is how they were defined in Sec.~\ref{Sec:chanwhat}: as functions that take histories as the inputs. A possibilistic (relational) channel relates each input history to a set of possible outputs, possibly empty. A probabilistic channel assigns a probability to each  input. \sindex{channel}\sindex{channel!possibilistic|see {channel, relational}} \sindex{channel!probabilistic} \sindex{channel!stochastic|see {channel, probabilistic}} A communication channel can be possibilistic, or probabilistic, or even quantum. What makes it into a \emph{communication}\/ channel is that the input and the output are under control of communicating parties, usually Alice and Bob.

\subsection{Communicating vs sharing}
\sindex{channel!communication}
If Alice and Bob are merchants using the English Channel, then they may interact in two ways:
\begin{itemize}
\item they may \emph{share}\/ the Channel: \begin{itemize}
\item each of them uses both the inputs and the outputs: loads the Channel on one side and unloads it on the other side, independently of the other merchant; or
\end{itemize}
\item they may \emph{communicate}\/ through the Channel:
\begin{itemize}
\item Alice enters the inputs on one side of the Channel, Bob extracts the outputs on the other side of the Channel; and then Bob may enter the inputs for Alice into another communication channel in the opposite direction.
\end{itemize}
\end{itemize}
The structural difference between sharing and communication is summarized in the following table: 
\begin{center}
\begin{tabular}{|c|c|}
\hline
shared channel & communication channels\\
\hline
\hline
$\left(\Inp_{A}+\Inp_{B}\right)^{+} \to \WP\Oup$
&
\begin{minipage}[c]{0.25\linewidth}
\centering
\ \ \ \ $\Inp_{A}^{+}\   \to\  \WP\Oup_{B}$\\
$\WP\Oup_{A}\ \ot\ \  \ \ \Inp_{B}^{+}$
\end{minipage}\\
\hline
\end{tabular}
\end{center}
\sindex{channel!point-to-point}
The displayed channels are the simplest possible: either of them is only shared by two subjects, or only used for communication between two subjects. In practice, most channels combine sharing and communication. Alice and Bob often talk at the same time, and share the ether while trying to communicate. The practice of everyone talking at the same time is scaled up to the extreme in packet-switching networks, where large crowds of subjects enter their input packets, the network as a shared channel mixes them all, and then unmixes them as a communication channel, to deliver each packet to the desired recipient. Modern networks, starting from the Internet, are mostly just implementations of such channels, hiding the network structure from the users in a black box, and delivering the shared communication channel functionality. In this section, we ignore sharing and focus on communication. 

\subsection{Example: the wren authentication} \label{Sec:wren} \sindex{authentication} \sindex{protocol!authentication}The wren on the left in Fig.~\ref{Fig:wren-chicks} receives its chicks' chirps as messages saying ``Feed me!'', transmitted by the data channel carried by the sound. 
\begin{figure}[!h]
\begin{center}
\includegraphics[height=5cm]{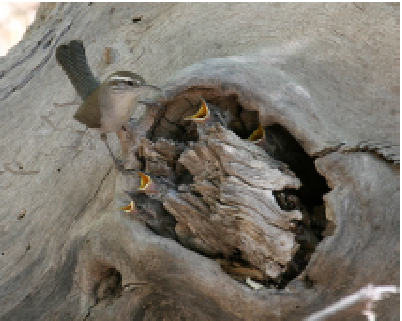}
\hspace{5em}
\includegraphics[height=5cm]{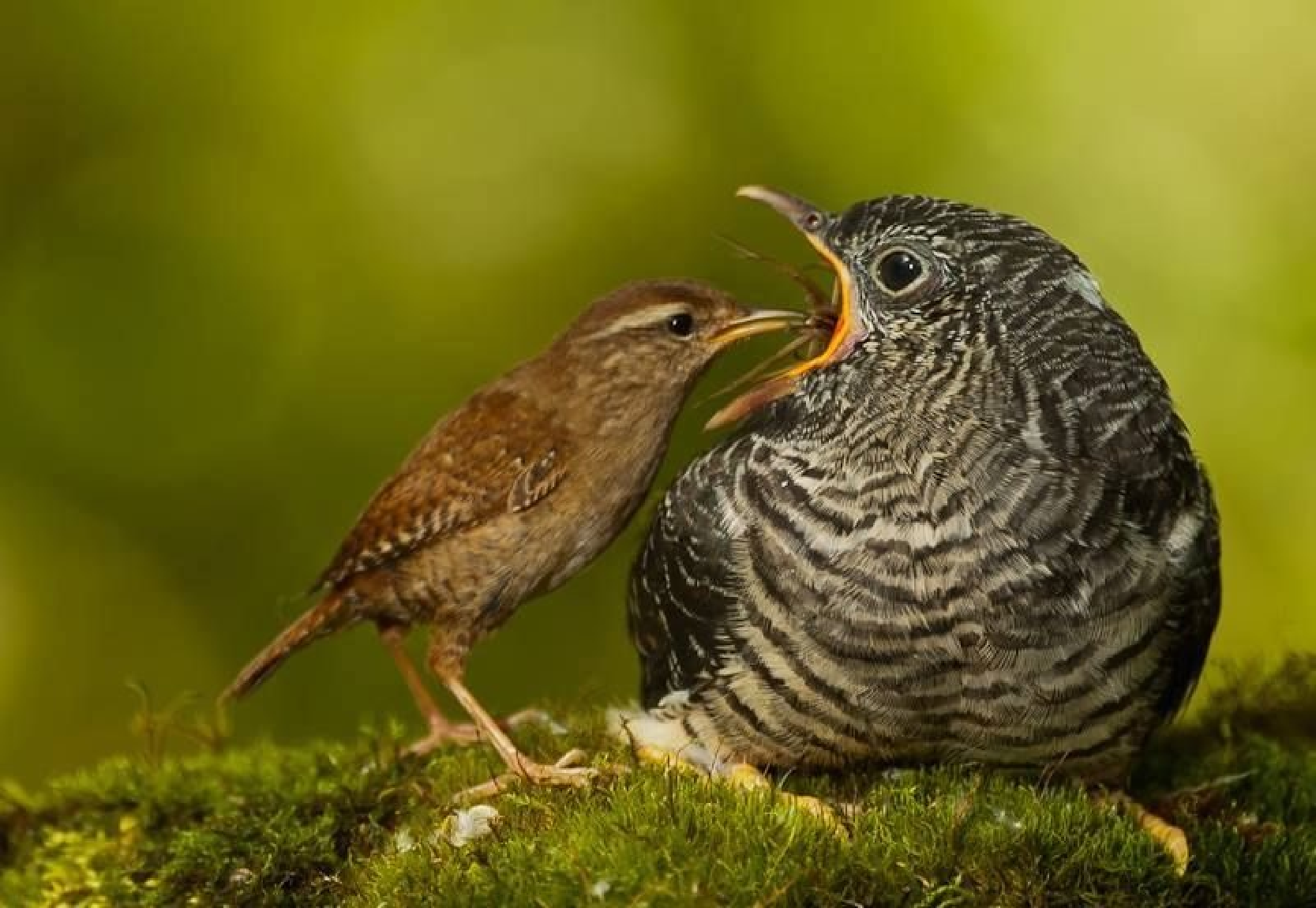}
\caption{Wren's chick-feeding protocol gets attacked by cuckoo}
\label{Fig:wren-chicks}
\end{center}
\end{figure}
The wren responds by delivering the food on the channel carried by the physical space. Since the chicks chirp together, the chirping channel is shared on the input side, whereas the food channel is shared on the output side. The competition for the food at the output provides an incentive for attacks on the chick feeding protocol. A big one is shown in Fig.~\ref{Fig:wren-chicks} on the right. The first step is that a cuckoo lays an egg in wrens' nest. Since the wrens cannot tell apart their own chicks from the cuckoo's, they feed them all. The second step of cuckoo's attack is that the cuckoo chick murders the wren chicks by pushing them out of the nest. It keeps its adoptive wren parents all for itself, and outgrows them before they know. 
\begin{figure}[!h]
\begin{center}
\includegraphics[width=11cm
]{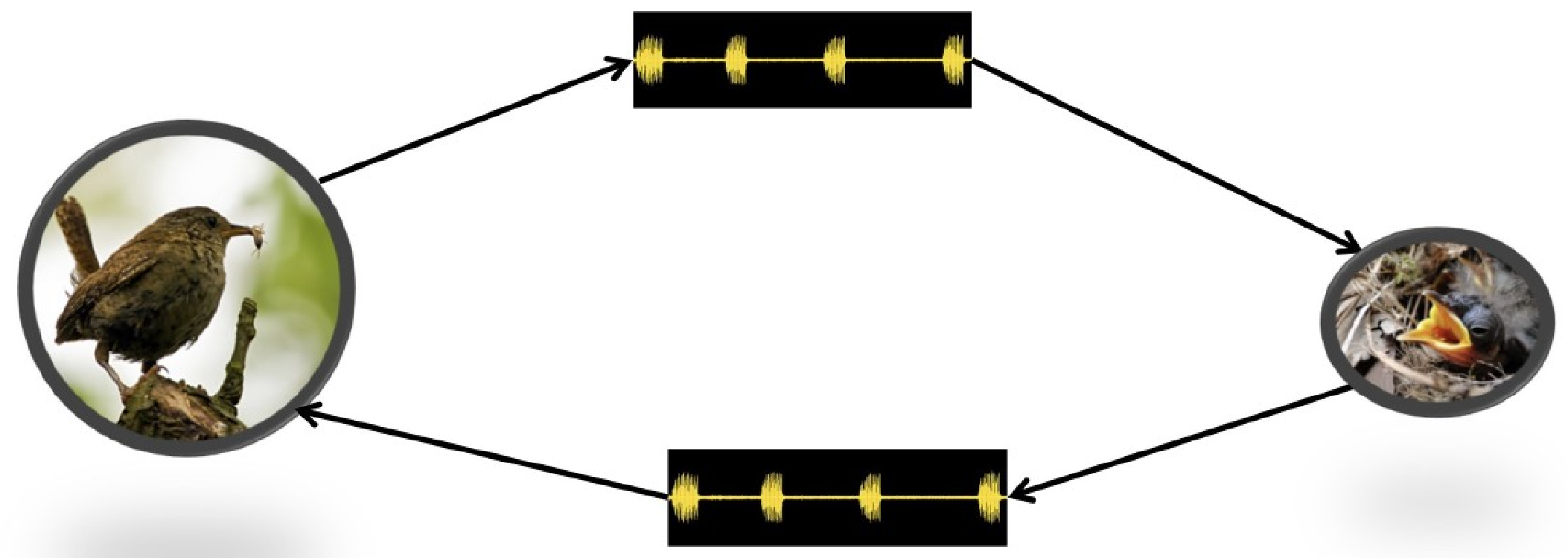}
\caption{Wren's chick authentication}
\label{Fig:wren-auth}
\end{center}
\end{figure}

To defend against this attack, the Fairy-wrens evolved the authentication protocol~\cite{wren-auth}, sketched in~\cref{Fig:wren-auth}.
As they lay on eggs, the wrens keep chirping a fixed but randomly chosen, and therefore unique, pattern. The chicks learn it from inside the eggs. This chirp is used as a shibboleth\footnote{If you don't remember what is a shibboleth, a quick web search will do.}. When they hatch, the chicks chirp back their parents' unique chirping pattern. The cuckoo chicks have so far not managed to evolve a capability to learn chirping from inside the egg. Without the shibboleth chirp, the cuckoo chicks do not get fed and starve without killing anyone. Authentication saves lives. \sindex{authentication!of wren chicks} 

\begin{figure}[!ht]
\begin{center}
\includegraphics[width=12cm
]{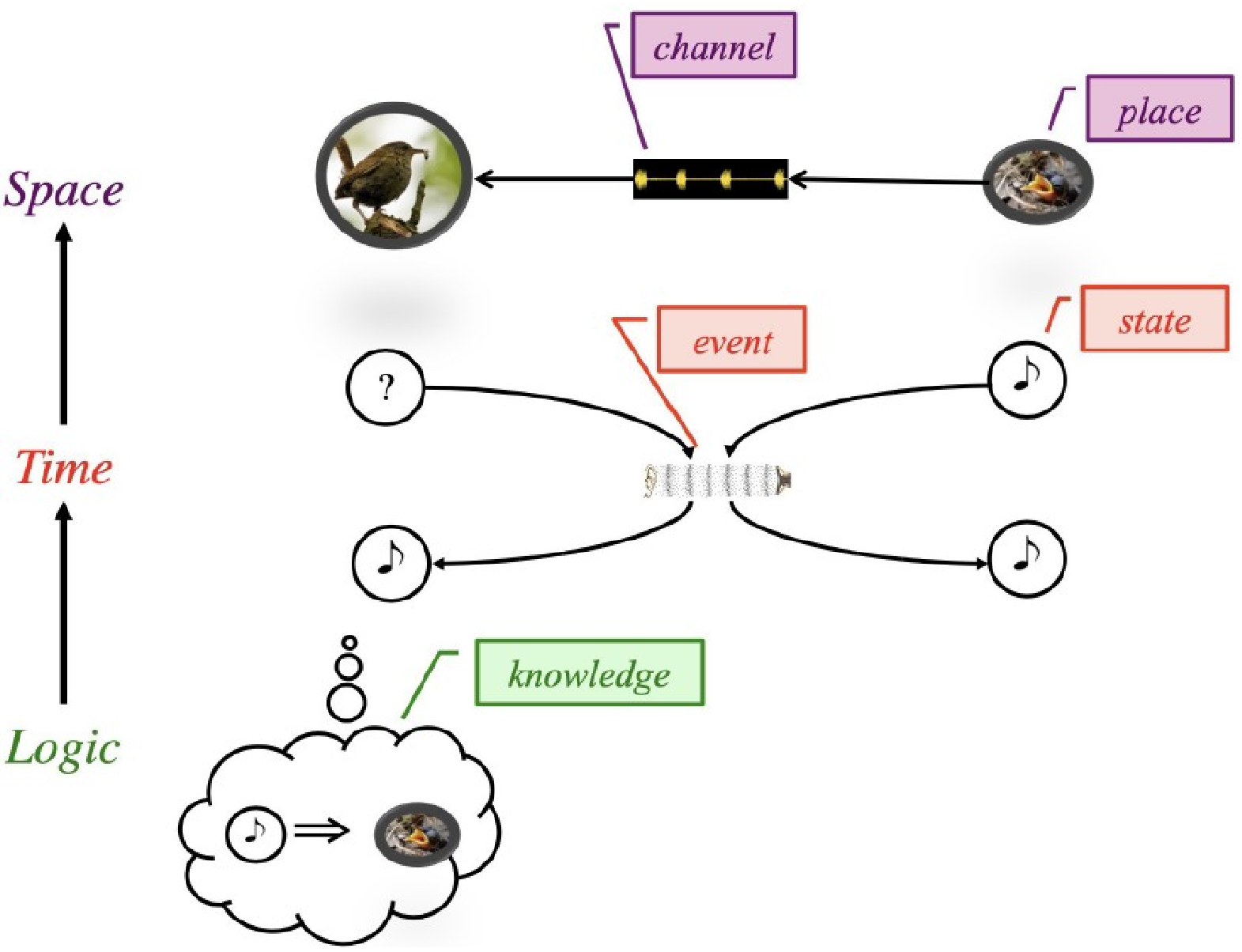}
\caption{Authentication framework}
\label{Fig:wren-auth-frame}
\end{center}
\end{figure}

Fig.~\ref{Fig:wren-auth-frame} shows the wren chick communication channel within the general authentication framework.\sindex{authentication!general framework} The wren parents on one hand and the wren chicks on the other are viewed as network nodes (``places''). Their interactions through the data channel, transmitting chick's shibboleth chirps, cause state transitions at each node. Each of the parents was ready to receive the message (as marked by ``?''). After a parent receives the chirp message (denoted by ``\textmusicalnote''), they know (on the level of Logic) that the chick was successfully authenticated.

\subsection{Channeling what you know, have, are}
\begin{figure}[!ht]
\begin{center}
\includegraphics[height=4cm
]{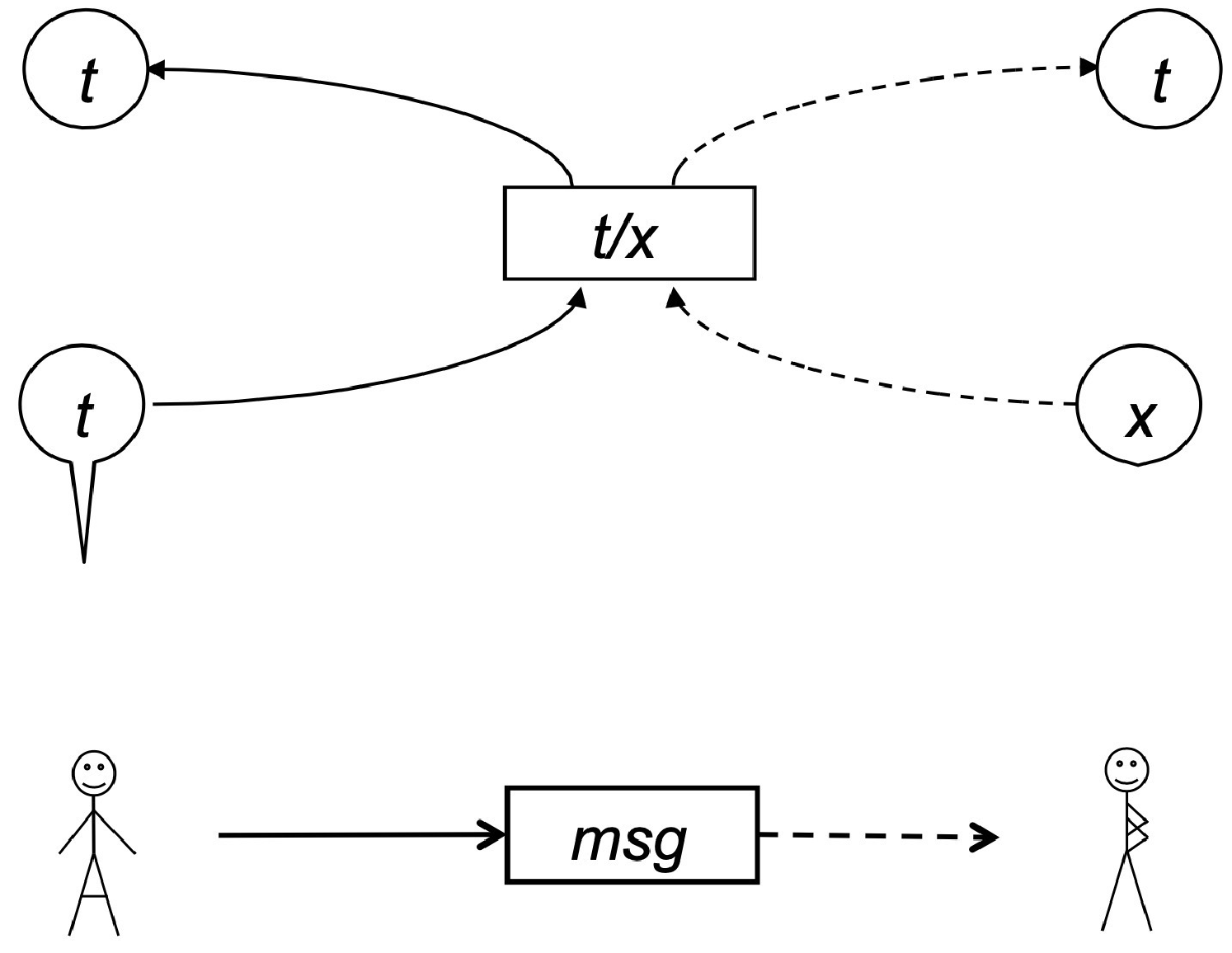}
\hspace{3em}
\includegraphics[height=4cm
]{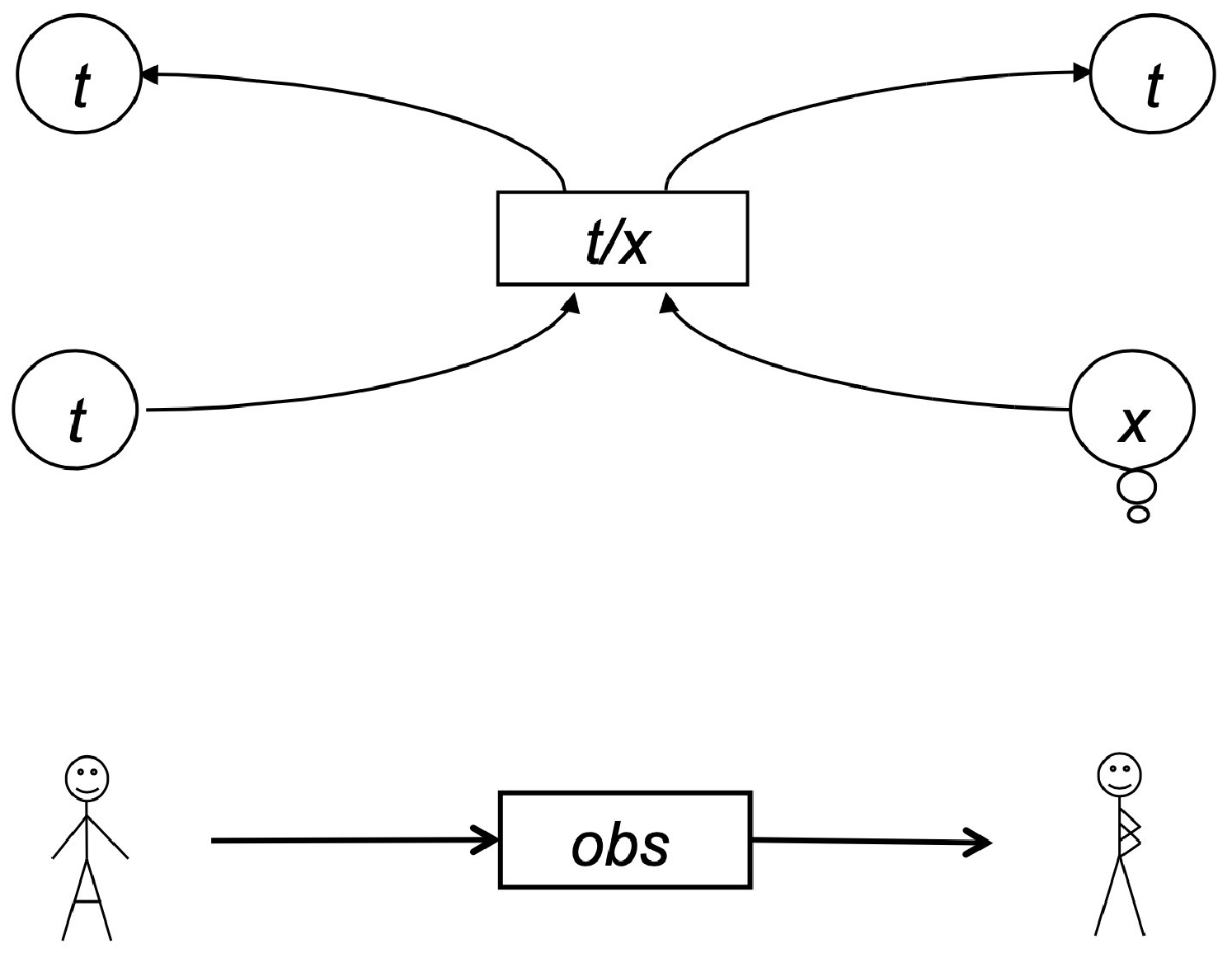}
\caption{Data transmissions: send/receive and observe/display}
\label{Fig:au-data}
\end{center}
\end{figure}
Leaving the logic aside for the moment,~\cref{Fig:au-data} displays just the space and the time: this time the network channel is at the bottom and the channel interaction is above it. The network nodes communicating through the channel are called Alice and Bob again, and the two possible interactions through a data channel are pushing a message (on the left) and pulling an observation (on the right). Alice pushes data to Bob by sending a message, whereas Bob pulls data from Alice by observing an observation.  But Alice's message can only be transmitted if Bob is ready to receive it. And Bob's observation can only be made if Alice displays the data to be observed. Alice as the sender initiates the message transmission, and Bob accepts it. Bob as the observer initiates the observation of an event from Alice's side, and Alice makes the observation possible. In both cases,~\cref{Fig:au-data} denotes Bob's state as recipient of the communication channel output by the variable $x$. The channel interaction is modeled as substitution $t/x$, whereby the transmitted data $t$ gets stored in $x$ and Bob's state changes from $x$ to $t$ (just like wren's state was changed from the listening ``?'' to the transmitted chirp ``\textmusicalnote'').  Note that Alice knows the data $t$ both before she sends a copy and after she sent it.

\begin{figure}[!hb]
\begin{center}
\includegraphics[height=4cm
]{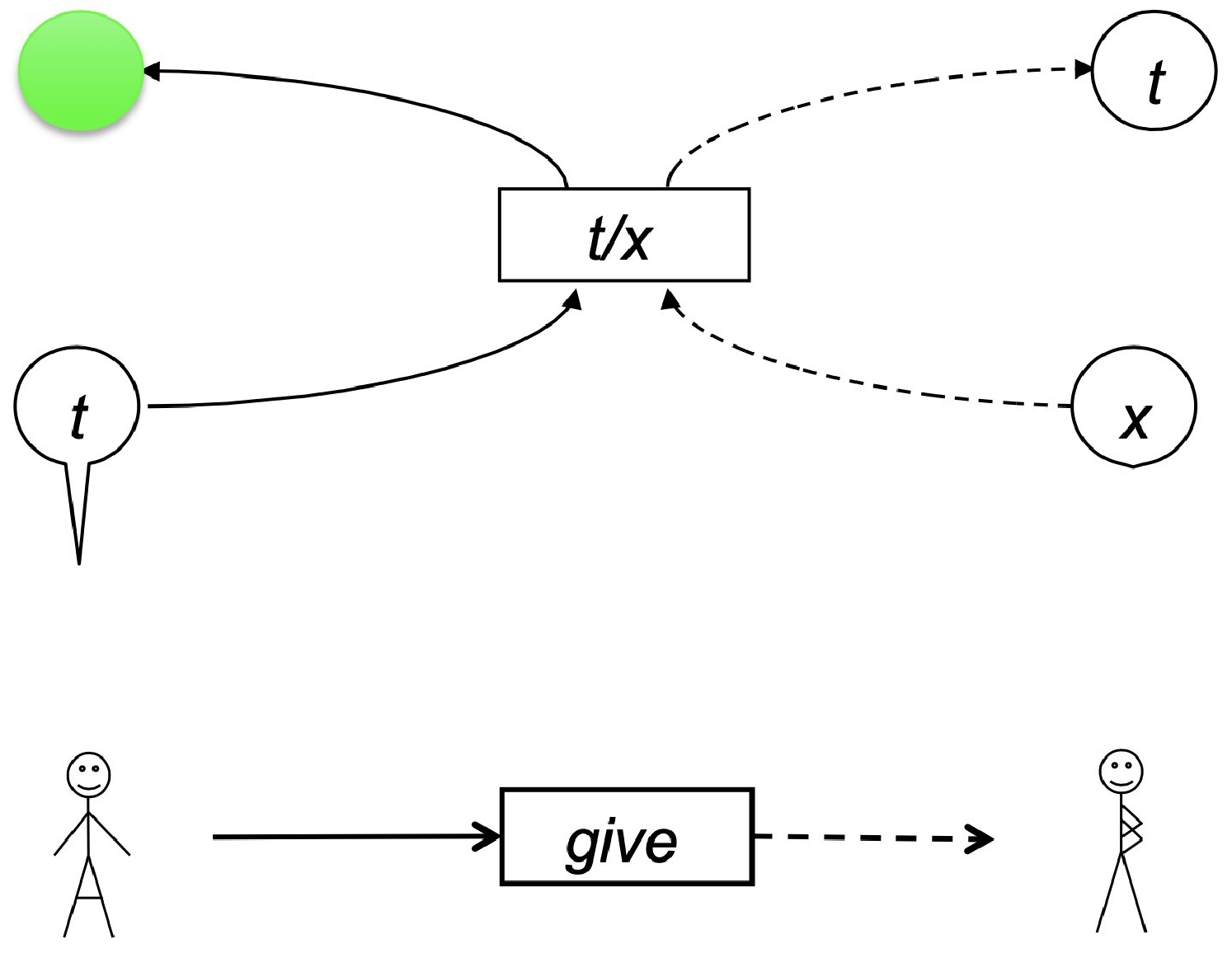}
\hspace{3em}
\includegraphics[height=4cm
]{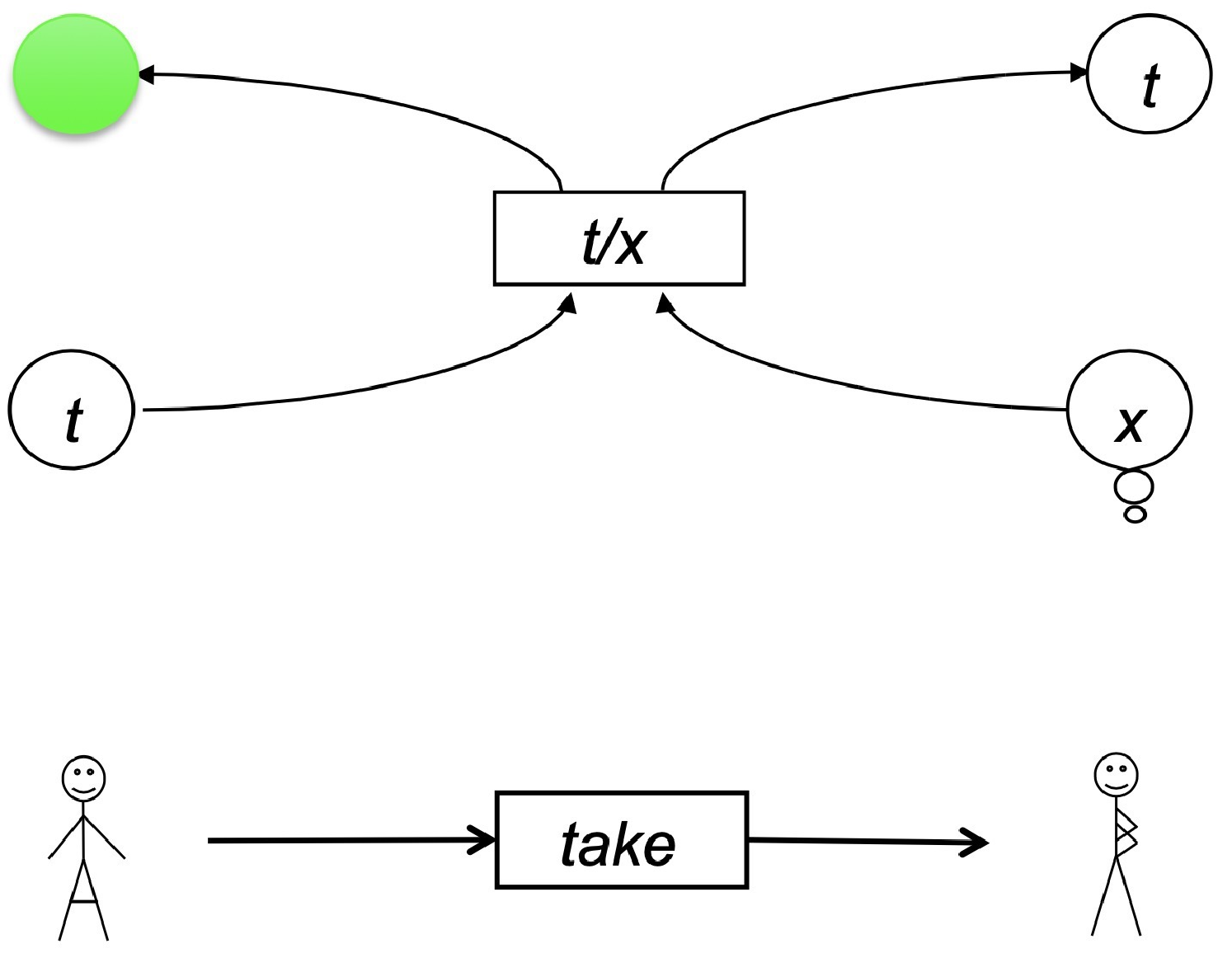}
\caption{Thing transmissions: give/have and have/take}
\label{Fig:au-thing}
\end{center}
\end{figure}
Fig.~\ref{Fig:au-thing} shows a channel transmitting things. The difference is that Alice this time does \emph{not}\/ have the thing $t$ after she transmits it to Bob. This is marked by the green balloons denoting her state after the channel interactions. This time the two possible interactions are that Alice may \emph{give}\/ Bob the thing $t$ on her own initiative, or that Bob may \emph{take} it on his initiative. The precondition for Bob's taking $t$ is that Alice had it; the postcondition of Alice's giving $t$ is that Bob has it. Bob's state changes like in the data channel, from $x$ to $t$, but Alice's state now changes differently, from $t$ to nothing.

In addition to the data channels and the thing channels, there are also trait channels, usually biometric. On the input side, someone has a trait $t$. On the output side, someone compares that trait with a trait pattern $?t$, and if both match, confirms the pattern as $!t$. 
\begin{figure}[!ht]
\begin{center}
\includegraphics[width=12cm
]{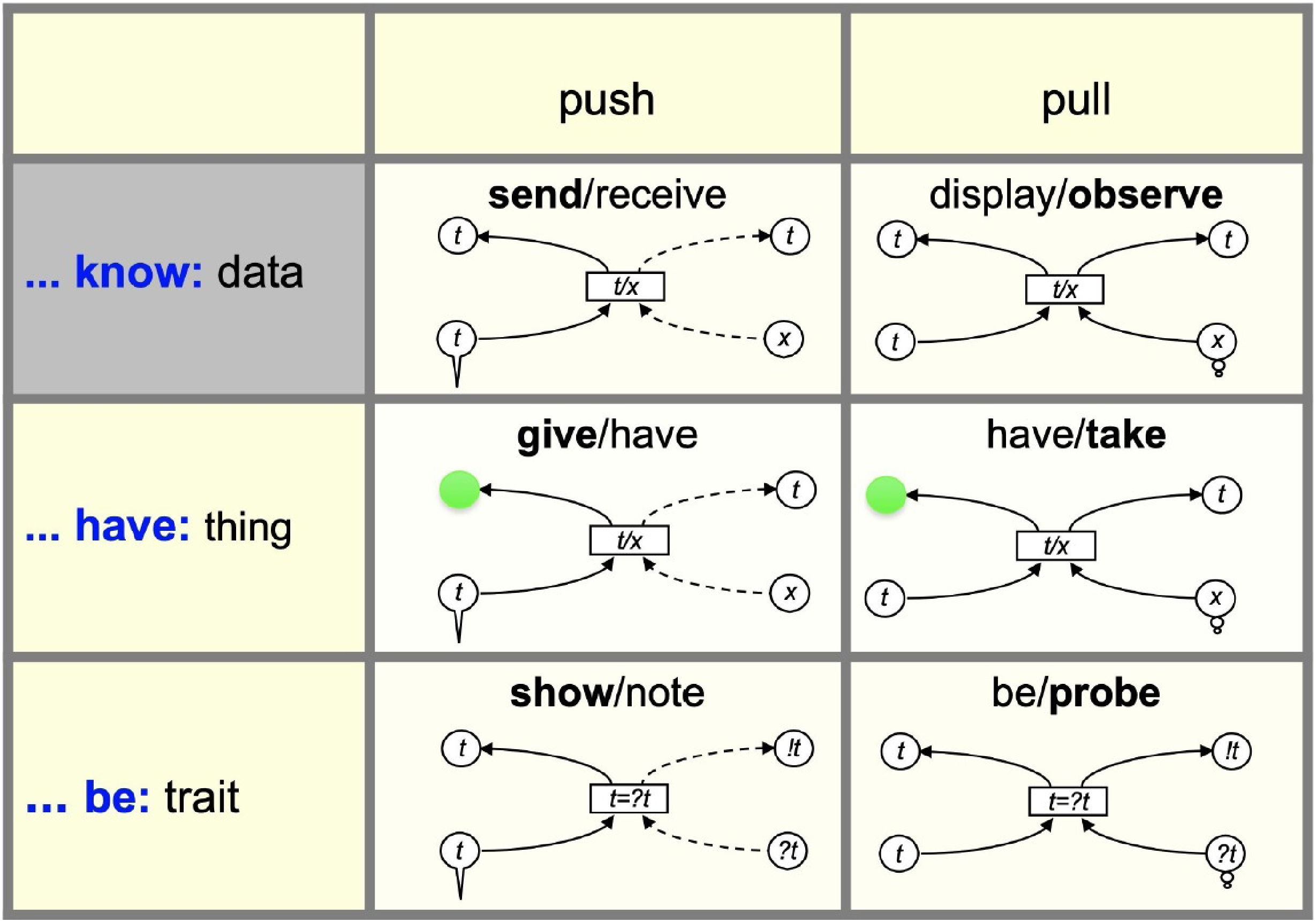}
\caption{Channel types and the supported interactions}
\label{Fig:know-have-be}
\end{center}
\end{figure}
Table~\ref{Fig:know-have-be} displays all three types of channels. The initiator actions are written in \textbf{bold}. The classification of channels into the three types is tidy and simple in general but often subtle in concrete cases. 
\begin{figure}[!h]
\begin{center}
\includegraphics[height=4cm
]{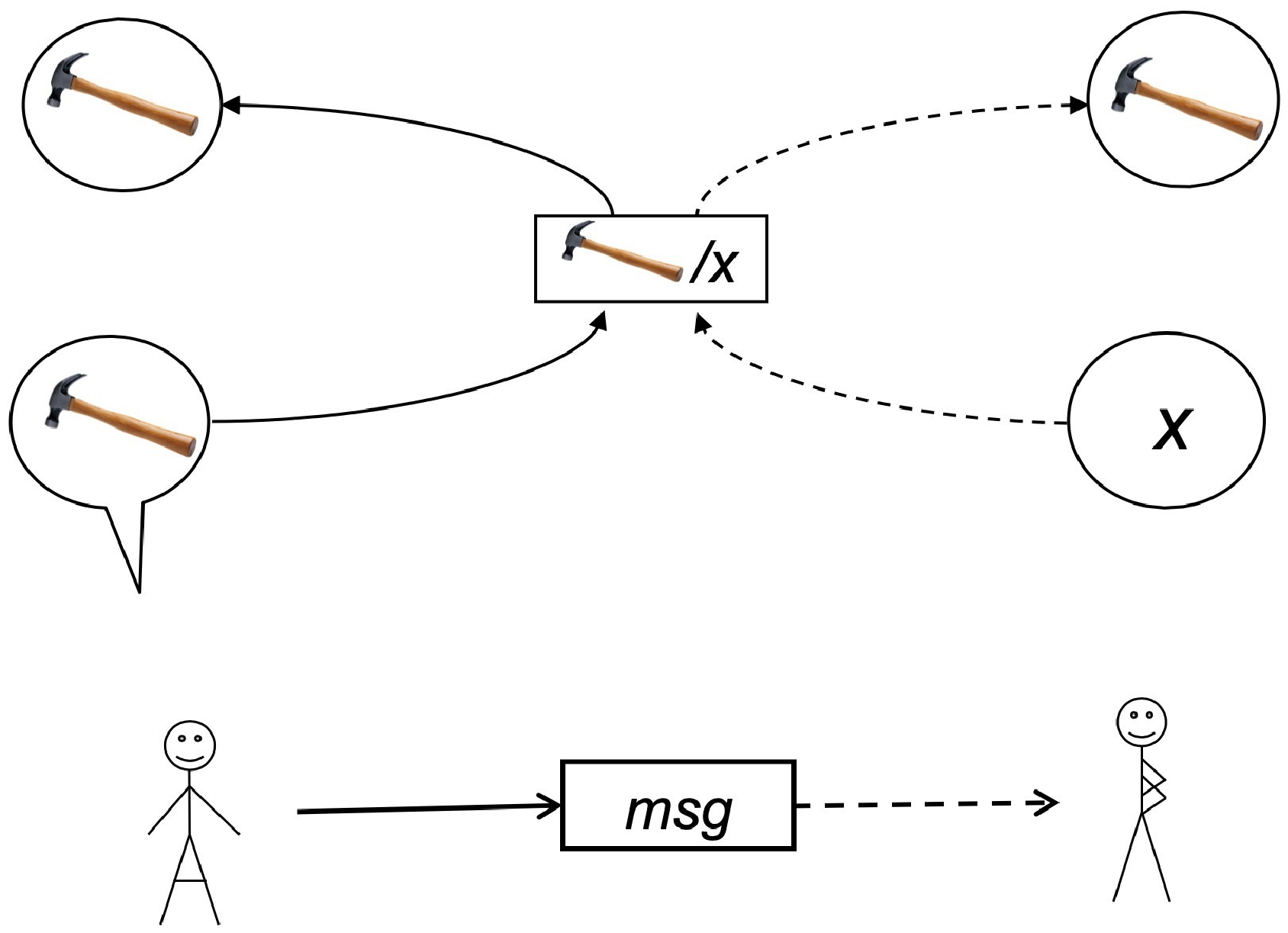}
\hspace{3em}
\includegraphics[height=4cm
]{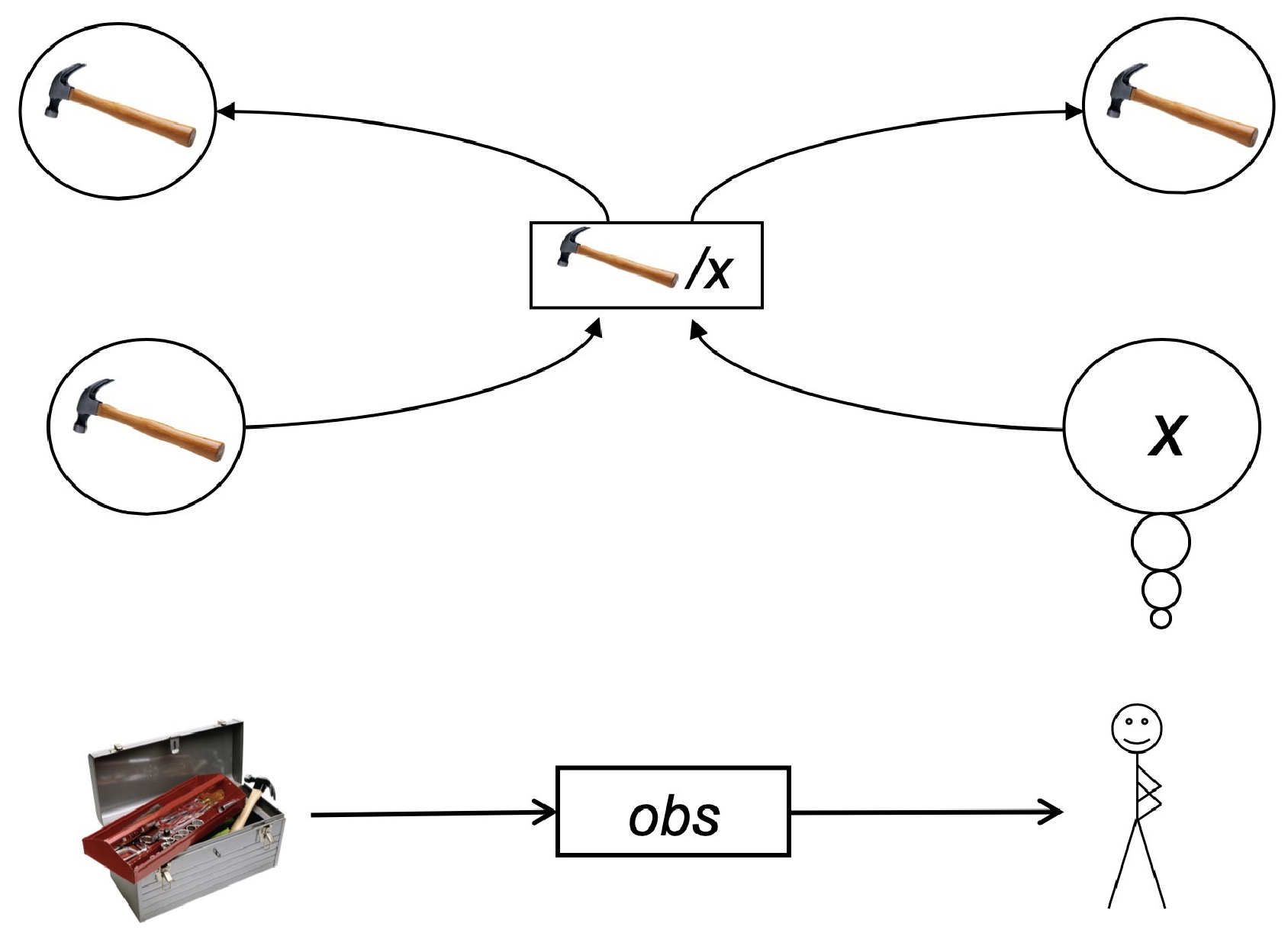}

\vspace{3ex}
\includegraphics[height=4cm
]{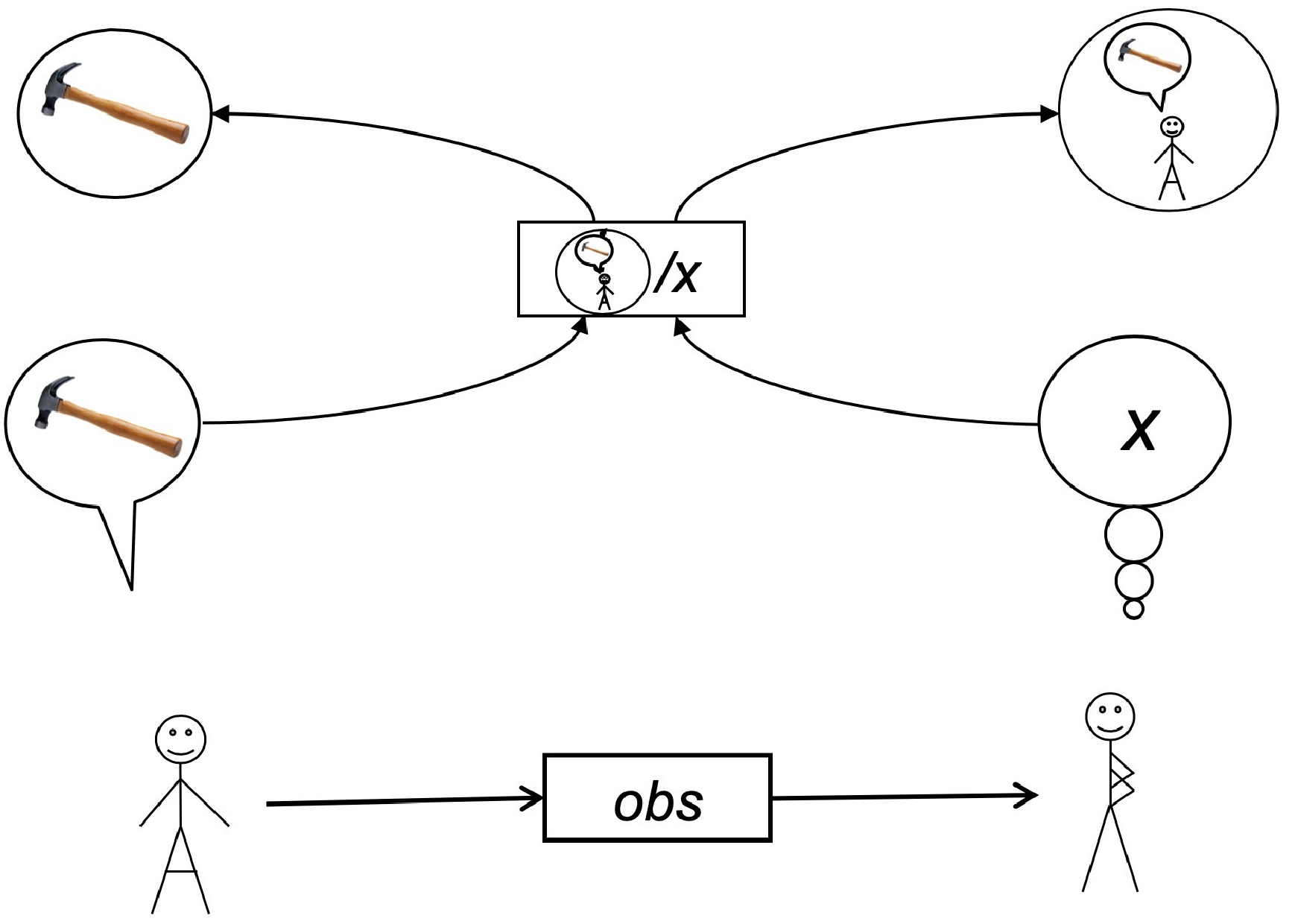}
\hspace{3em}
\includegraphics[height=4cm
]{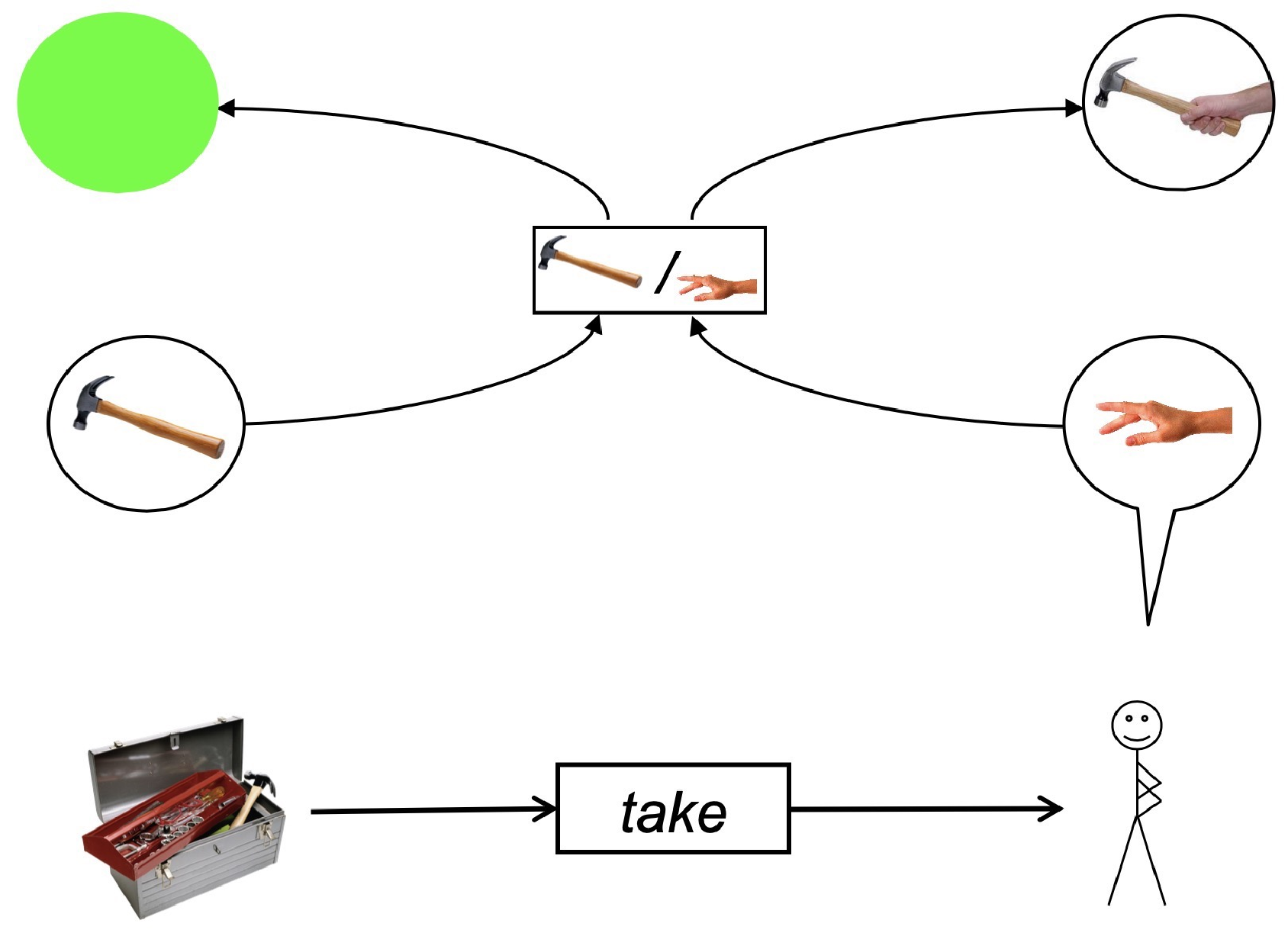}
\caption{A ``hammer'' mentioned and observed. Alice observed. A hammer taken.}
\label{Fig:hammer}
\end{center}
\end{figure}
Fig.~\ref{Fig:hammer} shows Alice asking Bob if he has a hammer, Bob observing a hammer in the toolbox, Bob observing Alice asking about the hammer,  Bob taking a hammer from the toolbox. 
\begin{figure}[!h]
\begin{center}
\includegraphics[height=4cm
]{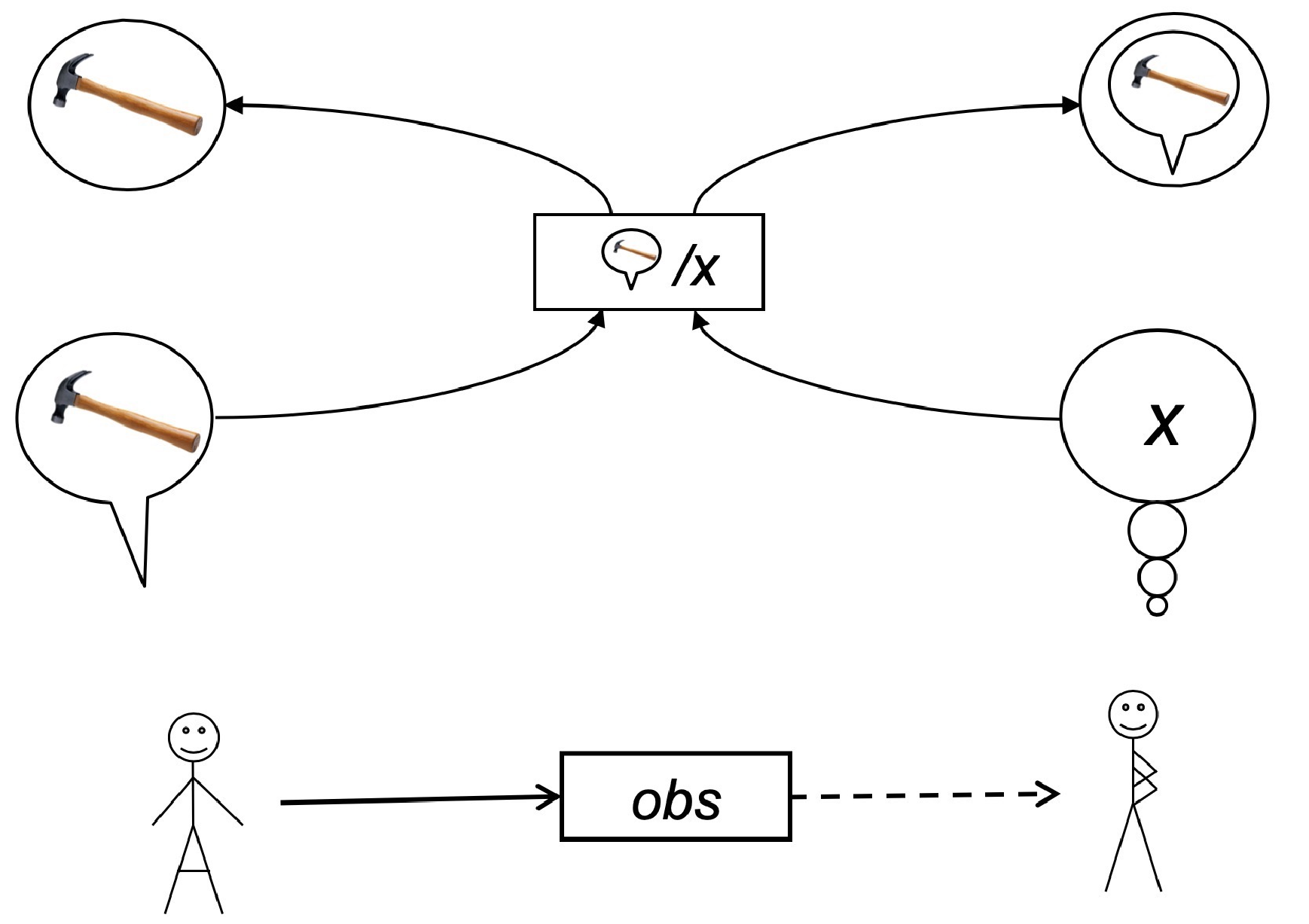}
\hspace{3em}
\includegraphics[height=4cm
]{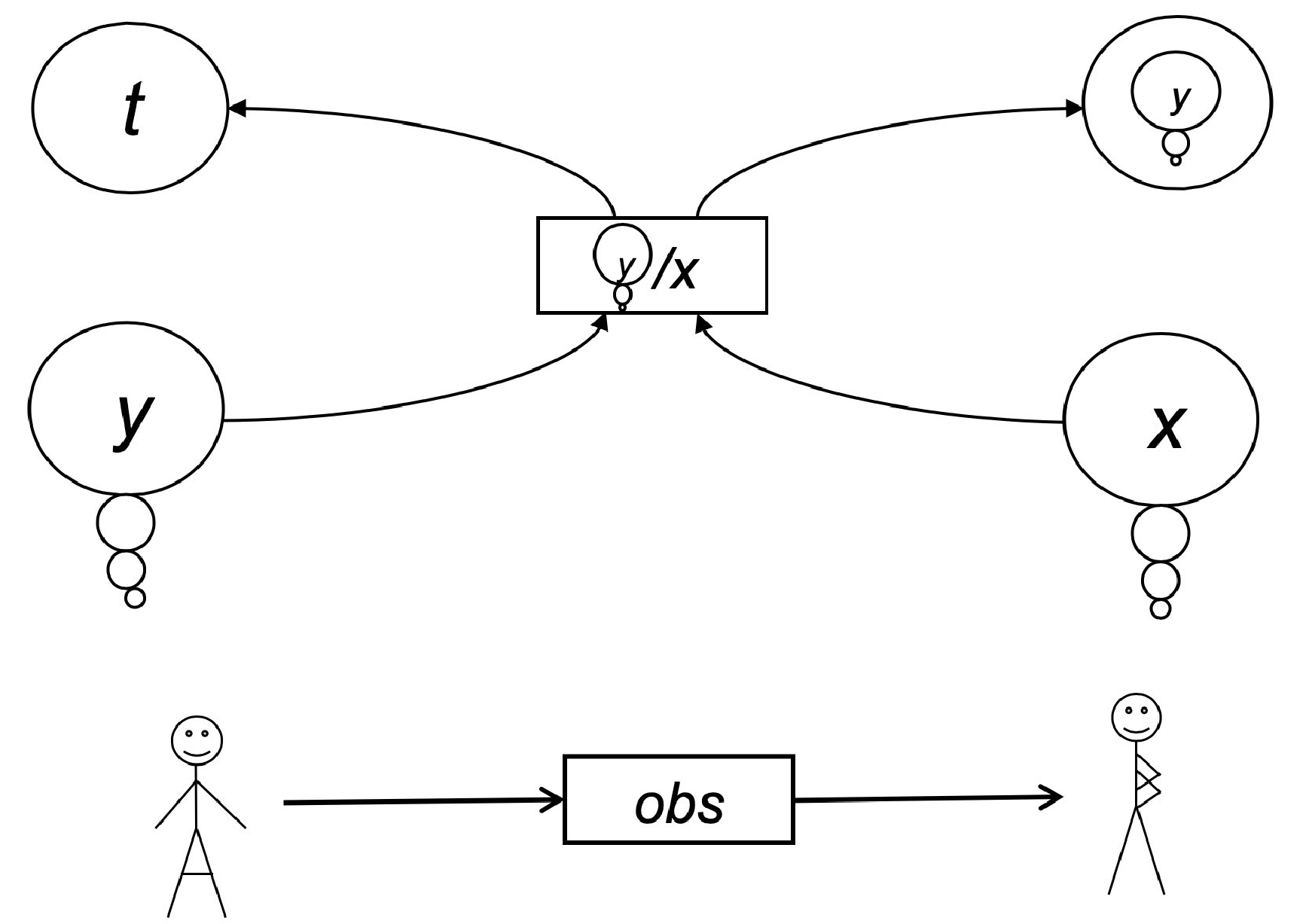}
\caption{Observing actions: sending vs observing}
\label{Fig:send-obs}
\end{center}
\end{figure}
Fig.~\ref{Fig:send-obs}, on the other hand, says that actions themselves may be observed, irrespective of the subjects performing them. On the left in this figure, Bob observes the action of asking about the hammer. This is different from the diagram in \cref{Fig:hammer} bottom left, where Bob observes Alice performing that action. Lastly, in Fig.~\ref{Fig:send-obs} on the right, Bob observes the action of observing. 

\label{Sec:comchan}

\section{Protocols}

\subsection{Composing networks from channels}\sindex{channel!composition}
Networks are built by composing channels. For example, given a channel $\AAA^{+}\tto f \WP\BBB$ from Alice to Bob and a channel $\BBB^{+}\tto t \WP \CCC$ from Bob to Carol, we compose a simple\footnote{A channel $\AAA^{+}\tto f \WP\BBB$ on its own is also a network, with the node set $\{\AAA,\BBB\}$ and the link set $\{f\}$. The pair of channels $\AAA^{+}\tto f \WP\BBB$ and $\CCC^{+}\tto h \WP\DDD$ is another one, with 4 nodes and 2 links. But these networks are trivial because the channels are isolated, and the transmissions cannot be composed.} network with the node set $\{\AAA,\BBB,\CCC\}$ and the link set $\{f,t\}$. The simplest way to pipe the outputs of $\AAA^{+}\tto f \WP\BBB$ as inputs into $\BBB^{+}\tto t \WP \CCC$ is to use \eqref{eq:cumul} and \eqref{eq:contchan} to present both channels as continuous
\[\prooftree
\AAA^{+}\to  \WP\BBB
\justifies
\WP\AAA^{\ast}\to  \WP\BBB^\ast
\endprooftree
\qquad\qquad
\prooftree
\BBB^{+}\to  \WP\CCC
\justifies
\WP\BBB^{\ast}\to  \WP\CCC^\ast
\endprooftree
\]
The network can now be viewed as a pair of composable functions:
\[
\WP\AAA^{\ast}\tto f\WP\BBB^{\ast} \tto t \WP\CCC^{\ast}
\]
Fig.~\ref{Fig:proto-comp} shows a concrete instance of such a composite. Bob observes that there is a hammer in the toolbox and tells Carol.
\begin{figure}[!h]
\begin{center}
\includegraphics[width=0.75\linewidth
]{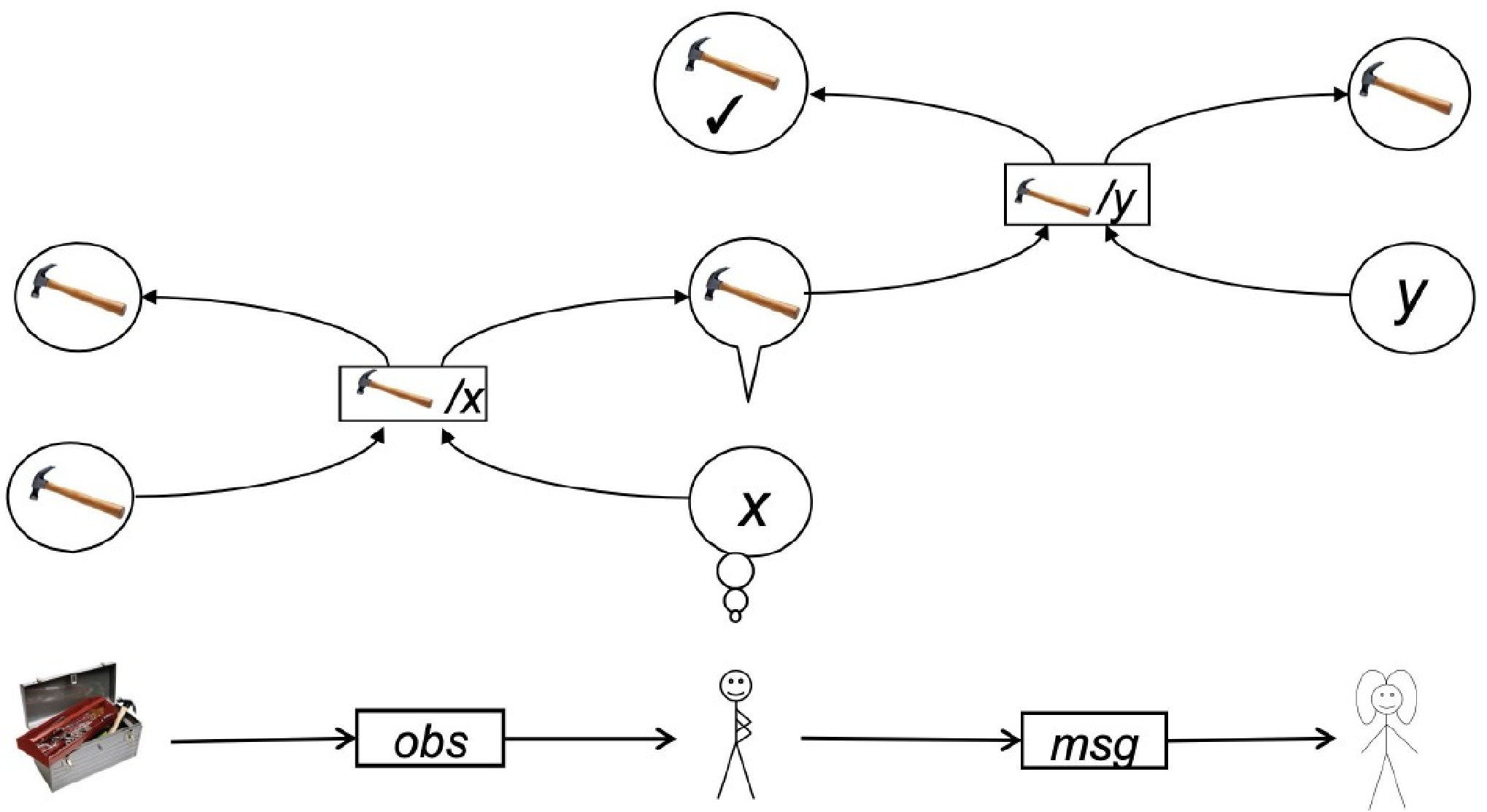}
\caption{Composing channels}
\label{Fig:proto-comp}
\end{center}
\end{figure}

\para{Remark.} Note that the subjects interacting in networks are not necessarily people, or computers. E.g. in this case, Alice is a toolbox. Any actor playing a role in a network computation is assigned a network node.

\subsection{Programming network computation}\sindex{network!computation}
The central feature of network models is that every event is localized at a node. A local event is only observable by local subjects. There is no global observer who sees everything and there are no events outside nodes. 

In a computer, there is a global observer \sindex{observer!global} who sees everything because everything happens in one place. You just look at that place and you are the global observer. For instance, in a Turing machine, all actions are performed by the head, reading and writing on the tape. Nothing ever happens anywhere else. You watch the head, and you are the global observer. 

In a network, every subject inhabits a network node and only observes the events at that node. But there are also communication channels between the nodes. The local computations at the nodes are coordinated through non-local communications through channels. A network is just a specification of nodes and links: the localities and the channels. Network computation is comprised of local computations coordinated through channel communications:
\bear
\mbox{ network computation }  & = & \mbox{ computation } + \mbox{ communication } 
\eear
\para{Protocols} prescribe the computations and the  communications distributed over networks, just like programs prescribe \sindex{protocol} the computations centralized in computers:\sindex{protocol}
\bear
\frac{\mbox{protocol}}{\mbox{program}} & = & \frac{\mbox{network}}{\mbox{computer}}
\eear
This indicates what protocols do. But how do they do it?

\subsection{What is a protocol?}
Protocols regulate network interactions, not just between computers, but also between people, between people and computers, between anyone and anything that plays a role in a network. Protocols prescribe the roles that network actors should agree to play. 

\para{Social life is a protocol suite.} A \sindex{protocol!supermarket} supermarket is a network of products waiting to be sold, shoppers coming to buy them, cashiers executing payment transactions. Fig.~\ref{Fig:protocol-supermarket} gives a high-level view of the payment protocol there.
\begin{figure}[!h]
\begin{center}
\includegraphics[width=0.5\linewidth
]{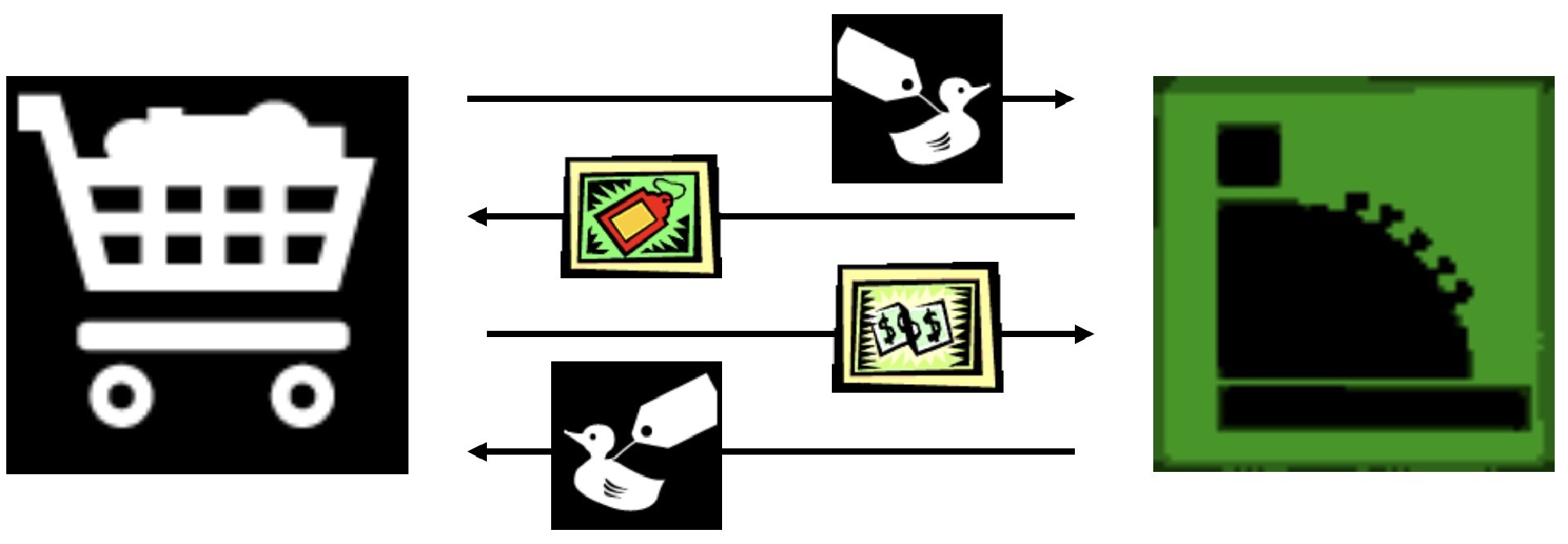}
\caption{Bird's eye view of the shopping protocol}
\label{Fig:protocol-supermarket}
\end{center}
\end{figure}
You go with the shopping cart to the cash register and present what you want to buy, say a rubber duck. The cashier establishes the price. You give the money and take the rubber duck. Every step of this abstract protocol refines to a much more complex protocol. The payment at the third step may be the most involved. It is also a multi-step protocol where every step refines to a protocol. If the shopper pays by contacting a bank, an entire suite of banking protocols is  run. Widening the perspective, layer upon layer of social protocols come into sight, from regulating how we wait in line before the cash register, through traffic protocols we obey as we drive home, to household protocols as we arrive home.

We follow protocols as instinctively as we follow grammars when we use languages. We enact complex social interactions with the same ease with which we generate complex sentences. Most of the time, we don't even become aware of the underlying structures. And the references that bind words into a sentence and messages into a protocol turn out to be of the same kind. --- The  referential structures underlying sentences and protocols  are the \emph{syntactic patterns}. \sindex{protocol!and grammar}

\textbf{Protocols are the \emph{syntactic patterns}\/ of network interactions.} But what is a syntactic pattern? Our communication channels are coded using languages. The crucial feature of a language, whether it is a natural language like English, or a programming language like Python, is that it has a \emph{syntax}\/ and a \sindex{syntax} \sindex{semantics} \emph{semantics}. Semantics conveys the meaning of words, syntax the form of sentences. But the word meanings depend on the sentence contexts, and the sentence structures depend on  the syntactic typing of words. Syntactically well-formed sentences are recognized through type-checking and type-matching. The subject of a sentence should be of type Noun Phrase, and it should be matched by a Verb Phrase. \sindex{type!checking} \sindex{type!matching}

In programming, we classify data \sindex{type!data} \sindex{datatype} into \emph{data types}\/ to make sure that the operations are correctly applied: e.g., that we only ever try to add or multiply numbers, and not people or tomatoes. If you try to multiply tomatoes, a \emph{type-checking}\/ error is raised.

In natural language, a similar process is based on assigning the \sindex{type!syntactic} \emph{syntactic types}\/ to words and phrases. To interpret this sentence that you are reading now, your mind assigns each word a type (e.g., the word ``mind'' is of type Noun) and it expects that, to make a sentence, the word of type Noun is matched by a word of type Verb (e.g., the word ``assigns'' is of type Verb). Hence, the syntactically well-formed phrase ``mind assigns''. But the verb ``assign'' is transitive, so your mind expects that the action of the verb ``assigns'' transitions to another noun. The word ``type'' is a Noun, and hence the phrase ``mind assigns type''. A couple of further refinements required by the English syntax yield the sentence ``your mind assigns a type''. By constraining the  candidate words in a sentence, syntax streamlines and simplifies the processes of generating and understanding language. The underlying computations are the \emph{syntactic processes}.

Protocols work in a similar way. In a conversational protocol extending Example 3 in Sec.~\ref{Sec:au-av}, the statements may be of type Question or of type Answer. If you interpret a statement as a question, you expect it to be matched by an answer. In an authentication protocol, the messages may be of type Challenge or of type Response. \sindex{protocol!challenge-response} A challenge (e.g., a wren challenges the chicks by offering the food) needs to be matched by a response (e.g., the chicks' respond by chirping). 

Security protocols have been studied, designed, and analyzed since the early days of network computation. A variety of models and theories has been developed and used \cite[to mention just a few]{BellaG,Boyd-Mathuria,Cremers-Mauw:book,Ryan-Schneider:book}. Most are beyond our scope and besides our goals. We will just take a quick look at a simple but central family: the  challenge-response authentication protocols.

\label{Sec:protocol}

\section{Authentication}

\subsection{Problem of communication}
The general problem of communication between Alice and Bob is that Alice only sees her side of a communication channel, whereas Bob only sees his side. 

\para{Communication as sharing meaning.} If Alice says ``I love you'', she would like to know that Bob has heard and understood what she said. Bob, on the other side, would like to know that she really said what he heard and that she really meant what she said. Alice and Bob are \emph{communicating}\/ if they achieve a \emph{common}\/ interpretation of the messages between them. If an eavesdropper Eve intercepts Alice's message and changes ``I love you'' to ``I hate you'', then the communication between Alice and Bob has been disrupted. Eve can achieve the same miscommunication effect if she prepares Bob to misinterpret Alice's message by telling him misleading gossip. 

The goal of communication from Alice to Bob is to transmit a message and assure a common meaning, accepted by the subject on both ends of the channel. 

However, communication is a process. Channels transmit strings, transmitting strings takes time, and the interpretations evolve as the transmission progresses. What Alice means and how Bob interprets it is not fixed and often remains ambiguous. The interpretations can be narrowed down through feedback, e.g., Bob questioning the meaning and Alice answering. But the problem of interpretation applies to the questions and the answers. To clarify what they mean, Alice and Bob may refer to other nodes on the network (e.g., a hammer as the meaning of the word ``hammer'', or the contexts where the word ``hammer'' occurs). But the evolution of meaning never ends. Communication remains a network process. 

\subsection{Network security}
Data, things, and traits that are communicated successfully, in the sense that the meaning at the source is transmitted to the target of communication, are said to be  \textbf{\emph{authentic}}. \sindex{authenticity} 

Examples of data, things, and traits that need to be authenticated include:
\begin{itemize}
\item a command claiming to be from the commander and not from pirates; 
\item a passport claiming to be original and not forged;
\item a voice claiming to be human and not artificial.
\end{itemize}
They are authentic if they are what they claim to be. 

\para{Authenticity vs integrity.} Integrity is a property closely related to authenticity. It is sometimes viewed as the guarantee that the contents of a channel are not only unchanged, but that they were not even accessed. \sindex{integrity} Integrity of data, things, or traits in that sense could be breached even if they are unchanged but were accessed by the adversary. In other contexts, integrity is identified with authenticity. Yet other times, authenticity is narrowed to mean integrity of the origination claims.

In network security, the \emph{good-stuff-should-happen}\/ requirements are authenticity and integrity; the \emph{bad-stuff-should-not-happen}\/ requirements are secrecy and  confidentiality. In this chapter we study the good stuff. The bad stuff is covered in the next chapter.

\para{Remark.} In Sec.~\ref{Sec:good-bad}, we defined authenticity and integrity as properties of data transmissions alone. Here we widen the angle and view them as properties of network transmissions, which include data, things, or traits. This does not change but rather refines the definition. Data have carriers. The primordial data carrier was our voice. Then we wrote on clay tablets and paper. Nowadays, data are carried by electronic signals. In any case, all channels have a physical aspect. The other way around, every thing means something. Recognizing a thing means giving it a name and a meaning. You recognize a hammer by recognizing the word ``hammer'' as its name and its use for hammering as its meaning. A person's trait similarly carries a social meaning. The distinctions between data, things, and traits, are convenient for classifying the families of security tools and channels; but at a closer look, the demarcation lines between data, things, and traits often shift and blur. The same security properties apply to all types of channel contents.

\para{Network security problem.} 
\begin{figure}[!h]
\begin{center}
\includegraphics[width=0.75\linewidth
]{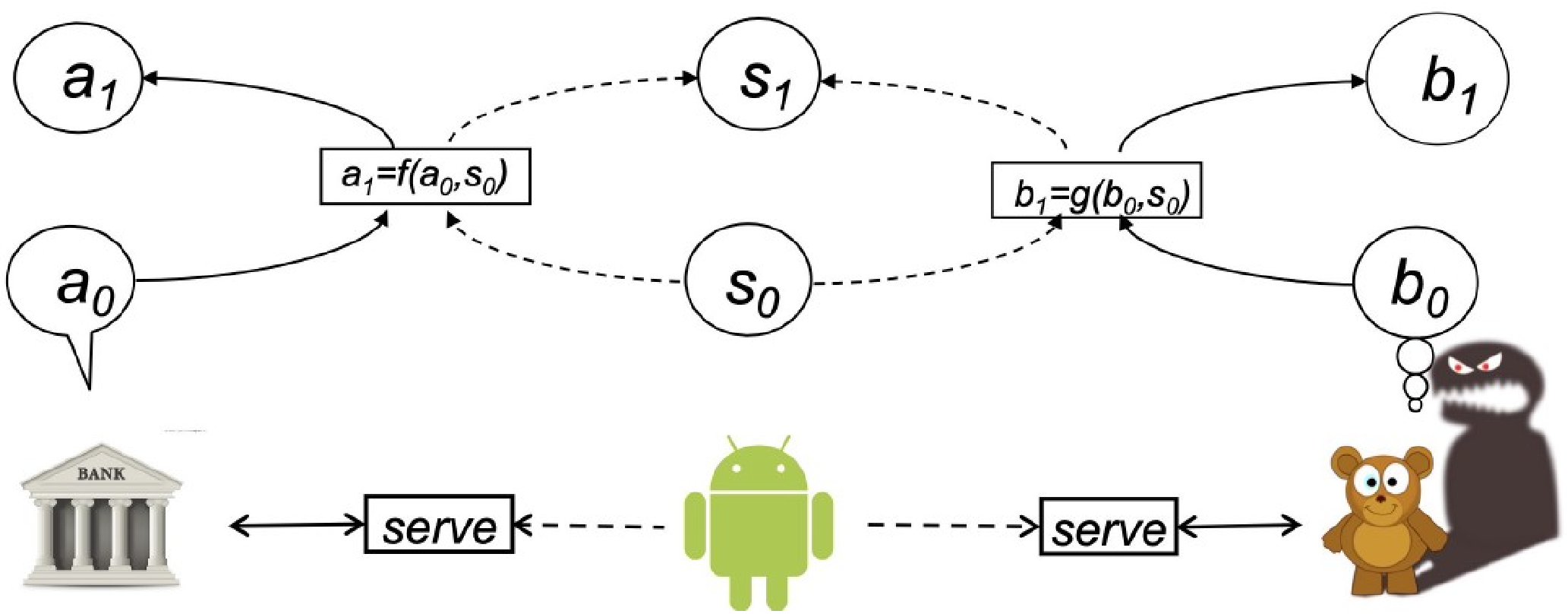}
\caption{Channel compositions lead to interference and attacks}
\label{Fig:proto-nocomp}
\end{center}
\end{figure}
If Bob interacts with Alice in one part of the network and with Carol in another, then Alice and Carol cannot prevent Bob from implementing a covert channel between them even if there are no overt or shared channels between them. Suppose that Bob is an operating system in a smart device, serving Alice as a banking application and Carol as a game, as illustrated in \cref{Fig:proto-nocomp}. Then Carol may be able to initiate banking transactions through a covert channel with Alice, although there is no overt channel between them.
 
This kind of security problem arises even in a single-channel as a network $\AAA^{+}\to \WP\BBB$ on its own. If Eve hijacks Bob's role, she may steal Alice's data or things. This happened in the wren feeding protocol in Sec.~\ref{Sec:wren}, where the cuckoo as Eve stole food and killed the wren chicks. If Eve hijacks Alice's role, then she may impersonate Alice to Bob. If Bob is a bank, Eve may take control of Alice's account.

\para{Authentication protocols}\sindex{protocol!authentication} provide solutions for such problems. The wren protocol in Sec.~\ref{Sec:wren} was an example of authentication. What is authentication?

\subsection{Idea of authentication}
Authentication is a proof that someone is who they claim to be. If a computer engaging in an online chat claims to be human, proving that claim is authentication. The \emph{Turing Test}\/ \sindex{Turing Test} is an authentication task. If a chick in wren's nest claims to be a wren chick, a proof of that claim is an authentication.  If an owner of a painting claims that the painting was made by Picasso, a proof of that claim is an authentication. After the claim is proven, the painting is said to be an \sindex{authenticity} \emph{authentic}\/ Picasso. If a visitor of an online banking website claims that they are Alice, then Bob, the bank, asks that they authenticate themselves. If Alice is supposed to prove who she is by providing to Bob her credentials, then she may also need to authenticate Bob if a malicious actor impersonating Bob could copy her credentials and impersonate her to the actual bank. So Alice and Bob need to authenticate each other. In the next section, we will see an example of a \emph{mutual}\/ authentication protocol.\sindex{authentication!mutual}

Logically speaking, authentication is an instance of \emph{hypothesis testing}. The claims that a chick is a wren, that a painting is a Picasso, that a visitor of an online banking service is Alice are hypotheses that require indirect verification, because the direct evidence is unobservable: the wren chick's distinguishing properties are unobservable for its parents, the act of Picasso painting his painting is remote in time, the identity of a website visitor is remote in space. Just like scientists seek ways to indirectly verify hypotheses about unobservable subatomic particles, authenticators seek ways to verify hypotheses about unobservable network nodes.

\subsection{Example: the Needham-Schroeder Public Key (NSPK) protocol}\label{Sec:NSPK}
\sindex{protocol!security}
The NSPK protocol is a mutual authentication protocol based on \emph{confidential channels}, implemented using public-key cryptography. It goes back to the early days of computer security \cite{Needham-Schroeder} and is almost as simple as the wren authentication protocol from Sec.~\ref{Sec:wren}. Yet, analyzing it took many years and even more publications. Its security features are independent of the cryptographic implementation of its confidential channels\footnote{Cryptographic protocols are usually analyzed assuming that the underlying cryptographic functions are perfectly secure. There are lots of cryptography textbooks that study what that means and how it is achieved.}, so we sketch a high-level view of that aspect.

\subsubsection{NSPK channels}
The NSPK protocol is run over confidential channels realized through public-key cryptography. 

\para{Idea of confidential channels.} \sindex{channel!confidential} If Alice wants to send to Bob the confidential data $t$, she applies Bob's encryption function $\{-\}_{B}$ and forms the message $\{t\}_{B}$. The encryption scrambles  $t$, so that it cannot be recognized or extracted from $\{t\}_{B}$ by anyone except Bob. The message $\{t\}_{B}$ can therefore be sent to Bob through the channels of a public network, while maintaining the confidentiality of $t$. 
You can think of  $\{-\}_{B}$ as a lockbox that anyone can close, but only Bob can open. Alice puts $t$ in the box $\{-\}_{B}$, locks it, and sends $\{t\}_{B}$ to Bob. 
\begin{figure}[!ht]
\begin{center}
\includegraphics[height=4.5cm]{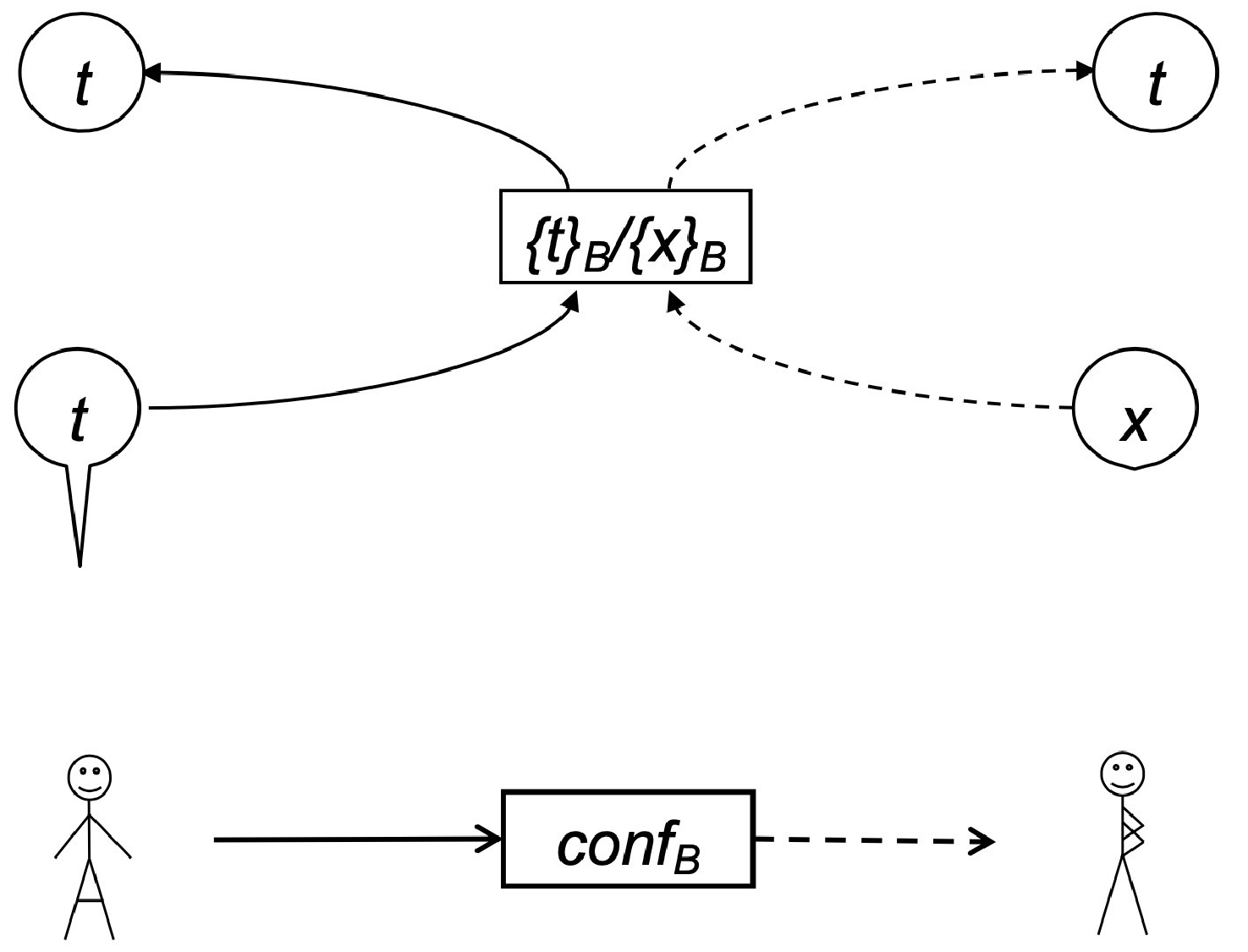}
\caption{A confidential channel}
\label{Fig:proto-confidential}
\end{center}
\end{figure}
Lots of posthandlers may handle the closed box $\{t\}_{B}$ on its way to Bob, but no one can learn anything about $t$ until Bob receives the box and opens it. That is the idea of a \emph{confidential channel}, \sindex{channel!confidential} displayed in \cref{Fig:proto-confidential}. Bob's capability to open the box $\{t\}_{B}$ and extract $t$ is presented by the \emph{pattern-matching}\/ operation $\{t\}_{B}/\{x\}_{B}$ which results in storing the data $t$ into the variable $x$ that Bob had ready for the data that he may receive through the confidential channel.  

The next paragraph provides a closer view of the last paragraph. You may skip it, but don't skip the protocol that follows.

\para{Implementation of confidential channels: public-key cryptosystems.\mbox{\large $^{\star}$}} \sindex{cryptosystem}The confidential channel schema in \cref{Fig:proto-confidential} can be implemented using public-key cryptography. A public-key cryptosystem is a triple $<G,E,D>$ of functions\footnote{More precisely, $G$ and $E$ are randomized algorithms.} where
\begin{itemize}
\item $G$ is a random generator of public/private key pairs $<k,\overline k>$,
\item $E_{k}$ is a $k$-indexed family of encryption functions, and
\item $D_{\overline k}$ is an $\overline k$-indexed family of decryption functions,
\end{itemize}
such that
\begin{itemize}
\item[{\sc(fun)}] all $t$ satisfy $D_{\overline k}\Big(E_{k}(t)\Big) = t$, and 
\item[{\sc(sec)}] if $A\Big(E_{k}(t)\Big) = t$, then $A=D_{\overline k}$, i.e., any decryption algorithm depends on the private key $\overline k$.
\end{itemize}
Cryptographers use more precise definitions, but this gives you an idea of what they say. 

The idea of how cryptosystems are used is that Bob first obtains from $G$ a key pair $<k,\overline k>$, of which he announces $k$ as his public key and keeps $\overline k$ as his private key. For a fixed key pair, the function pair $<E_{k},D_{\overline k}>$ is called a \emph{cipher}\sindex{cipher}. An  encrypted message $E_{k}(t)$ is called the \emph{ciphertext}\/ and $t$ is the corresponding \emph{plaintext}. When Alice wants to transmit in confidence a plaintext $t$ to Bob who announced the public key $k$, she generates the ciphertext $\{t\}_{B} = E_{k}(t)$, and sends it through a public channel (such as the Internet or a cellular network). By the functional property {\sc(fun)}, Bob can extract $t = D_{\overline k}\left(E_{k}(t)\right)$. By the security property {\sc(sec)}, anyone who can extract $t=A\Big(E_{k}(t)\Big)$ using some algorithm $A$ must know Bob's private key $\overline k$.

\para{High-level view of confidential channels: pattern-matching.} All that matters for the NSPK protocol is that Bob is the only one who can extract $t$ from $\{t\}_{B}$, because this is the capability used to identify him (like the chirps were used to identify baby wrens). Therefore, we reduce the whole story of public-key cryptography to the \emph{pattern-matching operation} 
\bea\label{eq:B-PKI}
\{t\}_{B}/\{x\}_{B} & \vdash & t/x
\eea
which compares the patterns $\{-\}_{B}$ of the terms on the left, and since they match, it stores the content $t$ of the first pattern into the variable $x$ of the second pattern. For a general pattern $p$ and an arbitrary term $T$, the pattern-matching operation $T/p(x)$ is defined by the reduction rule
\bea\label{eq:PKI}
p(t)/p(x) & \vdash & t/x
\eea
In words, the operation $T/p(x)$ substitutes $t$ for $x$ if $T=p(t)$.

\subsubsection{NSPK interactions}
Fig.~\ref{Fig:proto-NSPK} shows two views of the NSPK protocol.
\begin{figure}[!h]
\begin{center}
\includegraphics[width=0.4\linewidth
]{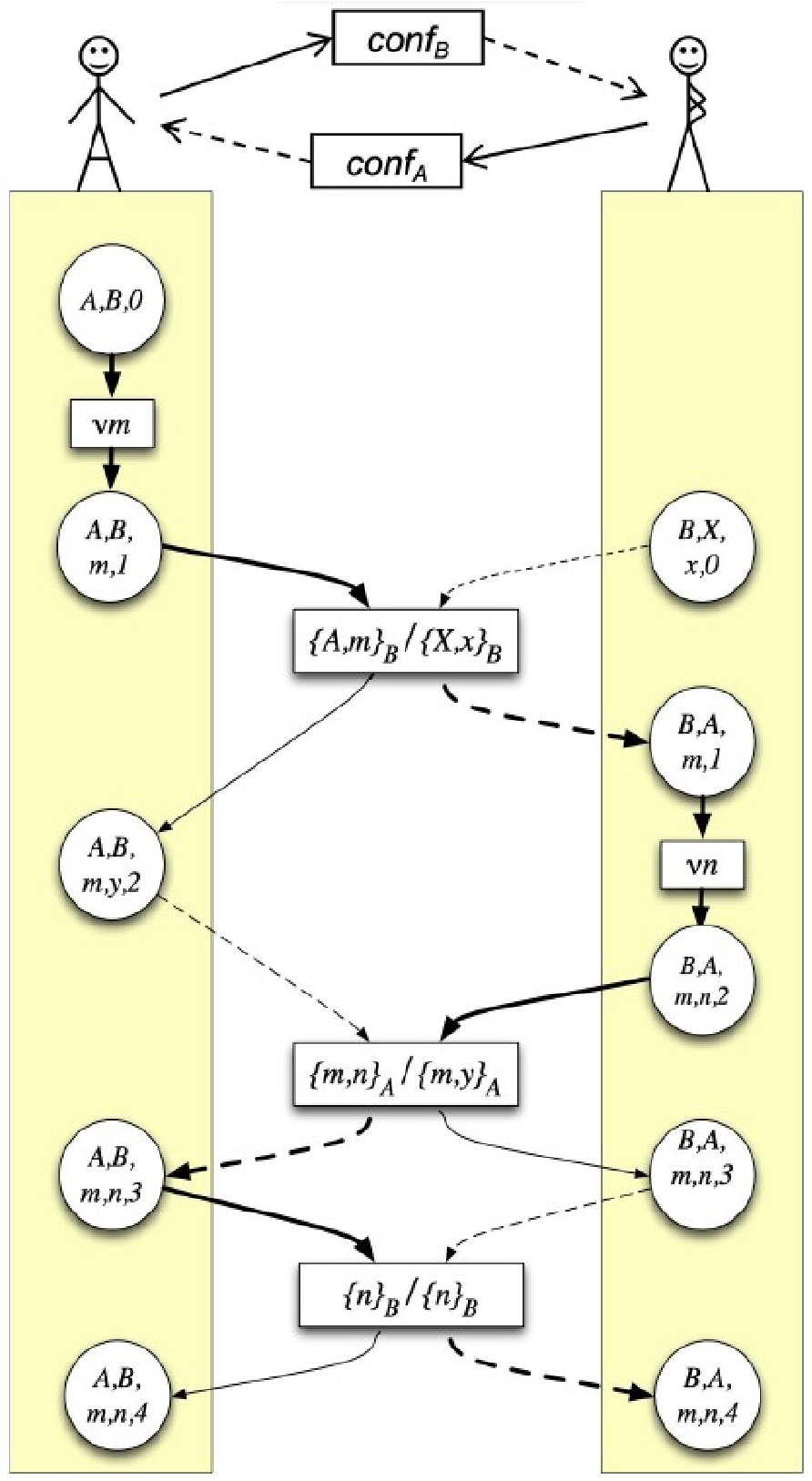}\hspace{0.1\linewidth}
\includegraphics[width=0.4\linewidth
]{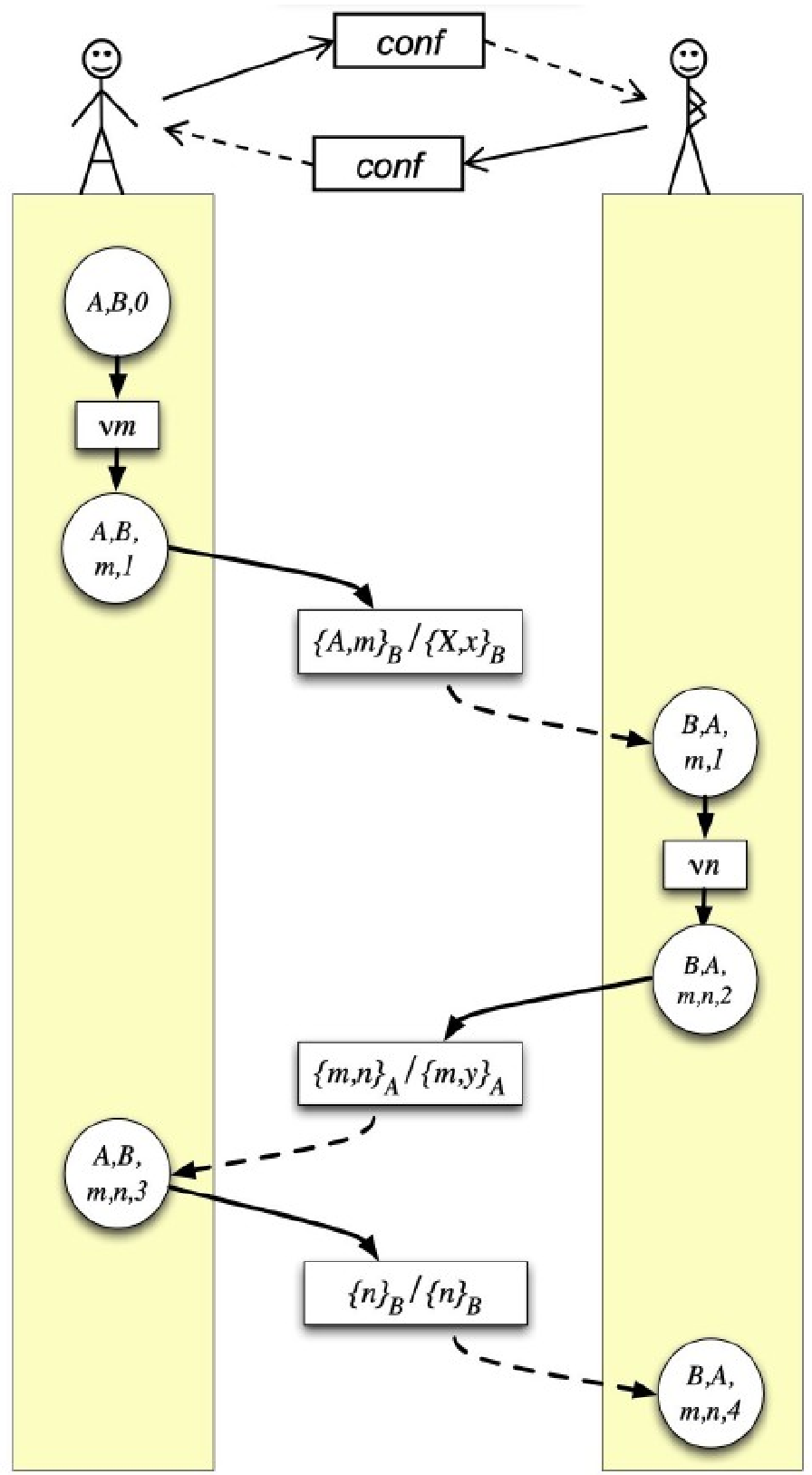}
\caption{The legendary Needham-Schroeder Public Key (NSPK) protocol}\sindex{protocol!Needham-Schroeder Public Key (NSPK)}
\label{Fig:proto-NSPK}
\end{center}
\end{figure}
On the top of each of the pictures is the network on which the protocol runs: just two nodes, Alice and Bob, and two confidential channels between them, one in each direction. The diagram on the left displays the protocol as a composition of three confidential channel interactions, instances of \cref{Fig:proto-confidential}. Each of the interactions is a message send/receive pair along a confidential channel. The diagram on the right omits the senders' state updates after the send actions and the receivers' ready states before their receive actions. This leaves on the right a protocol picture of the type that is usually found on protocol designers' whiteboards and in their papers, allowing you to follow the protocol actions as a linear sequence. This view abstracts away the fact that receivers cannot receive messages unless they are ready to receive them, and that senders need to remember that they sent a challenge to be able to receive and verify the response. But taking this into account requires looking at state changes in two places at the same time, like in the diagram on the left. We are primed to think in linear time and understand diagrams like the one on the right much easier. It is often a good idea to draw both views, to avoid confusion.

What happens in this protocol? Alice and Bob go through 4 state changes each. Initially, Alice is in the state $(A,B,0)$, which means that she knows that she is Alice, that she wants to talk to Bob, and that her counter is at $0$. Bob is in the state $(B,X,x,0)$, which means that he knows that he is Bob, that he is ready to talk to someone whose name he will store in the variable $X$, their session identifier into the variable $x$, and his counter is also at $0$. 

Alice begins a protocol session by performing the action $\nu m$, which stores a new value into the variable $m$ and changes Alice's state $(A,B,0)\mapsto (A,B,m,1)$.  The value in $m$ is Alice's \emph{nonce},~\sindex{nonce} a value that she will \emph{n}\/ot use but \emph{once}. It will allow her to distinguish the messages related to this protocol session from other messages that she may receive in the meantime. 

Alice then begins the conversation by sending the message  $\{A,m\}_{B}$, which here means:  ``Hi, this is Alice, my nonce is $m$. Please prove that you are Bob by decrypting this message.'' --- This is \textbf{Alice's \emph{challenge}\/ to Bob}, requesting that he decrypts from inside $\{-\}_{B}$. She records that she has sent it by changing the state $(A,B,m,1)\mapsto (A,B,m,y,2)$, remembering all her previous values, but now ready to receive Bob's response into $y$.

Bob, on the other hand, was ready to receive and decrypt messages in the form $\{X,x\}_{B}$ and store the content into $X$ and $x$. Upon receiving and decrypting $\{A,m\}_{B}$, he stores $A$ into $X$, $m$ into $x$, and changes his state $(B,X,x,0)\mapsto (B,A,m,1)$. Then Bob also generates a nonce $n$ by $\nu n$, changes his state $(B,A,m,1)\mapsto (B,A,m,n,2)$, and sends $\{m,n\}_{A}$, meaning: ``Hi, in response to $m$, here is my nonce $n$. Please prove that you are Alice by decrypting this message.'' --- This message is \textbf{Bob's \emph{response} to Alice} (because he proves by sending $m$ that he is able to extract data from $\{-\}_{B}$) and moreover, it is also  \textbf{Bob's \emph{challenge} to Alice} to decrypt from inside $\{-\}_{A}$.

Alice was ready to receive messages in the form $\{m,y\}_{A}$. Upon decrypting Bob's $\{m,n\}_{A}$, she verifies that the first component is her nonce $m$ and that her challenge has been responded to. She then stores $n$ in  $y$, and changes her state $(A,B,m,y,2)\mapsto (A,B,m,n,3)$. Alice's remaining task is to respond to Bob's challenge. She sends $\{n\}_{B}$. This is \textbf{Alice's \emph{response} to Bob} (because she proves by sending $n$ that she is able to extract data from $\{-\}_{A}$) and she transitions $(A,B,m,n,3)\mapsto (A,B,m,n,4)$ to settle in her final state.

Bob, in the meantime, is ready to receive a message in the form $\{n\}_{B}$, which verifies that his challenge has been responded to. Upon receiving such a message, Bob transitions $(B,A,m,n,3)\mapsto (B,A,m,n,4)$ and settles in his final state.

\para{What has been achieved?} Bob and Alice have not only responded to each other's challenges, and authenticated each other; they have also responded to the  challenges $m,n$ that were \emph{bound together}\/ in a single message. The protocol thus does not just provide two authentications: of Bob by Alice and of Alice by Bob. It provides a \emph{mutual}\/ \sindex{protocol!mutual authentication} authentication, with both of them requesting authentication and accepting to be authenticated, within the same session. Last but not least, since the nonces $m$ and $n$ were only ever sent under Alice's or Bob's public keys, they can be used to generate a shared key that only Alice and Bob know. NSPK is therefore not only a \emph{mutual authentication}\/ protocol, but also a \emph{key distribution}\/ protocol. \sindex{protocol!key distribution} 

\subsubsection{Attack!!!} \label{Sec:attack}
The problem with the NSPK protocol is that it does not require that Bob explicitly identifies himself, like Alice did with ``Hi, this is Alice'', at the beginning. The protocol binds Alice's authentication of Bob by the challenge and the response
\beq\label{eq:NSPK-A} c_{AB}^{(m)} = \{A,m\}_B \qquad\qquad\qquad r_{AB}^{(m)}= \{m\}_A\eeq
and Bob's authentication of Alice by the challenge and the response
\beq\label{eq:NSPK-B} \textcent_{BA}^{(n)} = \{n\}_A \qquad\quad\qquad\qquad r_{BA}^{(n)} = \{n\}_B\eeq 
Bob's response $r_{AB}^{(m)}=\{m\}_A$ and his challenge $\textcent_{BA}^{(n)} =\{n\}_A$ are fused into $\{m,n\}_A$ in the second message, which binds the Alice-Bob authentication and the Bob-Alice authentication into a \emph{mutual}\/ authentication between Alice and Bob. \sindex{authentication!mutual} However, the symmetry is broken because Alice's challenge to Bob $c_{AB}^{(m)} = \{A,m\}_B$ says who is the challenger, whereas Bob's challenge to Alice $\textcent_{BA}^{(n)}=\{n\}_A$ does not. Why is this a problem?

If Bob is dishonest, he may decrypt Alice's challenge $c_{AB}^{(m)}=\{A,m\}_{B}$, re-encrypt it by Carol's public key and send $c^{AC}=\{A,m\}_{C}$ --- as Alice's challenge to \emph{Carol}. Since Carol as the responder in the NSPK protocol is not required to identify herself, Bob can forward her response $\{m,n\}_{A}$ to Alice as his own. The originator of the challenge $n$ for Alice is not mentioned, and Alice will think that this is Bob, because the challenge $n$ is bound with the response to her challenge $m$, which she issued under Bob's public key, and Bob must have decrypted it. But Bob is dishonest, and he transformed the challenge for him into a challenge for Carol. Alice will decrypt $\{m,n\}_{A}$ and respond with $\{n\}_{B}$, which Bob will decrypt, re-encrypt as $\{n\}_{C}$ --- and thus transform into a response to Carol. In the end, Carol thinks she has a mutually authenticated confidential channel with Alice, Alice thinks she has a mutually  authenticated  confidential channel with Bob, and Bob is sitting in the middle, relaying the messages between the two of them. If Carol is a bank where Alice has an account, then Bob can withdraw Alice's money. 

The Needham-Schroeder Public-Key (NSPK) and Symmetric-Key (NSSK) protocols were proposed in 1978~\cite{Needham-Schroeder}. \sindex{protocol!Needham-Schroeder Public-Key (NSPK)}\sindex{protocol!Needham-Schroeder Symmetric-Key (NSSK)} The above attack on NSPK  was discovered 17 years later~\cite{Needham-Schroeder-Lowe}, by a computer, in an automated analysis. In the meantime, the protocol was widely used to secure shared computers and studied in 100s (if not 1000s) of excellent publications, without anyone noticing the problem. It is true that Needham and Schroeder never said that the protocol was secure against a dishonest responder. But it is also true that neither they nor anyone else said that it was not secure against a dishonest responder --- until the attack surfaced in \cite{Needham-Schroeder-Lowe}, from exhaustive search of protocol runs in a formal model. 

The attack is easily prevented by making sure that Bob's authentication of Alice is not \eqref{eq:NSPK-B}, but symmetric to Alice's authentication of Bob in \eqref{eq:NSPK-A}, with 
\beq\label{eq:NSPK-B-fix} c_{BA}^{(n)} = \{B,n\}_A \qquad\quad\qquad\qquad r_{BA}^{(n)} = \{n\}_B\eeq 
The second message $\{m,n\}_{A}$ now becomes $\{m,B,n\}_{A}$. The former was the fusion of $r_{AB}^{(m)}=\{m\}$ and $\textcent_{BA}^{(n)} = \{n\}_{A}$. The latter is the fusion of $r_{AB}^{(m)}=\{m\}$ and $c_{BA}^{(n)} = \{B,n\}_{A}$. Bob's impersonation of Alice now fails, because Carol's second message is $\{m,C,n\}_A$, and if Bob sends it through to Alice, she will know that her challenge for Bob was responded to by Carol, and she will not respond to complete the session. The attack has been preempted by adding Bob's identity in his challenge and the protocol now looks secure.

But if it took us so long to detect and eliminate this attack, what reasons do we have to  believe that there are no other attacks? How can we know whether a  protocol is secure? What exactly does it mean for a protocol to be secure? 

\subsection{Protocol modeling and analysis} 
No matter how communicative we are, each of us observes the world from a different standpoint and their standpoint is the center of their world. Understanding the different worldviews from different standpoints at different network nodes is hard. Reasoning about the mazes and knots of network links spanned by communication channels is even harder. Each of us participates in many protocols involving many timelines, looping and intertwining; yet we view protocol runs as linear time intervals, progressing from a beginning to an end. I send a message, then you receive the message, then you send a message and I receive it. In reality, only a half of these events is directly observable for either of us. We imagine the other half. Distributed processes in general, and communications in particular, are hard to understand and easy to misunderstand. Protocols are hard to secure. Yet they are everywhere and every aspect of life depends on their security. We need precise models, definitions, security proofs, and empiric testing to get better models.  We need Security Science (SecSci). 

\para{Towards the definition.}
An authentication protocol has at least two roles: someone authenticates someone. We call the authenticator Alice and the subject she authenticates is Bob. Since the authentication requires that Alice and Bob communicate, the protocol requires channels both ways between them. The protocol interactions through the channels should test the claim (hypothesis) that the subject that Alice is communicating with is Bob. The basic test of such claims is the \emph{challenge-response protocol schema}. 

\subsection{When is an authentication protocol secure?}

\begin{definition}\label{Def:cr}
A \emph{challenge-response protocol}\/ between Alice and Bob consists of
\begin{itemize}
\item communication channels $\AAA^+\to\WP\BBB$ and $\BBB^+\to\WP\AAA$
\item families of 
\begin{itemize}
\item challenges $c_{AB}^{(m)}\in \AAA^+$ and
\item responses $r_{AB}^{(m)}\in \BBB^+$
\end{itemize}
indexed by numbers $m\in \NNn$.  
\end{itemize}
This \emph{authentication requirement}\/ from a challenge-response protocol is the implication
\bea\label{eq:cr-auth}
\mathrm{(A)} & \implies & \mathrm{(B)}
\eea
where (A) and (B) are the following sequencers of events:
\begin{enumerate}[(A)]
\item Alice sends a challenge $c_{AB}^{(m)}$ and subsequently receives a response $r_{AB}^{(m)}$, 
\item Bob receives $c_{AB}^{(m)}$ after Alice sends it and sends $r_{AB}^{(m)}$ before she receives it. 
\end{enumerate}
\end{definition}

\para{Explanation.} The intuitive intent of the authentication requirement $\mathrm{(A)} \implies \mathrm{(B)}$  is that
\begin{enumerate}[(A)]
\item whenever it seems to Alice that her challenge to Bob was responded to, 
\item then Bob really responded to Alice's challenge. 
\end{enumerate}
In other words, the requirement is that, upon the completion of a protocol run, Alice's \emph{\textbf{state of mind}}\/ (A) should coincide with the \emph{\textbf{state of the world}}\/ (B). The protocol test is thus required to convey the truth. Note that this is a \emph{\textbf{requirement on the challenge function $c_{AB}$ and the response function $r_{AB}$}}. In data-based authentication, their cryptographic properties assure that the requirement is satisfied. In thing-based and trait-based authentications, security is assured by the tamper-resistance of  authentication tokens and features.  

\para{Examples.} One example are the challenge-response functions \eqref{eq:NSPK-A}. Another one is \eqref{eq:NSPK-B}. But this is a protocol where Bob authenticates Alice; i.e., the roles in Def.~\ref{Def:cr} are swapped. It is the simplest and the weakest form of challenge-response authentication, usually called the \emph{ping}\/ authentication.\sindex{authentication!ping} If $n$ is freshly generated by Bob, and no one can guess it, then he can only be certain that Alice must have been active between the time when he sent $c_{BA}^{(n)}=\{n\}_{A}$ and the time when he received $r_{BA}^{(n)}=\{n\}_{B}$. This he knows because only Alice could decrypt $n$ from $\{n\}_{A}$. He cannot be sure, though,  that Alice was the one who encrypted $n$ by his public key, since anyone could have done that if Alice sent $n$ unencrypted; or Carol could have encrypted it if Alice sent it encrypted by Carol's public key. So all Bob knows from authenticating Alice by \eqref{eq:NSPK-B} is that she has been alive between his send-challenge and receive-response actions. If he didn't freshly generate an unpredictable $n$, even that is not certain, since anyone could have replayed a known $n$ back to him. On the other hand, if Bob authenticates Alice by the challenge-response pair \eqref{eq:NSPK-B-fix}, then Alice knows that the challenge is from Bob, and if she is honest (meaning, if she follows the protocol) then she will respond by $r_{BA}^{(n)} = \{n\}_{B}$. Then Bob and Alice know that they are participating the protocol, and Alice agrees to be authenticated by Bob. This is a stronger form of authentication, often called the \emph{agreement}.\sindex{authentication!agreement} Last but not least, composing \eqref{eq:NSPK-A} and \eqref{eq:NSPK-B-fix} by gluing $r_{AB}^{(m)} = \{m\}_{A}$ and $c_{BA}^{(n)} = \{B,n\}_{A}$ into a single message $\{m,B,n\}_{A}$ we get the NSPK like in \cref{Fig:proto-NSPK}, but fixed to be secure even if Bob is not honest \cite{Needham-Schroeder-Lowe}. To be sure that Bob and Alice agree not only to be authenticated by each other, but that they agree to a \emph{mutual}\/ authentication, in a single session, it must be required that both of their states of mind coincide with the state of the world, and thus  \emph{\textbf{match}}.\sindex{authentication!by matching conversation} This is the strongest form of authentication, called the \emph{matching conversation}\/ authentication in \cite{STS}. The general idea that the notions of authentication form a hierarchy was discussed in \cite{LoweG:Hierarchy}.

\para{Problem.}  
The authenticity requirement is not a trace property in the sense of Def.~\ref{Def:property}. The reason is that Alice cannot see the global state of the world, but only her local worldview. More precisely, she cannot see the global histories but only her local projections. Taking into account the possible non-local events that she does not see, her worldviews are not histories, but sets of possible histories. For example, when Alice observes the events $c, r\in \Event_{A}$, then her worldview is the set of histories 
\bear
\overline{c\prec r} & = & \left\{\vec x\cons c\cons \vec y\cons r \cons \vec z\ |\ \vec x, \vec y, \vec z \in \Event_{\neg A}^{\ast}\right\}
\eear 
where $\Event_{\neg A} = \Event\setminus \Event_{A}$. The authenticity requirement is thus 
\begin{itemize}
\item not a \textbf{property}  $P\in \WP\Event^{\ast}$ 
\item but a \textbf{\emph{hyper}\/property} $\PPP\in \WP\WP\Event^{\ast}$.
\end{itemize}

\section{Authenticity as a hyperproperty}

\subsection{Hyperproperties}
A hyperproperty is a property of properties. While properties (defined in Sec.~\ref{Sec:histories}) are \textbf{\emph{sets}\/ of histories} $P\in \WP\Event^{\ast}$, hyperproperties are \textbf{\emph{sets of sets}\/ of histories} $\PPP\in \WP\WP\Event^{\ast}$. Their role in channel and protocol security was recognized and studied by Clarkson and Schneider \cite{Clarkson-Schneider:hyperproperties}.

\para{Examples.} We have seen many hyperproperties in Chapters~\ref{Chap:Process} and \ref{Chap:Geometry}. While the elevator safety requirement $\SafE\in\WP \Event^{\ast}$ defined in \eqref{eq:SafE} is a $\Event$-trace property, the set $\SAF\in \WP\WP\Event^{\ast}$ of all $\Event$-trace safety properties is a hyperproperty. Safety is a property of properties. The particular elevator requirements studied in Ch.~\ref{Chap:Process} are the properties $\LivE,\AuthoE,\AvailE\in\WP\Event^{\ast}$, but the kinds of requirements that we studied, i.e., liveness, authority, and availability, are properties of properties: $\LIV, \AU, \AV \in \WP\WP\Event^{\ast}$. 

\para{Programs} (executed on computers) and processes (executed on state machines) \textbf{generate histories}, recorded as $\Event$-traces or contexts. Their security was studied in Chapters~\ref{Chap:Process}--\ref{Chap:Geometry} in terms of \textbf{trace \emph{properties}}.

\para{Protocols} (executed on networks) \textbf{generate worldviews}, recorded \emph{\textbf{sets of}}\/ $\Event$-traces or contexts. Their security will here be formalized in terms of \textbf{trace \emph{hyperproperties}}. 

\subsection{Formalizing authenticity}\label{Sec:hyper-auth}
\para{Protocol events.} To formalize security hyperproperties and align them with security properties from Ch.~\ref{Chap:Process}, we use the event space similar to Example 2 in Ch.~\ref{Chap:Process}.
Given the types
\begin{itemize}
\item $\Obj  \ =\  \big\{c^{(m)}, r^{(m)}, \ldots \big\}$ of protocol messages as objects;
\item $\Act \ =\  \big\{\send{-}, \recv{-} \big\}$ of ``sends'' and ``receives'' as protocol actions; 
\item $\Subj \ =\ \big\{A, B,C,\ldots\big\}$ as a type of protocol roles as subjects,
\end{itemize}
a protocol event is a set of triples representing subjects' actions of sending and receiving objects:
\bear
\Event & = & \Obj \times \Act \times \Subj\ \ = \ \ \Bigg\{\send{p^{(m)}}_A, \recv{p^{(m)}}_A \  \big|\ p \in \Obj, A \in \Subj, m\in \NNn \Bigg\}
\eear
where an event $\send{p^{(m)}}_A$ is that ``Alice sends the message $p^{(m)}\in \Obj$'' whereas $\recv{p^{(m)}}_{A}$ is ``Alice receives the message $p^{(m)}\in \Obj$''.  
The objects acted upon are thus the messages sent or received. In practice, messages are usually the well-formed expressions of a language, suitably encoded. We model them as terms of a sequent algebra\footnote{Sequent algebra has been studied from different angles in the work of Axel Thue, Andrey Markov, Emil Post, Gerhard Gentzen since the early XX century. Noam Chomsky developed them as formal grammars; computer scientists as rewrite systems. See~\cite{PavlovicD:spider} for references and a general overview close to the current use.}.

\para{Term operations.} While the challenge-response protocols are modeled using the message templates presented as \emph{families}\/ of terms $c^{(m)}_{AB}$ and $r^{(n)}_{AB}$, pa\-ra\-me\-trized by indeterminates $m$ and $n$, like we did in  Sec.~\ref{Sec:NSPK}, in reality a message can only be sent if all indeterminates in its template are determined\footnote{Careless applicants often forget to complete the indeterminate parts of their cover letter templates, and submit letters beginning with \emph{``Dear \underline{\hspace{3em}}''}, leaving the space for recipient's name indeterminate. The email message format and the network packet header also have such indeterminate spaces for the recipient. They cannot be sent empty, since they are used as addresses.}. The message space $\Obj$ is therefore the set of \emph{closed}\/ terms $\JJJ[]$ of a term algebra $\JJJ$. The full term algebra $\JJJ$ contains variables $x,y,z\ldots \in \VVV$, and $\JJJ[x,y]$ denotes the subalgebra of terms with the free variables $x,y$. Any pair of values $m,n$ induces the substitution operations $m/x$ and $n/y$, which commute and can be performed concurrently:
\[\begin{tikzar}{}
\& \JJJ[y]\ar{dr}[description]{n/y}\\
\JJJ[x,y] \ar{ur}[description]{m/x} \ar{dr}[description]{n/y} \ar{rr}[description]{m/x, n/y} \&\& \JJJ[]\\
\& \JJJ[x] \ar{ur}[description]{m/x}
\end{tikzar}\]
This means that the variables are independent on each other. The message space $\Obj = \JJJ[]$ is thus the subalgebra consisting of the terms where all free variables have been instantiated to values. 

Where do the values come from? In the  challenge-response protocols, there are two main operations that generate and substitute  values for variables:
\begin{itemize}
\item \emph{new value generation}\/ $\nu m$, 
which results in the substitution $m/x$ of a fresh value $m$ for the variable $x$, i.e., 
\bea
\nu x. p(x) & \vdash & p(m/x) \ \ \mbox{ for a fresh } m
\eea
\item \emph{pattern-matching}\/ $p(m)/p(x)$, which results in the substitution $m/x$ of the value $m$ for the variable $x$ whenever $p(m)$ is matched with $p(x)$, i.e.
\bea
p(m)/p(x) & \vdash & m/x 
\eea
\end{itemize}
For examples of these operations in action, see the NSPK protocol in Sec.~\ref{Sec:NSPK}.

\para{Abbreviations.} The examples that we saw, as well as most other protocol models, suggest the following natural assumptions: 
\begin{itemize}
\item Whenever a message $p(m)$ to be sent contains an indeterminate $m$ not in sender's prior state, then $m$ must be freshly generated before the send action, which is thus $\Big<\nu m. p(m)\Big>$.
\item Whenever upon the receipt of a message $p(m)$, a variable $x$ takes a value $m$ that was not a part of recipient's prior state, then the recipient must have extracted $m/x$ by pattern-matching from the pattern $p(x)$, so that their receive action must have been $\Big(p(m)/p(x)\Big)$.
\end{itemize}
These assumptions allow simplifying  the notations by writing
\begin{itemize}
\item $\send{p^{(m)}}$ to abbreviate $\Big<\nu m. p(m)\Big>$, and
\item $\recv{p^{(m)}}$ to abbreviate $\Big(p(m)/p(x)\Big)$.
\end{itemize}

\begin{definition}\label{Def:CRAU}
The set of authenticity hyperproperties of the challenge-response protocols between Alice and Bob form the set
\begin{multline}\label{eq:CRAU}
\CRAU_{AB} \ \  = \ \  \Bigg\{\, \PPP\in \WP\WP\Event^{\ast}\ |\ \forall P\in \PPP\ \forall \vec t\in P.\ \\  
\vec t = \vec x\cons \send{c^{(m)}_{AB}}_{A}\cons \vec y \cons \recv{r_{AB}^{(m)}}_{A}\cons \vec z \ \ \implies\  
\vec y = \vec u\cons \recv{c_{AB}^{(m)}}_{B} \cons \vec v \cons \send{r_{AB}^{(m)}}_{B} \cons \vec w \, \Bigg\}
\end{multline}
A challenge-response protocol is said to be secure if Alice's worldviews (the sets of possible states of the world, consistent with her observations) are in 
$\CRAU_{AB}$.
\end{definition}

\para{Explanation.} The security requirement in Def.~\ref{Def:CRAU} is the hyperproperty formalization of the security requirement in Def.~\ref{Def:cr}. More precisely, the elements of $\CRAU$ are just those  hyperproperties that satisfy \eqref{eq:cr-auth}, because
\begin{itemize}
\item $\vec t = \vec x\cons \send{c^{(m)}_{AB}}_{A}\cons \vec y \cons \recv{r_{AB}^{(m)}}_{A}\cons \vec z$ in \eqref{eq:CRAU}
formalizes (A) in \eqref{eq:cr-auth},
\item $\vec y = \vec u\cons \recv{c_{AB}^{(m)}}_{B} \cons \vec v \cons \send{r_{AB}^{(m)}}_{B} \cons \vec w$  in \eqref{eq:CRAU} formalizes (B) in \eqref{eq:cr-auth}.
\end{itemize}
The implication $\mathrm{(A)} \implies \mathrm{(B)}$ thus means that whenever Alice observes
\[
\vec x\cons \send{c^{(m)}_{AB}}_{A}\cons \vec y \cons \recv{r_{AB}^{(m)}}_{A}\cons\vec z,
\]
the actual state of the world is
\[
\vec x\cons \send{c^{(m)}_{AB}}_{A}\cons \overbrace{\vec u\cons \recv{c_{AB}^{(m)}}_{B} \cons \vec v \cons \send{r_{AB}^{(m)}}_{B} \cons \vec w}^{\vec y} \cons \recv{r_{AB}^{(m)}}_{A}\cons\vec z.
\]

\section{Network computations are protocol runs} 

\begin{figure}[!h]
\begin{center}

\includegraphics[width=0.8\linewidth
]{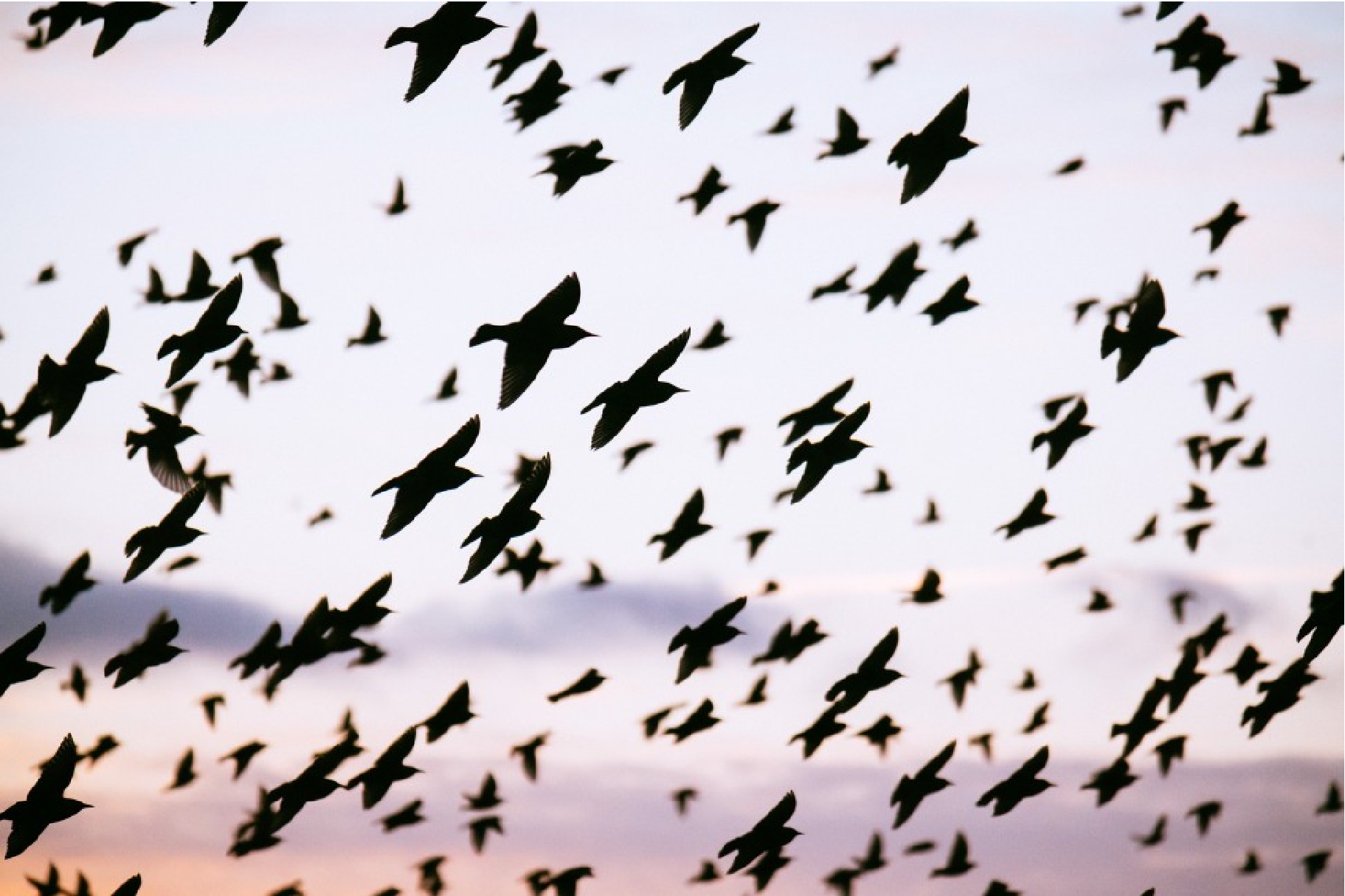}

\vspace{3ex}
\includegraphics[width=0.8\linewidth
]{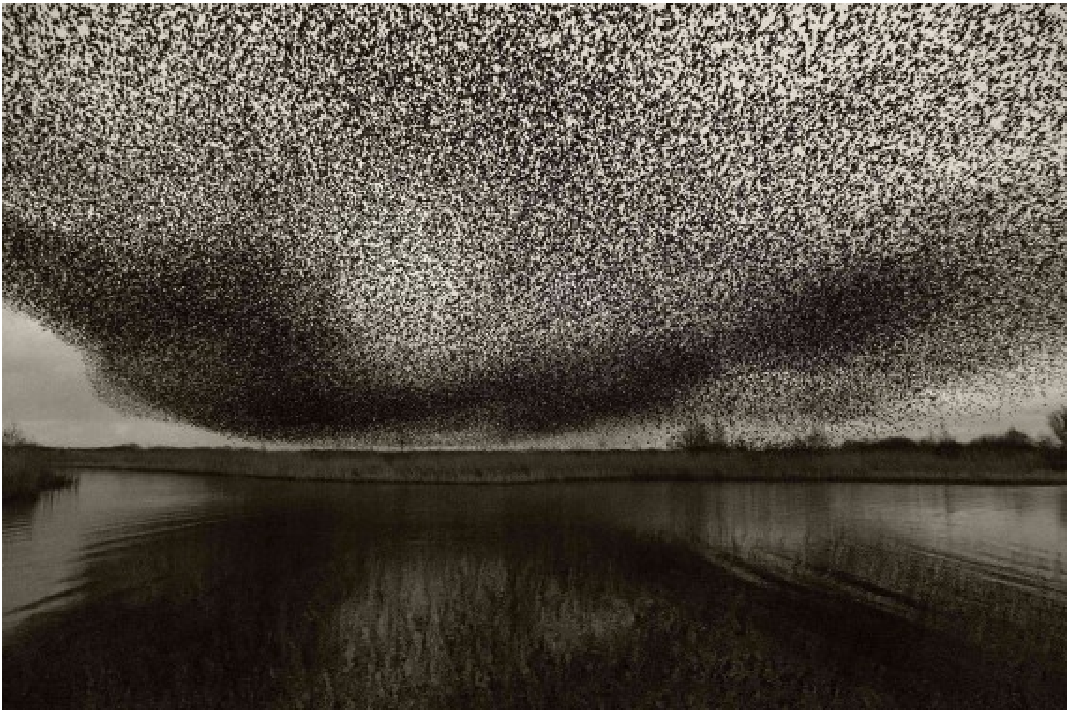}
\caption{Starling murmurations are protocol runs}
\label{Fig:muration}
\end{center}
\end{figure}
Syntactic processes streamline communications. \sindex{syntax} Many semantically different sentences share the same syntactic pattern. An algorithm is a syntactic pattern underlying a family of computations. A single algorithm can be instantiated to many applications,  implemented by many programs and executed on data from many sources. Different computations instantiate the same algorithm just like different sentences instantiate the same grammatical structure. 

A protocol is an algorithm for \emph{network}\/ computations. Each protocol role is assigned a network node\footnote{Sometimes the same actor plays several roles, which requires controlling several network nodes. Centralized control of multiple localities requires another network, with a hub and the channels to the multiple nodes. This lifts to network computation the kind of tricks played when the same actor plays several roles in a movie or in a theater production.}. Each protocol interaction is supported by a network channel. While a centralized computation is implemented by a program executed on a centralized computer\footnote{As explained in Sec.~\ref{Sec:Intro:Summary}, 
computation is centralized when there is a global observer, who sees all computational steps. In standard models of computation, this just means that all computational steps are executed in one place, e.g., by the writing and/or reading head of a state machine, by the rule set of a grammar, etc.}, a network computation is implemented by a troupe of communicating programs, usually cast over a network of computers,  interacting as prescribed by the protocol.  A \emph{protocol run}\/  consists of the runs of the protocol programs, each at their own node, coordinating with each other through communications prescribed by the protocol.  \sindex{protocol!run}

Just like an algorithm is implemented by programs and instantiated to computations, a distributed algorithm is implemented by protocols and instantiated to protocol runs. An example of a protocol run from the animal world is shown in Fig.~\ref{Fig:muration} showing a flock of starlings. The control of each individual flight path is distributed to the individual starlings, and their communication and interaction protocols establish a protocol run that enables the intricate flight patterns of the flock. 

\def\thechapter{8}
\setchaptertoc
\chapter{Information channels and secrecy}\label{Chap:Info}

\section{Channeling information and uncertainty}
\label{Sec:info-prob}

\subsection{Predictions are channels}

\para{Predictions.} Life persists by resisting environmental changes. The environmental changes can be resisted when they are predictable. The predictions are based on the assumption that  the future is like the past: if heavy rains caused floods last year and landslides a couple of years before, then the heavy rain today may cause a flood or a landslide tomorrow. We remember the past to predict the future from the present. If there were twice as many floods as landslides in the past, then the chance of floods in the future is twice greater than the chance of landslides in the future. Probabilistic channels track such chances. If in the past the word ``you'' occurred four times as frequently after the phrase ``I love'' than the word ``watermelon'', then its probability is assumed to be four times higher in the future. That is why language and other information channels are probabilistic. 

\para{From possibility to probability.} 
Predictions are channels from the input type $\Inp = \mbox{Causes}$ to the output type $\Oup = \mbox{Effects}$. The \emph{possible}\/ future effects of the present causes are captured by the \emph{possibilistic}\/ channels \sindex{channel!possibilistic}
\bea\label{eq:caus-chan}
\mbox{Causes}^{+} & \tto f & \WP(\mbox{Effects}).
\eea 
The output $f_{\vec x}\subseteq \mbox{Effects}$ is the set of possible effects of the causal context $\vec x \in \mbox{Causes}^{+}$ at the channel input. In Sec.~\ref{Sec:relchan}\eqref{eq:homchan}, it was convenient to present the possibilistic channel $f\colon\mbox{Causes}^{+} \to \WP(\mbox{Effects})$ in the form $\hhom - f - \colon \mbox{Causes}^{+} \times \mbox{Effects}\to \{0,1\}$, i.e., as \sindex{matrix!relational} \emph{relational matrix}\/ whose entries are the sequents 
\bea
\hhom{\vec x} f y & = & \begin{cases} 1 & \mbox{ if } y\in f_{\vec x},\\
0 & \mbox{ otherwise.}
\end{cases}
\eea
A more informative view of the \emph{probable}\/ future effects of the present causes is captured by the \emph{probabilistic}\/ channels:\sindex{channel!probabilistic}
\bea
\mbox{Causes}^{+} & \tto{\Probb} & \Delta(\mbox{Effects})
\eea
where $\Delta \Oup = \{\mu\colon \Oup \to [0,1]\ |\ \sum_{y\in \Oup}\mu(y) = 1\}$ is the set of probability distributions over $\Oup$. The probabilistic sequents induced by such probabilistic channels are the conditional  probabilities 
\bea\label{eq:stochentry}
\hhom{\vec x}\Probb y & = & \Prob{\vec x}(y).
\eea
Arranged together, they form \emph{stochastic matrices}\/ $\hhom - \Probb - \colon \mbox{Causes}^{+}\times \mbox{Effects}\to [0,1]$. A matrix $M\colon {\Inp \times \Oup}\to [0,1]$ is called \sindex{stochastic matrix| see{matrix, stochastic}} stochastic when $\sum_{y\in \Oup}M_{xy} = 1$ for all $x\in \Inp$. The entries defined in \eqref{eq:stochentry} induce  the bijective correspondence of probabilistic channels and stochastic matrices \sindex{matrix!stochastic}
\beq\label{eq:homchan-prob}
 \prooftree
\Probb\colon \mbox{Causes}^{+} \to \Delta(\mbox{Effects})
 \Justifies
 \hhom- \Probb - \colon \mbox{Causes}^{+}\times \mbox{Effects} \to [0,1]
 \endprooftree
 \eeq
When no confusion is likely, we omit the subscript $\Probb$ from $\hhom{\vec x}\Probb {y}$ and write the conditional probability of the effect $y$ in the causal  context $\vec x$ simply as $\hom{\vec x}{y}$. 

\para{Uncertainty.} \sindex{uncertainty} A prediction is uncertain when a causal context $\vec x$ allows multiple possible or probable effects $y$. This means that there are events $y_0, y_1,\ldots, y_n$ such that the sequents $\hom{\vec x} {y_{i}}$ are true (i.e., equal 1) for all $i=0,1,\ldots n$ in the possibilistic case, or they are all greater than $0$ in the probabilistic case. Each of them may happen, and it is uncertain which one will actually happen. 

\para{Sampling.} \sindex{sampling} The operation of \emph{sampling}\/ a channel consists of entering an input $\vec x$ and observing which of effects $y_{i}$ output, for $i=0,1,\ldots n$. Repeated sampling may produce different outputs. Recording all possibilities, we define the possibilistic channel. Seeing a possible effect multiple times makes no difference, since the repeated occurrences of the effect are not counted. Once all possibilities have been observed, a possibilistic channel yields no new  information. A probabilistic channel is defined by counting the number of times each of the different effects of the same cause occurs and by then calculating their frequencies. Suppose we sample a channel and collect a multiset $D$ of input-output pairs $<\vec x, y>$\footnote{It is not a set but a multiset because we want to count how many times each input-output pair occurs, and retain the multiple occurrences of $<\vec x, y>$ as different elements of $D$, distinguished by suitable counters or parameters.}. The multiset $D$ is partitioned into the disjoint unions
\[
D \   = \  \coprod_{\vec x\in \Inp^+} \vec xD\qquad\qquad\qquad\qquad \vec x D \ = \ \coprod_{y\in \Oup} \vec xDy
\]
where $\vec xD$ is the multiset input-output pairs where the input component is $\vec x$, whereas $\vec xD y$ is the multiset of the occurrences of the input-output pair $<\vec x, y>$. The probability that a context $\vec x$ will cause an effect $y$ can then be approximated by the frequency with which $\vec x$ was followed by $y$ in $D$
\bea\label{eq:freqxy}
\hom{\vec x}{y} & = & \frac{\#\vec xDy}{\#\vec xD}
\eea
where $\#S$ denotes the number of elements of set $S$. The probability that the channel will output $y\in \Oup$ in any context can then be approximated by sum of frequencies:
\beq\label{eq:freqy}
\homm y \ \  = \ \  \frac{\sum_{\vec x\in\Inp^+} \hom{\vec x}{y}}{\sum_{\substack{\vec x\in\Inp^+\\v\in \Oup}} \hom{\vec x}{v}}\ \ =\ \ \frac{\sum_{\vec x\in\Inp^+} \#\vec xDy}{\#D}.
\eeq

\subsection{Probabilistic channel types and structure}
\subsubsection{Cumulative probabilistic channels} 
To capture the dynamics of information transmission, the cumulative view of a probabilistic channel is defined along the lines of Sec.~\ref{Sec:chanlist}. Correspondence~\eqref{eq:cumul} now lifts to
\beq\label{eq:cumul-delta}
\begin{tikzar}[column sep = large]
\Big\{\Inp^{+}\to \Delta\Oup\Big\} \ar[bend left,tail]{r} \& \Big\{\Inp^{\ast}\to \Delta\Oup^{\ast}\Big\} \ar[bend left, two heads]{l}
\end{tikzar}
\eeq
The cumulative sequents are still derived from the single-output sequents by formula \eqref{eq:cumulseq}
\bear
 \hom{x_0 \ldots  x_n}{y_0 \ldots  y_n}\  &= &  \hom{x_0}{y_{0}}\cdot \hom{x_0 x_{1}}{y_{1}}\cdot \hom{x_0 x_{1}x_{2}}{y_{2}}\cdots  \hom{x_0 \ldots  x_n}{y_n}
\eear
This time, however, the sequents are probabilities and not mere relations, which means that they are evaluated in the interval $[0,1]$, and not merely in the set $\{0,1\}$. The meaning of this probabilistic view of formula \eqref{eq:cumulseq} is discussed in Sec.~\ref{Sec:genchan}. To go in \eqref{eq:cumul-delta} the other way around, from  cumulative sequents to single-output sequents, just project away all accumulated outputs except the last one.

\subsubsection{Continuous probabilistic channels}\label{Sec:contchan-prob}\sindex{channel!continuous probabilistic}
\sindex{channel!continuous}
A continuous possibilistic channel, defined in Sec.~\ref{Sec:contchan}, was a function $\gamma\colon \WP\Inp^{\ast}\tto\cup \WP\Oup^{\ast}$, preserving unions and finiteness. A continuous \emph{probabilistic}\/ channel is a function $\varphi\colon\Delta\Inp^{\ast}\tto\Sigma \Delta\Oup^{\ast}$, preserving the convex combinations and the finitely supported distributions. The continuation and the restriction
\beq\label{eq:contchan-delta}
\begin{tikzar}[column sep = large]
\Big\{\Inp^{\ast}\to \Delta\Oup^{\ast}\Big\} \ar[bend left]{r}{\overline{(-)}}\ar[phantom]{r}[description]{\cong} \& \Big\{\Delta\Inp^{\ast}\tto{\Sigma} \Delta\Oup^{\ast}\Big\}  \ar[bend left]{l}{\underline{(-)}}
\end{tikzar}
\eeq
form a bijection again, defined
\bear
\overline{\Probb}_{\mu}(\vec y)\ = \ \sum_{\vec x\in \Inp^\ast} \mu_{\vec x} \Prob{\vec x}(\vec y) &\quad\mbox{ and }\quad & \underline{\Probb}_{\vec x}(\vec y) \ =\ \Prob{\widehat{\vec x}}\left(\vec y\right)
\eear
for the inputs $\mu\in \Delta\Inp^\ast$ and  $\vec x\in \Inp^\ast$, with $\widehat{\vec x}\in \Delta\Inp^\ast$ denoting the point distribution, putting all weight on the point $\vec x$:
\bear
\widehat{\vec x}_{\vec{u}} & = &\begin{cases}
1 & \mbox{ if } \vec x = \vec{u}\\
0 & \mbox{ otherwise.}
\end{cases}
\eear
The stochastic matrix $\hom - - \colon \Inp^\ast \times \Oup^\ast \to [0,1]$ corresponding to the cumulative probabilistic channel $\Probb\colon \Inp^\ast\to \Delta \Oup^\ast$ now extends to the matrix $\hom - - \colon \Delta\Inp^\ast \times \Delta\Oup^\ast \to [0,1]$ 
of sequents
\bear
\hhom{\mu}\varphi {\nu} & = & \sum_{\substack{\vec x\in \Inp^\ast\\
\vec y\in \Oup^\ast}} \mu_{\vec x} \hom{\vec x}{\vec y}\nu_{\vec y}
\eear
induced by the continuous channel 
$\overline{\Probb}\colon \Delta\Inp^\ast\tto{\Sigma} \Delta \Oup^\ast$.

\subsubsection{Example: flipping a coin} 
If a coin is viewed as a probabilistic channel, then sampling the channel means flipping the coin. The channel inputs are the coin flips. Assuming that the coin cannot be manipulated, there is just one way to flip it, which means that the input type $\Inp = \mbox{Causes}$ has a single element, denoted $\ccoin$. A context $\vec x \in \mbox{Causes}^+$ is therefore just a number of coin flips $\vec x = \seq{\ccoin \ccoin \cdots \ccoin}$. Flipping the coin $D$ times produces a sample $\vec y = \seq{y_{1} y_{2}\ldots y_{D}}$, where each $y_{i} \in\{H,T\}$ says whether Heads or Tails came up in the $i$-th flip. Viewed as a possibilistic channel $f\colon \{\ccoin\}^{+}\to \WP\{H,T\}$, the coin supports all cumulative sequents with equal number of inputs and outputs:
\[\hhom{\overbrace{\ccoin\ccoin\cdots\ccoin}^{D}}f{y_{1} y_{2}\ldots y_{D}} = 1
\]
If the coin is biased, one side comes up less often than the other; yet as long as both are possible, all we know is that all possibilistic sequents are true. 

To determine whether the coin is biased (and unfair) or unbiased (and fair), the coin must be viewed as a probabilistic channel $\Probb \colon \{\ccoin_q\}^{+}\to \Delta\{H,T\}$, where $q$ now denotes the bias. It is assumed that flipping does not change the coin and that all flips obey the same probability distribution:
\[
\hhom{\ccoin_q}\Probb{H} \ \ =\ \ q\qquad\qquad  \hhom{\ccoin_q}\Probb{T}\ \ =\ \ 1-q
\] 
The coin is said to be \emph{unbiased}\/ if $\hhom{\ccoin_q}\Probb{H}= \hhom{\ccoin_q}\Probb{T}$, which means that $q=\frac 1 2$. Otherwise, the coin is said to be \emph{biased}. The assumption that flipping the coin does not change it also means that the past outcomes do not impact the present and the future outcomes, which implies
\bear
\hom{\overbrace{\ccoin_q\ccoin_q\cdots\ccoin_q}^{D}}{y_{1} y_{2}\ldots y_{D}} & = & \hom{\ccoin_q}{y_{1}}\cdot \hom{\ccoin_q}{y_{2}}\cdots \hom{\ccoin_q}{y_{D}} \ \ =\ \ q^{h}\cdot(1-q)^{D-h}
\eear
where $h$ is the number of Heads among the outcomes $y_{1}, y_{2}\ldots y_{D}$.  If the coin is unbiased, i.e., $q= 1-q = \frac 1 2$, then 
\bea\label{eq:unbiased}
\hom{\overbrace{\ccoin_{\frac 1 2}\ccoin_{\frac 1 2}\cdots\ccoin_{\frac 1 2}}^{D}}{\vec y} & = & \left(\frac 1 {2}\right)^{D}
\eea
for all $\vec y \in \{H,T\}^{D}$. On the other hand, if the coin is biased, say $q=\frac 3 4$ and $1-q = \frac 1 4$, then 
\bea\label{eq:biased}
\hom{\overbrace{\ccoin_{\frac 3 4}\ccoin_{\frac 3 4}\cdots\ccoin_{\frac 3 4}}^{D}}{\vec y} & = & \frac {3^{h}} {4^{D}}
\eea
where $h$ is again the number of Heads in $\vec y$. In particular
\[
\hom{\overbrace{\ccoin_{\frac 3 4}\ccoin_{\frac 3 4}\cdots\ccoin_{\frac 3 4}}^{D}}{\overbrace{HH\cdots H}^{D}} \ = \  \left(\frac {3} {4}\right)^{D}\qquad\qquad \hom{\overbrace{\ccoin_{\frac 3 4}\ccoin_{\frac 3 4}\cdots\ccoin_{\frac 3 4}}^{D}}{\overbrace{TT\cdots T}^{D}} \ =\  \left(\frac {1} {4}\right)^{D}.
\]
The bias of the coin can be established with arbitrary confidence if sufficiently large sample sets $D$ are available.  

\subsection{Quantifying information}
Before the coin is flipped, it is uncertain whether the Heads or the Tails will come up. After the coin is flipped and the outcome is observed, the uncertainty is eliminated. Information is the flip side of uncertainty: the more information, the less uncertainty. The other way around, \emph{probabilistic channels transmit information by decreasing the uncertainty}. That is why probabilistic channels are the \emph{information channels}. \sindex{channel!information} They were the starting point of Shannon's \emph{theory of information}\/ \cite{AshR:IT,Cover-Thomas:IT,ShannonC:communication}. 

\para{Increase of information is decrease of uncertainty.}\sindex{entropy} Sampling a channel decreases the uncertainty and increases the information. The amount of information transmitted through a channel can be measured by establishing how much uncertainty has been eliminated. Flipping a biased coin eliminates less uncertainty than flipping an unbiased coin. If the biased coin with
\[
\hom{\ccoin_{\frac 3 4}}{H} \ \ =\ \ \frac 3 4\qquad\qquad \hom{\ccoin_{\frac 3 4}}{T}\ \ =\ \ \frac 1 4
\] 
is flipped in a game, betting on Heads is a winning strategy in the long run. If an unbiased coin is flipped, there is no winning strategy and the outcome of the game is completely uncertain. Therefore, flipping a biased coin eliminates less uncertainty and furnishes less information than flipping an unbiased coin. Flipping a totally biased coin, with indistinguishable Heads on both sides, furnishes no information at all.

\para{Entropy.} The uncertainty of a channel, and the information that it transmits, can be measured as the average length of the bitstrings needed to describe the probabilities of its outputs. For example, for the unbiased coin, the probabilities are
\[\hom{\ccoin_{\frac 1 2}}{H}\  =\ \hom{\ccoin_{\frac 1 2}}{T}\ =\  \frac 1 2\ =\ (.1)_2\]
where $(.1)_2$ is written in binary. So writing each of the probabilities requires precisely one binary digit, and the lengths are 
\[\ell\hom{\ccoin_{\frac 1 2}}{H}\  =\  \ell\hom{\ccoin_{\frac 1 2}}{H} = 1\]
Their average length is calculated by weighing the length of each probability by the probability itself, and the amount of information obtained (i.e., the amount of uncertainty eliminated) by flipping the unbiased coin is
\bear
H(\ccoin_{\frac 1 2}) & = & \hom{\ccoin_{\frac 1 2}}{H}\cdot \ell\hom{\ccoin_{\frac 1 2}}{H} + \hom{\ccoin_{\frac 1 2}}{T} \cdot \ell\hom{\ccoin_{\frac 1 2}}{T} \ \ =\ \ \frac 1 2\cdot 1 + \frac 1 2\cdot 1 = 1
\eear
In words, this says that flipping a coin gives 1 bit of information. For 3 flips, \eqref{eq:unbiased} gives the probability of each outcome $\vec y \in \{H,T\}^3$ is
\[
\hom{\ccoin_{\frac 1 2}\ccoin_{\frac 1 2}\ccoin_{\frac 1 2}}{\vec y} \ = \ \frac 1 8\ =\ (.001)_2
\]
Writing the probability $\frac 1 8$ now requires 3 binary digits. The lengths of the probabilities that need to be averaged are this time 
\bear
\ell\hom{\ccoin_{\frac 1 2}\ccoin_{\frac 1 2}\ccoin_{\frac 1 2}}{\vec y} & = & 3
\eear
Averaging is easy again, since the lengths of the probabilities are weighed by the probabilities, which are the same again, just smaller: $\frac 1 8$ rather than $\frac 1 2$. The amount of information obtained (i.e., the amount of uncertainty eliminated) by flipping the unbiased coin 3 times is
\bear
H(\ccoin_{\frac 1 2}\ccoin_{\frac 1 2}\ccoin_{\frac 1 2}) & = & \sum_{\vec y \in \{H,T\}^3} \hom{\ccoin_{\frac 1 2}\ccoin_{\frac 1 2}\ccoin_{\frac 1 2}}{\vec y}\cdot \ell\hom{\ccoin_{\frac 1 2}\ccoin_{\frac 1 2}\ccoin_{\frac 1 2}}{\vec y}  \ \ =\ \ 8 \cdot \frac 1 8\cdot 3 = 3
\eear
Flipping a coin 3 times gives 3 bits of information. It sounds trivial all right, but the underlying general idea is not. The idea is that the amount of information (i.e., the decrease of uncertainty) about the causes $\vec x$  conveyed by the effects $\vec y$ transmitted by a channel $\hom {\vec x}{\vec y}$ is
\bear
H(\vec x) & = & \sum_{\vec y} \hom{\vec x}{\vec y}\ell \hom{\vec x}{\vec y}
\eear
But counting the expected lengths $\ell \hom{\vec x}{\vec y}$ of the binary representations of the probability of every output $\vec y$ is impractical. To simplify it, first note that the length of the binary notation for an integer $r\gt 1$ is the length of the smallest power of 2 above it, i.e., $\ell(r) = \lceil \log_2 r\rceil$. If we go beyond the integers and also allow fractional lengths, we can ignore the ceiling operation $\lceil - \rceil$, and write $\ell(r) = \log_2 r$. Since the logarithms of numbers $p\in (0,1]$ are negative, the fractional lengths become $\ell(p) = -\log_2 p$, the general formula for quantifying the information about $\vec x$ transmitted by the channel outputs $\vec y$ is thus
\bea\label{eq:entropy}
H(\vec x) & = & - \sum_{\vec y} \hom{\vec x}{\vec y}\log_2 \hom{\vec x}{\vec y} 
\eea
This is the most used and studied information measure:  Shannon's \emph{entropy} \cite{ShannonC:communication}. Applied to flipping the biased coin
\[\hom{\ccoin_{\frac 3 4}}{H}\  =\ \frac 3 4\ =\ (.11)_2
\qquad\qquad\qquad \hom{\ccoin_{\frac 3 4}}{T}\  =\ \frac 1 4\ =\ (.01)_2\]
formula \eqref{eq:entropy} tells how many bits of information it furnishes:
\begin{multline*} H(\ccoin_{\frac 3 4}) = -\Big( \hom{\ccoin_{\frac 3 4}}{H}\cdot \log\hom{\ccoin_{\frac 3 4}}{H} + \hom{\ccoin_{\frac 3 4}}{T} \cdot \log \hom{\ccoin_{\frac 3 4}}{T}\Big) \ \ =\\ - \Big( \frac 3 4\cdot \log \frac 3 4 + \frac 1 4\cdot \log \frac 1 4 \Big)\  =\   2 - \frac 3 4 \log_2 3\end{multline*}
Since $\log_2 3\gt \frac 3 2$, flipping this biased coin yields less than 1 bit of information. Flipping  $\ccoin_{\frac 1 2}$ gives $\frac 3 4 \log_2 3 - 1$ more information than flipping $\ccoin_{\frac 3 4}$. Information theory is a symphony of such measurements of information flows. The theories of communication, language, intelligence, data security, all depend on such measurements. As a quantitative model of information transmission, probabilistic channels play a central role in all sciences of  communication \cite{PavlovicD:LangEng,ShannonC:communication}. A rich theory of information flow security is built almost entirely in terms of channels as stochastic matrices \cite{SmithG:book}. In cryptography, simple memoryless channels are composed into complex secure constructs used in secure function evaluation and multi-party computation \cite{EvansD:MPC}.

\para{Examples: erasure channel, oblivious transfer.} A memoryless channel that transmits an input bit $b$ with probability $q$, and erases it otherwise, can be defined by the stochastic matrix $\hom - - \colon \{0,1\} \times \{0,1,e\}\to [0,1]$ comprised of the sequents
\[ \hom b b = q\qquad\qquad\qquad\qquad \hom b e = 1-q\]
This  \emph{erasure}\/ channel is displayed in Fig.~\ref{Fig:erase-OT} on the left. \begin{figure}[!h]
\begin{center}
\includegraphics[height=3cm
]{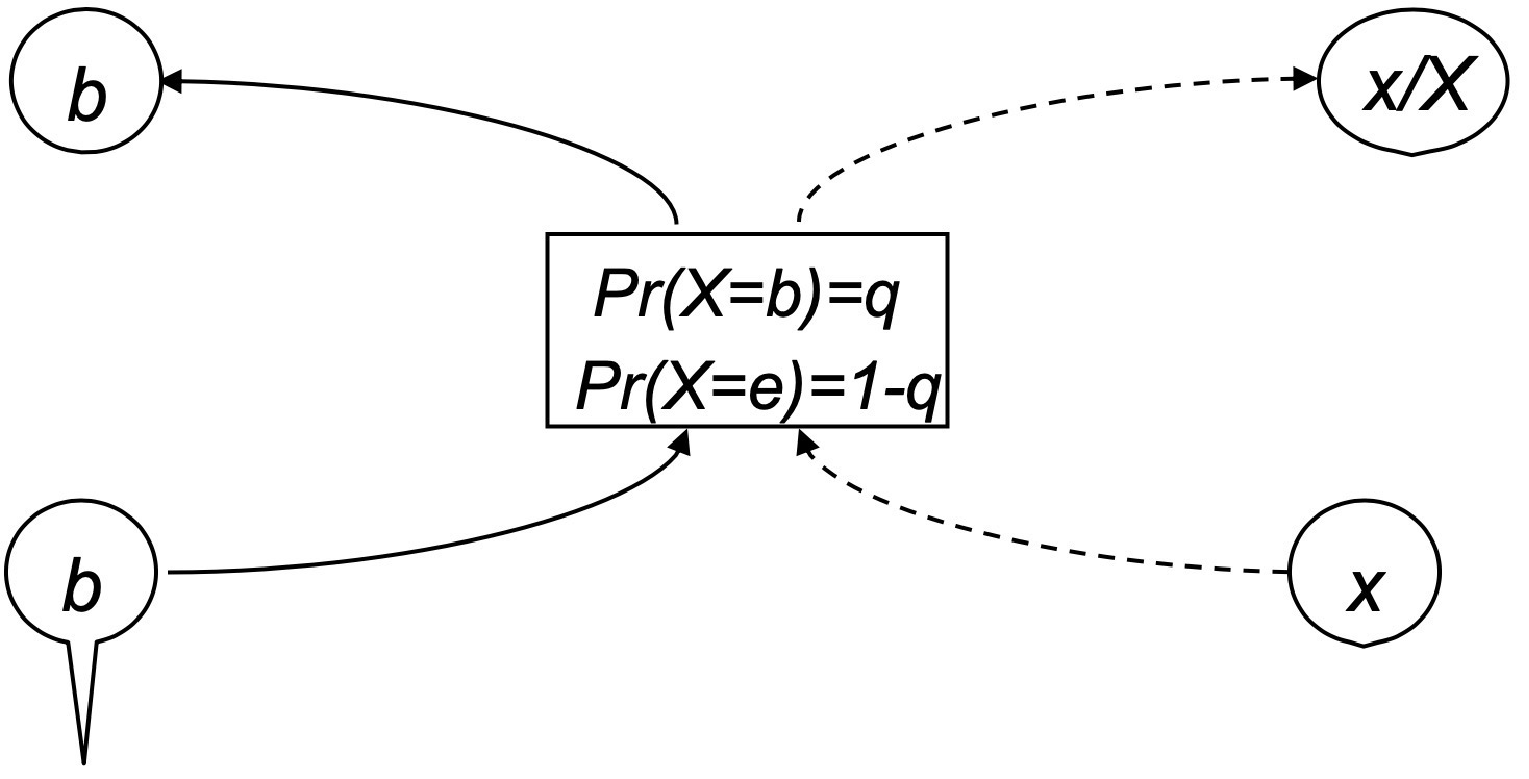}
\hspace{3em}
\includegraphics[height=3cm
]{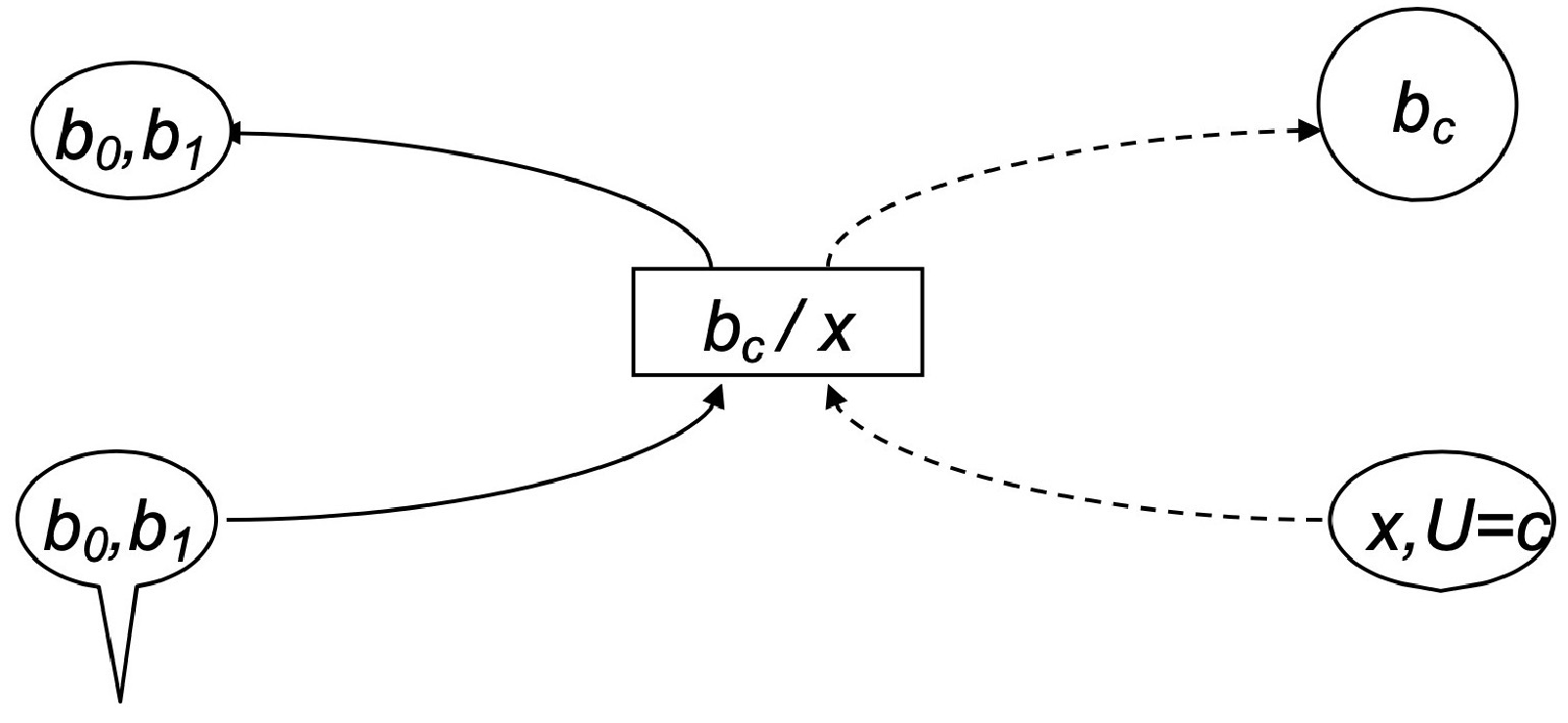}
\caption{The Erasure Channel and the Oblivious Transfer}
\label{Fig:erase-OT}
\end{center}
\end{figure}
On the right is the \emph{oblivious transfer}\/ channel, which transmits one of two bits from Alice to Bob, keeping Alice oblivious which of the two bits was transmitted. The stochastic matrix $\hom - - \colon \{0,1\}^2 \times \{0,1\}\to [0,1]$ and the sequents are
\[ \hom {b_0\, b_1}{b_0} = q\qquad\qquad\qquad\qquad \hom {b_0\, b_1}{b_1}= 1-q\]

\section{Predictive inference}
\label{Sec:info-lang}

\subsection{Example: language models}
As I speak, the context $\vec x$ of words that I have said induces a probability distribution over the next word $y$ that I will say. These \emph{conditional probabilities}\/ \sindex{probability!conditional} are the values of the sequents $\hom {\vec x} y$ describing the channel that I communicate on. We speak (and write) by sampling that distribution. The conditional probability that I sample as I speak may be something like 
\bear
\pderr{what is the next word that I am going to}{write} & = & {\textstyle \frac 1 8}\\
\pderr{what is the next word that I am going to}{say} & = & {\textstyle \frac 1 2}\\
\pderr{what is the next word that I am going to}{swallow} & = & {\textstyle \frac 1 {30}}
\eear
In the usual notation for conditional probabilities, the first line would be written 
\bear \Pr\left({\mbox{write}\, |\, \mbox{what is the next word that I am going to}}\right) & = & \textstyle \frac 1 8
\eear 
and the other two similarly. Here we don't write them like that, but as sequents. Changing standard notations is seldom a good idea, but they are seldom this bad.

\para{Speech is a probabilistic channel.} As you try to understand what I am trying to say, you also sample from a similar distribution and try to predict what I will say. If what I say is predictable, it is easier to understand. If I mumble, predictability enables error correction. When I deviate from your predictions, you may receive new information. The uncertainty increases with the unpredictability, and removing it yields more information. 

We speak the same language if we sample words from the same distribution. That distribution produces the language  that we speak and comprehend. In reality, everyone's distribution is slightly different, and we never speak a completely identical language. But if the distributions are close enough, we can communicate. That is what makes the probabilistic channels that generate streams of words into communication channels. They are the languages that we speak and write. But how do we come to share these channels? How do the probability distributions over the words settle in the minds of all people who use a language?

As we use a language, we sample it, and retain the word frequencies (\ref{eq:freqxy}--\ref{eq:freqy}). The conditional probability of the next word in a given context is the ratio in the form
\bear
\pderr{\small what is the next word that I am going to}{say} & = & \frac{\ppderr{\small what is the next word that I am going to say}}{\ppderr{\small what is the next word that I am going to}}
\eear
The numerator is the total frequency of the whole phrase. The denominator is the total frequency of the prefix, without the last word. The total frequencies are summed up over all contexts, as in \eqref{eq:freqy}|.

\para{Language generation.} The process of language generation proceeds by sampling the next word, adding it to the context, then sampling the next word, adding it to the context, and so on. This is what the chatbots do when they speak to us \cite{PavlovicD:LangEng}. Intuitively, the probability of a phrase that I (or a chatbot) might say should be the product of the probabilities of each of the words from its preceding context:
\bear
\ppderr{what is} & = & \ppderr{what}\cdot \pderr{what}{is}\\
\ppderr{what is the} & = &  \ppderr{what}\cdot \pderr{what}{is}\cdot \pderr{what is}{the}\\
\ppderr{what is the next} & = &  \ppderr{what}\cdot \pderr{what}{is}\cdot \pderr{what is}{the} \cdot \pderr{what is the}{next}
\eear
and so on. Justifying this intuition leads to a general law for reasoning about probabilistic channels.

\subsection{Generative channels and sources}
\label{Sec:genchan}\sindex{channel!generative} \sindex{source} 
\begin{definition}\label{Def:source}
A channel is called \emph{generative}\/ if it is in the form $\Inp^+\to \Delta \Inp$, i.e., its input and output types are the same.  A generative channel makes its underlying type $\Inp$ into a \emph{source}.
\end{definition}

The stochastic matrix of a generative channel is in the form $\hom - - \colon \Inp^\ast\times \Inp^\ast\to [0,1]$. On one hand, it is just a conditional probability distribution over $\Inp^\ast$, written in a matrix form.  On the other hand, when viewed as a channel, the fact that it emits outputs of the same type as its inputs puts it in the generative mode of operation. Starting from a context $\vec x= \seq{x_{1}x_{2}\ldots x_{m}}$, it generates a stream of outputs
\[
\pder {x_{1} \ldots x_{m}}{x_{m+1}},\ \pder {x_{1} \ldots x_{m}x_{m+1}}{x_{m+2}},\ldots , \pder {x_{1} \ldots x_{m}x_{m+1}\ldots x_{m+n-1}}{x_{m+n}}, \ldots
\] 
by appending each new output to the old context and thus forming a new context, which it consumes as the input to generate the next output. Then it appends that output to the previous input, and so on. At each point of the process, the probability of generating a particular cumulative sequence of outputs is completely determined by the single-output generations, as their product:
\bea
\pder {x_{1} \ldots x_{m}}{x_{m+1}\ldots x_{n}}& = & \pder {x_{1} \ldots x_{m}}{x_{m+1}}\cdot \notag\\
&& \pder {x_{1} \ldots x_{m}x_{m+1}}{x_{m+2}}\cdot\notag\\
&& \pder {x_{1} \ldots x_{m}x_{m+1}x_{m+2}}{x_{m+3}}\cdot\label{eq:chain}\\
&& \hspace{5em}\vdots\notag\\
&&\pder {x_{1} \ldots x_{m}x_{m+1}\ldots x_{m+n-1}}{x_{n}}\notag
\eea
The probability that a language channel generates a particular phrase is thus a product of the single step generations, starting from the empty context:
\bear
&& \ppderr{what is the next word that I am going to say}\ = \notag\\
 && \ppderr{what}\cdot \notag\\
 && \pderr{what}{is}\cdot \notag\\
 && \pderr{what is}{the}\cdot\notag\\
&& \pderr{what is the}{next} \cdot\label{eq:next}\\
&& \hspace{5em}\vdots\notag\\
&& \pderr{what is the next word that I am}{going} \cdot\notag\\
&& \pderr{what is the next word that I am going}{to} \cdot\notag\\
&& \pderr{what is the next word that I am going to}{say}\notag
\eear
Generalizing \eqref{eq:chain} yields the general transitivity law of generative channels.

\begin{proposition}\label{Prop:trans}
Every generative channel $\hom - - \colon \Inp^\ast\times \Inp^\ast\to [0,1]$ satisfies
\bea\label{eq:trans}
\hom{\vec x}{\vec y\vec z} & = & \hom{\vec x}{\vec y}\cdot \hom{\vec x\vec y}{\vec z}
\eea
for all $\vec x, \vec y, \vec z \in \Inp^\ast$.
\end{proposition}

\para{Remark.} Viewing the matrix $\hom - - \colon \Inp^\ast\times \Inp^\ast\to [0,1]$ as a listing of conditional probabilities makes equation \eqref{eq:bayes} into a familiar law, equivalent to the Bayesian law, which follows directly from the definition of the conditional probability given by Thomas Bayes:
\bear
\hom{\vec x}{\vec y} &= & \frac{\homm{\vec x\vec y}}{\homm{\vec x}}
\eear

\section{Inverse channels and Bayesian inference}

How can we derive causes from effects? How do we invert a probabilistic channel? --- Finding the unobservable causes of observable effects is the central problem of science. Finding the unknown inputs of a channel corresponding to its known outputs is the central problem of cryptanalysis and one of the main problems of channel security.

What does it mean that the channels $\Probb\colon \Inp^\ast\to \Delta \Oup^\ast$ and $\widetilde{\Probb}\colon \Oup^\ast\to \Delta\Inp^\ast$ are each other's inverses? --- Intuitively, it means that $<\vec x,\vec y>\in \Inp^\ast\times \Oup^\ast$ occurs as the input-output pair of $\Probb$ with the same probability as its reverse $<\vec y,\vec x>\in \Oup^\ast\times \Inp^\ast$ occurs as the input-output pair of $\widetilde\Probb$. But these probabilities are only determined if the probabilities of $\vec x\in \Inp^\ast$ and $\vec y\in \Oup^\ast$ are determined. If we cannot tell how frequently $\vec x\in \Inp^\ast$ occurs, how could we tell the frequency of $<\vec x,\vec y>\in \Inp^\ast\times \Oup^\ast$. But to say that the probabilities of $\vec x\in \Inp^\ast$ and $\vec y\in \Oup^\ast$ are determined means that $\Inp$ and $\Oup$ are sources, in the sense of Def.~\ref{Def:source}.

Let
\[\hhom - \Inp -\colon \Inp^\ast\times \Inp^\ast \to [0,1]\qquad\qquad\qquad
\hhom - \Oup -  \colon \Oup^\ast\times \Oup^\ast \to [0,1]\]
be the (stochastic matrices corresponding to the) generative channels making $\Inp$ and $\Oup$ into sources. (We usually omit the subscripts since the names suggest the types.) The probabilities that the source $\Inp$ may generate $\vec x$ and that $\Oup$ may generate $\vec y$ are respectively
\beq\label{eq:sourceXY}
\homm{\vec x} \ =\ \frac{\ \sum_{\vec u\in \Inp^\ast} \hom{\vec u}{\vec x}}{\sum_{\vec u,\vec w\in \Inp^\ast} \hom{\vec u}{\vec w}} \qquad\qquad\mbox{and}\qquad\qquad \homm{\vec y} \ =\ \frac{\ \sum_{\vec t\in \Oup^\ast} \hom{\vec t}{\vec y}}{\sum_{\vec t,\vec v\in \Oup^\ast} \hom{\vec t}{\vec v}}
\eeq
Recalling that $\Prob{\vec x}(\vec y)=\hom{\vec x}{\vec y}$ and $\widetilde{\Probb}_{\vec y}(\vec x)=\hom{\vec y}{\vec x}$, 
\begin{itemize}
\item the chance that $<\vec x,\vec y>$ will occur as the input-output pair of $\Probb$ is $\homm{\vec x}\cdot\hom{\vec x}{\vec y}$, whereas
\item the chance that $<\vec y,\vec x>$ will occur as the input-output pair of $\widetilde\Probb$ is $\homm{\vec y}\cdot\hom{\vec y}{\vec x}$.
\end{itemize}
The requirement that $\widetilde \Probb$ is the inverse of $\Probb$ is thus
\bea\label{eq:bayes}
\homm{\vec x}\cdot \hom{\vec x}{\vec y} & = & \homm{\vec y}\cdot \hom{\vec y}{\vec x}
\eea
The inverse of the channel $\Probb\colon \Inp^\ast \to \Delta \Oup^\ast$ can thus be defined
\bea\label{eq:bayes-inverse}
\widetilde\Probb\colon \Oup^\ast & \to & \Delta\Inp^\ast\\
\vec y &\mapsto & \hom{\vec y}{\vec x} 
= \frac{\homm{\vec x}\cdot \hom{\vec x}{\vec y}}{\homm{\vec y}}\notag
\eea

\para{Bayesian reasoning.} The definition of inverse probability goes back to Thomas Bayes' famous essay ``on the doctrine of chances'' \cite{BayesT:essay}, albeit in a convoluted form. Formulas \eqref{eq:bayes} and \eqref{eq:bayes-inverse} are often referred to as the \emph{Bayesian law}. \sindex{Bayesian law} Within a single source, without the concept of channel, it was rediscovered in the subsequent centuries by many others, until it was inaugurated as one of the basic tools of statistical inference by Fisher \cite{Fisher:1930}.
 
\subsubsection*{Example: the Monty Hall problem}\sindex{Monty Hall problem}
Monty Hall was the host of a TV game show \emph{Let's make a deal}. In one of his games, the prize was a car, hidden behind one of three doors. You will win the car if you select the correct door. After you have picked one door but before the door is open, Monty Hall opens one of the other doors, where he reveals a goat, and asks if you would like to switch from your current selection to the remaining door. How will your chances change if you switch?

Let us analyze the situation. Suppose that the doors are enumerated by $0,1,2$. Say the door that you have initially chosen is assigned the number 0.
\begin{itemize}
\item The available channel inputs $\Inp = \{C_0, C_1, C_2\}$ correspond to the situations that the $C$ar is behind the door $0, 1$, or $2$.
\item The available channel outputs $\Oup  =  \{G_0, G_1, G_2\}$ correspond to the situations where Monty Hall shows a $G$oat behind the door $0,1$, or $2$.

\item Let $\hom x y$ denote the probability that the car position input $x\in \{C_0, C_1, C_2\}$ causes Monty Hall's goat output $y\in \{G_0, G_1, G_2\}$. The matrix of the induced channel is given in the first of the following three tables.
\begin{center}
\begin{tabular}{|c||c|c|c|}
\hline
$\hom x -$ & $G_0$ & $G_1$ & $G_2$\\ 
\hline\hline
$\hom{C_0}-$ & $0$ & $\frac 1 2$ & $\frac 1 2$\\
\hline
$\hom{C_1}-$ & $0$ & $0$ & $1$\\
\hline
$\hom{C_2}-$ & $0$ & $1$ & $0$\\
\hline
\end{tabular}
\quad
\begin{tabular}{|c||c|c|c|}
\hline
$\homm{x,y}$ & $G_0$ & $G_1$ & $G_2$\\ 
\hline\hline
$C_0$ & $0$ & $\frac 1 6$ & $\frac 1 6$\\
\hline
$C_1$ & $0$ & $0$ & $\frac 1 3$\\
\hline
$C_2$ & $0$ & $\frac 1 3$ & $0$\\
\hline
\end{tabular}
\quad
\begin{tabular}{|c||c|c|c|}
\hline
$ \hom y -$ & $C_0$ & $C_1$ & $C_2$\\ 
\hline\hline
$\hom{G_0}-$ & $\uparrow$ & $\uparrow$ & $\uparrow$ \\
\hline
$\hom{G_1}-$ & $\frac 1 3$ & $0$ & $\frac 2 3$\\
\hline
$\hom{G_2}-$ & $\frac 1 3$ & $\frac 2 3$ & $0$\\
\hline
\end{tabular}
\end{center}

\begin{itemize}
\item If the car is  behind the door 0, then Monty Hall can open the doors 1 and 2 with equal probabilities $\hom{C_0}{G_1} = \hom{C_0}{G_2}  = \frac 1 2$, and show the goat there. This is the first line of the table.

\item If the car is behind the door 1, then Monty Hall can only show the goat behind the door 2, and thus $\hom{C_1}{G_2} = 1$.

\item If the car is behind the door 2, then Monty Hall can only show the goat behind the door 1, and thus $\hom{C_2}{G_1} = 1$.
\end{itemize}

\item The second table displays the joint probability distribution $\homm{x,y}\in \Distt{\Inp \times\Oup}$, which is computed from the first table using the formula $\homm{x,y} = \homm{x}\cdot\hom x y$, where we assume that the chances that the car is behind each of the door are equal, i.e., $\homm {C_0} = \homm{C_1} = \homm{C_2} = \frac 1 3$.

\item The third table displays the chance $\hom y x$ that Monty Hall's goat output $y\in \{G_0, G_1, G_2\}$ was caused by the car position $x\in \{C_0, C_1, C_2\}$ as the input. The computation of the probability that this might be the cause of the observed state is computed by going back along the inverse channel \eqref{eq:bayes-inverse}. The distribution of the car is $\homm x = \frac 1 3$ for $x\in \{C_0, C_1, C_2\}$, whereas the distribution of the goat $\homm y$ is computed by summing up the columns of the middle table. Hence, $\homm{G_0} = 0$ and $\homm{G_1} = \homm{G_2} = \frac 1 2$. Since $\homm{G_0} = 0$, the output $G_0$ can never be observed, and the chances $\hom{G_0}-$ are undefined, which is denoted by $\uparrow$. The second row is
\bear \hom{G_1}{C_0} &= &\frac{\homm{C_0}\cdot \hom{C_0}{G_1}}{\homm{G_1}}\ \ =\ \ \frac{\frac 1 3\cdot \frac 1 2}{\frac 1 2}\ \ =\ \ \frac 1 3\\
 \hom{G_1}{C_1} &= &\frac{\homm{C_1}\cdot \hom{C_1}{G_1}}{\homm{G_1}}\ \ =\ \ \frac{\frac 1 3\cdot 0}{\frac 1 2}\ \ =\ \ 0\\
 \hom{G_1}{C_2} &= &\frac{\homm{C_2}\cdot \hom{C_2}{G_1}}{\homm{G_1}}\ \ =\ \ \frac{\frac 1 3\cdot 1}{\frac 1 2}\ \ =\ \ \frac 2 3
\eear
The third row is computed analogously.
\end{itemize}
For both outputs $G_1$ and $G_2$, displaying the goat behind the door 1 or 2, the chance that the car is behind the remaining door 2, resp. 1, is twice bigger than the chance that it is behind the door 0 that you had chosen initially. You should switch.

\section{Probabilistic noninterference}

\subsection{Sharing information}
\sindex{channel!shared}
Suppose that a probabilistic channel  
$\hom - - \colon\Inp^\ast\times \Oup^\ast \to [0,1]$ is shared by Alice, Bob, and other subjects of type $\Subj$, like in Sec.~\ref{Sec:shared}. Alice can only enter the inputs from her own clearance type $\Inp_A$ \eqref{eq:privview} and only observes the corresponding outputs from her local channel view $\hom -  -  ^A\colon\Inp^\ast\times \Oup^\ast \to [0,1]$
\beq\label{eq:chanview-prob}
\hom {()}{()}^A= 1\qquad \qquad\quad \hom{\vec x\cons u}{\vec y \cons v}^A = \begin{cases}
\hom{\vec x}{\vec y}^A\cdot \hom{\vec x\cons u}{v}_\ast & \mbox{ if } u\propto A\\
\hom{\vec x}{\vec y}^A & \mbox{ otherwise}
 \end{cases}
 \eeq
where $\hom - -_\ast \colon \Inp^+\times \Oup\to[0,1]$ is the matrix of the single-output version of the channel:
\bea
\hom{\vec u}{v}_\ast & = & \sum_{\vec y\in \Oup^\ast} \hom{\vec u}{\vec y \cons v}
\eea
Definition \eqref{eq:chanview-prob} lifts to probabilistic channels the local view from the possibilistic framework in \eqref{eq:chanview}. 

\subsection{Probabilistic worldviews}  \sindex{worldview}
Alice's state of the world $\state_{\vec x_{A}}\colon \Inp^{\ast}\to  \{0,1\}$ was defined in \eqref{eq:stateA} as the characteristic function of the set $\left\{\vec x\colon \Inp^{\ast}\ |\ \vec x\restr_{A} = \vec x_{A}\right\}$. It is the inverse image of Alice's view $\vec x_{A}$ along the purge projection $\restr_{A}\colon \Inp^\ast\to\Inp^\ast$. The states of the world corresponding to the different projections that Alice may observe were collected into the possibilistic  channel $\state \colon \Inp^{\ast}_{A}\to \WP\Inp^{\ast}$,  conveniently written as the matrix  $\hom - - \colon \Inp^{\ast}_{A}\times \Inp^{\ast}\to \{0,1\}$ of entries $
\hom{\vec x_{A}}{\vec x} =  \state_{\vec x_{A}}(\vec x)$. 
In the stochastic case, assuming that $\Inp$ is a source with the frequency distribution $\Probb = \homm - \colon \Inp^\ast\to [0,1]$ defined as in \eqref{eq:sourceXY}, 
the local state of the world is the probabilistic channel $\state \colon \Inp^{\ast}_{A}\to \Delta\Inp^{\ast}$ corresponding to the stochastic matrix $\hom - - \colon \Inp^{\ast}_{A}\times \Inp^{\ast}\to [0,1]$ with the entries
\bea\label{eq:St-prob}
\hom{\vec x_{A}}{\vec x}\ \ =\ \  {\state_{\vec x_{A}}}(\vec x) & = &\begin{cases}\displaystyle \frac{[\vec x]}{[\vec x_{A}]} & \mbox{ if } \vec x\restr_{A}= \vec x_{A}\\[1ex]
0 & \mbox{ otherwise}
\end{cases}
 \eea 

\subsection{Probabilistic interference channels}
\label{Sec:prob-interchan}
By repeatedly entering the same input $\vec x_{A}\colon \Inp_{A}^{\ast}$ like in Sec.~\ref{Sec:interchan} \sindex{interference!channel}\sindex{channel!interference}, but this time not just  recording the possible outputs, but counting their frequencies, Alice can derive the interference version of a probabilistic channel, say in the matrix form
\[\prooftree
\hom - - \colon \Inp^{\ast}\times \Oup^{\ast}
\justifies
\intt^A \hom - - \colon \Inp^{\ast}_{A}\times \Oup^{\ast}
\endprooftree\]
by defining the matrix entries using \eqref{eq:chanview-prob} and \eqref{eq:St-prob} to be
\bea\label{eq:inttprob}
\intt^A \hom{\vec x_{A}} {\vec y}& = & \sum_{\vec x\in \Inp^\ast}  \hom{\vec x_{A}}{\vec x} \cdot \hom{\vec x}{\vec y}^A
\eea

\subsection{Definition and characterization of probabilistic noninterference}
With the above liftings of notions and notations from possibilistic channels to probabilistic channels, the definition of probabilistic noninterference can be written in the same way. We just rewrite Def.~\ref{Def:nonint} from its relational form to stochastic matrices.

\begin{definition}\label{Def:nonint-prob}\sindex{noninterference}
A channel $\hom - -  \colon \Inp^\ast \times\Oup^\ast \to [0,1]$ satisfies the noninterference requirement if for all subjects $A$, it is indistinguishable from the interference channels that it induces
\bea\label{eq:prob-nonint}
\intt^{A} \hom{\vec x_{A}}{\vec y} & = & \hom{\vec x_{A}}{\vec y}.
\eea
\end{definition}

\begin{proposition}\label{Prop:prob-nonint-char}
For every shared channel $\hom - - \colon \Inp^\ast\times\Oup^{\ast}\to [0,1]$ and every subject $A$, the following conditions are equivalent:
\begin{enumerate}[(a)]
\item the noninterference requirement:
\bear
\intt^{A} \hom{\vec x_{A}}{\vec y} & = & \hom{\vec x_{A}}{\vec y}
\eear

\item  for all $\vec x,\vec x' \in \Inp^\ast$ and $\vec y\in \Oup^\ast$, 
\bear
\vec x\restr_{A} = \vec {x'} \restr_{A} &  \Longrightarrow &  \hom{\vec x}{\vec y}^A = \hom{\vec {x'}}{\vec y}^A\eear
 
 \item for all $\vec x\in \Inp^\ast$ and $\vec y\in \Oup^\ast$ 
 \bear \hom{\vec x}{\vec y}^A  &  = &    \hom{\vec x\restr_{A}}{\vec y}\eear
 
\item for all $\vec x, \vec x'\in \Inp^\ast$ there is $\vec u\in \Inp^\ast$ such that  for all $\vec y\in\Oup^\ast$ holds
\bear 
\vec x\restr_{A} = \vec u \restr_{A}  \ \ \wedge \ \ \hom{\vec x}{\vec y}^A = \hom{\vec u}{\vec y}^A \ \ \wedge\ \ \  \vec u \restr_{\oth   A} = \vec x'\restr_{\oth   A} \eear
where $\oth A = \Subj \setminus \{A\}$.
\end{enumerate}
\end{proposition}

The proof is left as an exercise, since it again lifts from the possibilistic cases in Prop.~\ref{Prop:nonint-char}.

\subsection{Example: Car rental}
Alice and Bob used to visit Smallville together, and they always rented a car. Alice loved to rent a Porsche. They are not so close anymore, but Bob started missing Alice, and he is hoping that their paths might cross again. 

One day Bob visits Smallville again, and goes as always to the car rental office. He requests the Porsche. If the Porsche is available, then Alice is probably not in town. Or maybe she is, but the Porsche was not available when she tried to rent it. Or maybe the Porsche is not available, but someone else has rented it, and Alice is not in town. Or maybe\ldots

"Everything is possible. But how \emph{likely}\/ is it that Alice is in town if the Porsche is available?"

"Hmm. It is much easier to estimate how likely it is that Porsche is available if Alice is in town. No chance! Maybe a small chance. Say 1 in 5. On the other hand, no one ever wants to pay the pricey Porsche rental, so if Alice is \emph{not}\/ in town, then I think the chance that the Porsche is available should be something like 90\%."

Bob's estimates express his beliefs about the extent to which Alice's presence may be the cause of Porsche's unavailability. There are thus two possible causes, and two possible effects:
\begin{itemize}
\item  The causes are the inputs from $\Inp = \{a, \neg a\}$, where $a$ means that Alice is in town, and $\neg a$ that she is not in town.
\item The effects are the outputs from $\Oup  =  \{p, \neg p\}$, where $p$ means that the Porsche is available for rent, and $\neg p$ that it is out, and unavailable.
\item The channel $\hom - - \colon \Inp \times \Oup\to [0,1]$ now expresses Bob's beliefs. It is the leftmost matrix. The first row displays the probable effects of Alice's presence; the second row are the effects of her absence.
\begin{center}
\begin{tabular}{|r||r|r|}
\hline
$\hom x -$ & $p$ & $\neg p$\\ 
\hline\hline
$\hom{a}-$ & $\frac 1 5$ & $\frac 4 5$\\
\hline
$\hom{\neg a}-$ & $\frac 9 {10}$ & $\frac 1 {10}$\\
\hline
\end{tabular}
\qquad\qquad
\begin{tabular}{|r||r|r|}
\hline
$\homm{x,y}$ & $p$ & $\neg p$\\ 
\hline\hline
$a$ & $\frac 1 {10}$ & $\frac 2 5$\\
\hline
$\neg a$ & $\frac 9 {20}$ & $\frac {1}{20}$\\
\hline
\end{tabular}
\qquad\qquad
\begin{tabular}{|r||r|r|}
\hline
$\hom y -$ & $a$ & $\neg a$\\ 
\hline\hline
$\hom{p}-$ & $\frac 2 {11}$ & $\frac{9}{11}$ \\
\hline
$\hom{\neg p}-$ & $\frac 8 9$ & $\frac 1 9$\\
\hline
\end{tabular}
\end{center}

\item Bob has no idea how likely it is that Alice is in town, so he takes the chance to be fifty-fifty, i.e., that the probabilities are $\homm a = \homm{\neg a} = \frac 1 2$. Using the formula $\homm{x,y} = \homm{x}\cdot\hom x y$ again, the second table is obtained from the first table by multiplying all entries with $\frac 1 2$.

\item The third table displays the chance $\hom y x$ that the availability and unavailability of the Porsche, the output $y \in \{p, \neg p\}$ is caused by Alice's presence or absence, which is the input $x\in \{a, \neg a\}$. This chance is derived from the second table using the Bayesian law in \eqref{eq:bayes}. The probabilities $\homm y$ are obtained by summing up the columns of the second table. Hence, $\homm p = \frac{11}{20}$ and $\homm{\neg p} = \frac 9 {20}$.
\end{itemize}
So if the Porsche is available, then the chance that Alice is in town is $\hom p a = \frac 2 {11}$; if the Porsche is not available, then the chance is $\hom {\neg p} a = \frac 8 9$.

\section{Secrecy}

The earliest notion of \sindex{secrecy} secrecy, the \emph{"One Ring to rule them all"}\/ of cryptography, is Shannon's notion of \emph{perfect secrecy} \cite{ShannonC:Secrecy}. We show how this classical concept is subsumed under channel security, and then comment how its modern refinements fit into the same framework.

\subsection{Perfect secrecy}
\sindex{cipher} A \emph{cipher}\/ is a family of functions $\{<\Enc_k, \Dec_k>\}_{k\in \Keys}$, where
\begin{itemize}
\item $\Enc_k: \Msgs \to \Ciph$ are the \emph{encryption}\/ functions, 
\item $\Dec_k : \Ciph \to \Msgs$ are the \emph{decryption}\/ functions,
\item $\Msgs$ is the type of \emph{messages}, or \emph{plaintexts},
\item $\Ciph$ is the type of \emph{ciphertexts}, and
\item $\Keys$ is the type of keys,
\end{itemize}
such that the equations
\bea\label{eq:decrypt}
\Dec_k\left(\Enc_k(m)\right) & = & m
\eea
hold for every $k\in \Keys$ and every $m\in \Msgs$. These equations impose on ciphers the \emph{functional}\/ requirement that every message $m$ enciphered as $c = \Enc_k(m)$ can be deciphered as $m=\Dec_k(c)$, for every $k\in \Keys$. The \emph{security}\/ requirement, on the other hand, is that this is the \emph{only}\/ way to recover the plaintext, i.e., that it cannot be recovered without the key $k$ which allows the user to pick the correct $\Dec_k$ for \eqref{eq:decrypt}. This is the \emph{secrecy}\/ requirement. It can be formalized probabilistically, by requiring that chance that $m$ can be \emph{guessed}\/ from $c=\Enc_k(m)$ is negligible, i.e., close to 0. Ever since Shannon's seminal paper on \emph{"Communication theory of Secrecy Systems"} \cite{ShannonC:Secrecy}, where the formalization was proposed, the assumption was that "the enemy knows the system", in the sense that the attacker Alice is given the family $\{<\Enc_k, \Dec_k>\}_{k\in \Keys}$. Although she is not given the key $k$, and therefore does not know which particular encryption function $\Enc_k$ is used, knowing the whole family of encryption functions allows her to average out and derive a guessing channel
\[\prooftree
\Keys \times \Msgs \tto{\Enc} \Ciph
\justifies
\Ciph \tto{\hom - - } \Dist \Msgs
\endprooftree\]
where the distribution $\hom c - \in \Dist \Msgs$ measures the probability $\hom c m$ that $c = \Enc_k(m)$ by averaging over $\Keys$ the $m$-cylinder of $\Enc_k^{-1}(c)$, i.e., by counting which portion of $\Keys$ enciphers $m$ as $c$. Hence, we have the memoryless probabilistic channel
\beq\label{eq:Shan-chan}
\begin{array}{rccrcl}
\hom - - \ :\ \Ciph & \to & \Dist \Msgs\\
c & \mapsto & \big(\hom c - & :\ \Msgs & \to & [0,1]\big)\\
&&& m & \mapsto & \hom c m = \frac{\#\left\{k\in \Keys\ |\ \Enc_k(m) = c\right\}}{\# \Keys}
\end{array}
\eeq
where $\#X$ denotes the number of elements of the set $X$. This channel captures attacker Alice's view of the cipher. Her goal is to guess $m$ from $c$, and \eqref{eq:Shan-chan} gives the guessing odds. The value $\hom m c$ is the chance that the plaintext is $m$ if the ciphertext is $c$. In other words, we write the conditional probability $\Pr(m | c)$ in the form $\hom c m$. The guessing is thus driven by the probability distribution $\hom c - :\ \Msgs  \to  [0,1]$. If Alice is given a single chance to guess the plaintext, she should probably try an $m$ with the highest probability $\hom c m$; otherwise, if there are more guessing chances, or if the cryptanalytic attack is an ongoing process, then the whole area of betting strategies and information elicitation opens up. 

Shannon's \emph{perfect secrecy}\/ \sindex{secrecy!perfect}
requirement was that the ciphertext tells nothing about the plaintext that wouldn't be known without the ciphertext. It is assumed that the messages are sourced with a publicly known distribution $\homm - :\Msgs\to [0,1]$, which can be, e.g., the word frequency if $\Msgs$ is the lexicon of a language. So $\homm m$ is the \emph{a priori}\/ chance that an unknown word is $m$. The perfect secrecy requirement is that the ciphertext does not add any \emph{posterior}\/ information to this prior knowledge.

\bigskip
\begin{definition}\label{Def:perf-sec} 
A cipher 
 satisfies the\/ \emph{perfect secrecy} requirement if the equation 
\bea\label{eq:perf-sec}
\hom c m & = & \homm m
\eea
holds for all $m\in \Msgs$ and $c\in \Ciph$.
\end{definition} 

The statistical independency of $m$ and $c$, required by \eqref{eq:perf-sec}, can be equivalently expressed using the constant projector
\bea\label{eq:upsilon-bang}
\upsilon ! \ :\ \Dist \Ciph \tto ! 1 \tto\upsilon \Dist\Ciph
\eea
which projects all distributions over $\Ciph$ into the uniform distribution $
\upsilon  : \Ciph  \to [0,1]$ where $\upsilon (c) = \frac 1 {\#\Ciph}$ for all $c\in \Ciph$. Note that $1 = \Dist 1$, as there is a single probability distribution on the singleton.

\bigskip
\begin{proposition}
A cipher 
 satisfies the perfect secrecy requirement from Def.~
 \ref{eq:perf-sec} if and only if the memoryless channel
\bea\label{eq:overhom}
\overline{\hom - -}\ :\ \Dist{\Ciph} & \to & \Dist \Msgs\\
\gamma & \mapsto & \sum_{x\in \Ciph} \gamma(x) \cdot \hom x - \notag
\eea
satisfies the negative security requirement with respect to the projector $\upsilon ! $ from~\eqref{eq:upsilon-bang}, which means that the following triangle commutes
\beq\label{eq:Shan-triangle}
\begin{tikzar}[row sep = 1pc,column sep = 7pc]
\Dist \Ciph \arrow{d}[swap]{!} \arrow{dr}{\overline{\hom - -}}\\
\Dist 1 \arrow{d}[swap]{\upsilon}\& \Dist \Msgs\\
\Dist \Ciph \arrow{ur}[swap]{\overline{\hom - -}}
\end{tikzar}
\eeq
\end{proposition}
 
\bpr
Consider the point distribution $\chi_c\in \Dist \Ciph$ 
\bear
\chi_c:\Ciph & \to &  [0,1]\\
x & \mapsto & \begin{cases} 1 & \mbox{ if } x=c\\
0 & \mbox{ otherwise}
\end{cases}
\eear
for an arbitrary $c\in \Ciph$. The definition of $\overline{\hom - -}$ in \eqref{eq:overhom} implies
\beq\label{eq:point}
\overline{\hom - -}\left(\chi_c\right)(m) \ \ =\ \ \sum_{x\in \Ciph} \chi_c(x) \cdot \hom x m\ \  =\ \  \hom c m
\eeq
where we fix an arbitrary $m\in \Msgs$, for convenience. On the other hand, going down and right around the triangle in \eqref{eq:Shan-triangle}, we get
\beq\label{eq:unif}
\overline{\hom - -}\left(\upsilon\right)(m)  \ \ =  \ \ \sum_{x\in \Ciph} \frac 1 n \cdot  \hom x m \ \ =\  \ \frac 1 n \sum_{x\in \Ciph} \homm{x,m}
\eeq
where $n = \#\Ciph$. The commutativity of the triangle in \eqref{eq:Shan-triangle} for all $\chi_c$, i.e., the equations $\overline{\hom - -}\left(\chi_c\right)(m) = \overline{\hom - -}\left(\upsilon\right)(m)$ thus give $n$ equations in the form
\bear
(1-n)\hom c m + \sum_{\substack{x\in \Ciph\\x\neq c}} \hom x m & = & 0
\eear
one for each $c \in \Ciph$. This system of $n$ equations thus completely determines the values of $\hom c m$ for each of the $n$ ciphertexts $c\in \Ciph$, and for the fixed arbitrary $m\in \Msgs$. Since each $c$ occurs in exactly one equation with a coefficient $1-n$ and with the coefficient 1 in the remaining $n-1$ equations, the system is invariant under the permutations of $c$, which means that all solutions must be equal, i.e.
\[\hom {x} m \ = \  \hom {c} m \ \mbox{ for all } x, c \in \Ciph  \]
This means that the value of $\hom x m$ does not depend on the choice of $x$, and hence
\[ \homm m \ =\ \sum_{x\in \Ciph} \homm x \cdot \hom x m\ =\ 
\hom c m \cdot \sum_{x\in \Ciph} \homm x \ =\ 
\hom c m\].
\epr

\subsection{Computational secrecy}
Security of modern cryptosystems is based on  computational hardness assumptions. Such a system can satisfy the secrecy requirement even if the chance $\hom c m$ to guess $m$ from $c = \Enc_k(m)$ is greater than the chance $\homm m$ to guess $m$ without $c$ --- provided that \emph{computing}\/ the probabilities $\hom c m$ for the various plaintexts $m$ is computationally hard. This means that, although "the enemy [still] knows the system", and the attacker Bob is still given the cipher $\{<\Enc_k, \Dec_k>\}_{k\in \Keys}$, it is computationally unfeasible for him to derive the channel $\hom - - :\Ciph \to \Dist \Msgs$ and actually compute the guessing chances $\hom c m$. So the security requirements on modern crypto systems cannot be stated in terms of this channel.  But although they are stated in different terms, they still fit into the conceptual framework of channel security in an interesting way. The technical details are beyond the scope of this text, but we sketch the big picture.

For a cipher $\{<\Enc_k, \Dec_k>\}_{k\in \Keys}$ derived from a modern crypto system, the functional requirement is still that the plaintext is recovered from the ciphertext by
\bear
\Dec_k(\Enc_k(m)) & = & m
\eear
but now the security requirement, that this should be the \emph{only}\/ way to recover $m$ from $c$, is expressed not by saying that the chance to guess $m$ from $c$ without $k$ is negligible (i.e., not greater than guessing $m$ on its own), but by saying that the chance to find a feasible algorithm $A$ that satisfies
\bea\label{eq:AEnc}
A(\Enc_k(m)) & = & m
\eea
is negligible (i.e., not greater than guessing $k$ and using $A = \Dec_k$). Although echoing the structure of the functional requirement $\Dec_k(\Enc_k(m)) =  m$, this version of the secrecy requirement turns out to be independent of it. 
If we omit the keys and forget about the decryption functions for a moment, we are left with a function $\Enc:\Msgs \to \Ciph$, and the security requirement is simply that finding the inverse images along it is hard. This is the seminal concept of \emph{one-way}\/ functions, the stepping stone into modern cryptography \cite{Diffie-Hellman}. 

\bigskip
\begin{definition}\label{Def:onwy}
A function $\Enc:\Msgs \to \Ciph$ is said to be\/ \emph{one-way} if 
\begin{enumerate}[a)]
\item given $m\in \Msgs$, it is easy to compute $\Enc(m)$, but
\item given $c \in \Ciph$, it is hard to find $m\in \Msgs$ such that $\Enc(m) = c$.
\end{enumerate}
\end{definition}
Formally, 
\begin{itemize}
\item (a) means that there is a feasible algorithm that computes $\Enc:\Msgs \to \Ciph$; but \item (b) means that there is no feasible algorithm that computes a function $A:\Ciph \to \Msgs$ such that $A(\Enc(m)) \in \Enc^{-1}\left(\Enc(m)\right)$ for all $m\in \Msgs$. 
\end{itemize}

\para{Remark.} In general, it is not required that one-way functions $\Enc$ are injective, and the inverse images may not be unique. In the special case of crypto systems, however, the functional requirement  $\Dec_k \circ \Enc_k(m) = m$ implies that $\Enc_k(x) = \Enc_k(y)$ only when $x=y$, and $\Enc_k$ is thus injective. In the general case, the requirement that $A(\Enc(m)) \in \Enc^{-1}\left(\Enc(m)\right)$ is easily seen to be equivalent with the equation $\Enc \circ A \circ \Enc(m) = \Enc(m)$. But $\Enc \circ A \circ \Enc = \Enc$ means that $A\circ \Enc$ is a projector. 

\bigskip
\begin{proposition}
A function $\Enc:\Msgs\to \Ciph$ is one-way if and only if it is feasible, and if for any feasible function $A:\Ciph\to \Msgs$ where $A\circ  \Enc : \Msgs \to \Msgs$ is a projector, the chance that the following triangle commutes is negligible (which we express by $\star$).
\beq\label{eq:onwy-prop}
\begin{tikzar}[row sep = 1pc,column sep = 2pc]
\Dist \Msgs \arrow{d}[swap]{\Dist \Enc} \arrow{drrr}{\Dist \Enc}\\
\Dist \Ciph \arrow{d}[swap]{\Dist A}\& \mbox{\Large$\bigstar$} \&\& \Dist \Ciph\\
\Dist \Msgs \arrow{urrr}[swap]{\Dist\Enc}
\end{tikzar}
\eeq
\end{proposition}

\para{Comment.} The requirement that the triangles in \eqref{eq:onwy-prop}, indexed by all feasible functions $A:\Ciph\to \Msgs$, \emph{almost never}\/ commute, is of course \emph{not}\/ one of the channel security requirements from Def.~\ref{Def:nonint-prob}.
There, the triangles are required to commute to be secure. The idea is that the projectors $\proj$ in Def.~\ref{Def:nonint-prob} filter out the undesired flows, and let through the desired flows. A function $f$ thus behaves like $f\circ \proj$ precisely when it passes no undesired flows. In triangle \eqref{eq:onwy-prop}, the undesired flows are the composites $A\circ \Enc$ where $A$ reverts the one-way-function $\Enc$. Here the security requirement is that such flows are unlikely. This is expressed by minimizing the incidence of the equation $\Enc\circ A \circ \Enc = \Enc$. The requirement that equations are not satisfied are unusual in algebra, but crucial in security. 

\def\thechapter{9}
\setchaptertoc
\chapter{Privacy}\label{Chap:Privacy}

\section{Idea of privacy}\label{Sec:Intro}

\para{What is privacy?} 
Privacy is the right to be left alone. 

Privacy is \emph{not}\/ a security property. It is the source of security requirements. I claim that the water well is my  private property and that my health record contains my private data. These privacy claims give rise to security requirements: to prevent others from accessing my water well and my health record.

Privacy is the owners' \emph{right}\/ to enjoy their resources with no interference from others. This right protects the owner from access through law-abiding parties, but also declares access from adversaries to be illegal.
Security focuses on the adversarial \emph{process}\/ of defending and attacking some privately owned assets. 

Privacy is the realm of ownership claims:
\begin{itemize}
\item ``The data about me are owned by me.''
\item ``The sea is owned by the King.''
\end{itemize}
Security is the process of implementing privacy claims:
\begin{itemize}
\item ``My data are secured by encryption.'' 
\item ``Our maritime borders are secured by King's Navy.''
\end{itemize}

There is also a feedback loop from security to privacy. If I am unable to secure the water well, then I cannot claim the  right to be left alone with it. For example, if Lion wants to prevent other animals from using a water well, but the well is too big for Lion to defend it, then Lion cannot claim the right to be left alone with the well. He may try to chase away Gazelle drinking on the other side of the well. But by the time Lion gets there, Gazelle has had plenty of water. After trying to secure the well by running around it for a while, Lion will eventually have to give up his privacy claim.

\para{If an asset can be secured, then it can be made private in as many ways as it can be used.} Back in Sec.~\ref{Sec:good-bad}, we distinguished the dependability requirements that can be imposed on a car from the security requirements. The dependability is assured by the engine and the brakes; the security by the locks and the keys. In addition, the privacy requirements can be imposed on the various aspects of the use of the car. Driver's and passengers' privacy can be assured by tinted windows. The right to privately use the car can be secured by the ownership documents. The right to use a car for public transport can be issued as a license to private individuals, for taxi or ride-sharing. A range of privacy arrangements for vehicles is illustrated in \cref{Fig:cars}.
\begin{figure}[!ht]
\begin{minipage}{.48\linewidth}
\centering
\includegraphics[height = 3.9cm
]{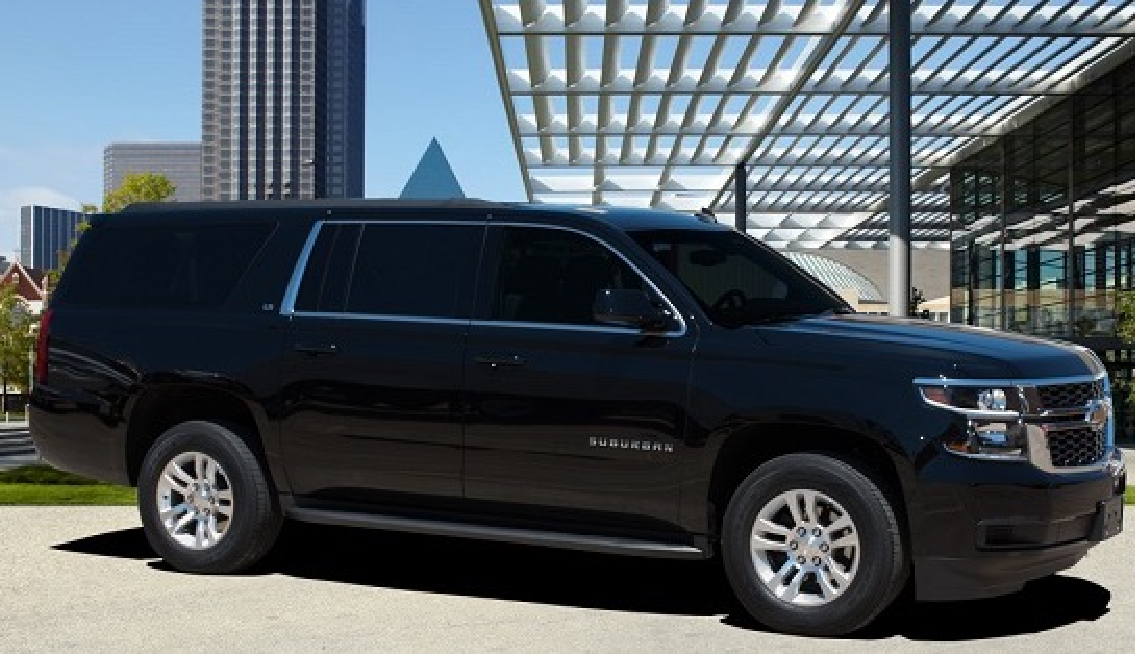}
\end{minipage}
\hspace{.02\linewidth}
\begin{minipage}{.48\linewidth}
\centering
\includegraphics[height = 3.9cm
]{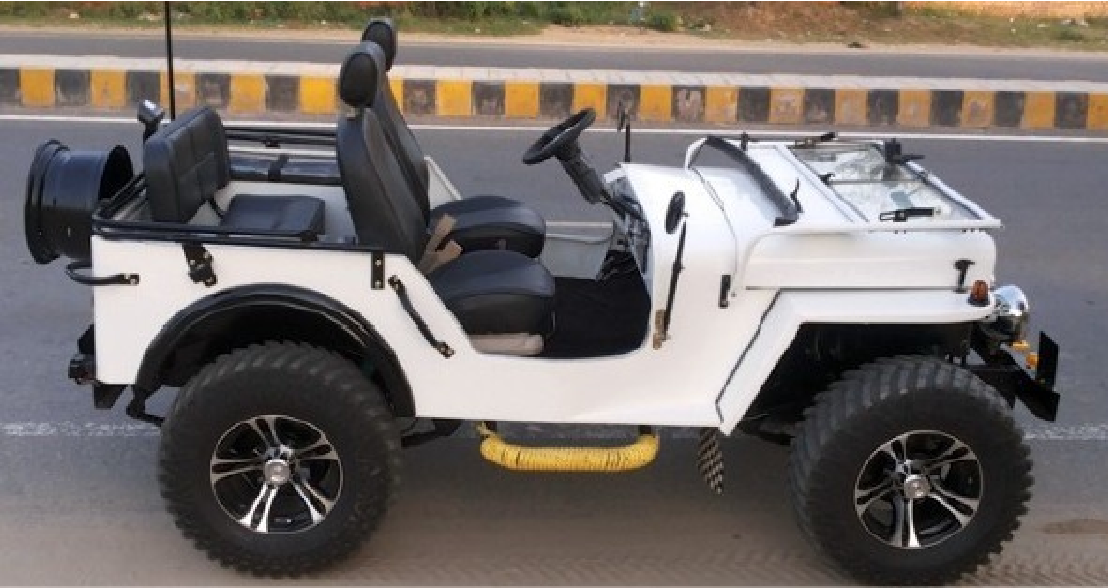}
\end{minipage}

\vspace{1.2\baselineskip}
\begin{minipage}{.48\linewidth}
\centering
\includegraphics[height = 3.8cm
]{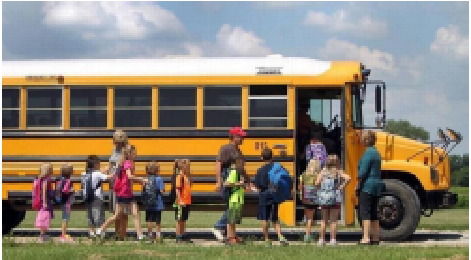}
\end{minipage}
\hspace{.02\linewidth}
\begin{minipage}{.48\linewidth}
\centering
\includegraphics[height = 3.8cm
]{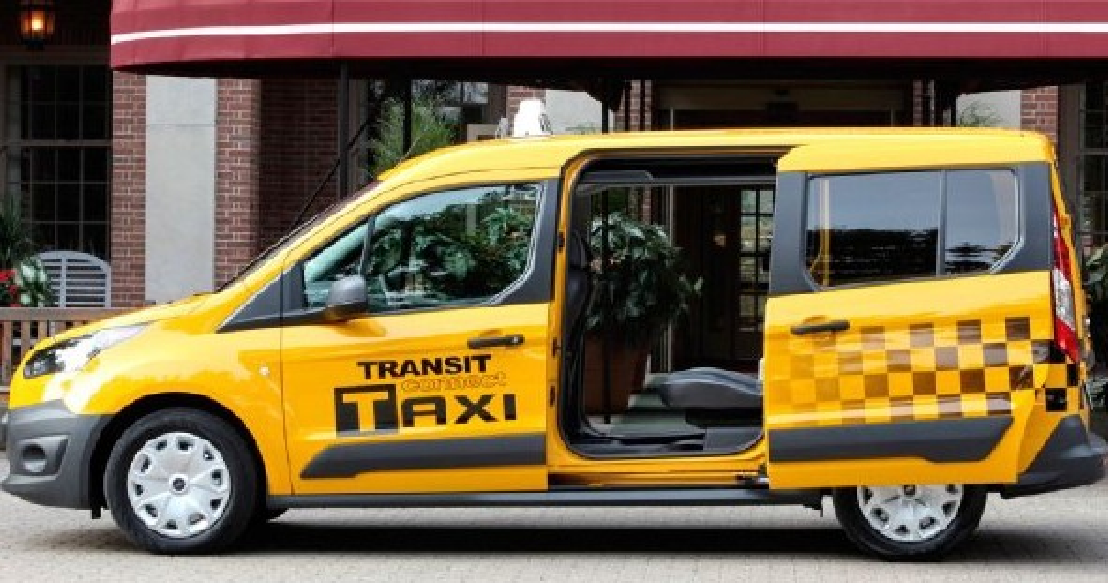}
\end{minipage}
\caption{Different types of vehicles address different privacy requirements}
\label{Fig:cars}
\end{figure}
Fig.~\ref{Fig:NYC} shows a city as a partition of space into public and private. The structure of this partition is determined by the privacy claims and arrangements: ``This is my private apartment and I have the right to be alone'', or ``This is public space and everyone can come here''. The purpose of security tools and the task of security policies is to implement such privacy claims and denials, as well as complex privacy policies.
\begin{figure}[!ht]
\begin{center}
\includegraphics[height = 7cm
]{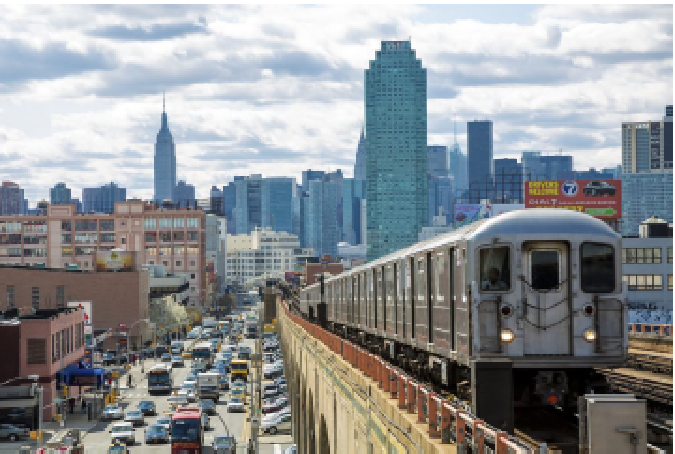}
\end{center}
\caption{The structure on living spaces is induced by the variety of privacy policies}
\label{Fig:NYC}
\end{figure}

\subsection[History]{History of privacy}
The social history is, first of all, the history of shifting demarcation lines between the public  sphere and the private sphere~\cite{BaileyJ:private,orlin2009locating,schoeman1984philosophical}. The communist revolutions usually start by abolishing not only private property, but also private rights. Tyrannies and oligarchies, on the other hand, erode the public ownership and rights, by privatizing resources, services, and social life. The distinction between the realms of 
\begin{itemize}
\item the public: city, market, warfare\ldots and
\item the private: family, household, childbirth\ldots 
\end{itemize}
was established and discussed in antiquity \cite{angela2009rome,burke2000delos}. It was a frequent topic in Greek tragedies: e.g., Sophocles' \emph{Antigone}\/ is torn between her private commitment to her brothers and her public duty to the king. The English word \emph{politics}\/ comes from the Greek word \textgreek{p'olis}, denoting the public sphere; the English word \emph{economy}\/ comes from the Greek word \textgreek{o'ikos}, denoting the private sphere.

\para{Disambiguations.} There are many aspects of privacy, conceptualized in different research communities, and studied by different methods; and perhaps even more aspects that are not conceptualized in research, but arise in practice, and in informal discourse. We carve a small part of the concept, and attempt to model it formally.

As an abstract requirement, privacy is a negative constraint, in the form \emph{"bad stuff (actions) should not happen"}. Note that secrecy and confidentiality are also such negative constraints, whereas authenticity and integrity are positive constraints, in the form \emph{"good stuff (actions) should happen"}. More precisely, authenticity and integrity require that some desirable information flows happen. For example, a message \emph{"I am Alice"}\/ is authentic if it originates from Alice. On the other hand, confidentiality, secrecy and privacy require that some undesirable information flows are prevented: Alice's password should be secret, her address should be confidential, and her health record should be private. 

But what is the difference between privacy, secrecy and confidentiality? Let us move out of the way the difference between the latter two first. In the colloquial usage, they allow subtle distinctions. For example, when a report is confidential, we don't know its contents; but when it is secret, we don't even know that it exists. For the moment, we ignore such differences, and focus on distinguishing privacy as a right from secrecy and confidentiality as security properties. The terms are used interchangeably. 

As for the difference between privacy and secrecy (or confidentiality), it comes in two flavors:
\begin{description}
\item[1)] while secrecy is a property of \emph{information}, privacy is a property of any \emph{asset}\/ or \emph{resource}\/ that can be secured\footnote{Information is, of course, a resource, so it can be private.}; and
\item[2)] while secrecy is a \emph{local}\/ requirement, usually imposed on information flowing through a given channel, privacy is a \emph{global}\/ requirement, usually imposed on all resource requests and provisions, along any channel of a given network. 
\end{description}
Let us take a look at some examples.

\para{Ad (1)} Alice's password is secret, whereas her bank account is private. Bob's health record is private, and it remains private after he shares some of it with Alice. It consists of his health information, but may also contain some of his tissue samples for later analysis. On the other hand, Bob's criminal record is in principle not private, as criminal records often need to be shared, to protect the public. Bob may try to keep his criminal record secret, but even if he succeeds, it will not become private. Any resource can be made private if the access to it can be secured. For example,  we speak of a private water well, private funds, private party if the public access is restricted. On the other hand, when we speak of a secret water well, secret funds, or a secret party, we mean that the public does not have any information about them. A water well can be secret, and it can be private, but not secret.

\para{Ad (2)} To attack Alice's secret password, Mallory eavesdrops at the secure channel to her bank; to attack her bank account, he can of course, also try to steal her credentials, or he can initiate a request through any of the channels of the banking network, if he can access it. Or he can coordinate an attack through many channels. To attack a secret, a cryptanalyst analyzes a given cipher. To attack privacy, a data analyst can gather and analyze data from many surveillance points. To protect secrecy, the cryptographer must assure that the plaintexts cannot be derived from the ciphertexts without the key, for a given cipher. To protect privacy, the network operator must assure that there are no covert channels anywhere in the network. 

\subsection{Attacking and securing privacy}
The general principle of \textbf{Privacy Policy Enforcement} is:
\begin{quote}
\em If a resource can be secured,\\ then it can be made private.
\end{quote}
This extends the general principle of \textbf{Security Policy Enforcement}, which states that
\begin{quote}
\em A resource can be secured  \\if and only if its value is greater than the cost of securing it.
\end{quote}
In practice, though, things get complicated due to the dynamic interactions across the layers of the \emph{policy enforcement stack}, displayed in \cref{Fig:PEStack}. 
\begin{figure}[!ht]
\begin{center}
\newcommand{\Sentences}{Privacy}
\newcommand{\Words}{Security}
\newcommand{\Alphabet}{Cryptography}
\newcommand{\notng}{}
\def\JPicScale{.4}
\ifx\JPicScale\undefined\def\JPicScale{1}\fi
\psset{unit=\JPicScale mm}
\psset{linewidth=0.3,dotsep=1,hatchwidth=0.3,hatchsep=1.5,shadowsize=1,dimen=middle}
\psset{dotsize=0.7 2.5,dotscale=1 1,fillcolor=black}
\psset{arrowsize=1 2,arrowlength=1,arrowinset=0.25,tbarsize=0.7 5,bracketlength=0.15,rbracketlength=0.15}
\begin{pspicture}(0,0)(250,80)
\psline[linewidth=0.6](10,60)(250,60)
\psline[linewidth=0.6](30,40)(210,40)
\psline[linewidth=0.6](70,20)(190,20)
\psline[linewidth=0.6](70,40)(70,20)
\psline[linewidth=0.6](190,40)(190,20)
\psline[linewidth=0.6](210,60)(210,40)
\psline[linewidth=0.6](30,60)(30,40)
\psline[linewidth=0.6](10,80)(10,60)
\psline[linewidth=0.6](250,80)(250,60)
\rput(130,70){\Sentences}
\rput(130,50){\Words}
\rput(130,30){\Alphabet}
\psline[linewidth=0.6](10,80)(250,80)
\rput(70,-15){\notng}
\end{pspicture}
\vspace{-1\baselineskip}
\caption{Privacy is built upon security, which is built upon cryptography}
\label{Fig:PEStack}
\end{center}
\end{figure}
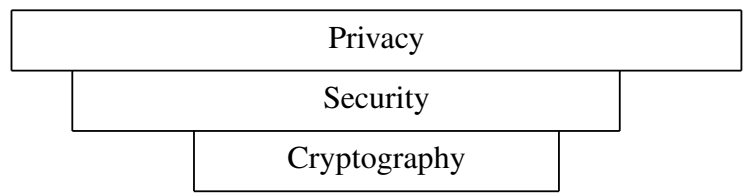
The complications arise from the fact that the strength of the foundations does not in general guarantee the strength of the buildings that they carry. The strong foundations are necessary, but not sufficient for a strong building. In security, this leads to the well-known phenomenon that security protocols may be broken without breaking the underlying crypto \cite{SimmonsG:TMN}. In privacy, the analogous phenomenon is that privacy protocols may be broken without breaking the underlying security protocols \cite{PavlovicD:CathyFest}. The idea is displayed in Figures \ref{Fig:breaking-sec} and \ref{Fig:breaking-priv}.
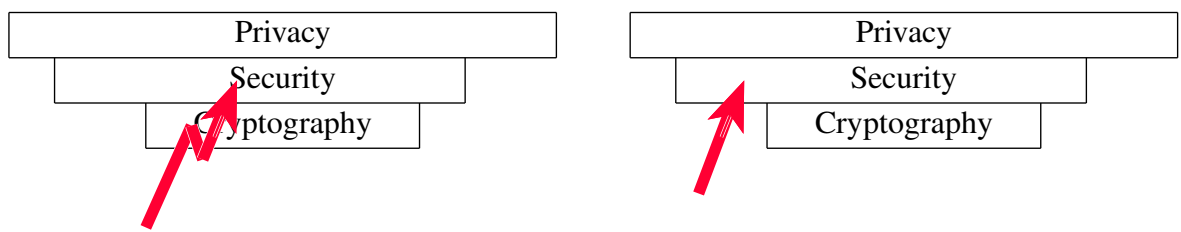
\begin{figure}[!b]
\begin{minipage}{.45\linewidth}
\newcommand{\Sentences}{Privacy}
\newcommand{\Words}{Security}
\newcommand{\Alphabet}{Cryptography}
\newcommand{\notng}{}
\def\JPicScale{.3}
\ifx\JPicScale\undefined\def\JPicScale{1}\fi
\psset{unit=\JPicScale mm}
\psset{linewidth=0.3,dotsep=1,hatchwidth=0.3,hatchsep=1.5,shadowsize=1,dimen=middle}
\psset{dotsize=0.7 2.5,dotscale=1 1,fillcolor=black}
\psset{arrowsize=1 2,arrowlength=1,arrowinset=0.25,tbarsize=0.7 5,bracketlength=0.15,rbracketlength=0.15}
\begin{pspicture}(0,0)(250,80)
\psline[linewidth=0.6](10,60)(250,60)
\psline[linewidth=0.6](30,40)(210,40)
\psline[linewidth=0.6](70,20)(190,20)
\psline[linewidth=0.6](70,40)(70,20)
\psline[linewidth=0.6](190,40)(190,20)
\psline[linewidth=0.6](210,60)(210,40)
\psline[linewidth=0.6](30,60)(30,40)
\psline[linewidth=0.6](10,80)(10,60)
\psline[linewidth=0.6](250,80)(250,60)
\rput(130,70){\Sentences}
\rput(130,50){\Words}
\rput(130,30){\Alphabet}
\psline[linewidth=0.6](10,80)(250,80)
\newrgbcolor{userLineColour}{1 0 0.2}
\psline[linewidth=5,linecolor=userLineColour,border=0.3](70,-15)(90,30)
\newrgbcolor{userLineColour}{1 0 0.2}
\psline[linewidth=5,linecolor=userLineColour,border=0.3](90,30)(95,15)
\newrgbcolor{userLineColour}{1 0 0.2}
\psline[linewidth=5,linecolor=userLineColour,border=0.3](95,15)(105,40)
\newrgbcolor{userLineColour}{1 0 0.2}
\psline[linewidth=1,linecolor=userLineColour,border=0.3,arrowsize=15 2,arrowlength=1.3]{->}(100,25)(110,50)
\end{pspicture}
\end{minipage}
\hspace{.05\linewidth}
\begin{minipage}{.45\linewidth}
\newcommand{\Sentences}{Privacy}
\newcommand{\Words}{Security}
\newcommand{\Alphabet}{Cryptography}
\newcommand{\notng}{}
\def\JPicScale{.3}
\ifx\JPicScale\undefined\def\JPicScale{1}\fi
\psset{unit=\JPicScale mm}
\psset{linewidth=0.3,dotsep=1,hatchwidth=0.3,hatchsep=1.5,shadowsize=1,dimen=middle}
\psset{dotsize=0.7 2.5,dotscale=1 1,fillcolor=black}
\psset{arrowsize=1 2,arrowlength=1,arrowinset=0.25,tbarsize=0.7 5,bracketlength=0.15,rbracketlength=0.15}
\begin{pspicture}(0,0)(250,80)
\psline[linewidth=0.6](10,60)(250,60)
\psline[linewidth=0.6](30,40)(210,40)
\psline[linewidth=0.6](70,20)(190,20)
\psline[linewidth=0.6](70,40)(70,20)
\psline[linewidth=0.6](190,40)(190,20)
\psline[linewidth=0.6](210,60)(210,40)
\psline[linewidth=0.6](30,60)(30,40)
\psline[linewidth=0.6](10,80)(10,60)
\psline[linewidth=0.6](250,80)(250,60)
\rput(130,70){\Sentences}
\rput(130,50){\Words}
\rput(130,30){\Alphabet}
\psline[linewidth=0.6](10,80)(250,80)
\newrgbcolor{userLineColour}{1 0 0.2}
\psline[linewidth=5,linecolor=userLineColour,border=0.3](40,0)(55,40)
\newrgbcolor{userLineColour}{1 0 0.2}
\psline[linewidth=1,linecolor=userLineColour,border=0.3,arrowsize=15 2,arrowlength=1.3]{->}(50,25)(60,50)
\rput(90,-35){\notng}
\rput(80,-25){\notng}
\rput(70,-15){\notng}
\end{pspicture}
\end{minipage}
\caption{Security protocols can be attacked, by breaking the underlying crypto, but also directly}
\label{Fig:breaking-sec}
\end{figure}
\begin{figure}[!t]
\begin{minipage}{.45\linewidth}
\newcommand{\Sentences}{Privacy}
\newcommand{\Words}{Security}
\newcommand{\Alphabet}{Cryptography}
\newcommand{\notng}{}
\def\JPicScale{.3}
\ifx\JPicScale\undefined\def\JPicScale{1}\fi
\psset{unit=\JPicScale mm}
\psset{linewidth=0.3,dotsep=1,hatchwidth=0.3,hatchsep=1.5,shadowsize=1,dimen=middle}
\psset{dotsize=0.7 2.5,dotscale=1 1,fillcolor=black}
\psset{arrowsize=1 2,arrowlength=1,arrowinset=0.25,tbarsize=0.7 5,bracketlength=0.15,rbracketlength=0.15}
\begin{pspicture}(0,0)(250,80)
\psline[linewidth=0.6](10,60)(250,60)
\psline[linewidth=0.6](30,40)(210,40)
\psline[linewidth=0.6](70,20)(190,20)
\psline[linewidth=0.6](70,40)(70,20)
\psline[linewidth=0.6](190,40)(190,20)
\psline[linewidth=0.6](210,60)(210,40)
\psline[linewidth=0.6](30,60)(30,40)
\psline[linewidth=0.6](10,80)(10,60)
\psline[linewidth=0.6](250,80)(250,60)
\rput(130,70){\Sentences}
\rput(130,50){\Words}
\rput(130,30){\Alphabet}
\psline[linewidth=0.6](10,80)(250,80)
\newrgbcolor{userLineColour}{1 0 0.2}
\psline[linewidth=5,linecolor=userLineColour,border=0.3](190,6.25)(165,51.25)
\newrgbcolor{userLineColour}{1 0 0.2}
\psline[linewidth=5,linecolor=userLineColour,border=0.3](159.69,37.01)(146.42,60.44)
\newrgbcolor{userLineColour}{1 0 0.2}
\psline[linewidth=1,linecolor=userLineColour,border=0.3,arrowsize=15 2,arrowlength=1.3]{->}(155,46.25)(141.73,69.68)
\newrgbcolor{userLineColour}{1 0 0.2}
\psline[linewidth=5,linecolor=userLineColour,border=0.3](160.62,36.25)(165,51.25)
\rput(70,-15){\notng}
\end{pspicture}
\end{minipage}
\hspace{.05\linewidth}
\begin{minipage}{.45\linewidth}
\newcommand{\Sentences}{Privacy}
\newcommand{\Words}{Security}
\newcommand{\Alphabet}{Cryptography}
\newcommand{\notng}{}
\def\JPicScale{.3}
\ifx\JPicScale\undefined\def\JPicScale{1}\fi
\psset{unit=\JPicScale mm}
\psset{linewidth=0.3,dotsep=1,hatchwidth=0.3,hatchsep=1.5,shadowsize=1,dimen=middle}
\psset{dotsize=0.7 2.5,dotscale=1 1,fillcolor=black}
\psset{arrowsize=1 2,arrowlength=1,arrowinset=0.25,tbarsize=0.7 5,bracketlength=0.15,rbracketlength=0.15}
\begin{pspicture}(0,0)(250,80)
\psline[linewidth=0.6](10,60)(250,60)
\psline[linewidth=0.6](30,40)(210,40)
\psline[linewidth=0.6](70,20)(190,20)
\psline[linewidth=0.6](70,40)(70,20)
\psline[linewidth=0.6](190,40)(190,20)
\psline[linewidth=0.6](210,60)(210,40)
\psline[linewidth=0.6](30,60)(30,40)
\psline[linewidth=0.6](10,80)(10,60)
\psline[linewidth=0.6](250,80)(250,60)
\rput(130,70){\Sentences}
\rput(130,50){\Words}
\rput(130,30){\Alphabet}
\psline[linewidth=0.6](10,80)(250,80)
\newrgbcolor{userLineColour}{1 0 0.2}
\psline[linewidth=5,linecolor=userLineColour,border=0.3](243.75,33.75)(226.42,60.44)
\newrgbcolor{userLineColour}{1 0 0.2}
\psline[linewidth=1,linecolor=userLineColour,border=0.3,arrowsize=15 2,arrowlength=1.3]{->}(235,46.25)(221.73,69.68)
\rput(70,-15){\notng}
\end{pspicture}
\end{minipage}
\caption{Privacy protocols can be attacked, by breaking the underlying security, but also directly}
\label{Fig:breaking-priv}
\end{figure}
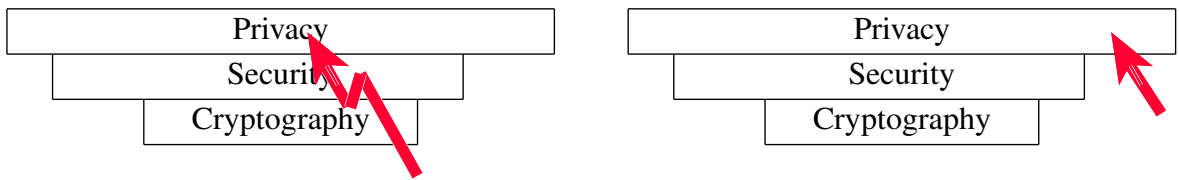
The privacy attacks that leave the underlying security intact are implemented through \textbf{deceit}. The methods of deception allow attackers to
\begin{itemize}
\item pull private data and resources, or 
\item push private actions
\end{itemize}
without breaking the security measures that make them private.

\section{Surveillance and sousveillance} \label{Sec:veillance}

\subsection{The paradox of data as a resource} 
Private data are a privately owned resource. How can this be, if data are easy to copy, whereas a resource was defined in Sec.~\ref{Sec:what-resource} by being hard to come by? The reason is that some data must be protected as private in order to protect some physical or financial resources as private. 

\para{Examples.} A password or a cryptographic key are easily copied, but should be kept private if they are used to protect private funds in a bank account. The intrinsic value of Alice's health record is to provide the information needed by Alice's physician Bob to treat her if she becomes ill, but the extrinsic value of Alice's health record is to help Alice's insurers withdraw her health coverage just before she may need it, and also perhaps to tell burglars when Alice might be in a hospital. 

\para{Private data are thus kept private not because of their utility for the owner, but because of their potential value for others, as means for attacks.} If private data are effectively secured, they become hard to come by for the attackers, but easy to use in attacks when they are available. Paradoxically, the privately owned data are resources for the non-owners.

\subsection{The two sides of data security}

\para{The \emph{Digital Rights Management (DRM)}\/ technologies} are tasked with controlling the distribution of specified digital data. While the individually owned data are also digitized and private, the DRM technologies have been developed mainly to protect the digital assets owned and marketed by organizations, and in some cases by governments. The marketed digital assets include software and media, such as digital music, movies, and digital text, such as news, study materials, and popular literature. 

\para{The \emph{surveillance}\/ technologies} are tasked with enabling and facilitating data collection about the behaviors of specified subjects. In market surveillance, the subject behaviors usually include the marketing habits and interests, which are used for targeted advertising, influence campaigns, sponsored search, and recommender systems in general. In investigative surveillance, the subject behaviors include a wide range of subject behaviors and contacts, often open for inclusion of any correlatable information.

It is not hard to see that \textbf{the surveillance tasks and the DRM tasks are two sides of the same coin}: the former strives to establish a data flow, the latter to prevent data from flowing. A DRM copy protection may be used to prevent surveillance; surveillance techniques may be used to break or disable a DRM protection. 

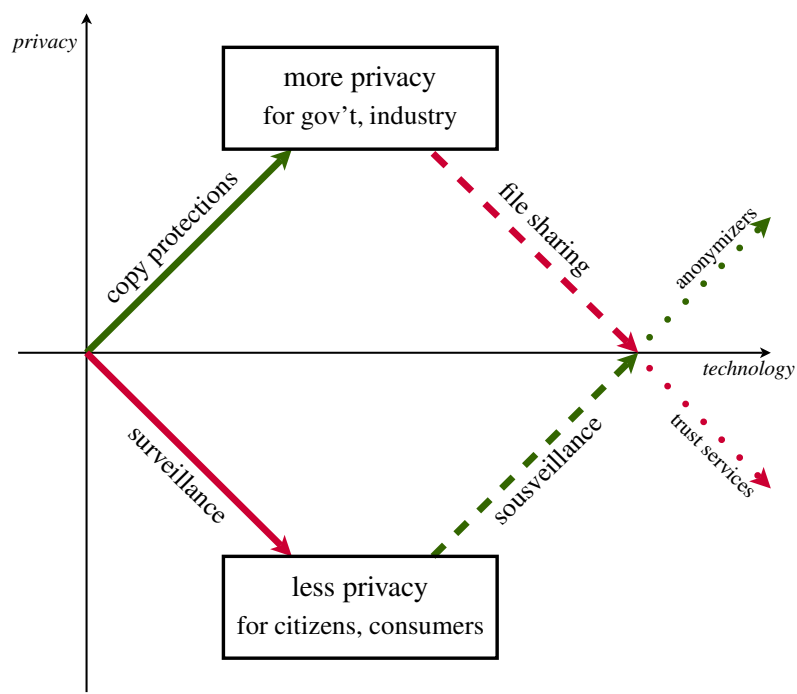
\begin{figure}[!ht]
\begin{center}
\newcommand{\priv}{\scriptsize \it privacy}
\newcommand{\tech}{\scriptsize \it technology}
\newcommand{\more}{\begin{minipage}[c]{3cm}\centering{\small more privacy}\\{\footnotesize for gov't, industry}\end{minipage}}
\newcommand{\less}{\begin{minipage}[c]{3.3cm}\centering{\small less privacy}\\{\footnotesize for citizens, consumers}\end{minipage}}
\newcommand{\surveillance}{\footnotesize surveillance}
\newcommand{\sousveillance}{\footnotesize sousveillance}
\newcommand{\DRM}{\footnotesize copy protections}
\newcommand{\filesharing}{\footnotesize file sharing}
\newcommand{\FISA}{}
\newcommand{\arabspring}{}
\newcommand{\Doubleclick}{}
\newcommand{\cyberanarchy}{}
\newcommand{\anon}{\scriptsize anonymizers}
\newcommand{\trust}{\scriptsize trust services}
\def\JPicScale{.9}
\ifx\JPicScale\undefined\def\JPicScale{1}\fi
\psset{unit=\JPicScale mm}
\psset{linewidth=0.3,dotsep=1,hatchwidth=0.3,hatchsep=1.5,shadowsize=1,dimen=middle}
\psset{dotsize=0.7 2.5,dotscale=1 1,fillcolor=black}
\psset{arrowsize=1 2,arrowlength=1,arrowinset=0.25,tbarsize=0.7 5,bracketlength=0.15,rbracketlength=0.15}
\begin{pspicture}(0,0)(110,100)
\psline{->}(0,50)(110,50)
\psline{->}(10,0)(10,100)
\rput[t](106.88,48.75){\tech}
\rput[r](8.75,95.62){\priv}
\newrgbcolor{userLineColour}{0.2 0.4 0}
\psline[linewidth=1,linecolor=userLineColour]{->}(10,50)(40,80)
\newrgbcolor{userLineColour}{0.8 0 0.2}
\psline[linewidth=1,linecolor=userLineColour]{->}(10,50)(40,20)
\newrgbcolor{userLineColour}{0.2 0.4 0}
\psline[linewidth=1,linecolor=userLineColour,linestyle=dashed,dash=2.5 2.5]{->}(60.62,20)(90.62,50)
\newrgbcolor{userLineColour}{0.8 0 0.2}
\psline[linewidth=1,linecolor=userLineColour,linestyle=dashed,dash=2.5 2.5]{->}(60.62,79.38)(90.62,50)
\rput[t]{-45}(24.38,33.12){\surveillance}
\rput[b]{45}(23.75,65.62){\DRM}
\rput[t]{45}(76.25,33.75){\sousveillance}
\rput[b]{-45}(75.62,66.88){\filesharing}
\pspolygon[linewidth=0.5](30,95)(70,95)(70,80)(30,80)
\pspolygon[linewidth=0.5](30,20)(70,20)(70,5)(30,5)
\rput(50,87.5){\more}
\rput(50,12.5){\less}
\rput[t]{-45}(73.12,65.62){\arabspring}
\rput[b]{45}(73.12,33.75){\cyberanarchy}
\rput[b]{-45}(26.88,35){\Doubleclick}
\rput[t]{45}(26.88,65){\FISA}
\newrgbcolor{userLineColour}{0.2 0.4 0}
\psline[linewidth=1,linecolor=userLineColour,linestyle=dotted,dotsep=3]{->}(90,50)(110,70)
\newrgbcolor{userLineColour}{0.8 0 0.2}
\psline[linewidth=1,linecolor=userLineColour,linestyle=dotted,dotsep=3]{->}(90,50)(110,30)
\rput[b]{45}(103.12,64.38){\anon}
\rput[t]{-45}(102.5,35.62){\trust}
\end{pspicture}
\caption{Technology investments favor the profitable side of data privacy}
\label{Fig:data-tech}
\end{center}
\end{figure}

The tasks of securing data privacy on one hand and of the intellectual property on the other correspond to \textbf{the same security problem}: to control the data flows in digital networks. However, the technologies developed for surveillance on the one hand and for the DRM on the other hand lead to \textbf{the opposite solutions}: they weaken privacy and strengthen the intellectual property. Fig.~\ref{Fig:data-tech} illustrates how the technical advances support privacy protections of intellectual property and commercial digital assets, but help break the privacy protections of individual behaviors and personal digital assets.

\section[Data privacy]{Data privacy and pull attacks} \label{Sec:datapriv}

\newcommand{\prx}[2]{\raisebox{.18ex}{\scriptsize |}\!\frac{\ #1\ }{\ #2\ }\!\raisebox{.18ex}{\scriptsize |}}

\subsection{Defining data privacy} 
Data privacy defines the boundaries between the private sphere and the public sphere and provides the legal right to be left alone with your private data. This right intends to restrict the access to the private data, just like secrecy intends to restrict the access to secret information. But while secrecy is concerned with access to data through a specific channels, the right to privacy is not specific to a particular channel, it legally protects from access to the private data through any global channel. 

Secrecy is formally defined in cryptography. The earliest definition, due to Shannon~\cite{ShannonC:Secrecy}, says that it is a property of a channel where the outputs are statistically independent of the inputs. It is tempting, and seems natural, to define privacy in a similar way. This was proposed by Dalenius back in the 1970s~\cite{DaleniusT:desideratum}: A database is private if the public data that it discloses publicly say nothing about the private data that it does not disclose. This \emph{desideratum}, as Dalenius called it, persisted in research for a number of years, before it became clear that it was generally impossible as a requirement when considering access through any global channel. For example, if everybody knows that Alice eats a lot of chocolate, but there is an anonymized database that shows a statistical correlation between eating a lot of chocolate and heart attacks, then this database discloses that Alice may be at a risk of heart attack, which should be Alice's private information, and thus breaches Dalenius' desideratum. Notably, this database breaches Alice's privacy \emph{even if}\/ Alice's record does not come about in it. Indeed, it is not necessary that Alice occurs in the database either for establishing the correlation between chocolate and heart attack, or for the public knowledge that Alice eats lots of chocolate; the two pieces of information can arise independently. Alice's privacy can be breached by linking two completely independent pieces of information, one about Alice and chocolate, the other one about chocolate and heart attack. But since Alice's record does not come about in the database, it cannot be removed from it, or anonymized in it. The public knowledge of background information can establish covert channels that require the definition for data privacy to be different from the definition of secrecy.

Alice's privacy is her global right. But the definition of privacy of her data in a database cannot have global requirements and depend on the (non)-availability of background information and covert channels that are in no control of the database. It has to depend on the records in the database that contain Alice's information and how they are disclosed.
Limiting the focus to what private information can be learned just from Alice's record in the database leads to the more practical measure of {\em differential privacy}. 
In differential privacy, it is required that all sensitive data about Alice that can be learned from a database $D$ with Alice's record, could also be learned from the database $D'$ where Alice's record is replaced by an alternative record. This sounds very close to Dalenius' definition, 
which says that all sensitive data about Alice that can be learned from a database $D$ with Alice's record, could also be learned without access to the database $D$, 
but on a closer look it is very different. The definition of differential privacy addresses the impossibility of privacy according to Dalenius' definition and removes the issue of background information. 

\para{Data privacy in databases.} The life of any data set consists of first gathering the data (veillance), then storing and releasing the data, and finally processing the data.
Databases are the central tool for the management of large collections of data. 
Conceptually, 
databases organize data as large matrices, called {\em tables}, 
storing values of {\em attributes} organized in {\em records}. 
\begin{definition}
Given the sets $\RRR$ of \emph{records}, $\AAA$ of \emph{attributes}, and $V_a$ of \emph{values} for each $a\in \AAA$, 
a database is a matrix
\bear
D &:& \RRR\times \AAA\  \to V
\eear
where $V = \bigcup_{a\in \AAA} V_a$ and $D(r,a) \in V_a$ for all $r\in \RRR$ and $a\in \AAA$.
\end{definition} \label{Def:database}
The rows of the matrix represent the tuples of data in a record of the database, and the columns of the matrix represent the attributes, i.e., the data for an attribute for each record.

If an attribute has a unique value for each record, this attribute is considered to be an {\em identifier}.

\begin{definition}
    An \emph{identifier} (ID)  is an attribute $a \in \AAA$ that uniquely determines all entities.
    
    More precisely, there is a function $f:V_a \to \EEE$ such that for all $e\in \EEE$ holds
    \bear
    f(D_{R(e)}^a) & = & e
    \eear
    where $D_{R(e)}^a$ is value for the attribute $a$ that occurs in a record $R(e)$ in the database $\DDD$.
\end{definition}

Attributes of private data, like health conditions, educational records, or financial information, that are particularly sensitive regarding the protection of their access, are called \textbf{\em sensitive attributes}.

\begin{figure}[htbp]
    \begin{center}
     \includegraphics[height=5.5cm]
                        {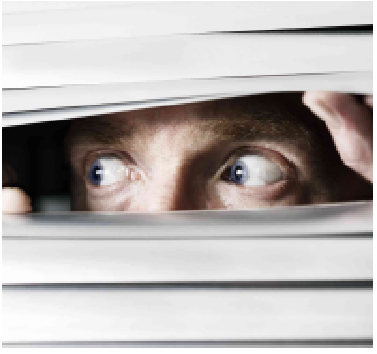}
\hspace{2em}
        \includegraphics[height=5.5cm]
                        {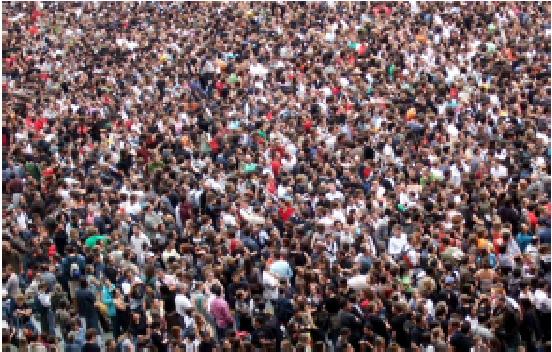}
        \caption{Privacy can be achieved by being hard to see or hard to find}
        \label{Fig:Crowd}
    \end{center}
\end{figure}

We described some instances of data gathering in section~\ref{Sec:veillance}, in particular in the context of public data. The fact that the gathered data in 
these cases has been public does not imply for every stored datum, possibly containing sensitive attributes, to be publicly standing out. 
One can be completely private and anonymous if lost in a crowd as suggested in~\cref{Fig:Crowd}, even if the data about the crowd is public and contains sensitive information. 
The surveillance of this data alone does not kill the privacy of any individual's data in the crowd. It is the search in the collected data that can kill the privacy. 
The simple look at the image does not highlight any private data of any subject, but if there is sufficient reason to spend any effort 
on processing the image, private data can be identified and extracted. The controlled access to the records of the database through its interface can 
restrict the searches and help protect the privacy of the data.

\begin{figure}[t!hbp]
    \begin{center}    
\begin{tabular}{|c||c|c|c||c|}
    \hline
    $ID$      & \multicolumn{3}{c ||}{$QID$} & $SA$ \\ \hline
    $Name$    & $Zipcode$ & $Age$ & $Sex$  & $Disease$ \\ \hline \hline 
    $Alice$   & $47677$   & $29$  & $F$    & \textit{ovarian cancer}  \\ \hline
    $Betty$   & $47602$   & $22$  & $F$    & \textit{ovarian cancer}  \\ \hline
    $Charles$ & $47678$   & $27$  & $M$    & \textit{prostate cancer} \\ \hline
    $David$   & $47905$   & $43$  & $M$    & \textit{flu}             \\ \hline
    $Emily$   & $47909$   & $52$  & $F$    & \textit{heart disease}   \\ \hline
    $Fred$    & $47906$   & $47$  & $M$    & \textit{heart disease}   \\ \hline
\end{tabular}
\caption{Data contain identifiers (ID) and sensitive attributes (SA).}
        \label{Fig:database}
    \end{center}
\end{figure}
Medical databases contain sensitive attributes, like e.g., health conditions or medication records, that are privately owned before they are brought into the database.
Data are assumed to be secured while stored, their release makes them public, and the concern of privacy and anonymity arises with the controlled access and the release of any record. 
The database in~\cref{Fig:database} contains the identifier $Name$ and the sensitive attribute $Disease$. Neither Alice nor Betty would like the information that they have been diagnosed with cancer to be publicly known,
they would want to protect access to this private data. 

As a simple solution to support the restriction of access to sensitive private data, access could simply be prohibited to anyone other than the owners. This simple solution is not very practical, because direct access to the private data in not just needed by the owners, but is also important for health care professionals to have to do their work.  

Accumulative and anonymized access is important for the public and for research to improve general understanding of health conditions and the spread and causes of diseases.
In many cases, the sensitive attributes only need to be presented in accumulated form for statistical analysis where they can be detached from the identifiers of the records.
Individuals would not be willing to have their sensitive data released if their identity can be identified from the released record. This is where 
the privacy protection of statistical databases come in.

\begin{definition}
    Data are collected from a set of \emph{entities} $\EEE$.
\emph{Data gathering} is a map $R:\EEE\to \RRR$, so that $D_{R(e)}$ is the tuple of the data corresponding to the entity $e\in \EEE$.
\emph{Data identification} is a map $E:\RRR\to \EEE$, such that $E(R(e))=e$.
\end{definition}

Simply removing the identifier from a record before release removes the data identification by the means of the identifier, but it does not detach the record from the entity.
The removal of the identifier of the name attribute in the example in Fig.~\ref{Fig:database} leaves the attributes shown in~\cref{Fig:database_anon}. In each record the sensitive attributes are only linked to non-sensitive attributes that do not reveal the identity of the record by themselves - the records seem to be anonymized. 
\begin{figure}[t!hbp]
    \begin{center}    
    \begin{tabular}{|c|c|c||c|}
        \hline
         \multicolumn{3}{| c ||}{\textit{QID}} & \textit{SA} \\ \hline
         \textit{Zip code} & \textit{Age} & \textit{Sex}  & \textit{Disease} \\ \hline \hline 
         $47677$   & $29$  & $F$    & \textit{ovarian cancer}  \\ \hline
         $47602$   & $22$  & $F$    & \textit{ovarian cancer}  \\ \hline
         $47678$   & $27$  & $M$    & \textit{prostate cancer} \\ \hline
         $47905$   & $43$  & $M$    & \textit{flu}             \\ \hline
         $47909$   & $52$  & $F$    & \textit{heart disease}   \\ \hline
         $47906$   & $47$  & $M$    & \textit{heart disease}   \\ \hline
    \end{tabular}
    \caption{Attempt at anonymizing records by removing identifiers (ID).}
    \label{Fig:database_anon}                
    \end{center}
\end{figure}

Cross-referencing the records in the anonymized database against records in public databases, like a voter register shown in Fig.~\ref{Fig:reident}, allows for re-identification of the records and link the sensitive attribute entries back with the individuals' identities.

\begin{figure}[htbp]
    \begin{center}   
    \small
    \begin{minipage}{.45\linewidth}
    \centering
    \begin{tabular}{|c|c|c||c|} \hline
         \multicolumn{3}{| c ||}{\textit{QID}} & \textit{SA} \\ \hline
         \textit{Zip code} & \textit{Age} & \textit{Sex}  & \textit{Disease} \\ \hline \hline 
         \colorbox{red!30}{\bf $47677$}  & \colorbox{red!30}{\bf $29$}  & \colorbox{red!30}{\bf $F$}    & \textit{ovarian cancer}  \\ \hline
         $47602$   & $22$  & $F$    & \textit{ovarian cancer}  \\ \hline
         $47678$   & $27$  & $M$    & \textit{prostate cancer} \\ \hline
         $47905$   & $43$  & $M$    & \textit{flu}             \\ \hline
         $47909$   & $52$  & $F$    & \textit{heart disease}   \\ \hline
         $47906$   & $47$  & $M$    & \textit{heart disease}   \\ \hline
    \end{tabular}
    \end{minipage}\hspace{.05\linewidth}~\begin{minipage}{.45\linewidth}
    \centering
    \begin{tabular}{|c||c|c|c|} \hline
         \textit{Name} & \textit{Zip code} & \textit{Age} & \textit{Sex}  \\ \hline \hline 
         \textit{Alice} & \colorbox{red!30}{\bf $47677$}  & \colorbox{red!30}{\bf $29$}  & \colorbox{red!30}{\bf $F$}  \\ \hline
         \textit{Bob}   & $47983$   & $65$  & $M$    \\ \hline
         \textit{Carol} & $47677$   & $22$  & $F$    \\ \hline
         \textit{Dan}   & $47532$   & $23$  & $M$    \\ \hline
         \textit{Ellen} & $46789$   & $43$  & $F$    \\ \hline
    \end{tabular}
    \end{minipage}
    \caption{Medical database linked with Voter Register.}
    \label{Fig:reident}                
    \end{center}
\end{figure}
\normalsize
Such examples are real. In the 90s, a case involving the governor of Massachusetts got highly publicized and was subsequently used to drive privacy policy in the health care industry. 
After the governor collapsed at a public event, a graduate student demonstrated that linking the anonymized data of health records that his health insurance generally made available for research studies, with the anonymized data that could be bought from the voter register, allowed the re-identification of the governor's health record related to the public collapse and exposed the governor's health condition. 
The anonymized records of the health insurance had the obvious personal identifiers removed, but the published data included zip code, date of birth, and gender. 
When matched with the voter records, there were only six possible patient records that could belong to the governor. Only three matched the gender, and only one shared the zip code. 
Those attributes were sufficient to identify the governor uniquely.
This case got later published in~\cite{Sweeney97}. Interestingly, the names and other details of the case from public media got changed in this publication, as listed in~\cref{Fig:reident_real}, to protect the privacy of the involved entities.

\begin{figure}[htbp]
    \begin{center}       
Medical Data Released as Anonymous
\small
\begin{tabular}{|c|c|c|c|c|c|c|c|}
        \hline
        \textit{SSN} & \textit{Name} & \textit{Race} & \textit{Date of Birth} & \textit{Gender} & \textit{ZIP} & \textit{Martial Status} & \textit{Problem} \\ \hline \hline
        & & \textit{asian} & \textit{09/27/64} & \textit{female} & \textit{02139} & \textit{divorced} & \textit{hypertension} \\ \hline
        & & \textit{asian} & \textit{09/30/64} & \textit{female} & \textit{02139} & \textit{divorced} & \textit{obesity} \\ \hline
        & & \textit{asian} & \textit{04/18/64} & \textit{male}   & \textit{02139} & \textit{married}  & \textit{chest pain} \\ \hline
        & & \textit{asian} & \textit{04/15/64} & \textit{male}   & \textit{02139} & \textit{married}  & \textit{obesity} \\ \hline
        & & \textit{black} & \textit{03/13/63} & \textit{male}   & \textit{02138} & \textit{married}  & \textit{hypertension} \\ \hline
        & & \textit{black} & \textit{03/18/63} & \textit{male}   & \textit{02138} & \textit{married}  & \textit{shortness of breath} \\ \hline
        & & \textit{black} & \textit{09/13/64} & \textit{female} & \textit{02141} & \textit{married}  & \textit{shortness of breath} \\ \hline
        & & \textit{black} & \textit{09/07/64} & \textit{female} & \textit{02141} & \textit{married}  & \textit{obesity} \\ \hline
        & & \textit{white} & \textit{05/14/61} & \textit{male}   & \textit{02138} & \textit{single}   & \textit{chest pain} \\ \hline
        & & \textit{white} & \textit{05/08/61} & \textit{male}   & \textit{02138} & \textit{single}   & \textit{obesity} \\ \hline
        & & \textit{white} & \textit{09/15/61} & \textit{female} & \textit{02142} & \textit{widow}    & \textit{shortness of breath} \\ \hline
    \end{tabular}
\vspace*{0.5cm}

    Voter List
    \begin{tabular}{|c|c|c|c|c|c|c|c|}
        \hline
        \textit{Name} & \textit{Address} & \textit{City} & \textit{ZIP} & \textit{DOB} & \textit{Gender} & \textit{Party} & \textit{$\ldots\ldots$} \\ \hline \hline
        \textit{$\ldots\ldots$} & \textit{$\ldots\ldots$} & \textit{$\ldots\ldots$} & \textit{$\ldots\ldots$} & \textit{$\ldots\ldots$} & \textit{$\ldots\ldots$} & \textit{$\ldots\ldots$} & \textit{$\ldots\ldots$} \\ \hline  
        \textit{$\ldots\ldots$} & \textit{$\ldots\ldots$} & \textit{$\ldots\ldots$} & \textit{$\ldots\ldots$} & \textit{$\ldots\ldots$} & \textit{$\ldots\ldots$} & \textit{$\ldots\ldots$} & \textit{$\ldots\ldots$} \\ \hline
        \textit{Sue J. Carlson} & \textit{1459 Main St.} & \textit{Cambridge} & \textit{02142} & \textit{09/15/61} & \textit{female} & \textit{democrat} & \textit{$\ldots\ldots$} \\ \hline
        \textit{$\ldots\ldots$} & \textit{$\ldots\ldots$} & \textit{$\ldots\ldots$} & \textit{$\ldots\ldots$} & \textit{$\ldots\ldots$} & \textit{$\ldots\ldots$} & \textit{$\ldots\ldots$} & \textit{$\ldots\ldots$} \\ \hline
    \end{tabular}
    \normalsize
    \caption{Medical record of the Governor of Massachusetts identified.}
        \label{Fig:reident_real}                
    \end{center}
\end{figure}

There are several other similar cases that have been publicized over the years, but this specific case has been influential in determining the criteria for the 2003 HIPAA act 
and in defining the related criterion of $k$-anonymity, that has been used to advance privacy protection in statistical databases. In the identification of the governor, the attributes of gender, birthdate and ZIP code were not sensitive by themselves. Their critical feature was that the specific value combination of these attributes for the governor were unique in both databases. This allowed them to act - in combination - as an identifier. We call such sets of attributes a {\em quasi-identifier}.
\begin{definition}
    A \emph{quasi-identifier} (QID)  is a {\color{red}set of} attributes $\QQQ\subseteq \AAA$ that uniquely determine some entities.
    
    More precisely, there is a partial function $f:\prod_{i\in \QQQ}V_i \pto \EEE$ such that for some $e\in \EEE$ holds
    \bear
    f(D_{R(e)}^\QQQ) & = & e
    \eear
    where $D_{R(e)}^\QQQ$ is a $\QQQ$-tuple of attributes in the database $\DDD$.
\end{definition}

In figure~\ref{Fig:reident}, the set of attributes $\QQQ = \{\textit{ZIP}, \textit{Age}, \textit{Sex} \} $ is a quasi-identifier because 
$$f(D_{R(\textit{Alice})}^{\{\textit{ZIP}, \textit{Age}, \textit{Sex}\}}) \ = \ \textit{Alice},$$ 
where $D_{R(\textit{Alice})}^{\{\textit{ZIP}, \textit{Age}, \textit{Sex}\}} \ = \ (\textit{47677, 29, F})$ and the function $f$ is determined from the voter register uniquely to be $f(\textit{47677, 29, F}) \ = \ \textit{Alice}.$
In figure~\ref{Fig:reident_real}, the set of attributes $\QQQ = \{\textit{Gender}, \textit{DOB}, \textit{ZIP}\} $ is a quasi-identifier because 
$$f(D_{R(\textit{Sue J. Carlson})}^{\{\textit{Gender}, \textit{DOB}, \textit{ZIP}\}}) \ = \ \textit{Sue J. Carlson},$$ 
where uniquely $D_{R(\textit{Sue J. Carlson})}^{\{\textit{Gender}, \textit{DOB}, \textit{ZIP}\}} \ = \ (\textit{female, 09/15/61, 02142})$ and the function $f$ is determined from the voter register to be uniquely $f(\textit{female, 09/15/61, 02142}) \ = \ \textit{Sue J. Carlson}.$

\subsection{$K$-Anonymity}
The measure of {\em $k$-anonymity}\sindex{k-anonymity}~\cite{SweeneyL:kanonymity} has been introduced to make re-identification more difficult. With $k$-anonymity, a database ensures that the values of the quasi-identifier attributes of any record cannot be distinguished among at least $k$ records; 
the larger the $k$, the larger the set of records among which one entity can be hiding. This measure helps protect privacy when the focus is one individual query. 

\begin{definition}
    A database $D$ satisfies the \emph{$k$-anonymity requirement} if for every quasi-identifier $\QQQ$ and every $\QQQ$-tuple of values $x \ \in \ V^\QQQ = \prod_{i\in \QQQ}V_i$, there are either at least $k$ records with the same value $x$, or 
no such records exist in $D$. Formally, the requirements 
are $\forall \QQQ \subseteq \AAA, \forall x \in V^\QQQ$:
        $$ \# \RRR_x^\QQQ \ = \ \# \{ r \in \RRR \ | \ D_r^\QQQ = x \} \ \ \ \ge \ k, \mbox{ \ \ \ \ \ \ or \ \ \ \ \ \ }        
                    x \ \notin \ D^\QQQ.$$
\end{definition}
This measure needs to consider any $\QQQ$-tuple of attributes as potential quasi-identifiers.

\para{Techniques.}
The two techniques to achieve $k$-anonymity are generalization and suppression~\cite{SweeneyL:suppression}.

In \textit{generalization} some details of some attributes are removed to make their values more common, and less identifiable. 

The choice which attribute values are used to generalize to larger ranges has an impact on the information that each query to the database can provide. In their use to provide averages and statistical information about the attribute values of the records, it is important to chose generalizations that limit the bias that they have on the averages and cumulative information that is processed for any queries.

\textit{Suppression} removes records that are considered outliers that cannot be easily generalized with other records.

\para{Example.}
Figure~\ref{Fig:kanonex} illustrates the records of a database that, before anonymization has a $QID$ of the attributes $<ZIP, car, child>$, where the value 
\textit{$<$96822, Subaru Outback 1999, 8 year old$>$} only occurs in one record of the database as shown in the figure on the left. But generalizing the values of the attribute $ZIP$ to the first two digits, the attribute $car$ to just the brand of the car, and the attribute of the age of the child to just the information whether the child is a minor or not, the same anonymized values for the previous QID attributes for the entity occur in $k$ records now, as shown in the figure on the right. The record of the entity that could previously be uniquely identified, now shares its value for these attributes with $k-1$ other records. The database is $k$-anonymized if there are at least $k$ records not just for this specific $\QQQ$-tuple value for this specific $\QQQ$, but if there are at least $k$ records in the database for any $\QQQ$-tuple value for any subset of attributes $\QQQ$. 

\begin{figure}[t!hb]
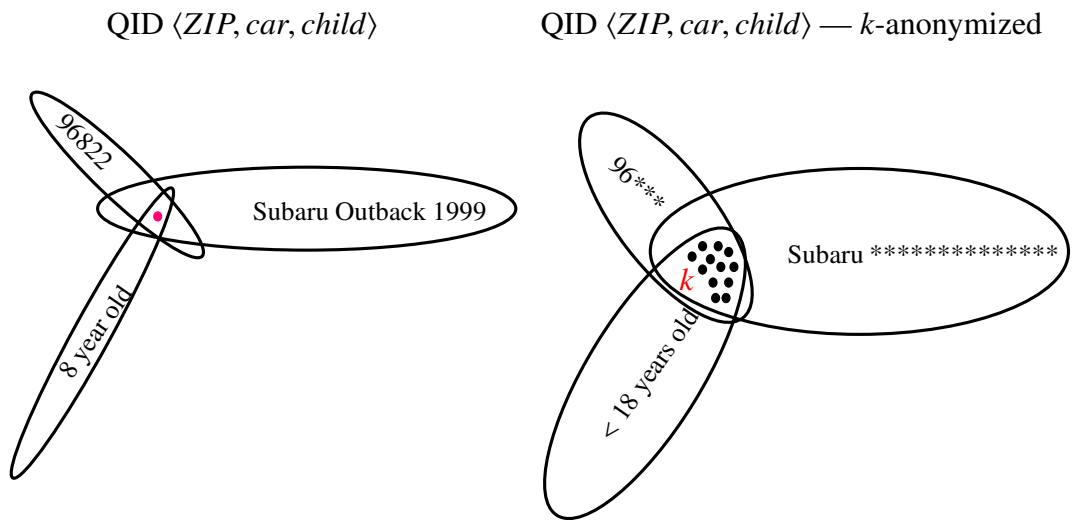

\begin{minipage}[t]{0.45\linewidth}
\centering
QID $<ZIP, car, child>$
\newcommand{\ZIP}{\footnotesize 96822}
\newcommand{\car}{\footnotesize Subaru Outback 1999}
\newcommand{\child}{\footnotesize 8 year old}
\begin{center}
\def\JPicScale{.55}
\input{QID-unique.tex}
\end{center}
\end{minipage}
    %
\begin{minipage}[t]{0.45\linewidth}
\centering
QID $<ZIP, car, child>$ --- $k$-anonymized
\newcommand{\ZIP}{\footnotesize 96***}
\newcommand{\car}{\footnotesize Subaru **************}
\newcommand{\child}{\footnotesize $\lt 18$ years old}
\newcommand{\kthem}{\bf \color{red} $k$}
\begin{center}
\def\JPicScale{.55}
\input{QID-k-anon.tex}
\end{center}
\end{minipage}
\caption{Database $k$-anonymization by generalization of $QID$-attributes.}
\label{Fig:kanonex}
\end{figure}

\para{Limitations of $k$-anonymity.}
The utility of using $k$-anonymity to protect privacy in databases has its limits:

\textbf{Lack of diversity:} Even if a database provides $k$-anonymity by the appropriate application of generalization and suppression, an entity could still be associated with their private information, e.g., the sensitive attribute of their health condition if the database lacks diversity and the same SA value occurs in more than $k$ records. Then, $k$-anonymity does not conceal the value. It reveals the SA not just for the entity $e$, but for all the more than $k$ entities with the same SA value 
$xs \ = \ D_{R(e)}^{SA}$, the same QID tuple $qs = D_{R(e)}^\QQQ$ and with records $\RRR_{qs}^\QQQ$.
In this case the database would be $k$-anonymous, but still discloses the SA of a group of more than $k$ entities.

\textbf{Background information:} General anonymized data may also disclose an individual SA value when combining the data from the database with any background information about an individual.

The data relating smoking and cancer from a database $D$ can be used together with the public knowledge that Bob smokes, and, in effect, link Bob with the SA of having a cancer risk --- even if Bob does not occur in $D$.
The background information that is used to reveal Bob's private information in this case is a \textbf{false problem}, because its release is not under the control of the database. No anonymization of the database $D$ can eliminate the information available outside of $D$. Therefore, it must be acceptable that a database $D$ may disclose some sensitive information about me to those who know me --- even if I do not occur in $D$.

\para{Example.} Alice writes an exam at school. Her teacher grades the exams and returns them to the students. He also gives them a summary of the results and tells them, among others, that $3$ students failed the exam. If the records for each student in a database $D$ contain their grade for this exam, the teacher's summary does not release the grades for any student but just provides the information:  
\bear \# \RRR^{grade}_{fail} = \#\{ r\in \RRR \ |\  D_r^{grade} = \mbox{fail}\} \ \  =  \ \  3
\eear
Charlie likes to harass students, and asks all of his friends 
how they did in the exam. Alice does not want to tell Charly about her exam, she wants to keep her exam results private. After Charly finds out the grades of all other students in the class, and that only two of them had failed the exam, he concludes, that Alice must have failed the exam, and harasses her about it. 
Through his collection of information, Charly builds the database $D'$ that contains all records of $D$ except for Alice's. For this database $D'$, Charlie determines that   
\bear 
\#\{ r\in \RRR \ |\  {D'}_r^{grade} = \mbox{fail}\} \ \  =  \ \  2.
\eear
and can conclude that $D\setminus D' = R(Alice)$ and $D^{grade}_R(Alice) \ = \ fail$.
The release of the summary of the grades by the teacher did not preserve the privacy of Alice's grade on the exam, although Alice's result had been made $3$-anonymous in the disclosure of results from database $D$ by the teacher.
If the teacher had made her statement less precise and told the students instead that either $2$ or $3$ students had failed the exam, the privacy of Alice's results would have been preserved.
The question remains how the release of aggregate data that has been processed by tools like mining or classification affect the privacy and anonymity of sensitive data.
This issue occurs and has been studied in a refined form in the context of statistical databases and has lead to the notion of k-privacy. 

\para{Statistical databases.}
Statistical databases are collecting data for statistical analysis and reporting. 
Their purpose is to provide cumulative and statistical information of public interest and for research while protecting the privacy of the individual records.
Their focus is on the classification and analysis of the dataset as a whole, they are not interested in the details of a specific datum. To be sharing sensitive data with the database, participants need to be convinced that their record cannot be identified, and any part of their sensitive data cannot be reconstructed from the data provided in the disclosures of the database.

\subsection{Differential Privacy}

\textit{Differential privacy}~\cite{Dwork06} is a requirement on the disclosure algorithm $F$ of 
the database $D$, and not a requirement on the data in the database itself.
It implements the indistinguishability of databases $D$ and $D'$, where database $D$ differs from database $D'$ just in one individual records, in terms of an equivalence kernel. 
Differential privacy requires that the flow leakage of individual information from any single record is negligible when releasing a statistical result from the database $D$.
Intuitively, the risk to one's privacy incurred by participating in a database
is expressed by the parameter of the privacy budget $\varepsilon$ in the model.
The techniques that have been developed to achieve differential privacy can be controlled to achieve an arbitrary level of privacy under this measure.

\begin{definition}\label{Def:diffpriv}
Let $\DDD$ be a family of databases, $\PPP \subseteq \sum_{a\in \AAA}V_a$, a family of properties (viewed as sets of values in some attributes), and $\varepsilon\gt 0$ a real number, called the privacy budget.

A disclosure algorithm $F:\DDD\to \PPP$ is \emph{$\varepsilon$-differentially private} if for every property $Y\in \PPP$ holds
\bear
1 \ \ \ge \ \ \prx{\Pr(F(x)\in Y)}{\Pr(F(x')\in Y)} \ \  \ge \ \ e^{-\varepsilon}
\eear
for any pair of databases $x,x'\in \DDD$ which differ in at most one record, and where the normalized ratio is denoted as
$$\prx{x}{y} \  = \left\{ \begin{array}{ll} x/y & \mbox{if $x \le y$,} \\
                          y/x & \mbox{otherwise.} \end{array} \right. $$
\end{definition}


\para{Implementation.}
Differential privacy is a property of a database disclosure algorithm $F$. 
To make a disclosure algorithm meet the different privacy requirement, the data in the database is not getting affected at all, so that 
no bias is added to the recorded data. Instead, differential privacy can be achieved by perturbing the disclosure 
algorithm in different phases of the process.
In principle, one could perturb the disclosure algorithm 
\begin{itemize}
\item at the inputs of the queries,
\item at the calculation of intermediate values for the query, or
\item at the output values of the disclosure.
\end{itemize}
In practice, the perturbation of the output values is the most common approach by adding noise to the output values. 
In this way, the statistical properties of the added noise and the relation of the perturbed values to the unperturbed values can be best controlled to meet the differential privacy requirements from definition~\ref{Def:diffpriv} for a specific privacy budget $\varepsilon$.

\para{Output Perturbation Method.}
The standard method to achieve differential privacy is the controlled 
addition of Laplace noise to generate the outputs of the disclosure algorithm~\cite{dwork2014}.
The amount of noise is adjusted with respect to the parameter $\varepsilon$
to achieve $\varepsilon$-differential privacy and trade-off the accuracy of 
the disclosed results versus the risk of private individual data to be leaked. 
The parameter choices are guided based on the following properties.

To implement differential privacy and guarantee a bound on the ratio of the probabilities between $F(x)$ and to $F(x')$, the density function of the added noise needs to have a bound on its ratios in shifted versions. This is why Laplace noise is chosen. The Laplace density function for various parameters $\lambda$ is shown in Fig.~\ref{Fig:laplace_prop}. Two versions of the density functions are shown that are shifted by $\Delta f_{x,x'} \ = \ \|f(x) - f(x')\|$ to represent the densities of perturbed results disclosed by $F(x)$ and $F(x')$ from databases $x, x' \ \in \ \DDD$ - this is the situation we are in when adding Laplace noise to the unperturbed results of a feasible disclosure algorithm $f(x)$ and $f(x')$.

A bound on the ratio of the two densities for a disclosed value $y \in Y \in \PPP$ that could have been 
generated from $y = F(x)$ or $y = F(x')$ then allows to determine and control the Laplace 
parameters to achieve $\varepsilon$-privacy for the disclosure algorithm $F(x)$.

\begin{figure}[t!hbp]
    \begin{center}
        \scalebox{0.2}{ 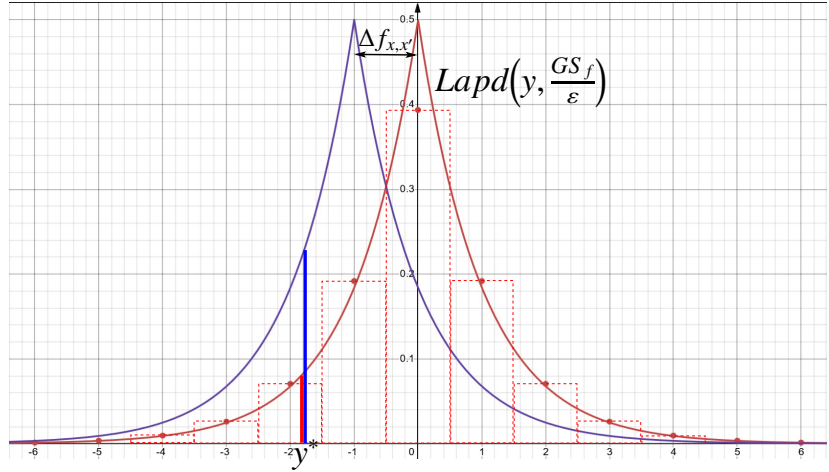}
        \caption{Laplace density functions shifted by $\Delta f_{x, x'}$}
        \label{Fig:laplace_prop}                
    \end{center}
\end{figure}

\begin{theorem}\label{Thm:DPreal}
Let $f:\DDD\to \PPP$ be a feasible disclosure algorithm. Then 
\bear
F(x) & = & f(x) + Lap\left(\frac{GS_f}{\varepsilon}\right)
\eear
is $\varepsilon$-differentially private, where
$GS_f = \bigwedge_{x,x'} \| f(x) - f(x') \|$ is the \emph{global sensitivity}, and 
$Lap(\lambda)$ describes randomly sampled noise with density function $Lapd(y, \lambda) \ = \ \frac{1}{2\lambda} exp\Large( - \frac{|y|}{\lambda} \Large)$, and 
where databases $x,x' \ \in \ \DDD$ differ in at most one record.
\end{theorem}
\begin{proof}
To get $\varepsilon$-differential privacy, we need to compare the disclosures of databases $x$ and $x'$, where the two databases are different in just one record $r$. The global sensitivity $GS_f$ expresses a bound of how much the unperturbed 
disclosure $f(x)$ can differ from $f(x')$. With the difference denoted by $\Delta f_{x,x'} = f(x) - f(x')$, we can write 
$$| \Delta f_{x,x'} |  \ = \ |f(x) - f(x') |\ \le \ GS_f.$$  
When an arbitrary value $y' \in Y$ is disclosed, this value could have been generated as
$y' = F(x) = f(x) + N^x_{Lap}$   
or as
$y' = F(x') = f(x') + N^{x'}_{Lap}$, where $N^x_{Lap}$ and $N^{x'}_{Lap}$ are continuous random variables representing the Laplace noise sampled from $Lap\left(\frac{GS_f}{\varepsilon}\right)$.  
For $F(x) =  F(x')$, we have $N^x_{Lap} = N^{x'}_{Lap} + \Delta f_{x,x'}$, and 
if we set $y^* = F(x) - f(x)$, the ratio of the probabilities can be expressed as
\bear
\Pr\Big(F(x)=y'\Big)/\Pr\Big(F(x')=y'\Big)  & = & \Pr\Big(N^x_{Lap}=y^*\Big)/\Pr\Big(N^{x'}_{Lap} = (y^* + \Delta f_{x,x'})\Big) \\
& = & Lapd\Big(y^*, \frac{GS_f}{\varepsilon}\Big)/Lapd\Big(y^* + \Delta f_{x,x'}, \frac{GS_f}{\varepsilon}\Big).
\eear
Note that $y^*$ is shown in Fig.~\ref{Fig:laplace_prop}, $Lapd(y^*, \frac{GS_f}{\varepsilon})$ is depicted by the red bar and $Lapd(y^* + \Delta f_{x,x'}, \frac{GS_f}{\varepsilon})$ is depicted by the blue bar in the figure.
To get to the requirement for differential privacy from Def.~\ref{Def:diffpriv}, we consider the normalized ratio:
\bear
    \prx{\Pr(F(x)= y')}{\Pr(F(x')=y')}  & = & \prx{Lapd(y^*, \frac{GS_f}{\varepsilon})}{Lapd(y^* + \Delta f_{x,x'}, \frac{GS_f}{\varepsilon})} \\
    & = & \prx{exp\Big(- \frac{\varepsilon \cdot y^*}{GS_f}\Big)}{exp\Big(- \sum_{a \in \AAA}\frac{\varepsilon |y^* + \Delta f_{x,x'}|}{GS_f}\Big)} \\
    & = & exp\Big( - \Big| \frac{\varepsilon \cdot |y^*|}{GS_f} - \frac{\varepsilon |y^* - \Delta f_{x,x'}|}{GS_f} \Big| \Big) \\
    &  \ge & exp\Big( - \frac{\varepsilon \big| \Delta f_{x,x'} \big|}{GS_f}\Big) \\
    & \ge & exp\Big(- \frac{\varepsilon \cdot GS_f}{GS_f}\Big) \\
    &  = & e^{- \varepsilon}.
    \eear
The third line follows from the definition of the normalized ratio and its relation to the absolute value of the exponent difference: $\prx{exp(a)}{exp(b)}= exp(|a-b|)$.
The first inequality after that applies the triangle inequality, and the second inequality applies the global sensitivity $GS_f$ 
as the upper bound on $\| \Delta f_{x,x'} \|$. 
At the end we get the lower bound required by Def.~\ref{Def:diffpriv}. 
This bound was determined taking an arbitrary value of $y' \in Y$. Therefore, we generally have as required
\bear
   1 \ \ \ge \ \ \prx{\Pr(F(x)\in Y)}{\Pr(F(x')\in Y)}  & \ge &  e^{-\varepsilon}.
\eear
This shows that the choice of parameter $\lambda = \frac{GS_f}{\varepsilon}$ for the additive Laplace noise is successful in making the disclosure algorithm $F(x)$ 
$\varepsilon$-differentially private.
\end{proof}

\section[Influence]{Public influence and push attacks
}\label{Sec:deceit}

\subsection{Layered architectures}
The internet stack, displayed in \cref{Fig:stacks} on the right, is studied in most network courses\footnote{Many courses actually follow the OSI network stack, which has two additional layers. The additional layers do not seem to arise from additional functionalities,  but they are useful nevertheless, as a fascinating illustration of design-by-committee and of accidents-of-evolution.}. It is less often noted that most codes and languages, both natural and artificial, are built according to the same layered architecture, as displayed in \cref{Fig:stacks} on the left. 
\begin{figure}[!ht]
    \centering
    \begin{minipage}[b]{0.3\textwidth}
    \centering
\includegraphics[width=5cm]{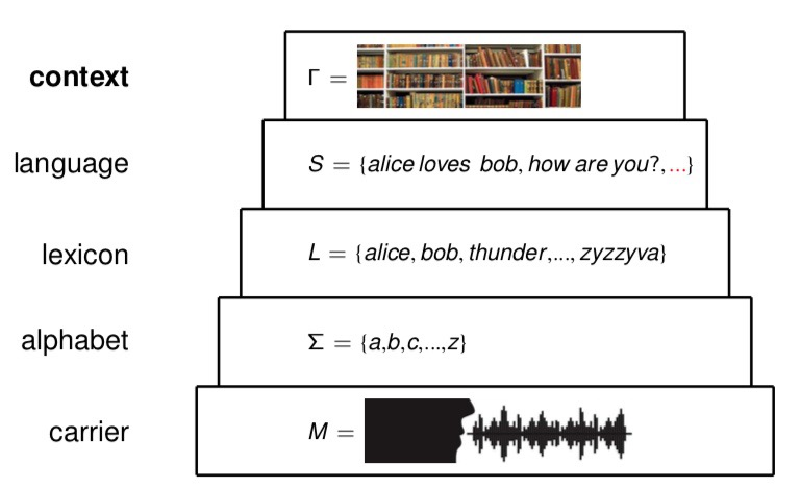}
    \end{minipage}%
\hspace{1.3em}
    \begin{minipage}[b]{0.3\textwidth}
\centering
\includegraphics[width=5cm]{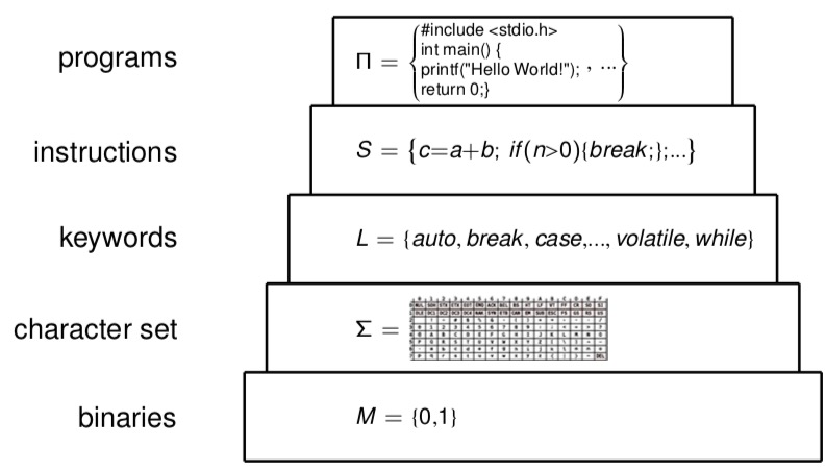} 
    \end{minipage}
\hspace{1.3em}
        \begin{minipage}[b]{0.3\textwidth}
        \centering
\includegraphics[width=5cm]{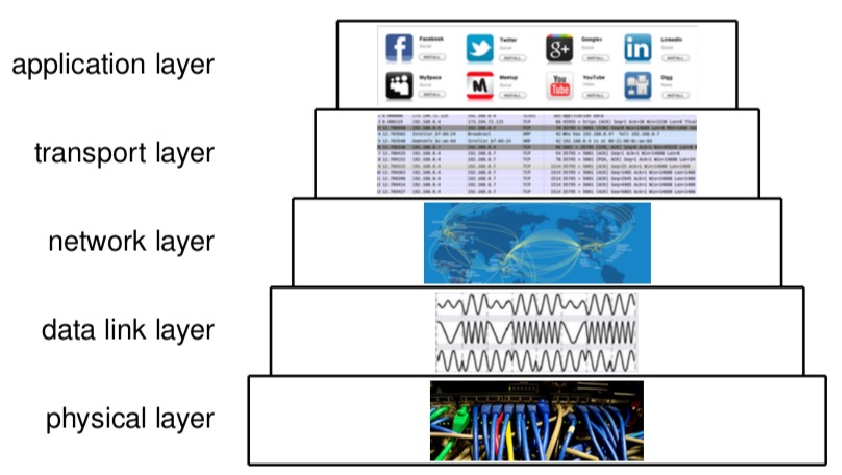} 
    \end{minipage}
            \caption{Languages and networks transfer messages across layers}
                    \label{Fig:stacks}
\end{figure}
The signals or messages transferred at each layer provide the particles of which the signals or messages at the next layer above are composed. In a natural language, the words are composed of letters, the sentences of words, texts of sentences, and so on. On the Internet, the network packets are composed of the frames transferred on the data link layer, the transport layer frames encapsulate sequences of network packets, and so on. Just like the syntax of a natural language determines which sentences are well-formed, and thus facilitates parsing, the network protocols determine formats of the packets, to facilitate assembling the messages that may have been split into particles. The familiar network protocol architecture of the Internet is displayed in \cref{Fig:stack-protocol}. 
\begin{figure}[!ht]
    \centering
        \begin{minipage}[b]{0.5\textwidth}
        \centering
\includegraphics[width=7cm]{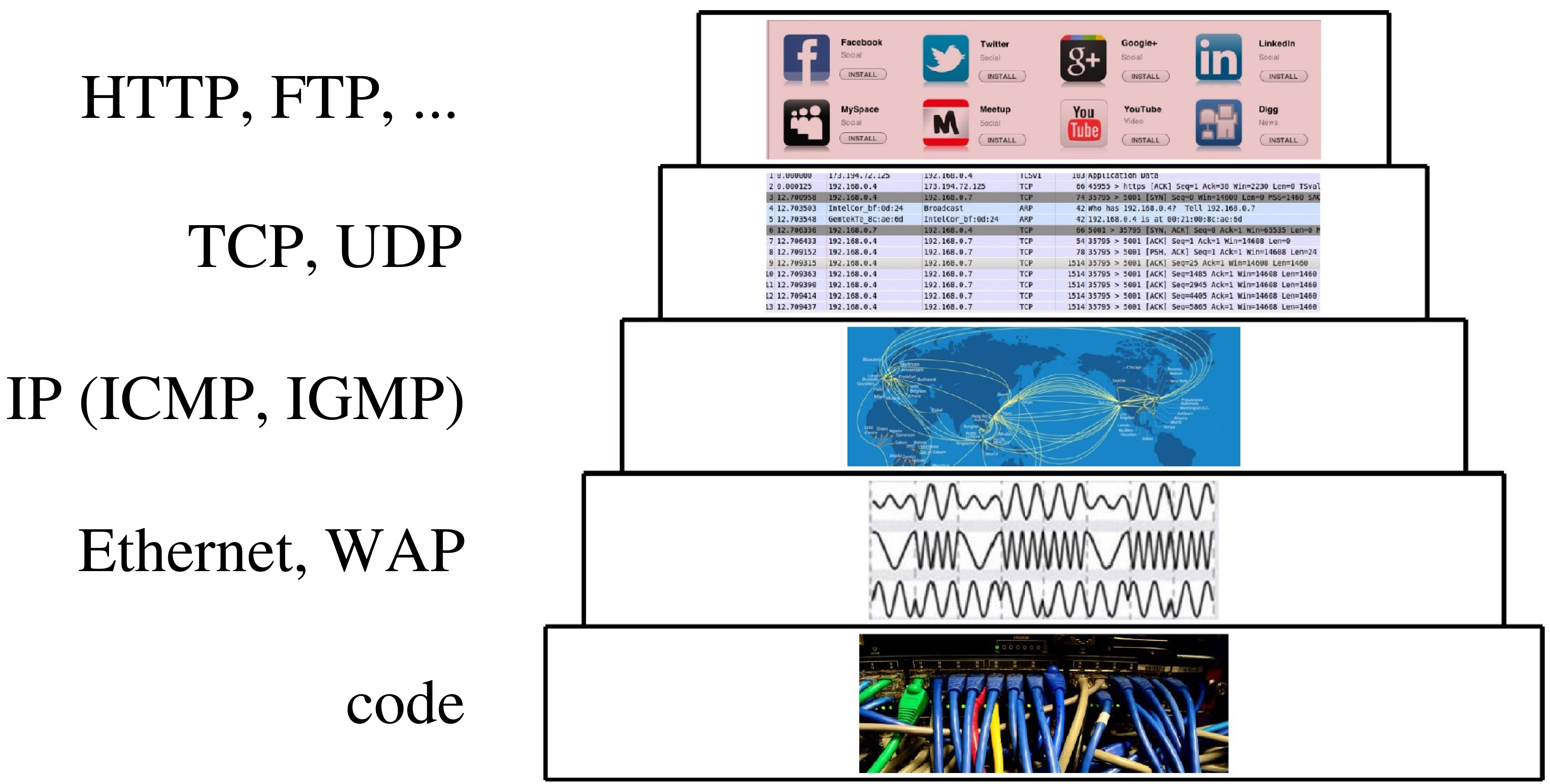} 
    \end{minipage}
            \caption{Protocols are syntactic rules for communication across layers}
                    \label{Fig:stack-protocol}
\end{figure}
The application layer protocols are designed by the application developers, and follow their own layered architecture, often several layers deep. Communications in natural languages are also regulated by protocols, which distinguish the sentence flow of a newspaper article from the formal languages of legal documents, of engineering, and the style and language of literary narrative, of poetry, and so on. These high-level rules tend to be more complex than the grammatical rules that generate sentences, and they are only known for some limited cases and situations.

\subsection{Level-above attacks}
In general, a context can be viewed as  a prefix of a well-formed message: it conveys some information but leaves some uncertainty. A sequence of contexts is a higher-level context, that may require error correction or abstraction. For example, Alice's email message to Bob is a context, as is Bob's response: they both narrow the possible conversations that may ensue, but usually do not determine them completely. A completed conversation, or a protocol run, is also a context at a still higher level, where a sequence of conversations may constitute a transaction. 

Deceit and outsmarting arise when such level-above channels are established covertly. For example, Bob sends a message in the context of one conversation, but that message may shift the conversation into a different context, which Alice may or may not notice, and may respond to it at the higher level, or remain at the lower level. Through a sequence of transactions, with or without the level shifts and covert messages, Alice and Bob may establish a covert context of trust, separate from the overt contexts of any of the conversations.

\subsection{Upshot}
Fig.~\ref{Fig:stacks} shows the computational  versions of the language stack, from the carriers at the bottom to the contexts at the top. The \emph{natural language stack}\/ on the left is echoed by the \emph{programming language stack}\/ in the middle, which is echoed by the  \emph{network stack}\/ on the right\footnote{The OSI model, usually taught in networks courses, separates a \emph{session}\/ layer above the transport layer, and a \emph{presentation}\/ layer below the application layer. A quick look at protocols shows that this is just a legacy quirk.}. 

\begin{figure}[!ht]
    \centering
    \begin{minipage}[b]{0.3\textwidth}
\centering
\includegraphics[width=5cm]{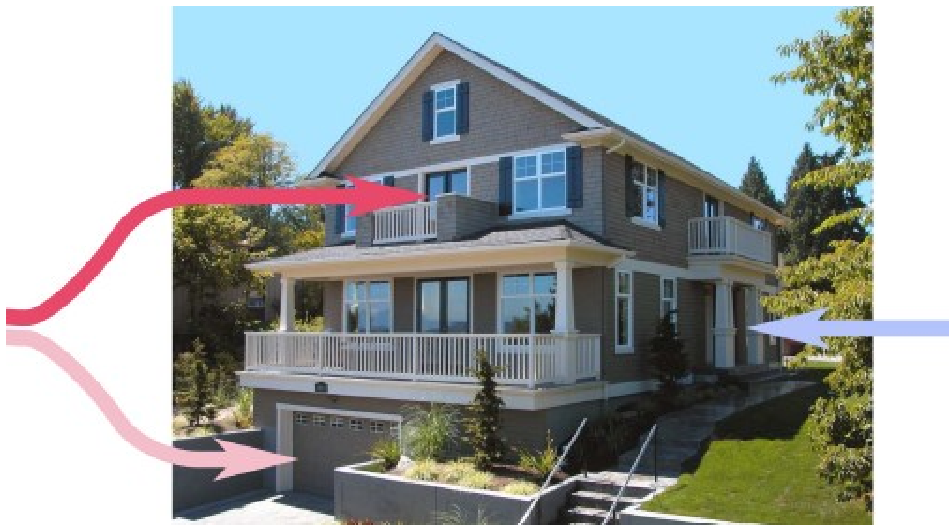}
    \end{minipage}%
    \hspace{1.3em} 
    \begin{minipage}[b]{0.3\textwidth}
\centering
\includegraphics[width=5cm]{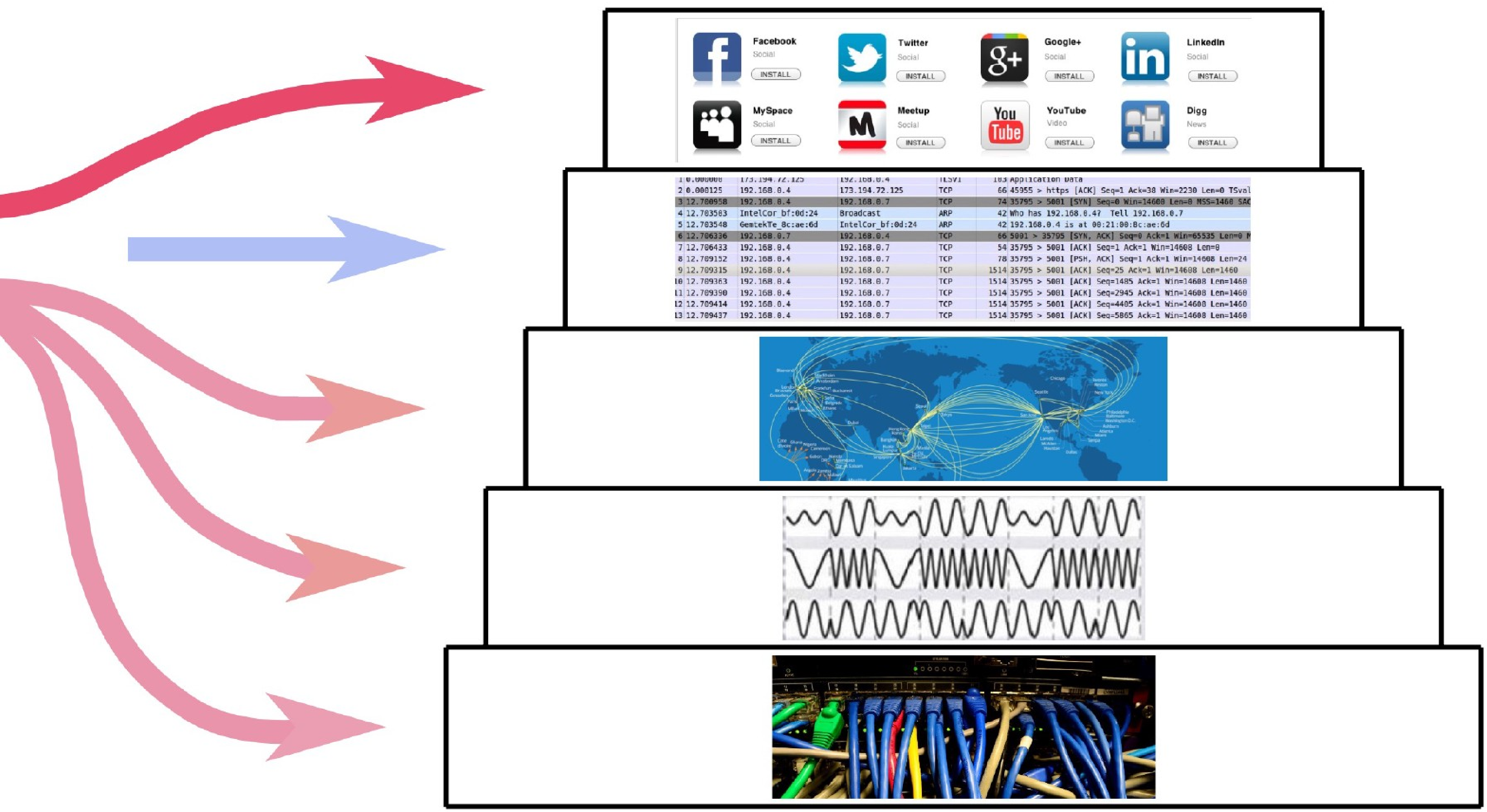} 
    \end{minipage}
    \hspace{1.3em} 
        \begin{minipage}[b]{0.3\textwidth}
\centering
\includegraphics[width=5cm]{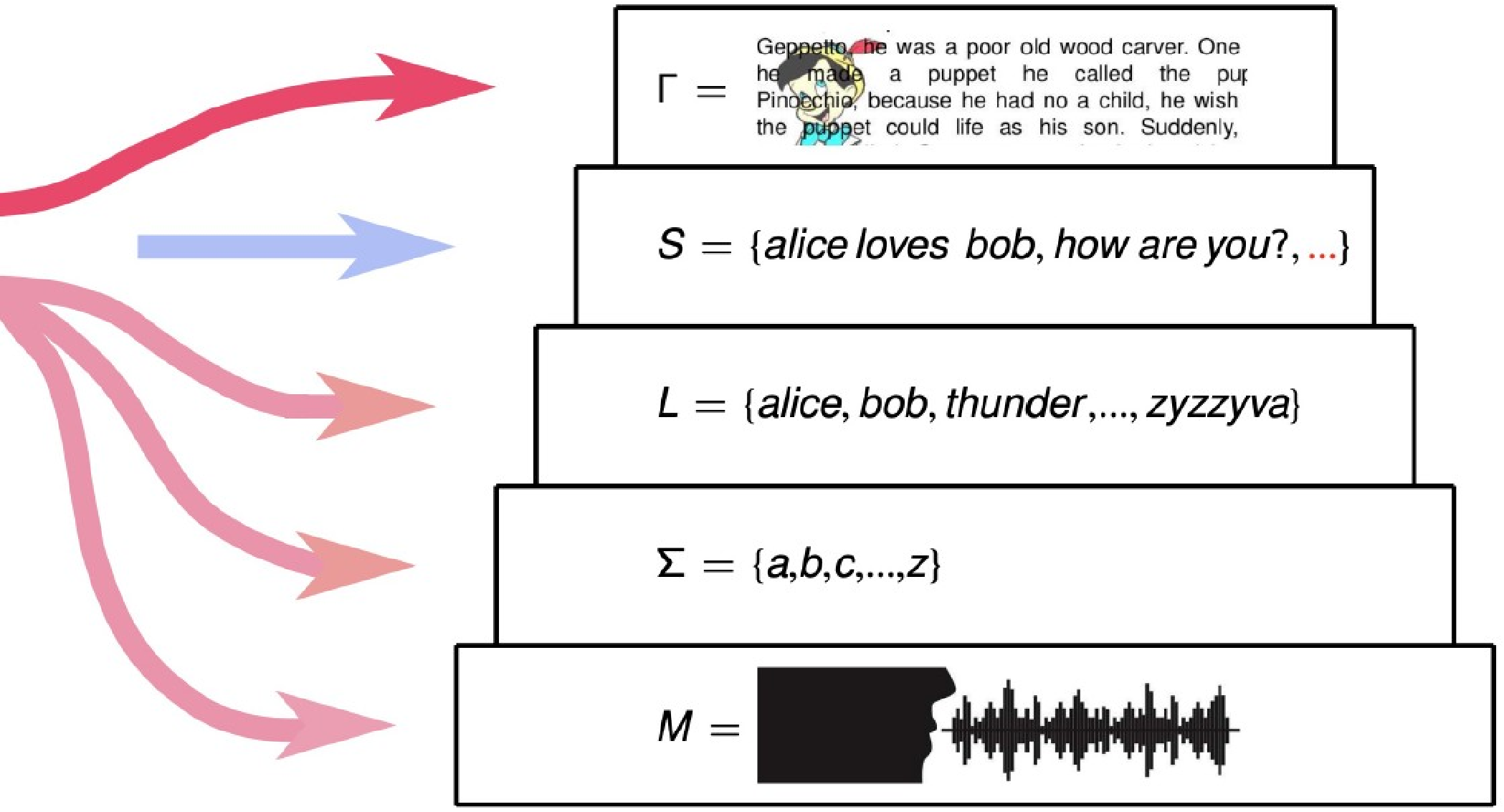} 
    \end{minipage}
            \caption{{\color{lightBlue}Attacks} are easy, {\color{lightRed}level-below} attacks are easier, {\color{red}level-above} attacks are the easiest}
                    \label{Fig:subver-lim}
\end{figure}

\section{What did we learn?}

The takeaway ideas of this chapter are:

\begin{itemize}
\item Privacy is the right to be left alone.

\item Data privacy prevents the information pull.

\item Influence implements the action push.

\item Deceit defeats privacy through level-above attacks.
\end{itemize}

But where do these ideas leave the tasks of the privacy protections in reality? For example, implementing the data privacy in a network requires at least the guarantees that
\begin{enumerate}[a)]
\item the private data cannot be derived from the publicly released data;
\item the publicly released data cannot be compiled and expanded through cross-referencing.
\end{enumerate}
But requirement (a) is a typical example of an intractable logical problem. What is derivable from what is demonstrated by constructing derivations. What is not derivable is even harder to characterize in general. Requirement (b) requires controlling the information flows, which is a typical example of an intractable problem in a network. So it seems that the data privacy problem has no general solution, as it involves intractable problems.\footnote{Differential privacy is a general method to protects the \emph{statistical}\/ databases, where the problems of cross-referencing and cross-derivations are filtered out together.}

Yet in real life, we do manage to maintain certain levels of privacy. Some say less and less, some say there are new forms of data that remain private. Either way, there is empiric evidence that privacy is not completely impossible. 

\backmatter

\chapter{Epilogue}\label{Chap:Postface}

Aloha Students of Security Science,

The grades have been entered. You have all done well. Thanks for taking this course.

Interestingly, you seem to have coordinated the votes to evenly distribute the points for the presentations. (The chance that everyone might get the same number of votes with no coordination seems to be less than 1 in 8000.) This is obviously in the interest of those who have invested less work and could expect an outcome below the average, and against the interest of those who have invested more work, and could expect an outcome above the average. The advantages and disadvantages of egalitarian social contracts are well-known.

On the other hand, the agreement may be signaling that the utility of points for those who gained some was greater than for those who lost some, in which case the collusion has increased the total utility.  This allows us to close the course by repeating the side remark that we never had a chance to fully develop:
\begin{quote} \emph{\textbf{Security and Economy are two sides of the same coin:}}
\begin{itemize}
\item a resource is an economic asset only if it can be secured, while
\item a security protection is effective only if it is cost-effective.
\end{itemize}
\end{quote} 
For instance, the lion cannot claim a water well as his asset if it is so big that he cannot prevent the gazelle from drinking on the other side. A \$100 lock is not an effective protection of a \$50 bike.

But valuations can vary wildly. My bike can be priceless to me. To prevent it from being stolen, I might be willing to steal a lock. In a market economy, my goal is to minimize my costs, maximize my revenue, and secure my profits. The less I give and the more I take, the better off  I am. But there are also social processes, such as love, religion, art, and science, where our goal is to give as much as we can and take out as little as possible. Many security failures arise from confusing different security goals that arise in different economic environments. To be good security engineers and scientists, you must remember do not waste money on the market and to not save efforts in love and science.

In any case, we respected your preferences and assigned the maximum of 25 points to everyone.

Happy holidays!

-- Peter and Dusko

\bibliographystyle{plain} 
\bibliography{TEXT-ref}  

\begin{thebibliography}{10}

\bibitem{JFK}
W.~Aiello and et~al S.M.~Bellovin.
\newblock {Just fast keying: Key agreement in a hostile internet}.
\newblock {\em ACM Trans. Inf. Syst. Secur.}, 7(2):242--273, May 2004.

\bibitem{Alpern-Schneider:safliv}
Bowen Alpern and Fred~B. Schneider.
\newblock Recognizing safety and liveness.
\newblock {\em Distributed Computing}, 2:117--126, 1987.

\bibitem{SmithG:book}
M.S. Alvim, K.~Chatzikokolakis, A.~McIver, C.~Morgan, C.~Palamidessi, and
  G.~Smith.
\newblock {\em The Science of Quantitative Information Flow}.
\newblock Information Security and Cryptography. Springer, 2020.

\bibitem{angela2009rome}
A.~Angela and G.~Conti.
\newblock {\em A Day in the Life of Ancient Rome}.
\newblock Europa Editions, 2009.

\bibitem{AshR:IT}
Robert~B. Ash.
\newblock {\em Information Theory}.
\newblock Dover Publications, 1990.

\bibitem{BaileyJ:private}
Joe Bailey.
\newblock From public to private: The development of the concept of "private".
\newblock {\em Social Research}, 69(1):15--31, 2002.

\bibitem{BayesT:essay}
Thomas Bayes.
\newblock An essay towards solving a problem in the doctrine of chances.
\newblock {\em Philosophical Transactions of the Royal Soceity. of London},
  53:370--418, 1763.

\bibitem{Bell-Lapadula}
David~E. Bell and Leonard~J. LaPadula.
\newblock Secure computer systems: Mathematical foundations.
\newblock Technical Report ESD-TR-73-278, Vol. I - IV, Electronic Systems
  Division, Air Force Systems Command, 1973.

\bibitem{BellaG}
Giampolo Bella.
\newblock {\em Formal Correctness of Security Protocols}.
\newblock Information Security and Cryptography. Springer Berlin Heidelberg,
  2007.

\bibitem{Biba}
Ken Biba.
\newblock Integrity considerations for secure computer systems.
\newblock Technical report, The Mitre Corporation, 06 1975.

\bibitem{Boyd-Mathuria}
Colin Boyd and Anish Mathuria.
\newblock {\em Protocols for Authentication and Key Establishment}.
\newblock Information Security and Cryptography. Springer, 2013.

\bibitem{ChineseWall}
D.F.C. Brewer and M.J. Nash.
\newblock The chinese wall security policy.
\newblock In {\em Proceedings. 1989 IEEE Symposium on Security and Privacy},
  pages 206--214, 1989.

\bibitem{burke2000delos}
S.~Burke.
\newblock {\em Delos: Investigating the Notion of Privacy Within the Ancient
  Greek House}.
\newblock PhD thesis, University of Leicester, 2000.

\bibitem{PavlovicD:CathyFest}
Jason Castiglione, Dusko Pavlovic, and Peter-Michael Seidel.
\newblock {Privacy protocols}.
\newblock In Joshua~Guttman et~al., editor, {\em Foundations of Security,
  Protocols, and Equational Reasoning}, volume 11565 of {\em Lecture Notes in
  Computer Science}, pages 167--192. Springer, 2019.

\bibitem{Clark-Wilson}
David~D. Clark and David~R. Wilson.
\newblock { A Comparison of Commercial and Military Computer Security Policies
  }.
\newblock In {\em Proceedings IEEE Symposium on Security and Privacy}, pages
  184--184, Los Alamitos, CA, USA, April 1987. IEEE Computer Society.

\bibitem{Clarkson-Schneider:hyperproperties}
Michael~R. Clarkson and Fred~B. Schneider.
\newblock Hyperproperties.
\newblock {\em J. of Computer Security}, 18(6):1157--1210, 2010.

\bibitem{wren-auth}
Diane Colombelli-N\'egrel, Mark~E. Hauber, Jeremy Robertson, Frank~J. Sulloway,
  Herbert Hoi, Matteo Griggio, and Sonia Kleindorfer.
\newblock Embryonic learning of vocal passwords in superb fairy-wrens reveals
  intruder cuckoo nestlings.
\newblock {\em Current Biology}, 22(22):2155--2160, December 2012.

\bibitem{Cover-Thomas:IT}
Thomas~M. Cover and Joy~A. Thomas.
\newblock {\em Elements of information theory}.
\newblock Wiley-Interscience, New York, NY, USA, 1991.

\bibitem{Cremers-Mauw:book}
Cas Cremers and Sjouke Mauw.
\newblock {\em Operational Semantics and Verification of Security Protocols}.
\newblock Information Security and Cryptography. Springer, 2012.

\bibitem{DaleniusT:desideratum}
Tore Dalenius.
\newblock Towards a methodology for statistical disclosure control.
\newblock {\em Statistik Tidskrift}, 15:429--444, 1977.

\bibitem{Diffie-Hellman}
Whitfield Diffie and Martin~E. Hellman.
\newblock New directions in cryptography.
\newblock {\em IEEE Transactions on Information Theory}, 22(6):644--654, 1976.

\bibitem{STS}
Whitfield Diffie, Paul~C. van Oorschot, and Michael~J. Wiener.
\newblock Authentication and authenticated key exchanges.
\newblock {\em Designs, Codes, and Cryptography}, 2:107--125, 1992.

\bibitem{Dwork06}
Cynthia Dwork.
\newblock Differential privacy.
\newblock In Michele Bugliesi, Bart Preneel, Vladimiro Sassone, and Ingo
  Wegener, editors, {\em Proceedings of {ICALP} 2006, Part {II}}, volume 4052
  of {\em Lecture Notes in Computer Science}, pages 1--12. Springer, 2006.

\bibitem{dwork2014}
Cynthia Dwork and Aaron Roth.
\newblock The algorithmic foundations of differential privacy.
\newblock {\em Foundations and Trends in Theoretical Computer Science},
  9(3-4):211--407, 2014.

\bibitem{EvansD:MPC}
David Evans, Vladimir Kolesnikov, and Mike Rosulek.
\newblock {\em A Pragmatic Introduction to Secure Multi-Party Computation}.
\newblock Foundations and Trends{\textregistered} in Privacy and Security
  Series. Now Publishers, 2018.

\bibitem{FeynmanR:character}
Richard~P. Feynman.
\newblock {\em The character of physical law}, volume~66.
\newblock MIT Press, Cambridge MA, USA, 1965.

\bibitem{Fisher:1930}
Ronald~A. Fisher.
\newblock Inverse probability.
\newblock {\em Proceedings of the Cambridge Philosophical Society},
  28:528--535, 1930.

\bibitem{rfc6071}
S.~Frankel and S.~Krishnan.
\newblock {IP Security (IPsec) and Internet Key Exchange (IKE) Document
  Roadmap}.
\newblock RFC 6071, February 2011.

\bibitem{rfc2409}
D.~Harkins and D.~Carrel.
\newblock {The Internet Key Exchange (IKE)}.
\newblock RFC 2409, November 1998.

\bibitem{HRU}
Michael~A. Harrison, Walter~L. Ruzzo, and Jeffrey~D. Ullman.
\newblock {Protection in Operating Systems}.
\newblock {\em {Communications of the ACM}}, 19(8):461--471, August 1976.

\bibitem{KaliskiB:MQV}
Burton~S. Kaliski.
\newblock {An unknown key-share attack on the MQV key agreement protocol}.
\newblock {\em ACM Trans. Inf. Syst. Secur.}, 4(3):275--288, August 2001.

\bibitem{Koblitz-Menezes:another-perspective}
Neal Koblitz and Alfred Menezes.
\newblock {Critical perspectives on provable security: Fifteen years of
  ``another look'' papers}.
\newblock {\em Advances in Mathematics of Communications}, 13(4):517--558,
  2019.

\bibitem{KrawczykH:SKEME}
Hugo Krawczyk.
\newblock {SKEME: A versatile secure key exchange mechanism for internet}.
\newblock In {\em Proceedings of Internet Society Symposium on Network and
  Distributed Systems Security}, pages 114--127. IEEE, 1996.

\bibitem{HMQV}
Hugo Krawczyk.
\newblock {HMQV: A High-Performance Secure Diffie-Hellman Protocol}.
\newblock In Victor Shoup, editor, {\em CRYPTO 2005}, pages 546--566, Berlin,
  Heidelberg, 2005. Springer.

\bibitem{KuhnT:structure}
Thomas~S. Kuhn.
\newblock {\em The Structure of Scientific Revolutions}.
\newblock University of Chicago Press, 2012.

\bibitem{Lamport:safliv}
Leslie Lamport.
\newblock Proving the correctness of multiprocess programs.
\newblock {\em IEEE Transactions on Software Engineering}, SE-3(2):125--143,
  1977.

\bibitem{MQV}
Laurie Law, Alfred Menezes, Minghua Qu, Jerome~A. Solinas, and Scott~A.
  Vanstone.
\newblock An efficient protocol for authenticated key agreement.
\newblock {\em Designs, Codes and Cryptography}, 28:119--134, 2003.

\bibitem{Needham-Schroeder-Lowe}
Gavin Lowe.
\newblock An attack on the {Needham-Schroeder} public-key authentication
  protocol.
\newblock {\em Information Processing Letters}, 56:131--133, 1995.

\bibitem{LoweG:Hierarchy}
Gavin Lowe.
\newblock A hierarchy of authentication specifications.
\newblock In {\em Proceedings of the 10th IEEE Computer Security Foundations
  Workshop}, pages 31--43. IEEE, 1997.

\bibitem{MacKayD:book}
David~J.C. MacKay.
\newblock {\em Information Theory, Inference and Learning Algorithms}.
\newblock Cambridge University Press, 2003.

\bibitem{PavlovicD:ESORICS04}
Catherine Meadows and Dusko Pavlovic.
\newblock Deriving, attacking and defending the {GDOI} protocol.
\newblock In Peter Ryan, Pierangela Samarati, Dieter Gollmann, and Refik Molva,
  editors, {\em Proceedings of ESORICS 2004}, volume 3193 of {\em Lecture Notes
  in Computer Science}, pages 53--72. Springer Verlag, 2004.

\bibitem{MeadowsC:GDOI}
Catherine Meadows, Paul Syverson, and Iliano Cervesato.
\newblock {Formalizing GDOI group key management requirements in NPATRL}.
\newblock In {\em Proceedings of the 8th ACM CCS}, CCS '01, pages 235--244, New
  York, NY, USA, 2001. ACM.

\bibitem{MenezesA:HMQV}
Alfred Menezes.
\newblock Another look at hmqv.
\newblock {\em Journal of Mathematical Cryptology}, 1(1):47--64, 2007.

\bibitem{Needham-Schroeder}
Michael~D. Needham, Roger M.;~Schroeder.
\newblock Using encryption for authentication in large networks of computers.
\newblock {\em Communications of the ACM}, 21:993--999, 1978.

\bibitem{orlin2009locating}
Lena~C. Orlin.
\newblock {\em Locating Privacy in Tudor London}.
\newblock Oxford University Press, 2009.

\bibitem{rfc2412}
H.~Orman.
\newblock {The OAKLEY Key Determination Protocol}.
\newblock RFC 2412, November 1998.

\bibitem{Pareto}
Vilfrido Pareto.
\newblock {\em Considerations on the Fundamental Principles of Pure Political
  Economy}.
\newblock Routledge Studies in the History of Economics. Taylor \& Francis,
  2007.

\bibitem{PaulsonL:bull}
Lawrence~C. Paulson.
\newblock Mechanized proofs for a recursive authentication protocol.
\newblock In {\em Proceedings of CSFW '97}, CSFW, page~84, USA, 1997. IEEE
  Computer Society.

\bibitem{PavlovicD:spider}
Dusko Pavlovic.
\newblock {Lambek pregroups are Frobenius spiders in preorders}.
\newblock {\em Compositionality}, 4(1):1--21, 2022.
\newblock arxiv.org/abs/2105.03038.

\bibitem{PavlovicD:LangEng}
Dusko Pavlovic.
\newblock Language processing in humans and computers.
\newblock {\em CoRR}, arxiv.org/abs/2405.14233, May 2024.

\bibitem{PopperK:refutations}
Karl~R. Popper.
\newblock {\em Conjectures and Refutations: The Growth of Scientific
  Knowledge}.
\newblock Classics Series. Routledge, 2002.

\bibitem{Ryan-Schneider:book}
Peter Ryan and Steve Schneider.
\newblock {\em The Modelling and Analysis of Security Protocols: The CSP
  Approach}.
\newblock Addison-Wesley, 2001.

\bibitem{Ryan-Schneider:bull}
P.Y.A. Ryan and S.~A. Schneider.
\newblock An attack on a recursive authentication protocol. a cautionary tale.
\newblock {\em Inf. Process. Lett.}, 65(1):7--10, January 1998.

\bibitem{schoeman1984philosophical}
Ferdinand~D. Schoeman.
\newblock {\em Philosophical Dimensions of Privacy: An Anthology}.
\newblock Cambridge University Press, 1984.

\bibitem{ShannonC:communication}
Claude~E. Shannon.
\newblock A mathematical theory of communication.
\newblock {\em The Bell System Technical Journal}, 27:379--423 and 623--656,
  1948.

\bibitem{ShannonC:Secrecy}
Claude~E. Shannon.
\newblock Communication theory of secrecy systems.
\newblock {\em The Bell System Technical Journal}, 28(4):656--715, 1949.

\bibitem{SimmonsG:TMN}
Gustavus~J. Simmons.
\newblock Cryptanalysis and protocol failures.
\newblock {\em Communications of the ACM}, 37:56--65, 1994.

\bibitem{Sweeney97}
Latanya Sweeney.
\newblock Weaving technology and policy together to maintain confidentiality.
\newblock {\em Journal of Law, Medicine and Ethics}, 25:98--110, 1997.

\bibitem{SweeneyL:suppression}
Latanya Sweeney.
\newblock Achieving k-anonymity privacy protection using generalization and
  suppression.
\newblock {\em International Journal of Uncertainty, Fuzziness and
  Knowledge-Based Systems}, 10(5):571--588, 2002.

\bibitem{SweeneyL:kanonymity}
Latanya Sweeney.
\newblock k-anonymity: {A} model for protecting privacy.
\newblock {\em International Journal of Uncertainty, Fuzziness and
  Knowledge-Based Systems}, 10(5):557--570, 2002.

\bibitem{rfc6407}
B.~Weis, S.~Rowles, and T.~Hardjono.
\newblock {The Group Domain of Interpretation}.
\newblock RFC 6407, October 2011.

\end{thebibliography}
\addcontentsline{toc}{chapter}{Bibliography}

\appendix

\def\thechapter{A}
\chapter{Appendix A. Prerequisites}

\section{List and string constructors and notations}\label{Prereq:list}
\sindex{string}\sindex{list}
\begin{itemize}
\item\textbf{lists:} $X^\ast \ = \ \Big\{\seq{x_1\ x_2 \ldots x_n} \in X^n\ |\ n = 0,1,2,\ldots\Big\}$
\item\textbf{strings:} $X^+  = \Big\{\seq{x_1\ x_2 \ldots x_n} \in X^n\ |\ n = 1,2,3\ldots\Big\}$ 
\end{itemize}
The only \sindex{list!difference from string} \sindex{string!difference from list} difference between the type $X^{\ast}$ of lists of elements of $X$ and the type $X^{+}$ of strings of elements of $X$ is that $X^{\ast}$ contains the empty list $\seq{}$, whereas $X^{+}$ does not. Strings are the nonempty lists, whereas a list is either a string or empty:
\bear
X^\ast & = & X^+ \cup \Big\{\sseq{}\Big\}
\eear
Lists are inductively generated by the list constructor $(\cons)$ \sindex{list!constructor} and the \sindex{list!empty} empty list $()$:
\bea
X^\ast \times X \tto{\ (\cons) \ } & X^\ast & \oot{\  ()\  } 1\label{eq:list}\\
\Big<\seq{x_0 \ldots x_n}, x_{n+1}\Big> \longmapsto &  \seq{x_0\ldots x_n\ x_{n+1}}&\notag\\
&\seq{}&\longmapsfrom \emptyset \notag
\eea
whereas strings are generated by the string constructor $(\cons)$ \sindex{string!constructor} and the \sindex{string!letter inclusion}  letter inclusion $(-)$ operation:
\bea
X^+ \times X \tto{\ (\cons) \ } & X^+ & \oot{\  (-)\  } X\label{eq:string}\\
\Big<\seq{x_0 \ldots x_n}, x_{n+1}\Big> \longmapsto &  \seq{x_0\ldots x_n\ x_{n+1}}&\notag\\
&\seq{x}&\longmapsfrom x \notag
\eea

\para{Notation.} We write the names of lists and strings in bold\footnote{Functional programmers write $xs \ = \ \seq{x1\  x2\ldots xn}$}
\bear
\vec x & = & \seq{x_0\ x_1\ x_2\ \ldots x_n}
\eear
When convenient, the indices can also be written right to left, i.e., $\vec x  =  \seq{x_n\ x_{n-1}\ \ldots x_1\ x_0}$. 

\subsubsection*{Induction} 
Inductive constructors allow inductive definitions. Here are a couple of basic examples.

\para{List concatenation.} The list \sindex{concatenation} \emph{concatenation}\/ appends two lists
\bear
X^\ast \times X^\ast &\tto{@} & X^\ast\ \ \oot{\sseq{}} 1
\\
\big<\vec z, \sseq{}\big> & \longmapsto & \vec z\\
\big<\vec z, \vec y\cons x \big > &\longmapsto & \left(\vec z @ \vec y\right)\cons x
\eear
Since concatenation is obviously associative, it makes $X^\ast$ into a monoid, with the empty string $()$ as the unit. More precisely, $X^\ast$ is the free monoid over $X$, whereas $X^+$ is the free semigroup. The monoid operations easily extend from lists to the sets of lists
\[
\prooftree
X^\ast\times X^\ast \tto{@} X^\ast \oot{\sseq{}} 1
\justifies
\WP (X^\ast) \times \WP \left(X^\ast\right) \tto{@} \WP \left(X^\ast\right) \oot{\{\sseq{}\}} 1
\endprooftree
\]
by defining the concatenation of $C,D\subseteq X^\ast$ to be the set of concatenations of their elements
\bear
CD\ =\ C@D & = & \left\{c@d\in X^\ast\ |\ c\in C, d\in D \right\}
\eear

\para{List length.} The simplest inductive definition assigns to each list the number of its symbols:
\bear
X^\ast & \tto{\ell} & \NNn\\
() & \mapsto & 0\\
\vec y\cons x & \mapsto & \ell(\vec y)+1
\eear
Note that $\NNn\cong \{1\}^\ast$ is the free monoid over one generator, with the operation of addition on $\NNn$ corresponding to the concatenation of  $\{1\}^\ast$, providing the arithmetic in base 1. The list length operation can thus be viewed as the monoid homomorphism $X^\ast \to \{1\}^\ast$ induced by the unique function $X\to 1$, identifying all symbols of the alphabet $X$ with a single symbol.

\para{Abbreviations.} When no confusion is likely we abbreviate, not just $C@D$ to $CD$ as above, but also $\vec x @ \vec y$ to $\vec x :: \vec y$, and even $\vec x\vec y$. 

\para{Prefix order.}\sindex{prefix order} The \emph{prefix}\/ relation $\sqsubseteq$  defined by
\bea\label{eq:prefix-order}
\vec x \sqsubseteq \vec y & \iff & \exists \vec z.\ \vec x :: \vec z = \vec y
\eea
is a partial order, both on lists and on strings. Writing $\vec x\sqsubseteq \vec y$ and saying that \emph{$\vec x$ is a prefix of $\vec y$} means that there is $\vec z = \seq{z_1 \ \ldots\  \ z_{n-k}}$ such that
\bear
&& \seq{x_1\ x_2 \ldots x_k\ z_1 \ \ldots\  \ z_{n-k}} \\
&=& \seq{y_1\ y_2 \ldots y_k\  y_{k+1} \ldots y_n} 
\eear
so that $x_i = y_i$ for $i\leq k$ and $z_i = y_{i+k}$ for $i\leq n-k$. This partial ordering makes any set of lists $X^\ast$ into a meet semilattice, with the meet $\vec x \sqcap \vec y$ extracting the greatest common prefix of $\vec x$ and $\vec y$. The set of strings $X^+$ is not a semilattice just because the greatest common prefix of strings starting  differently is the empty list, which is not a string.


\para{Specifying properties.} For any pair of events  $a,b\in \Event$, and for arbitrary sets of events $C,D\subseteq \Event$, 
some of the properties that can be defined are: 
\bear
\overline a & =  & \{ \vec x:: a :: \vec y\ |\ \vec x, \vec y \in \Event^\ast\}\\
\overline{a\prc b} & = &\left\{ \vec x:: a :: \vec y :: b :: \vec z\ \  |\ \ \vec x, \vec y, \vec z \in \Event^\ast\right\}\\
\overline{a\prc \exists b} & = & \big\{ \vec t \in \Event^\ast\ |\ \exists \vec x \vec y.\ \vec t = \vec x:: a:: \vec y\ \implies \  \exists \vec {y'}\vec {y''}.\ \vec y = \vec {y'} :: b :: \vec {y''}\big \}\\
\overline{\exists a\prc  b} & = & \big\{ \vec t \in \Event^\ast\ |\ \exists \vec x \vec y.\ \vec t = \vec x:: b:: \vec y\ \implies \  \exists \vec {x'}\vec {x''}.\ \vec x = \vec {x'} :: a :: \vec {x''}\big \}\\
\overline C &= &  \{ \vec x:: c :: \vec y\ |\ \vec x, \vec y \in \Event^\ast, c\in C\}\\
\overline{C\prc D} &  = & \left\{ \vec x:: c :: \vec y :: d :: \vec z\ \  |\ \ \vec x, \vec y, \vec z \in \Event^\ast,\  c\in C,\ d\in D\right\}\\
\overline{C\prc \exists D} & = & \big\{ \vec t \in \Event^\ast\ |\ 
\forall c\in C.\ \vec t = \vec x:: c:: \vec y\ \implies \  \exists d\in D.\  \vec y = \vec {y'} :: d :: \vec {y''}\big \}\\
\overline{\exists C\prc D} & = & \big\{ \vec t \in \Event^\ast\ |\ 
\forall d\in D.\ \vec t = \vec x:: d :: \vec y\ \implies \ \exists c\in C.\  \vec x = \vec {x'} :: c :: \vec {x''}\big \}
\eear

\section{Neighborhoods}
\label{Prereq:Topology}
\para{Neighborhoods.} Topology is the most general theory of space. Spaces usually model the realms of observations, and their structures therefore correspond to what is observed: e.g., metric spaces capture distances, vector spaces also capture angles, etc. Topological spaces only capture \emph{neighborhoods}, viewed as sets that are inhabited by some objects together. Smaller neighborhoods suggest that the objects are closer together, and the set inclusion of neighborhoods thus tells which objects are closer together, and which ones are further apart. But these distinctions are made without assigning any numeric values to the distances between objects, as it is done in metric spaces. Metric spaces can thus be viewed as a special case of topological spaces.

For any set $\Space$, a (topological) space structure is defined by either of the following:
\begin{itemize}
\item \emph{open  neighborhoods}\/ are represented by a family $\OOO_\Space \subseteq \WP(\Space)$  closed under finite $\cap$ and arbitrary $\bigcup$; whereas
\item \emph{closed  neighborhoods}\/ are represented by a family $\FFF_\Space \subseteq \WP(\Space)$  closed under finite $\cup$ and arbitrary $\bigcap$
\end{itemize}
The two families determine each other, in the sense that a set is open if and only if its complement is closed, i.e.
\[
\OOO_\Space\ = \ \left\{U\in \WP(\Space)\ |\ \neg U \in \FFF_\Space\right\}\qquad\qquad
\FFF_\Space \ = \ \left\{F\in \WP(\Space)\ |\ \neg F \in \OOO_\Space\right\}
\]
Any family $\BBB_\Space\subseteq \WP(\Space)$ can be declared to be open (or closed) neighborhoods, and used to generate  the full space structure.  If we want $\BBB_\Space$ to be 
\begin{itemize}
\item open, then set
\[
\BBB^\cap_\Space \ = \ \left\{\bigcap_{i=0}^n B_i\ | \   B_0,\ldots, B_n \in \BBB_\Space\right\}\qquad\mbox{and}\qquad
\OOO_\Space \ = \   \left\{\bigcup \VVV\ |\ 
  \VVV \subseteq \BBB^\cap_\Space\right\}
\]
\item closed, then set
\[
\BBB^\cup_\Space \  = \  \left\{\bigcup_{i=0}^n B_i\ | \  B_0,\ldots, B_n \in \BBB_\Space\right\}\qquad\mbox{and}\qquad
\FFF_\Space \ = \  \left\{\bigcap \VVV\ |\ 
  \VVV \subseteq \BBB^\cup_\Space\right\}\]
\end{itemize}

\para{Closure operators.} \sindex{closure}For any given topology on $\Space$, the family $\FFF_\Space$ can be equivalently presented in terms of the \emph{closure operator}\/ that it induces:
\bea\label{eq:closure}
\closure{} : \WP(\Space) & \to & \WP(\Space)\\
X & \longmapsto & \bigcap\{F\in \FFF_\Space\ |\ X\subseteq F\}\notag
\eea
It is easy to see that the general requirements from a closure operator are satisfied:
\[ X\  \subseteq\  \closure X \ = \ \closure{\closure X}\] 
and $\FFF_\Space = \{X\in \WP(\Space)\ |\ \closure X = X\}$.

\para{Interior operators} are dual to closure operators.\sindex{interior} This means that for any given  topology on $\Space$, the family $\OOO_\Space$ of open sets (dual to the closed sets $\FFF_\Space$) can be equivalently presented in terms of the \emph{interior operator}\/ that it induces:
\bea\label{eq:interior}
\interior{} : \WP(\Space) & \to & \WP(\Space)\\
X & \longmapsto & \bigcup\{U\in \OOO_\Space\ |\ U\subseteq X\}\notag
\eea
It is easy to see that the general requirements from an interior operator are satisfied:
\[ X\  \supseteq\  \interior X \ = \ \interior{\interior X}\] 
and $\OOO_\Space = \{X\in \WP(\Space)\ |\ \interior X = X\}$.

\para{Closures and interiors determine each other} by
\[
 \closure X = \neg\interior{\neg X}\qquad \qquad\qquad \qquad  \interior X = \neg \closure{\neg X} 
\]

\para{Density} can be defined by saying for $X, Y \in \WP(\Space)$ that \emph{$X$ is dense on $Y$}\/ if $\closure{X} \supseteq Y$. This can be equivalently written in the form $Y \cap \interior{\neg X} = \emptyset$. We sometimes describe such situations by saying that \emph{$\neg Y$ is codense on $\neg X$}. In general, $U$ is codense on $V$ when $\int V \setminus U= \emptyset$. The family of sets that are dense on $Y$, or on which $Y$ is codense, is
\bear
\DDD_Y  & = & \left\{ D \in\WP(Y) \ |\ \forall F\in \FFF_\Space.\ D\subseteq  F\implies Y\subseteq F \right\}\\
        & = & \left\{ D \in\WP(Y) \ |\ \forall U\in \OOO_\Space.\ Y\cap U\neq \emptyset \ \implies D\cap U \neq \emptyset  \right\}
\eear

\para{Closed-dense $\cap$-decomposition.} Any inclusion $X\subseteq \Space$ clearly factors as $X\subseteq \closure X \subseteq \Space$ where $\closure X\in \FFF_\Space$ is closed, and $X \in \DDD_{\closure X}$ is relatively dense.
Therefore, any $X\subseteq \Space$ is an intersection of a closed and a dense set
\bear
X & = & \closure X\ \cap\ (X\cup \neg \closure X)
\eear
where $\closure X\in\FFF_\Space$ is closed, and $X\cup \neg \closure X \in \DDD_\Space$ is dense everywhere.

\para{Open-codense $\cup$-decomposition} is induced by the closed-dense $\cap$-decomposition of the complement $\neg X\subseteq \Space$ through $\closure{\neg X} = \neg \interior X \subseteq \Space$. More precisely $\neg X = \closure{\neg X} \cap (\neg X\cup \neg \closure{\neg X}) = \neg\interior X \cap \neg\left(X \cap \neg \interior X\right )$ becomes
\bear
X & = & \interior X \cup \big(X\setminus \interior X \big)
\eear
where $\interior X \in \OOO_\Space$ is open, and $X \setminus \interior X$ is codense.

\section{Images, cylinders and cylindrifications}\label{Sec:localization}
\para{Direct and inverse images.} \sindex{image} Any function $f:A\to B$ induces two direct image maps and one inverse image map:
\bear
f_! :\WP A & \longrightarrow & \WP B\\
X & \longmapsto &\{y\ |\ \exists x.\ f(x) = y\wedge x\in X\}\\[2ex]
f^\ast :\WP B & \longrightarrow & \WP A\\
V&\longmapsto & \{x\ |\ f(x)\in V\}\\[2ex]
f_\ast :\WP A & \longrightarrow & \WP B\\
X & \longmapsto &\{y\ |\ \forall x.\ f(x) = y\Rightarrow x\in X\}
\eear
It is easy to show that they satisfy\sindex{image!direct}\sindex{image!inverse}
\bear
f_!(X)\subseteq Y & \iff & X\subseteq f^\ast(Y)\mbox{ and}\\
f^\ast(Y) \subseteq X & \iff & Y\subseteq f_\ast(X)
\eear
and hence for any $X\in \WP A$ holds
\[
f^\ast f_\ast(X)\  \subseteq\  X\  \subseteq\  f^\ast f_!(X) \]

\para{Cylinders.}\sindex{cylinder}  We call $f^\ast f_\ast(X)$ the \emph{internal $f$-cylinder}\/ contained in $X$, whereas $f^\ast f_!(X)$ is the \emph{external $f$-cylinder}\/ containing $X$. An external $f$-cylinder is thus the smallest inverse image of a set in $B$ along $f$ that contains $X$, whereas an internal $f$-cylinder is the largest inverse image of a set in $B$ along $f$ that is contained in $X$. 

\para{Cylinder closures and interiors.} \sindex{cylinder closure|see{closure,cylinder}} \sindex{cylinder interior|see{interior,cylinder}}Suppose that we are given a family of functions $V= \{v:A\to B_v\ |\ v\in \VVv\}$, where $\VVv$ is an arbitrary set of indices. We consider such families as \emph{cylinder localizations}: we build $v$-cylinders, internal and external, for all $v\in V$, and approximate arbitrary sets using cylinders. \sindex{closure!cylinder}\sindex{interior!cylinder}

\begin{figure}[ht!]
\begin{center}
\newcommand{\Exx}{X}
\newcommand{\Exxx}{X}
\newcommand{\ExxClosure}{\closure X}
\newcommand{\ExxxInterior}{\interior X}
\newcommand{\Alispace}{B_0}
\newcommand{\Bobspace}{B_1}
\def\JPicScale{.5}
\input{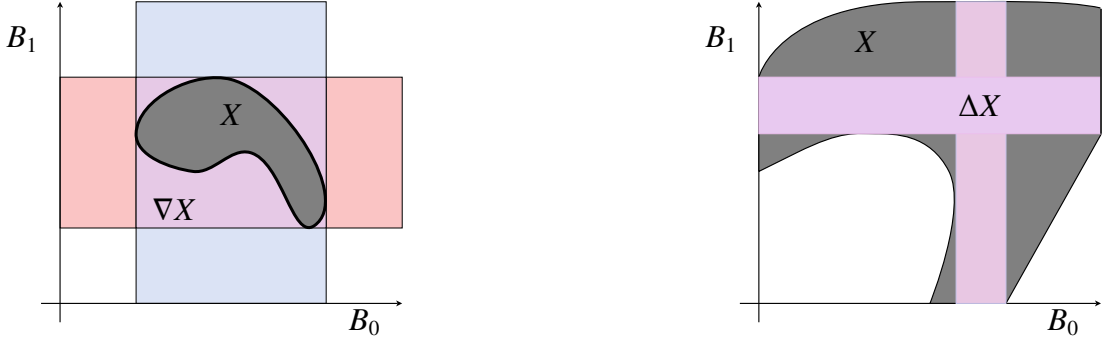}
\hspace{8em}
\ifx\JPicScale\undefined\def\JPicScale{1}\fi
\psset{unit=\JPicScale mm}
\psset{linewidth=0.3,dotsep=1,hatchwidth=0.3,hatchsep=1.5,shadowsize=1,dimen=middle}
\psset{dotsize=0.7 2.5,dotscale=1 1,fillcolor=black}
\psset{arrowsize=1 2,arrowlength=1,arrowinset=0.25,tbarsize=0.7 5,bracketlength=0.15,rbracketlength=0.15}
\begin{pspicture}(0,0)(130.92,102.31)
\rput(100,15){$\Alispace$}
\rput(10,90){$\Bobspace$}
\psline{->}(15,20)(110,20)
\psline{->}(20,15)(20,100)
\rput(65,70){$\Exx$}
\pscustom[fillcolor=gray,fillstyle=solid]{\psline(20,55)(20,80)
\psbezier(20,80)(25,100)(63.75,100)
\psbezier(102.5,100)(130.92,102.31)(118.12,79.38)
\psline(118.12,79.38)(85,20)
\psline(85,20)(65,20)
\psbezier(65,20)(75,45)(70,55)
\psbezier(65,65)(55.62,65)(50.62,65)
\psbezier(45.62,65)(42,66)(30,60)
\psline(30,60)(20,55)
}
\newrgbcolor{userLineColour}{0.78 0.77 0.92}
\newrgbcolor{userFillColour}{0.91 0.78 0.89}
\pspolygon[linewidth=0,linecolor=userLineColour,fillcolor=userFillColour,fillstyle=solid](71.88,100)(85,100)(85,20)(71.88,20)
\newrgbcolor{userLineColour}{0.92 0.75 0.93}
\newrgbcolor{userFillColour}{0.91 0.79 0.94}
\pspolygon[linewidth=0,linecolor=userLineColour,fillcolor=userFillColour,fillstyle=solid](20,80)(110,80)(110,65)(20,65)
\pspolygon[linewidth=0,linecolor=white,fillcolor=white,fillstyle=solid](110,100)(121.88,100)(121.88,65)(110,65)
\psline[fillstyle=solid](110,98.12)(110,65)
\rput(48.12,89.38){$\Exxx$}
\rput(78.12,72.5){$\ExxxInterior$}
\end{pspicture}
\caption{Cylinder closure and interior}
\label{default}
\end{center}
\end{figure}

The \emph{cylinder closure}\/ and \emph{cylinder interior}\/ operators  are defined by
\bea
\ana{-} \ :\ \WP A &\to & \WP A\label{eq:ana}\\
X &\longmapsto & \ana X = \bigcap_{v\in \VVv} v^\ast v_! (X)\notag\\
\cata{-} \ :\ \WP A &\to & \WP A\label{eq:cata}\\
X &\longmapsto & \cata X = \bigcup_{v\in \VVv} v^\ast v_\ast (X) \notag
\eea
where the sets
\bear
\ana{X}_v\ =\ v^\ast v_! (X) & = & \{z\in A\ |\ \exists x. v(z) = v(x)\wedge x\in X\} \\
\cata{X}_v\ =\ v^\ast v_\ast (X) & = & \{z\in A\ |\ \forall x. v(z) = v(x)\Rightarrow x\in X\} 
\eear
are respectively the external $v$-cylinders around $X$ and the internal $v$-cylinders inside $X$. It is easy to see that the $\ana -$ is a closure operator and that $\cata -$ is an interior operator, which means that they satisfy
\begin{gather*}
\cata{\cata X} = \cata X \qquad\qquad  \ana{\ana X} = \ana X\\
\cata X\  \subseteq\  X\  \subseteq\  \ana X
\end{gather*}
They are complementary in the sense
\[
\neg\ana X = \cata{\neg X}\qquad \qquad \qquad \neg\cata X = \ana{\neg X} 
\]
which means that they determine each other:
\[
\ana X = \neg\cata{\neg X}\qquad \qquad \qquad \cata X = \neg\ana{\neg X} 
\]

\begin{definition}\label{Def:cylindric}
$X\in \WP A$ is \emph{external cylindric}\/ if $X= \ana X$ and \emph{internal cylindric}\/ if $X = \cata X$.
\end{definition}

\begin{lemma}\label{Lemma:schedule}
Any history $\vec t$ of events $\Event = \coprod_{w\in\WWw} \Event_w$ is completely determined by 
\begin{itemize}
\item the restrictions $\vec t\restr_w$ for $w\in \WWw$, and 
\item their {\sschedule} , which is  function $\schedule : \length{\vec t} \to  \WWw$ such that 
\bea\label{eq:sched}
\schedule(k) = w & \iff & t_k \in\Event_w
\eea
where $\vec t = \seq{t_0\  t_1  \ldots t_n }$ and $\length{\vec t} = n+1$.
\end{itemize}
\end{lemma}

\bpr The claim is that the function  $\schedule : \length{\vec t} \to  \WWw$ satisfying \eqref{eq:sched} provides enough information to reconstruct  $\vec t$ from its restrictions $\vec t\restr_w$, given for all $w\in \WWw$. The restrictions can be {\sschedule}d in many ways, but  \eqref{eq:sched} says that $k$-th component $t_k$ of $\vec t$ comes from the restriction $\vec t \restr_{\schedule(k)}$. To simplify notation, we write $\vec t \restr_{\schedule(k)}$ as $\vec \tau\, ^{(k)}$, so that $\sigma$-th component of $\vec t \restr_{\schedule(k)}$ becomes $\tau^{(k)}_\sigma$. So we know that $t_k = \tau^{(k)}_\sigma$ for some $\sigma$. The question is: \emph{what is $\sigma$?}

Towards the answer, we use $\schedule : \length{\vec t} \to  \WWw$ to define the function $\varsigma : \length{\vec t}\times \Subj  \to \length{\vec t}+1$ which counts the number of $w$-actions in the $k$-length prefix of $\vec t$. The definition is by recursion:
\bear
\varsigma(0,w) & = & \begin{cases}
1 & \hspace{3em}\mbox{if } \schedule(0) = w\\
0 & \hspace{3em}\mbox{otherwise}
\end{cases}\\
\varsigma(n+1,w) & = & \begin{cases}
\varsigma(n)+1 & \mbox{ if } \schedule(n+1) = w\\
\varsigma(n) & \mbox{ otherwise}
\end{cases}
\eear
It is easy to show by induction that  $\varsigma(k,w)$ is the length of the $w$-restriction of the $k$-prefix ${\vec t }_k \sqsubseteq \vec t$, i.e., $\varsigma(k,w) = \length{\vec t_k\restr_w}$. Hence, 
$t_k \  = \ \tau^{(k)}_{\varsigma(k,\schedule(k))}$ is thus the $\varsigma(k,\schedule(k))$-th component of $\vec \tau\, ^{(k)} = \vec t\restr_{\schedule(k)}$.
\epr

\begin{proposition}
A property $\Property$ is localized if and only if together with every $\vec t$ it contains all $\vec s$ such that $\vec t \restr_w = \vec s\restr_w$ for all $w\in\WWw$.
\end{proposition}

\para{Localization.} In general, the smallest local property containing $\Property\subseteq \Event^\ast$ with respect to the partition  $\Event = \coprod_{i\in \IIi} \Event_i$ is
\bear
\widehat\Property & = & \left\{\vec t\in \Event^\ast\ |\ \forall i\in \IIi.\ \vec t\restr_i\, \in \Property_i
\right\} \\
& = & \bigcap_{i\in \IIi} \widehat\Property_i\quad\mbox{ where}\quad \widehat\Property_i = \{\vec t\in \Event^\ast\ |\ \vec t\restr_i \in \Property_i\}
\eear
The property $\widehat \Property$ is called the \emph{localization} of $\Property$. It is easy to see that $\Property \subseteq \widehat \Property = \widehat{\widehat P}$, i.e., that localization is a closure operator. The property $\Property$ is thus \emph{localized}\/ when $\Property = \widehat \Property$.

\label{Appendix:Prereq}

\addcontentsline{toc}{chapter}{Index}
\printindex

\end{document}